\documentclass[11pt, reqno]{amsart}
\usepackage{cosmological-rigidity}

\makeatletter
\renewcommand*\env@matrix[1][\arraystretch]{%
  \edef\arraystretch{#1}%
  \hskip -\arraycolsep
  \let\@ifnextchar\new@ifnextchar
  \array{*\c@MaxMatrixCols c}}
\makeatother

\DeclarePairedDelimiter{\abs}{\lvert}{\rvert}
\DeclarePairedDelimiter{\norm}{\lVert}{\rVert}

\DeclarePairedDelimiter{\curlyBrace}{\{}{\}}

\DeclarePairedDelimiter{\evalAt}{.}{\vert}

\newtheorem{theorem}{Theorem}[section]
\newtheorem{prop}[theorem]{Proposition}
\crefname{prop}{proposition}{propositions}
\AddToHook{env/prop/begin}{\crefalias{theorem}{prop}}
\newtheorem{lemma}[theorem]{Lemma}
\crefname{lemma}{lemma}{lemmas}
\AddToHook{env/lemma/begin}{\crefalias{theorem}{lemma}}
\newtheorem{remark}[theorem]{Remark}
\AddToHook{env/remark/begin}{\crefalias{theorem}{remark}}
\newtheorem{definition}[theorem]{Definition}
\crefname{definition}{definition}{definitions}
\AddToHook{env/definition/begin}{\crefalias{theorem}{definition}}
\newtheorem{corollary}[theorem]{Corollary}
\crefname{corollary}{corollary}{corollarys}
\AddToHook{env/corollary/begin}{\crefalias{theorem}{corollary}}
\newtheorem{conjecture}[theorem]{Conjecture}
\crefname{conjecture}{conjecture}{conjectures}
\AddToHook{env/conjecture/begin}{\crefalias{theorem}{conjecture}}

\newcommand{\LHS}{left-hand side}
\newcommand{\RHS}{right-hand side}
\newcommand{\Real}{\mathbb{R}}
\newcommand{\Complex}{\mathbb{C}}

\newcommand{\ImagUnit}{\mathbbm{i}}
\newcommand{\trace}{\operatorname{tr}}
\newcommand{\closure}{\operatorname{cl}}
\newcommand{\Divergence}{\operatorname{div}}
\newcommand{\supp}{\operatorname{supp}}
\newcommand{\curl}{\operatorname{curl}}
\newcommand{\Err}{\operatorname{Err}}
\newcommand{\LOT}{\operatorname{l.o.t}}
\newcommand{\Span}{\operatorname{span}}

\newcommand{\AdmissibleRHS}{\mathcal{N}_{\Metric}}
\newcommand{\QProj}{\Pi}
\newcommand{\LieZRiem}{\mathbf{S}}

\newcommand{\Manifold}{\mathcal{M}}
\newcommand{\LieDerivative}{\mathcal{L}}

\newcommand{\Ric}{\mathbf{Ric}}
\newcommand{\Riem}{\mathbf{R}}
\newcommand{\Weyl}{\mathbf{W}}

\newcommand{\ChristoffelTypeTwo}[3][]{
  \ifthenelse{\equal{#1}{}}
  {\tensor{\Gamma}{^{#2}_{#3}}} 
  {\tensor{\Gamma(#1)}{^{#2}_{#3}}}
}
\newcommand{\ChristoffelTypeOne}[3][]{
  \ifthenelse{\equal{#1}{}}
  {\tensor{\Gamma}{_{{#2#3}}}}
  {\tensor{\Gamma(#1)}{_{{#2#3}}}}
}

\newcommand{\Trace}{\operatorname{tr}}

\newcommand{\Sphere}{\mathbb{S}}
\newcommand{\DeformationTensor}[3]{
  \ifthenelse{\equal{#1}{0}}
  {\tensor[^{({#2})}]{\pi}{#3}}
  {\ifthenelse{\equal{#1}{1}}{\tensor[^{({#2})}]{\Gamma}{#3}}{ERRORCODE1}}
}

\newcommand{\CovariantDeriv}{\mathbf{D}}

\newcommand{\ABar}{\underline{A}}
\newcommand{\BBar}{\underline{B}}
\newcommand{\ZBar}{\underline{Z}}
\newcommand{\XBar}{\underline{X}}
\newcommand{\XHat}{\widehat{X}}
\newcommand{\XHatBar}{\underline{\XHat}}
\newcommand{\XiBar}{\underline{\Xi}}
\newcommand{\HBar}{\underline{H}}
\newcommand{\omegaBar}{\underline{\omega}}

\newcommand{\LeftDual}[1]{\tensor[^*]{{#1}}{}}
\newcommand{\Metric}{\mathbf{g}}
\newcommand{\chiBar}{\underline{\chi}}
\newcommand{\etaBar}{\underline{\eta}}
\newcommand{\xiBar}{\underline{\xi}}

\newcommand{\chiTF}{\widehat{\chi}}

\newcommand{\aTrace}[1]{\tensor[^{(a)}]{\Trace{#1}}{}}
\newcommand{\SymTracelessTensorProd}{\widehat{\otimes}}
\newcommand{\ComplexDeriv}{\mathcal{D}}

\newcommand{\horProj}[1]{\tensor[^{(h)}]{#1}{}}

\newcommand{\TransformScalarWaveOp}[1][]{
  \ifthenelse{\equal{#1}{}}
  {\widehat{\Box}^{(0)}}
  {\widehat{\Box}^{(0)}_{#1}}
}
\newcommand{\ScalarWaveOp}[1][]{
  \ifthenelse{\equal{#1}{}}
  {\Box^{(0)}}
  {\Box^{(0)}_{#1}}
}
\newcommand{\VectorWaveOp}[1][]{
  \ifthenelse{\equal{#1}{}}
  {\Box^{(1)}}
  {\Box^{(1)}_{#1}}
}
\newcommand{\TensorWaveOp}[1][]{
  \ifthenelse{\equal{#1}{}}
  {\Box^{(2)}}
  {\Box^{(2)}_{#1}}
}
\newcommand{\ScalarWaveConjOp}[1][]{
  \ifthenelse{\equal{#1}{}}
  {\overline{\Box}^{(0)}}
  {\overline{\Box}^{(0)}_{#1}}
}
\newcommand{\ScalarWaveLaplaceOp}[1][]{
  \ifthenelse{\equal{#1}{}}
  {\widehat{\Box}^{(0)}}
  {\widehat{\Box}^{(0)}_{#1}}
}
\newcommand{\ReducedWaveOp}[1][]{
  \ifthenelse{\equal{#1}{}}
  {\widetilde{\Box}^{(0)}}
  {\widetilde{\Box}^{(0)}_{#1}}
}

\newcommand{\EMT}{\mathbb{T}}
\newcommand{\KillT}{\mathbf{T}}
\newcommand{\KillPhi}{\mathbf{\Phi}}

\newcommand{\GeodesicVF}{\mathbf{G}}

\newcommand{\KdS}{Kerr-de Sitter}

\newcommand{\EVE}{Einstein vaccuum equations}
\newcommand{\Horizon}{\mathcal{H}}
\newcommand{\EventHorizon}{\mathcal{H}}
\newcommand{\CosmologicalHorizon}{\overline{\mathcal{H}}}
\newcommand{\EventHorizonFuture}{\mathcal{H}^+}
\newcommand{\CosmologicalHorizonFuture}{\overline{\mathcal{H}}^+}
\newcommand{\EventHorizonPast}{\mathcal{H}^-}
\newcommand{\CosmologicalHorizonPast}{\overline{\mathcal{H}}^-}

\newcommand{\Id}{\operatorname{Id}}

\usepackage[foot]{amsaddr}
\usepackage[sortcites,
backend=biber,
date=year,
style=numeric-comp,
doi=false,
isbn=false,
url=false,
eprint=true]{biblatex}

\makeatletter
\def\paragraph{\@startsection{paragraph}{4}%
  \z@\z@{-\fontdimen2\font}%
  {\normalfont\bfseries}}
\makeatother

\makeatletter
\def\@tocline#1#2#3#4#5#6#7{\relax
  \ifnum #1>\c@tocdepth 
  \else
    \par \addpenalty\@secpenalty\addvspace{#2}%
    \begingroup \hyphenpenalty\@M
    \@ifempty{#4}{%
      \@tempdima\csname r@tocindent\number#1\endcsname\relax
    }{%
      \@tempdima#4\relax
    }%
    \parindent\z@ \leftskip#3\relax \advance\leftskip\@tempdima\relax
    \rightskip\@pnumwidth plus4em \parfillskip-\@pnumwidth
    #5\leavevmode\hskip-\@tempdima
      \ifcase #1
       \or\or \hskip 1em \or \hskip 2em \else \hskip 3em \fi%
      #6\nobreak\relax
    \dotfill\hbox to\@pnumwidth{\@tocpagenum{#7}}\par
    \nobreak
    \endgroup
  \fi}
\makeatother

\DeclareNameAlias{author}{family-given}
\AtEveryBibitem{
  \clearfield{month}
  \clearfield{url}
  \clearfield{urlyear}
  \clearfield{urlmonth} 
  \ifentrytype{online}{
    \clearfield{year}}
  {
    \clearfield{eprint}}
}
\renewbibmacro*{volume+number+eid}{%
  \printfield{volume}%
  \setunit*{\addnbspace}
  \printfield{number}%
  \setunit{\addcomma\space}%
  \printfield{eid}}
\DeclareFieldFormat[article]{volume}{\textbf{#1}}
\DeclareFieldFormat[article]{number}{\mkbibparens{#1}}
\DeclareFieldFormat*{title}{\textit{#1}}
\DeclareFieldFormat{journaltitle}{#1\isdot}
\renewbibmacro{in:}{}
 \numberwithin{equation}{section}

 \title{On the uniqueness of Kerr-de Sitter spacetimes}
 \author{Allen Juntao Fang}
\address{Universit\"at M\"unster, M\"unster,
  Deutschland
  (\href{mailto:allen.juntao.fang@uni-muenster.de}{allen.juntao.fang@uni-muenster.de})}

\date{}
\begin{document}
\maketitle

\begin{abstract}
  In this paper, we prove a series of results concerning the
  uniqueness of \KdS{} as a family of smooth stationary black hole
  solutions to the nonlinear Einstein vacuum equations with positive
  cosmological constant $\Lambda$. The results only assume smoothness
  rather than analyticity of the solution in question. The results use
  a two-sided approach to rigidity, requiring assumptions on both the
  event horizon and the cosmological horizon (or a neighborhood
  thereof) to formulate an appropriate unique continuation argument to
  prove the rigidity of \KdS.
\end{abstract}
\maketitle

\tableofcontents

\section{Introduction}
\label{sec:introduction}

The Einstein vacuum equations with cosmological constant $\Lambda$ are given by
\begin{equation*}
  \tag{$\Lambda$-EVE}
  \Ric(\Metric) = \Lambda \Metric,
\end{equation*}
where $(\mathcal{M},\Metric)$ is a $3+1$ Lorentzian metric with
signature $(-,+,+,+)$. 
Despite the difficulty in analyzing the Einstein vacuum equations,
many explicit solutions are known. These often however feature some
amount of symmetry. A natural question then is whether the explicit
solutions we have on hand are the unique solutions within their
symmetry class, or whether we have only discovered one of
many possible symmetric solutions.  The first black hole solution to
($0$-EVE) to be discovered was the Schwarzschild solution. Birkhoff's
theorem\footnote{The result is traditionally attributed to Birkhoff in
  \cite{birkhoffRelativityModernPhysics1923}, but seems to be independently
  discovered by \cite{jebsenGeneralSphericallySymmetric1921}. We refer the interested
  reader to \cite{johansenDiscoveryBirkhoffsTheorem2006} for a more
  thorough discussion.}  proves that any spherically symmetric
solution to ($0$-EVE) is locally isometric to a region in the
Schwarzschild spacetime, with analogous results for $\Lambda\neq
0$. In other words, Schwarzschild is the unique spherically symmetric
solution to ($0$-EVE).

There are numerous similar rigidity results showing that a particular
solution of ($\Lambda$-EVE) is the unique solution within some
symmetry class. For instance, Israel's theorem
\cite{israelEventHorizonsStatic1967} proves that Schwarzschild is the
unique asymptotically flat static solution to ($0$-EVE). An analogous
result in $\Lambda<0$ with conditions at conformal timelike infinity
taking the place of the asymptotic flatness condition in the
$\Lambda=0$ setting was first shown in
\cite{boucherUniquenessTheoremAntide1984} and later revisited in
\cite{wangUniquenessADSSpacetime2005} and
\cite{chruscielMassAsymptoticallyHyperbolic2003}. While a complete
equivalent of Israel's theorem in the $\Lambda>0$ setting does not
exist, a similar type of uniqueness result for Schwarzschild-de Sitter
was proven in \cite{borghiniUniquenessSchwarzschildSitter2023} for
$\Lambda>0$ under an assumption relating to the virtual mass of the
manifold. A more complete account of the uniqueness of static
solutions to ($\Lambda$-EVE) can be found in
\cite{borghiniUniquenessSchwarzschildSitter2023}.

Under the conditions of only stationarity and axisymmetry, Carter
\cite{carterAxisymmetricBlackHole1971} and Robinson
\cite{robinsonUniquenessKerrBlack1975} proved that Kerr is the unique
stationary and axisymmetric asymptotically flat black hole solution to
($0$-EVE).
\begin{theorem}[Carter-Robinson uniqueness theorem]
  If $(\mathcal{M}, \Metric)$ is a stationary, axisymmetric,
  asymptotically flat solution to ($0$-EVE) with a connected, regular
  event horizon, then the domain of outer communication of
  $\mathcal{M}$ is isometrically diffeomorphic to the domain of outer
  communications of a Kerr black hole.
\end{theorem}
A more comprehensive review of results pertaining to stationary and
axisymmetric uniqueness of ($0$-EVE) can be found in the excellent
account in \cite{chruscielStationaryBlackHoles2012}. It should be
noted that the Carter-Robinson uniqueness theorem is a powerful result
that does not have a known analogue in the $\Lambda\neq 0$ setting.

Finally, one can consider the case of only assuming stationarity. This
is the natural setting for the Kerr rigidity conjecture and the \KdS{}
rigidity conjecture (when $\Lambda=0$ and $\Lambda>0$ respectively).

\begin{conjecture}[Kerr rigidity]
  \label{conjecture:Kerr}
  The domain of exterior communication of a regular, stationary,
  four-dimensional vacuum black hole solution to the Einstein vacuum
  equations with a vanishing cosmological constant $\Lambda$ is
  isometrically diffeomorphic to the domain of exterior communication
  of a Kerr black hole.
\end{conjecture}

\begin{conjecture}[\KdS{} rigidity]
  \label{conjecture:KdS}
  The domain of exterior communication of a regular, stationary
  (appropriately defined), four-dimensional vacuum black hole solution
  to the Einstein vacuum equations with a positive cosmological
  constant $\Lambda$ is isometrically diffeomorphic to the domain of
  exterior communication of a \KdS{} black hole.
\end{conjecture}
The present paper is primarily concerned with \KdS{}
rigidity. Nonetheless, a brief review is first presented below of the
much more well-studied case of Kerr rigidity in the $\Lambda=0$
setting.

\subsection{Kerr-rigidity: a review}

This section reviews progress made towards resolving
\Cref{conjecture:Kerr}, which has been resolved in two particular
cases. First, \Cref{conjecture:Kerr} has been resolved in the case
that the solution $(\mathcal{M}, \Metric)$ in question is a small
perturbation of a member of the Kerr family in
\cite*{alexakisUniquenessSmoothStationary2010}. Second,
\Cref{conjecture:Kerr} has been resolved in the case that the event
horizon of the black hole satisfies certain technical conditions that
make it appropriately compatible with the asymptotically flat end of
the black hole in \cite{ionescuUniquenessSmoothStationary2009}. These two results are reviewed below. There has also
been progress towards resolving \Cref{conjecture:Kerr} under other
assumptions. The interested reader is referred to the account in
\cite{ionescuRigidityResultsGeneral2015} for a more complete review of
the literature on the uniqueness of Kerr.

\subsubsection{Rigidity via axisymmetry}

In view of the Carter-Robinson uniqueness theorem, one natural
approach to prove Kerr rigidity is to construct a global vectorfield
generating axisymmetry. Combined with the stationarity assumption,
this would immediately yield that $(\mathcal{M}, \Metric)$ is actually
isometric to a member of Kerr. Hawking proved Kerr rigidity under this
assumption under the additional critical assumption of real
analyticity. Hawking's construction of the vectorfield generating
axisymmetry had two key steps. First, a second Killing vectorfield
$\mathbf{K}$ was constructed along the event horizon, using the
existence of the stationary vectorfield $\KillT$ and the non-degeneracy
of the horizon. Second, $\mathbf{K}$ was extended into the domain of
exterior communication by a Cauchy-Kowalevski type argument, taking
advantage of the assumed analyticity of the spacetime. It should be
noted that while the extension argument used by Hawking relies
critically on analyticity, the Carter-Robinson uniqueness theorem does
not.

The analyticity assumption is extremely strong, and indeed unreasonable
from a mathematical point of view. Using classical elliptic theory, it
is possible to show that the regions of stationary spacetimes in which
$\KillT$ is timelike are real analytic. However, in regions where the
stationary vectorfield $\KillT$ fails to be timelike, there is no
reason to expect real analyticity. In the context of black hole
rigidity, this is a critical obstacle since Kerr and \KdS{} black
holes both feature ergoregions where $\KillT$ becomes spacelike. 
However, the removal of the real analyticity condition is
highly non-trivial, as in general, extending the second Killing
vectorfield $\mathbf{K}$ off the horizon even in the smooth category involves
solving an \emph{ill-posed} problem. 

Significant progress was made towards reducing the analyticity assumption
to merely a $C^{\infty}$ condition under the additional assumption
that $(\mathcal{M}, \Metric)$ is close (in an appropriate sense) to
some Kerr spacetime by Alexakis-Ionescu-Klainerman in
\cite{alexakisHawkingsLocalRigidity2010,alexakisUniquenessSmoothStationary2010}, where they proved the following theorem.
\begin{theorem}[Main result in \cite{alexakisUniquenessSmoothStationary2010}]
  \label{thm:AIK}
  A black hole spacetime that solves ($0$-EVE) and is a suitably regular
  stationary perturbation of Kerr possesses a domain of outer
  communications isometric to that belonging to a member of the Kerr
  family.
\end{theorem}

The main thrust of the proof in
\cite{alexakisUniquenessSmoothStationary2010} is to construct a
vectorfield generating axisymmetry in the domain of exterior communication.
This is achieved by first showing that there exists a second Killing
vectorfield on the domain of exterior communication, $\mathbf{K}$ (the
Hawking vectorfield) and then showing that $\mathbf{K}$ and the
stationary vectorfield $\KillT$ generate a vectorfield generating axisymmetry
in the domain of exterior communication.

To show the existence of $\mathbf{K}$, the authors in
\cite{alexakisUniquenessSmoothStationary2010} make use of a series of
unique continuation arguments, based on the existence of a coupled
wave-transport system of the form
\begin{equation}
  \label{eq:AIK:wave-transport:schematic}
  \begin{split}
    \Box_{\Metric}\phi^{\mathbf{K}}_i &= \AdmissibleRHS(\phi^{\mathbf{K}}_i, \psi^{\mathbf{K}}_j, D\psi^{\mathbf{K}}_j, D\phi^{\mathbf{K}}_i),\\
    \GeodesicVF \psi^{\mathbf{K}}_j &= \AdmissibleRHS(\phi^{\mathbf{K}}_i, \psi^{\mathbf{K}}_j),
  \end{split}  
\end{equation}
where $\phi^{\mathbf{K}}_i, \psi^{\mathbf{K}}_j$ are quantities,
including $\DeformationTensor{0}{\mathbf{K}}{}$, generated from
$\mathbf{K}$ that vanish when $\mathbf{K}$ is Killing, and where
$\AdmissibleRHS$ is a smooth and at-least-linear combination of its
arguments. A key structure of the $\AdmissibleRHS$ terms in
\Cref{eq:AIK:wave-transport:schematic} is that they are
\emph{order-reducing} with respect to the number of derivatives on the
left-hand side.

It is well-known that unique continuation arguments for wave-transport
systems of the form \Cref{eq:AIK:wave-transport:schematic} rely on
some form of pseudoconvexity. On black hole backgrounds, the classical
pseudoconvexity of Hörmander
\cite{hormanderLinearPartialDifferential1964} breaks down due to the
presence of trapped null geodesics.  The main breakthrough used in
\cite{ionescuUniquenessSmoothStationary2009} and also used
subsequently in
\cite{alexakisRigidityStationaryBlack2014,alexakisUniquenessSmoothStationary2010,wongUniquenessKerrNewmanBlack2009,wongNonExistenceMultipleBlackHoleSolutions2014},
was that in the setting of black hole rigidity, since one is already
assuming the existence of a stationary vectorfield $\KillT$, it
suffices to use a weaker notion of pseudoconvexity adapted to the
existence of $\KillT$, namely, $\KillT$-pseudoconvexity. The obstacles
to $\KillT$-pseudoconvexity, rather than being all trapped null
geodesics in general, are only $\KillT$-orthogonal trapped null
geodesics (also referred to as $\KillT$-trapped null
geodesics). Critically, Kerr and small perturbations of Kerr do not
possess any $\KillT$-orthogonal trapped null geodesics. In fact, on
Kerr, one can construct a globally $\KillT$-pseudoconvex foliation
simply by taking $r$-constant hypersurfaces, where $r$ is the
Boyer-Lindquist radial coordinate.

Below is a brief review of the main steps in proving the perturbative
$C^{\infty}$ Kerr rigidity theorem in \Cref{thm:AIK}.

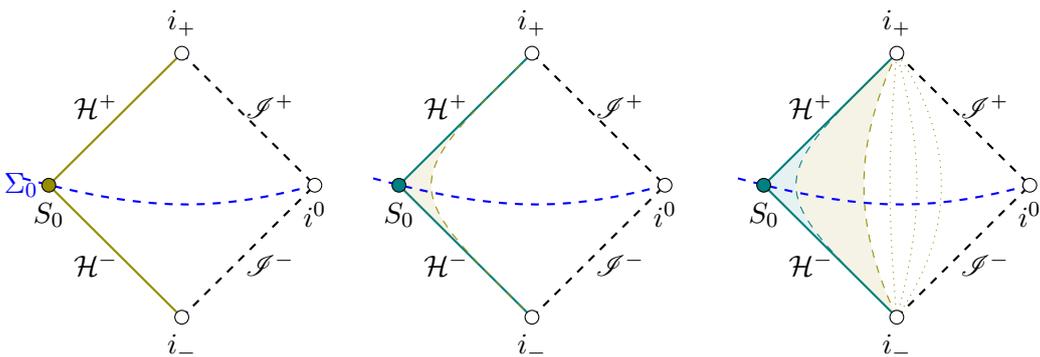
\begin{figure}[ht]
  \centering
  \begin{minipage}[t]{.3\textwidth}
    \centering
    \begin{tikzpicture}[scale=0.35,every node/.style={scale=1.0}]


  \def \s{5} 
  \def \exts{0.1} 
  \def \t{0.3}
  \def \Tlen{1}
  \def \ulen{0.7}
  \def \offset{0}
  \def \initExt{-0.05}

  \coordinate (tInf) at (0,\s); 
  \coordinate (EventZero) at (-\s,0); 
  \coordinate (CosmoZero) at (\s,0); 
  \coordinate (tNegInf) at (0,-\s);

  \coordinate (tInfL) at (-\offset,\s); 
  \coordinate (EventZeroL) at (-\s-\offset,0); 
  \coordinate (CosmoZeroL) at (\s-\offset,0); 
  \coordinate (tNegInfL) at (-\offset,-\s);
  \coordinate (initExtControlL)  at (-\s-\offset+\initExt,0);

  \coordinate (tInfR) at (\offset,\s); 
  \coordinate (EventZeroR) at (-\s+\offset,0); 
  \coordinate (CosmoZeroR) at (\s+\offset,0); 
  \coordinate (tNegInfR) at (\offset,-\s);
    \coordinate (initExtControlR)  at (\s+\offset-\initExt,0);

  \draw[olive,thick,name path=EventFuture] (tInfL) --
  node[pos=0.4,left]{\textcolor{black}{$\EventHorizonFuture$}} (EventZeroL) ;  

  \draw[olive,thick,name path=EventPast] (tNegInfL) --
  node[pos=0.4,left]{\textcolor{black}{$\EventHorizonPast$}} (EventZeroL) ;

  \draw[black,dashed,thick,name path=CosmoFuture] (tInfR) --
  node[pos=0.4,right]{$\mathscr{I}^+$} (CosmoZeroR) ;
  
  \draw[black,dashed,thick,name path=CosmoPast] (tNegInfR) --
  node[pos=0.4,right]{$\mathscr{I}^-$} (CosmoZeroR);



  

  

  \path[blue,thick,dashed,out=-15,in=-165,shorten <= -10]
  (EventZeroL) edge node[pos=0,left]{$\Sigma_0$} node[pos=0.6](TInit){} (CosmoZeroR) ;

  \node[scale=0.5,fill=white,draw,circle,label=above:$i_+$]at(tInf){};
  \node[scale=0.5,fill=white,draw,circle,label=below:$i_-$]at(tNegInf){};
  \node[scale=0.5,fill=olive,draw,circle,label=below:$S_0$]at(EventZeroL){};
  \node[scale=0.5,fill=white,draw,circle,label=below:$i^0$]at(CosmoZeroR){};

\end{tikzpicture}

  \end{minipage}%
  \begin{minipage}[t]{.3\textwidth}
    \centering
    \begin{tikzpicture}[scale=0.35,every node/.style={scale=1.0}]


  \def \s{5} 
  \def \exts{0.1} 
  \def \t{0.3}
  \def \Tlen{1}
  \def \ulen{0.7}
  \def \offset{0}
  \def \initExt{-0.05}

  \coordinate (tInf) at (0,\s); 
  \coordinate (EventZero) at (-\s,0); 
  \coordinate (CosmoZero) at (\s,0); 
  \coordinate (tNegInf) at (0,-\s);

  \coordinate (tInfL) at (-\offset,\s); 
  \coordinate (EventZeroL) at (-\s-\offset,0); 
  \coordinate (CosmoZeroL) at (\s-\offset,0); 
  \coordinate (tNegInfL) at (-\offset,-\s);
  \coordinate (initExtControlL)  at (-\s-\offset+\initExt,0);

  \coordinate (tInfR) at (\offset,\s); 
  \coordinate (EventZeroR) at (-\s+\offset,0); 
  \coordinate (CosmoZeroR) at (\s+\offset,0); 
  \coordinate (tNegInfR) at (\offset,-\s);
    \coordinate (initExtControlR)  at (\s+\offset-\initExt,0);

  \draw[teal,thick,name path=EventFuture] (tInfL) --
  node[pos=0.4,left]{\textcolor{black}{$\EventHorizonFuture$}} (EventZeroL) ;  

  \draw[teal,thick,name path=EventPast] (tNegInfL) --
  node[pos=0.4,left]{\textcolor{black}{$\EventHorizonPast$}} (EventZeroL) ;

  \draw[black,dashed,thick,name path=CosmoFuture] (tInfR) --
  node[pos=0.4,right]{$\mathscr{I}^+$} (CosmoZeroR) ;
  
  \draw[black,dashed,thick,name path=CosmoPast] (tNegInfR) --
  node[pos=0.4,right]{$\mathscr{I}^-$} (CosmoZeroR);



  
 \path[fill=olive,fill opacity=0.1,draw=none] (tInfL) ..controls(initExtControlL)..(tNegInfL)-- (EventZeroL) -- cycle;
    \draw[olive,dashed] (tInfL) .. controls(initExtControlL).. (tNegInfL);

  

  \path[blue,thick,dashed,out=-15,in=-165,shorten <= -10]
  (EventZeroL) edge  (CosmoZeroR) ;

  \node[scale=0.5,fill=white,draw,circle,label=above:$i_+$]at(tInf){};
  \node[scale=0.5,fill=white,draw,circle,label=below:$i_-$]at(tNegInf){};
  \node[scale=0.5,fill=teal,draw,circle,label=below:$S_0$]at(EventZeroL){};
  \node[scale=0.5,fill=white,draw,circle,label=below:$i^0$]at(CosmoZeroR){};

\end{tikzpicture}

  \end{minipage}%
  \begin{minipage}[t]{.3\textwidth}
    \centering
    \begin{tikzpicture}[scale=0.35,every node/.style={scale=1.0}]


  \def \s{5} 
  \def \exts{0.1} 
  \def \t{0.3}
  \def \Tlen{1}
  \def \ulen{0.7}
  \def \offset{0}
  \def \initExt{-0.05}

  \coordinate (tInf) at (0,\s); 
  \coordinate (EventZero) at (-\s,0); 
  \coordinate (CosmoZero) at (\s,0); 
  \coordinate (tNegInf) at (0,-\s);

  \coordinate (tInfL) at (-\offset,\s); 
  \coordinate (EventZeroL) at (-\s-\offset,0); 
  \coordinate (CosmoZeroL) at (\s-\offset,0); 
  \coordinate (tNegInfL) at (-\offset,-\s);
  \coordinate (initExtControlL)  at (-\s-\offset+\initExt,0);

  \coordinate (tInfR) at (\offset,\s); 
  \coordinate (EventZeroR) at (-\s+\offset,0); 
  \coordinate (CosmoZeroR) at (\s+\offset,0); 
  \coordinate (tNegInfR) at (\offset,-\s);
    \coordinate (initExtControlR)  at (\s+\offset-\initExt,0);

  \draw[teal,thick,name path=EventFuture] (tInfL) --
  node[pos=0.4,left]{\textcolor{black}{$\EventHorizonFuture$}} (EventZeroL) ;  

  \draw[teal,thick,name path=EventPast] (tNegInfL) --
  node[pos=0.4,left]{\textcolor{black}{$\EventHorizonPast$}} (EventZeroL) ;

  \draw[black,dashed,thick,name path=CosmoFuture] (tInfR) --
  node[pos=0.4,right]{$\mathscr{I}^+$} (CosmoZeroR) ;
  
  \draw[black,dashed,thick,name path=CosmoPast] (tNegInfR) --
  node[pos=0.4,right]{$\mathscr{I}^-$} (CosmoZeroR);



  
 \path[fill=teal,fill opacity=0.1,draw=none] (tInfL) ..controls(initExtControlL)..(tNegInfL)-- (EventZeroL) -- cycle;
 \draw[teal,dashed] (tInfL) .. controls(initExtControlL).. (tNegInfL);

 \path[fill=olive,fill opacity=0.1,draw=none] (tInfL)..controls(initExtControlL)..(tNegInfL) to [out=115,in=-115](tInf);
 \draw[olive, dashed] (tNegInfL) to [out=115,in=-115](tInf);

 \draw[olive, dotted] (tNegInfL) to [out=95,in=-95](tInf);
 \draw[olive, dotted] (tNegInfL) to [out=75,in=-75](tInf);
 \draw[olive, dotted] (tNegInfL) to [out=55,in=-55](tInf);

  

  \path[blue,thick,dashed,out=-15,in=-165,shorten <= -10]
  (EventZeroL) edge  (CosmoZeroR) ;

  \node[scale=0.5,fill=white,draw,circle,label=above:$i_+$]at(tInf){};
  \node[scale=0.5,fill=white,draw,circle,label=below:$i_-$]at(tNegInf){};
  \node[scale=0.5,fill=teal,draw,circle,label=below:$S_0$]at(EventZeroL){};
  \node[scale=0.5,fill=white,draw,circle,label=below:$i^0$]at(CosmoZeroR){};

\end{tikzpicture}

  \end{minipage}%
  \caption{The basic steps in
    \cite{alexakisUniquenessSmoothStationary2010} to prove perturbative
    $C^{\infty}$ Kerr rigidity.  The olive green regions show the main area
    of interest in each step while teal regions depict regions
    previously shown to be axisymmetric.}
\label{fig:AIK}
\end{figure}
We detail the main steps in the proof of perturbative $C^{\infty}$
Kerr rigidity in \cite{alexakisUniquenessSmoothStationary2010} below
(see \Cref{fig:AIK}).
\begin{enumerate}
\item The first step is to construct an additional Killing vectorfield
  along the horizons. The authors construct the
  Hawking vectorfield, a Killing vectorfield that is tangent to the
  event horizon, by solving a characteristic initial
  value problem in the causal future and causal past of $S_0$.
  A complete proof of this result in the Kerr case can be found in
  \cite{alexakisHawkingsLocalRigidity2010}. 
\item The second step is to extend the Hawking vectorfield from the
  event horizon to a neighborhood of the event horizon in the domain
  of exterior communication. This step makes use of a unique
  continuation argument at the bifurcate sphere $S_0$ where the desired
  pseudoconvexity is a consequence of the null bifurcate geometry of
  the event horizon.
\item The third step, which is the most delicate, is the extension of
  the Hawking vectorfield into the entire domain of exterior
  communication. This step is also conducted by a unique continuation
  argument. In this case, classical pseudoconvexity fails completely
  due to the possible presence of trapped null geodesics. The authors instead
  construct a $\KillT$-pseudoconvex foliation. This relies on the
  critical observation that Kerr does not possess any $\KillT$-trapped
  null geodesics, and moreover, that this property is robust under
  small perturbations. 
\item The final step is to recover a global vectorfield generating axisymmetry
  from the Hawking vectorfield. \emph{A priori}, what is constructed
  is simply an additional Killing symmetry, not an axisymmetric
  vectorfield. To show that the Hawking vectorfield and the stationary
  vectorfield span a rotational vectorfield, the authors again start
  at the bifurcate sphere $S_0$ and show that there is an axisymmetric
  vectorfield on $S_0$ which extends to the entire domain of exterior
  communication. 
\end{enumerate}

\subsubsection{Rigidity via the Mars-Simon tensor}

Outside of the perturbative regime, an approach that has seen some
success in proving rigidity results for Kerr has been to use an
alternative method of characterizing the Kerr family, namely, the vanishing of
the Mars-Simon tensor $\mathcal{S}$
\cite{marsSpacetimeCharacterizationKerr1999,marsUniquenessPropertiesKerr2000}. The
Mars-Simon tensor played a crucial role in Ionescu and Klainerman's
result in \cite{ionescuUniquenessSmoothStationary2009}.
\begin{theorem}[Main result of \cite{ionescuUniquenessSmoothStationary2009}]
  \label{thm:IK}
  A stationary, asymptotically flat, regular solution
  $(\mathcal{M},\Metric)$ to Einstein's vacuum equations with $\Lambda=0$ with
  null-bifurcate event horizon possessing a smooth bifurcate sphere
  $S_0$ that is compatible with the asymptotically flat end has a
  domain of outer communication isometric to that belonging to a
  member of the Kerr family.
\end{theorem}

Instead of assuming that $(\mathcal{M}, \Metric)$ was globally close
to a member of the Kerr family, the authors in
\cite{ionescuUniquenessSmoothStationary2009} posed certain rigidity
conditions on the bifurcate sphere of the event horizon $S_0$. These
rigidity conditions essentially imply a certain compatibility between
the null bifurcate horizon and the asymptotically flat end of
$\mathcal{M}$ (see \Cref{remark:thm:e1c1} for more discussion).  To
prove the uniqueness of Kerr using its characterization by the
Mars-Simon tensor $\mathcal{S}$, Ionescu and Klainerman again used a
unique continuation argument
\cite{ionescuUniquenessSmoothStationary2009}. This takes advantage of
the fact that the Mars-Simon tensor satisfies
\begin{equation}
  \label{eq:intro:S-T-coupled-system}
  \begin{split}
    \Box_{\Metric}\mathcal{S} &= \mathcal{N}_{\Metric}(\mathcal{S}, \CovariantDeriv \mathcal{S}),\\
    \LieDerivative_{\KillT}\mathcal{S} &= 0,
  \end{split}
\end{equation}
where $\KillT$ is the stationary vectorfield. $\KillT$-pseudoconvexity
once again played a substantial role in the unique continuation
argument as classical pseudoconvexity fails due to the possible
existence of trapping.

Below is a summary of the main steps taken in
\cite{ionescuUniquenessSmoothStationary2009} to prove \Cref{thm:IK}
(see \Cref{fig:IK}).
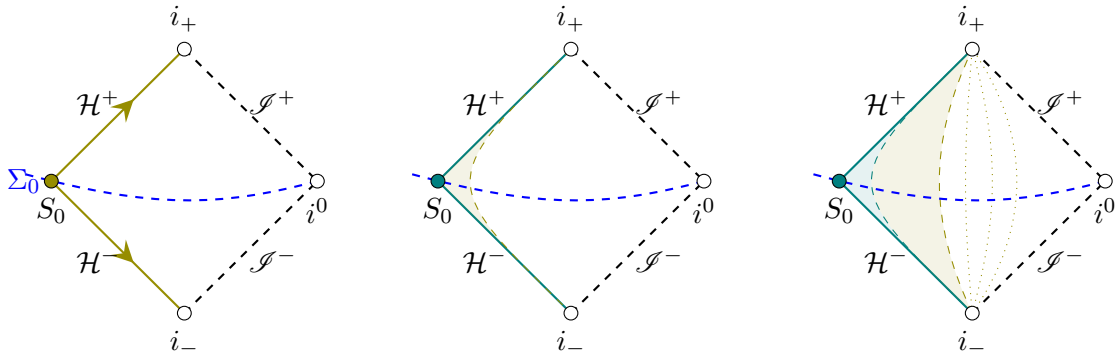
\begin{figure}[ht]
  \centering
  \begin{minipage}[t]{.33\textwidth}
    \centering
    \begin{tikzpicture}[scale=0.35,every node/.style={scale=1.0}]


  \def \s{5} 
  \def \exts{0.1} 
  \def \t{0.3}
  \def \Tlen{1}
  \def \ulen{0.7}
  \def \offset{0}
  \def \initExt{-0.05}

  \coordinate (tInf) at (0,\s); 
  \coordinate (EventZero) at (-\s,0); 
  \coordinate (CosmoZero) at (\s,0); 
  \coordinate (tNegInf) at (0,-\s);

  \coordinate (tInfL) at (-\offset,\s); 
  \coordinate (EventZeroL) at (-\s-\offset,0); 
  \coordinate (CosmoZeroL) at (\s-\offset,0); 
  \coordinate (tNegInfL) at (-\offset,-\s);
  \coordinate (initExtControlL)  at (-\s-\offset+\initExt,0);

  \coordinate (tInfR) at (\offset,\s); 
  \coordinate (EventZeroR) at (-\s+\offset,0); 
  \coordinate (CosmoZeroR) at (\s+\offset,0); 
  \coordinate (tNegInfR) at (\offset,-\s);
    \coordinate (initExtControlR)  at (\s+\offset-\initExt,0);

    \begin{scope}[decoration={
        markings,
        mark=at position 0.5 with {\arrow[scale=2]{stealth[reversed]}}}]
      \draw[olive,thick,name path=EventFuture,postaction={decorate}] (tInfL) --
      node[pos=0.4,left]{\textcolor{black}{$\EventHorizonFuture$}} (EventZeroL) ;
      
      \draw[olive,thick,name path=EventPast,postaction={decorate}] (tNegInfL) --
      node[pos=0.4,left]{\textcolor{black}{$\EventHorizonPast$}} (EventZeroL) ;
    \end{scope}

  \draw[black,dashed,thick,name path=CosmoFuture] (tInfR) --
  node[pos=0.4,right]{$\mathscr{I}^+$} (CosmoZeroR) ;
  
  \draw[black,dashed,thick,name path=CosmoPast] (tNegInfR) --
  node[pos=0.4,right]{$\mathscr{I}^-$} (CosmoZeroR);



  

  

  \path[blue,thick,dashed,out=-15,in=-165,shorten <= -10]
  (EventZeroL) edge node[pos=0,left]{$\Sigma_0$} node[pos=0.6](TInit){} (CosmoZeroR) ;

  \node[scale=0.5,fill=white,draw,circle,label=above:$i_+$]at(tInf){};
  \node[scale=0.5,fill=white,draw,circle,label=below:$i_-$]at(tNegInf){};
  \node[scale=0.5,fill=olive,draw,circle,label=below:$S_0$]at(EventZeroL){};
  \node[scale=0.5,fill=white,draw,circle,label=below:$i^0$]at(CosmoZeroR){};

\end{tikzpicture}

  \end{minipage}%
  \begin{minipage}[t]{.33\textwidth}
    \centering
    \begin{tikzpicture}[scale=0.35,every node/.style={scale=1.0}]


  \def \s{5} 
  \def \exts{0.1} 
  \def \t{0.3}
  \def \Tlen{1}
  \def \ulen{0.7}
  \def \offset{0}
  \def \initExt{-0.05}

  \coordinate (tInf) at (0,\s); 
  \coordinate (EventZero) at (-\s,0); 
  \coordinate (CosmoZero) at (\s,0); 
  \coordinate (tNegInf) at (0,-\s);

  \coordinate (tInfL) at (-\offset,\s); 
  \coordinate (EventZeroL) at (-\s-\offset,0); 
  \coordinate (CosmoZeroL) at (\s-\offset,0); 
  \coordinate (tNegInfL) at (-\offset,-\s);
  \coordinate (initExtControlL)  at (-\s-\offset+\initExt,0);

  \coordinate (tInfR) at (\offset,\s); 
  \coordinate (EventZeroR) at (-\s+\offset,0); 
  \coordinate (CosmoZeroR) at (\s+\offset,0); 
  \coordinate (tNegInfR) at (\offset,-\s);
    \coordinate (initExtControlR)  at (\s+\offset-\initExt,0);

  \draw[teal,thick,name path=EventFuture] (tInfL) --
  node[pos=0.4,left]{\textcolor{black}{$\EventHorizonFuture$}} (EventZeroL) ;  

  \draw[teal,thick,name path=EventPast] (tNegInfL) --
  node[pos=0.4,left]{\textcolor{black}{$\EventHorizonPast$}} (EventZeroL) ;

  \draw[black,dashed,thick,name path=CosmoFuture] (tInfR) --
  node[pos=0.4,right]{$\mathscr{I}^+$} (CosmoZeroR) ;
  
  \draw[black,dashed,thick,name path=CosmoPast] (tNegInfR) --
  node[pos=0.4,right]{$\mathscr{I}^-$} (CosmoZeroR);



  
 \path[fill=olive,fill opacity=0.1,draw=none] (tInfL) ..controls(initExtControlL)..(tNegInfL)-- (EventZeroL) -- cycle;
    \draw[olive,dashed] (tInfL) .. controls(initExtControlL).. (tNegInfL);

  

  \path[blue,thick,dashed,out=-15,in=-165,shorten <= -10]
  (EventZeroL) edge  (CosmoZeroR) ;

  \node[scale=0.5,fill=white,draw,circle,label=above:$i_+$]at(tInf){};
  \node[scale=0.5,fill=white,draw,circle,label=below:$i_-$]at(tNegInf){};
  \node[scale=0.5,fill=teal,draw,circle,label=below:$S_0$]at(EventZeroL){};
  \node[scale=0.5,fill=white,draw,circle,label=below:$i^0$]at(CosmoZeroR){};

\end{tikzpicture}

  \end{minipage}%
  \begin{minipage}[t]{.33\textwidth}
    \centering
    \begin{tikzpicture}[scale=0.35,every node/.style={scale=1.0}]


  \def \s{5} 
  \def \exts{0.1} 
  \def \t{0.3}
  \def \Tlen{1}
  \def \ulen{0.7}
  \def \offset{0}
  \def \initExt{-0.05}

  \coordinate (tInf) at (0,\s); 
  \coordinate (EventZero) at (-\s,0); 
  \coordinate (CosmoZero) at (\s,0); 
  \coordinate (tNegInf) at (0,-\s);

  \coordinate (tInfL) at (-\offset,\s); 
  \coordinate (EventZeroL) at (-\s-\offset,0); 
  \coordinate (CosmoZeroL) at (\s-\offset,0); 
  \coordinate (tNegInfL) at (-\offset,-\s);
  \coordinate (initExtControlL)  at (-\s-\offset+\initExt,0);

  \coordinate (tInfR) at (\offset,\s); 
  \coordinate (EventZeroR) at (-\s+\offset,0); 
  \coordinate (CosmoZeroR) at (\s+\offset,0); 
  \coordinate (tNegInfR) at (\offset,-\s);
    \coordinate (initExtControlR)  at (\s+\offset-\initExt,0);

  \draw[teal,thick,name path=EventFuture] (tInfL) --
  node[pos=0.4,left]{\textcolor{black}{$\EventHorizonFuture$}} (EventZeroL) ;  

  \draw[teal,thick,name path=EventPast] (tNegInfL) --
  node[pos=0.4,left]{\textcolor{black}{$\EventHorizonPast$}} (EventZeroL) ;

  \draw[black,dashed,thick,name path=CosmoFuture] (tInfR) --
  node[pos=0.4,right]{$\mathscr{I}^+$} (CosmoZeroR) ;
  
  \draw[black,dashed,thick,name path=CosmoPast] (tNegInfR) --
  node[pos=0.4,right]{$\mathscr{I}^-$} (CosmoZeroR);



  
 \path[fill=teal,fill opacity=0.1,draw=none] (tInfL) ..controls(initExtControlL)..(tNegInfL)-- (EventZeroL) -- cycle;
 \draw[teal,dashed] (tInfL) .. controls(initExtControlL).. (tNegInfL);

 \path[fill=olive,fill opacity=0.1,draw=none] (tInfL)..controls(initExtControlL)..(tNegInfL) to [out=115,in=-115](tInf);
 \draw[olive, dashed] (tNegInfL) to [out=115,in=-115](tInf);

 \draw[olive, dotted] (tNegInfL) to [out=95,in=-95](tInf);
 \draw[olive, dotted] (tNegInfL) to [out=75,in=-75](tInf);
 \draw[olive, dotted] (tNegInfL) to [out=55,in=-55](tInf);

  

  \path[blue,thick,dashed,out=-15,in=-165,shorten <= -10]
  (EventZeroL) edge  (CosmoZeroR) ;

  \node[scale=0.5,fill=white,draw,circle,label=above:$i_+$]at(tInf){};
  \node[scale=0.5,fill=white,draw,circle,label=below:$i_-$]at(tNegInf){};
  \node[scale=0.5,fill=teal,draw,circle,label=below:$S_0$]at(EventZeroL){};
  \node[scale=0.5,fill=white,draw,circle,label=below:$i^0$]at(CosmoZeroR){};

\end{tikzpicture}

  \end{minipage}%
  \caption{The basic steps in
    \cite{ionescuUniquenessSmoothStationary2009} to prove $C^{\infty}$
    Kerr rigidity conditional on the compatibility of
    $S_0$.
    The olive green shows the main area of interest in
    each step while teal regions depict regions where $\mathcal{S}=0$
    was previously shown.}
  \label{fig:IK}
\end{figure}

\begin{enumerate}
\item The first step in \cite{ionescuUniquenessSmoothStationary2009}
  is to show that $\mathcal{S}=0$ along the event horizon. This is
  done by first showing that $\mathcal{S}=0$ on the bifurcate sphere
  $S_0$. This is a consequence of the compatibility
  condition. Notably, this is the main place where the compatibility
  condition is used. To extend the vanishing of $\mathcal{S}$ to the
  entire horizons, the null Bianchi equations are then used to
  transport $\mathcal{S}$ along the future and past event horizons (as
  indicated by the arrows in the leftmost Penrose diagram in
  \Cref{fig:IK}).
\item The second step is to show that $\mathcal{S}$ vanishes in a
  small neighborhood of the event horizon. This uses a unique
  continuation argument based on the classical pseudoconvexity of
  $S_0$ which is a consequence of the null bifurcate geometry.
\item The third step is to extend the fact that $\mathcal{S}=0$ into the
  entire domain of exterior communication. This is done via a
  bootstrap argument and the construction of a $\KillT$-pseudoconvex
  foliation. The main difficulty in this step is extending
  $\mathcal{S}=0$ within the ergoregion. Once it is shown that
  $\mathcal{S}=0$ in the ergoregion, $\KillT$-pseudoconvexity in the
  rest of $\mathcal{M}$ essentially follows for free (in particular,
  in the asymptotically flat case, this is true in the asymptotic
  region). A key fact is that the foliation is constructed using only
  the stationary structure of $\mathcal{M}$.
\end{enumerate}

\subsection{\KdS{} rigidity}

The following sections provide a brief review of some of the key tools
used in the ensuing proofs of \KdS{} uniqueness. While in general,
less is known about black hole rigidity in the $\Lambda>0$ setting,
perturbative rigidity of slowly-rotating Kerr(-Newman)-de Sitter was
proven by Hintz \cite{hintzUniquenessKerrNewman2018} using a different
framework than the previously discussed rigidity results. In
particular, the perturbative rigidity result in
\cite{hintzUniquenessKerrNewman2018} relied on having a stability
result, which is typically substantially more difficult to prove than
a rigidity result (see \Cref{sec:other-literature} for further
discussion).

\subsubsection{Stationarity and the definition of $\KillT$}
\label{sec:intro:stationarity}

The starting point of Kerr and \KdS{} rigidity is the stationarity of
the solution at hand. In the $\Lambda=0$
setting, the primary solutions of interest are the asymptotically flat
solutions. By the standard definition of stationarity on
asymptotically flat spacetimes, $\KillT$ is then restricted to being
asymptotically timelike. This is a convenient definition from the
perspective of black hole rigidity given that one of the primary
difficulties of proving $C^{\infty}$ Kerr rigidity is the appearance
of ergoregions and the spacelike character of $\KillT$ at certain
points. In the case of a fixed Kerr background, this uniquely
identifies $\KillT=\partial_t$ as the stationary vectorfield.

Unfortunately in \KdS, stationarity is slightly more nuanced.  In the
$\Lambda>0$ setting, there is no asymptotic flatness assumption and
thus, even on a fixed \KdS{} background, there is no immediate
canonical definition of $\KillT$. Note that this is inherently
linked to $\KillT$-pseudoconvexity, which played a key role in
enabling the unique continuation arguments in
\cite{alexakisUniquenessSmoothStationary2010,ionescuUniquenessSmoothStationary2009}.

For the subextremal \KdS{} family itself, a natural choice for the
definition of $\KillT$, given that the presence of any
$\KillT$-trapped null geodesics is undesirable, is
$\KillT = \partial_t + \frac{a}{r_{*}^2+a^2}\partial_{\varphi}$, where
$r_{*}$ is the unique maximizer of
$\Delta = (r^2+a^2)\left(1- \frac{\Lambda}{3}r^2\right) - 2Mr$. This
is both related to the subextremality condition for \KdS, and also
reduces to $\partial_t$ in the $\Lambda\to 0$ limit in the sense that
$\Delta\vert_{\Lambda=0}$ is monotonically increasing in the region of
interest and therefore $r_{*}\to \infty$.  This assumption also
ensures that $\KillT$ is timelike in a neighborhood of $r=r_{*}$. A
more complete description of $\KillT$-trapped null geodesics and their
absence in subextremal \KdS, can be found in the articles by Petersen
and Vasy \cite{petersenStationarityFredholmTheory2024,
  petersenWaveEquationsKerr2024}, where the issue is explored in
depth.

However, with this choice, $r$ does not induce a globally
$\KillT$-pseudoconvex foliation on \KdS. Rather, at every point in
\KdS, either $r$ or $-r$ is $\KillT$-pseudoconvex (see
\Cref{fig:outline}). Critically however, on any subextremal \KdS{}
solution, this definition of $\KillT$ ensures that in some domain away
from the horizons, both $r$ and $-r$ are $\KillT$-pseudoconvex.  To
make use of $\KillT$-pseudoconvexity on a general stationary solution
$(\mathcal{M}, \Metric)$ to ($\Lambda$-EVE), a scalar function that
mimics the role of $r$ on subextremal \KdS{} solutions will be
constructed using only the fact that $(\mathcal{M}, \Metric)$ is
stationary.

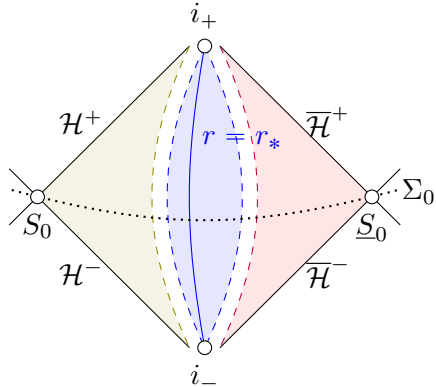
\begin{figure}[ht]
  \centering
  \begin{tikzpicture}[scale=0.4,every node/.style={scale=1.0}]


  \def \s{5} 
  \def \exts{0.1} 
  \def \t{0.3}
  \def \Tlen{1}
  \def \ulen{0.7}
  \def \offset{0.5}

  \coordinate (tInf) at (0,\s); 
  \coordinate (EventZero) at (-\s,0); 
  \coordinate (CosmoZero) at (\s,0); 
  \coordinate (tNegInf) at (0,-\s);

  \coordinate (tInfL) at (-\offset,\s); 
  \coordinate (EventZeroL) at (-\s-\offset,0); 
  \coordinate (CosmoZeroL) at (\s-\offset,0); 
  \coordinate (tNegInfL) at (-\offset,-\s);

  \coordinate (tInfR) at (\offset,\s); 
  \coordinate (EventZeroR) at (-\s+\offset,0); 
  \coordinate (CosmoZeroR) at (\s+\offset,0); 
  \coordinate (tNegInfR) at (\offset,-\s);

  \draw[shorten >= -15,name path=EventFuture] (tInfL) --
  node[pos=0.5,left]{$\EventHorizonFuture$} (EventZeroL) ;  

  \draw[shorten >= -15,name path=EventPast] (tNegInfL) --
  node[pos=0.5,left]{$\EventHorizonPast$} (EventZeroL) ;

  \draw[shorten >= -15,name path=CosmoFuture] (tInfR) --
  node[pos=0.5,right]{$\CosmologicalHorizonFuture$} (CosmoZeroR) ;
  
  \draw[shorten >= -15,name path=CosmoPast] (tNegInfR) --
  node[pos=0.5,right]{$\CosmologicalHorizonPast$} (CosmoZeroR);


  \path[fill=red,fill opacity=0.1,draw=none] (tInfR) to [out=-65,in=65](tNegInfR) -- (CosmoZeroR) -- cycle;
  \draw[purple, dashed,out=-65,in=65] (tInfR) edge node[pos=0.9,right]{} (tNegInfR);

  \path[fill=olive,fill opacity=0.1,draw=none] (tInfL) to [out=-115,in=115](tNegInfL) -- (EventZeroL) -- cycle;
  \draw[olive,dashed,out=-115, in=115] (tInfL) edge node[pos=0.9,left]{} (tNegInfL);

  
  \draw[blue,out=-100,in=100] (tInf) edge node[pos=0.3,right]{$r=r_{*}$} (tNegInf);
  \path[fill=blue,fill opacity=0.1,draw=none] (tInf) to [out=-115,in=115](tNegInf) to [out=65,in=-65](tInf);
  \draw[blue,dashed,out=-115,in=115](tInf) edge (tNegInf);
  \draw[blue,dashed,out=-65,in=65](tInf) edge (tNegInf);

  \path[black,thick,dotted,out=-15,in=-165,shorten >= -10, shorten <= -10]
  (EventZeroL) edge node[pos=1.05,right]{$\Sigma_0$} node[pos=0.6](TInit){} (CosmoZeroR) ;

  \node[scale=0.5,fill=white,draw,circle,label=above:$i_+$]at(tInf){};
  \node[scale=0.5,fill=white,draw,circle,label=below:$i_-$]at(tNegInf){};
  \node[scale=0.5,fill=white,draw,circle,label=below:$S_0$]at(EventZeroL){};
  \node[scale=0.5,fill=white,draw,circle,label=below:$\underline{S}_0$]at(CosmoZeroR){};

\end{tikzpicture}

  \caption{A Penrose diagram of a subextremal member of the \KdS{}
    family. The olive green region shows where $r$ is
    $\KillT$-pseudoconvex.  The red region depicts where $-r$ is
    $\KillT$-pseudoconvex. The blue region is where both $r$ and $-r$
    are $\KillT$-pseudoconvex, and is nonempty as a result of the
    subextremality assumption.}
  \label{fig:outline}
\end{figure}

\subsubsection{Stationarity implies axisymmetry}

The Carter-Robinson uniqueness theorem played a crucial role in
proving the rigidity results for Kerr by providing a way of
characterizing Kerr.  The Carter-Robinson uniqueness theorem does not
require analyticity, but it is restricted to the $\Lambda=0$
setting. No equivalent theorem is known for $\Lambda\neq 0$.
Nonetheless, one can ask whether, analogous to the result in
\cite{alexakisUniquenessSmoothStationary2010}, a stationary
perturbation of \KdS{} must have an additional symmetry, and in
particular, if it must be axisymmetric. This turns out to be precisely
the case.
\begin{theorem}[Rough version of \Cref{thm:main:e2c2}]
  \label{thm:main:e2c2:rough}
  A regular stationary cosmological black hole perturbation of \KdS{}
  which is also a solution to ($\Lambda$-EVE) where $\Lambda>0$ with regular
  null-bifurcate horizons possesses a second (rotational) Killing
  vectorfield.
\end{theorem}

In a similar vein as the proof in
\cite{alexakisUniquenessSmoothStationary2010}, the rotational
vectorfield generating axisymmetry is not constructed
directly. Instead, analogues to the Hawking vectorfield of Kerr are
constructed first, then the vectorfield generating axisymmetry is
constructed \textit{a posteriori}. The main difference is that unlike
in the case of Kerr, where the Hawking vectorfield was constructed
from the event horizon and extended into the entire domain of exterior
of communication, to handle \KdS, two analogues of the Hawking
vectorfield are constructed and extended, one from the event horizon
and the other from the cosmological horizon.

The main steps in the proof of \Cref{thm:main:e2c2:rough} are summarized
below (see \Cref{fig:E2C2}).
\begin{figure}[ht]
  \centering
  \begin{minipage}[t]{0.24\linewidth}
    \centering
    \begin{tikzpicture}[scale=0.27,every node/.style={scale=1.0}]


  \def \s{5} 
  \def \exts{0.1} 
  \def \t{0.3}
  \def \Tlen{1}
  \def \ulen{0.7}
  \def \offset{0}
  \def \initExt{-0.05}

  \coordinate (tInf) at (0,\s); 
  \coordinate (EventZero) at (-\s,0); 
  \coordinate (CosmoZero) at (\s,0); 
  \coordinate (tNegInf) at (0,-\s);

  \coordinate (tInfL) at (-\offset,\s); 
  \coordinate (EventZeroL) at (-\s-\offset,0); 
  \coordinate (CosmoZeroL) at (\s-\offset,0); 
  \coordinate (tNegInfL) at (-\offset,-\s);
  \coordinate (initExtControlL)  at (-\s-\offset+\initExt,0);

  \coordinate (tInfR) at (\offset,\s); 
  \coordinate (EventZeroR) at (-\s+\offset,0); 
  \coordinate (CosmoZeroR) at (\s+\offset,0); 
  \coordinate (tNegInfR) at (\offset,-\s);
  \coordinate (initExtControlR)  at (\s+\offset-\initExt,0);


  \draw[olive,thick,name path=EventFuture] (tInfL) --
  node[pos=0.4,left]{\textcolor{black}{$\EventHorizonFuture$}} (EventZeroL) ;
  
  \draw[olive,thick,name path=EventPast] (tNegInfL) --
  node[pos=0.4,left]{\textcolor{black}{$\EventHorizonPast$}} (EventZeroL) ;

  \draw[olive,thick,name path=CosmoFuture] (tInfR) --
  node[pos=0.4,right]{\textcolor{black}{$\,\CosmologicalHorizonFuture$}} (CosmoZeroR) ;
  
  \draw[olive,thick,name path=CosmoPast] (tNegInfR) --
  node[pos=0.4,right]{\textcolor{black}{$\,\CosmologicalHorizonFuture$}} (CosmoZeroR);




  

  

  \path[blue,thick,dashed,out=-15,in=-165,shorten <= -10,shorten >= -10]
  (EventZeroL) edge node[pos=0,left]{$\Sigma_0$} node[pos=0.6](TInit){} (CosmoZeroR) ;

  \node[scale=0.5,fill=white,draw,circle,label=above:$i_+$]at(tInf){};
  \node[scale=0.5,fill=white,draw,circle,label=below:$i_-$]at(tNegInf){};
  \node[scale=0.5,fill=olive,draw,circle,label=below:$S_0$]at(EventZeroL){};
  \node[scale=0.5,fill=olive,draw,circle,label=below:$\underline{S}_0$]at(CosmoZeroR){};

\end{tikzpicture}

  \end{minipage}%
  \begin{minipage}[t]{0.24\linewidth}
    \centering
    \begin{tikzpicture}[scale=0.27,every node/.style={scale=1.0}]


  \def \s{5} 
  \def \exts{0.1} 
  \def \t{0.3}
  \def \Tlen{1}
  \def \ulen{0.7}
  \def \offset{0}
  \def \initExt{-0.05}

  \coordinate (tInf) at (0,\s); 
  \coordinate (EventZero) at (-\s,0); 
  \coordinate (CosmoZero) at (\s,0); 
  \coordinate (tNegInf) at (0,-\s);

  \coordinate (tInfL) at (-\offset,\s); 
  \coordinate (EventZeroL) at (-\s-\offset,0); 
  \coordinate (CosmoZeroL) at (\s-\offset,0); 
  \coordinate (tNegInfL) at (-\offset,-\s);
  \coordinate (initExtControlL)  at (-\s-\offset+\initExt,0);

  \coordinate (tInfR) at (\offset,\s); 
  \coordinate (EventZeroR) at (-\s+\offset,0); 
  \coordinate (CosmoZeroR) at (\s+\offset,0); 
  \coordinate (tNegInfR) at (\offset,-\s);
    \coordinate (initExtControlR)  at (\s+\offset-\initExt,0);

  \draw[teal,thick,name path=EventFuture] (tInfL) --
  node[pos=0.4,left]{\textcolor{black}{$\EventHorizonFuture$}} (EventZeroL) ;  

  \draw[teal,thick,name path=EventPast] (tNegInfL) --
  node[pos=0.4,left]{\textcolor{black}{$\EventHorizonPast$}} (EventZeroL) ;

  \draw[teal,thick,name path=CosmoFuture] (tInfR) --
  node[pos=0.4,right]{\textcolor{black}{$\,\CosmologicalHorizonFuture$}} (CosmoZeroR) ;
  
  \draw[teal,thick,name path=CosmoPast] (tNegInfR) --
  node[pos=0.4,right]{\textcolor{black}{$\,\CosmologicalHorizonPast$}} (CosmoZeroR);

 \path[fill=olive,fill opacity=0.1,draw=none] (tInfL) ..controls(initExtControlL)..(tNegInfL)-- (EventZeroL) -- cycle;
 \draw[olive,dashed] (tInfL) .. controls(initExtControlL).. (tNegInfL);

  \path[fill=olive,fill opacity=0.1,draw=none] (tInfR) ..controls(initExtControlR)..(tNegInfR)-- (CosmoZeroR) -- cycle;
    \draw[olive,dashed] (tInfR) .. controls(initExtControlR).. (tNegInfR);

  

  \path[blue,thick,dashed,out=-15,in=-165,shorten <= -10,shorten >= -10]
  (EventZeroL) edge  (CosmoZeroR) ;

  \node[scale=0.5,fill=white,draw,circle,label=above:$i_+$]at(tInf){};
  \node[scale=0.5,fill=white,draw,circle,label=below:$i_-$]at(tNegInf){};
  \node[scale=0.5,fill=teal,draw,circle,label=below:$S_0$]at(EventZeroL){};
  \node[scale=0.5,fill=teal,draw,circle,label=below:$\underline{S}_0$]at(CosmoZeroR){};

\end{tikzpicture}

  \end{minipage}%
  \begin{minipage}[t]{0.24\linewidth}
    \centering
    \begin{tikzpicture}[scale=0.27,every node/.style={scale=1.0}]


  \def \s{5} 
  \def \exts{0.1} 
  \def \t{0.3}
  \def \Tlen{1}
  \def \ulen{0.7}
  \def \offset{0}
  \def \initExt{-0.05}

  \coordinate (tInf) at (0,\s); 
  \coordinate (EventZero) at (-\s,0); 
  \coordinate (CosmoZero) at (\s,0); 
  \coordinate (tNegInf) at (0,-\s);

  \coordinate (tInfL) at (-\offset,\s); 
  \coordinate (EventZeroL) at (-\s-\offset,0); 
  \coordinate (CosmoZeroL) at (\s-\offset,0); 
  \coordinate (tNegInfL) at (-\offset,-\s);
  \coordinate (initExtControlL)  at (-\s-\offset+\initExt,0);

  \coordinate (tInfR) at (\offset,\s); 
  \coordinate (EventZeroR) at (-\s+\offset,0); 
  \coordinate (CosmoZeroR) at (\s+\offset,0); 
  \coordinate (tNegInfR) at (\offset,-\s);
  \coordinate (initExtControlR)  at (\s+\offset-\initExt,0);

  \draw[teal,thick,name path=EventFuture] (tInfL) --
  node[pos=0.4,left]{\textcolor{black}{$\EventHorizonFuture$}} (EventZeroL) ;  

  \draw[teal,thick,name path=EventPast] (tNegInfL) --
  node[pos=0.4,left]{\textcolor{black}{$\EventHorizonPast$}} (EventZeroL) ;

  \draw[teal,thick,name path=CosmoFuture] (tInfR) --
  node[pos=0.4,right]{\textcolor{black}{$\,\CosmologicalHorizonFuture$}} (CosmoZeroR) ;
  
  \draw[teal,thick,name path=CosmoPast] (tNegInfR) --
  node[pos=0.4,right]{\textcolor{black}{$\,\CosmologicalHorizonPast$}} (CosmoZeroR);


  \path[fill=teal,fill opacity=0.1,draw=none] (tInfL) ..controls(initExtControlL)..(tNegInfL)-- (EventZeroL) -- cycle;
  \draw[teal,dashed] (tInfL) .. controls(initExtControlL).. (tNegInfL);
  \path[fill=teal,fill opacity=0.1,draw=none] (tInfR) ..controls(initExtControlR)..(tNegInfR)-- (CosmoZeroR) -- cycle;
  \draw[teal,dashed] (tInfR) .. controls(initExtControlR).. (tNegInfR);

  \path[fill=olive,fill opacity=0.1,draw=none] (tInfL)..controls(initExtControlL)..(tNegInfL) to [out=105,in=-105](tInf);
  \draw[olive, dashed] (tNegInfL) to [out=105,in=-105](tInf);
  \path[fill=olive,fill opacity=0.1,draw=none] (tInfR) ..controls(initExtControlR)..(tNegInfR) to [out=75,in=-75] (tInf);
  \draw[olive,dashed] (tNegInfR) to [out=75,in=-75] (tInf);

  

  \path[blue,thick,dashed,out=-15,in=-165,shorten <= -10, shorten >= -10]
  (EventZeroL) edge  (CosmoZeroR) ;

  \node[scale=0.5,fill=white,draw,circle,label=above:$i_+$]at(tInf){};
  \node[scale=0.5,fill=white,draw,circle,label=below:$i_-$]at(tNegInf){};
  \node[scale=0.5,fill=teal,draw,circle,label=below:$S_0$]at(EventZeroL){};
  \node[scale=0.5,fill=teal,draw,circle,label=below:$\underline{S}_0$]at(CosmoZeroR){};

\end{tikzpicture}

  \end{minipage}%
  \begin{minipage}[t]{0.24\linewidth}
    \centering
    \begin{tikzpicture}[scale=0.27,every node/.style={scale=1.0}]


  \def \s{5} 
  \def \exts{0.1} 
  \def \t{0.3}
  \def \Tlen{1}
  \def \ulen{0.7}
  \def \offset{0}
  \def \initExt{-0.05}

  \coordinate (tInf) at (0,\s); 
  \coordinate (EventZero) at (-\s,0); 
  \coordinate (CosmoZero) at (\s,0); 
  \coordinate (tNegInf) at (0,-\s);

  \coordinate (tInfL) at (-\offset,\s); 
  \coordinate (EventZeroL) at (-\s-\offset,0); 
  \coordinate (CosmoZeroL) at (\s-\offset,0); 
  \coordinate (tNegInfL) at (-\offset,-\s);
  \coordinate (initExtControlL)  at (-\s-\offset+\initExt,0);

  \coordinate (tInfR) at (\offset,\s); 
  \coordinate (EventZeroR) at (-\s+\offset,0); 
  \coordinate (CosmoZeroR) at (\s+\offset,0); 
  \coordinate (tNegInfR) at (\offset,-\s);
  \coordinate (initExtControlR)  at (\s+\offset-\initExt,0);

  \draw[teal,thick,name path=EventFuture] (tInfL) --
  node[pos=0.4,left]{\textcolor{black}{$\EventHorizonFuture$}} (EventZeroL) ;  

  \draw[teal,thick,name path=EventPast] (tNegInfL) --
  node[pos=0.4,left]{\textcolor{black}{$\EventHorizonPast$}} (EventZeroL) ;

  \draw[teal,thick,name path=CosmoFuture] (tInfR) --
  node[pos=0.4,right]{\textcolor{black}{$\,\CosmologicalHorizonFuture$}} (CosmoZeroR) ;
  
  \draw[teal,thick,name path=CosmoPast] (tNegInfR) --
  node[pos=0.4,right]{\textcolor{black}{$\,\CosmologicalHorizonPast$}} (CosmoZeroR);


  \path[fill=teal,fill opacity=0.1,draw=none] (tInfL) ..controls(initExtControlL)..(tNegInfL)-- (EventZeroL) -- cycle;
  \path[fill=teal,fill opacity=0.1,draw=none] (tInfR) ..controls(initExtControlR)..(tNegInfR)-- (CosmoZeroR) -- cycle;

  \path[fill=teal,fill opacity=0.1,draw=none] (tInfL)..controls(initExtControlL)..(tNegInfL) to [out=105,in=-105](tInf);
  \draw[teal, dashed] (tNegInfL) to [out=105,in=-105](tInf);
  \path[fill=teal,fill opacity=0.1,draw=none] (tInfR) ..controls(initExtControlR)..(tNegInfR) to [out=75,in=-75] (tInf);
  \draw[teal,dashed] (tNegInfR) to [out=75,in=-75] (tInf);

  \path[fill=olive,fill opacity=0.1,draw=none] (tInf) to [bend left=30] (tNegInf) to [bend left=30] (tInf);
  \draw[olive,dashed] (tInf) to [bend left=30] (tNegInf) to [bend left=30] (tInf);

  

  \path[blue,thick,dashed,out=-15,in=-165,shorten <= -10, shorten >= -10]
  (EventZeroL) edge  (CosmoZeroR) ;

  \node[scale=0.5,fill=white,draw,circle,label=above:$i_+$]at(tInf){};
  \node[scale=0.5,fill=white,draw,circle,label=below:$i_-$]at(tNegInf){};
  \node[scale=0.5,fill=teal,draw,circle,label=below:$S_0$]at(EventZeroL){};
  \node[scale=0.5,fill=teal,draw,circle,label=below:$\underline{S}_0$]at(CosmoZeroR){};

\end{tikzpicture}

  \end{minipage}%
  \caption{Penrose diagrams describing the main steps in proving
    \Cref{thm:main:e2c2:rough}.
    The olive
    green shows the main area of interest in each step while teal
    regions depict regions previously shown to be axisymmetric. }
  \label{fig:E2C2}
\end{figure}
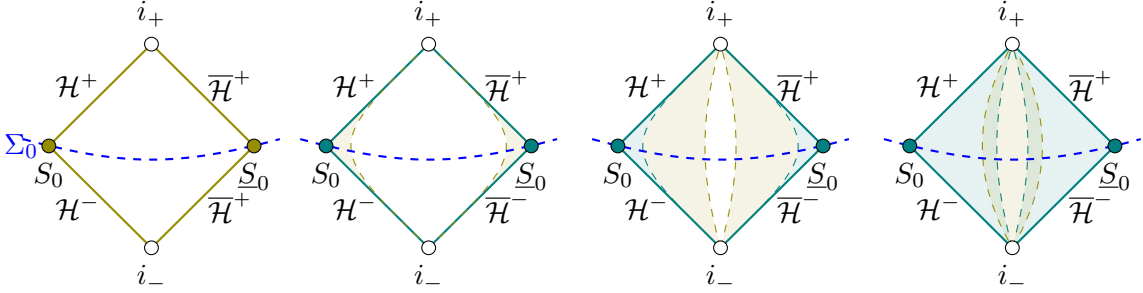

\begin{enumerate}
\item The first step is to show the existence of a second Killing
  vectorfield $\mathbf{K}$ on the event horizon and the existence of a
  second Killing vectorfield $\underline{\mathbf{K}}$ on the
  cosmological horizon. These are the analogues of the Hawking
  vectorfield. In a similar manner to the construction of the Hawking
  vectorfield in \cite{alexakisUniquenessSmoothStationary2010}, these
  are constructed by solving a characteristic initial value problem in
  the domain of dependence of $S_0$ and $\underline{S}_0$
  respectively. 
\item The second step is to extend $\mathbf{K}$ in a neighborhood of
  the event horizon and $\underline{\mathbf{K}}$ in a neighborhood of
  the cosmological horizon. This step uses a unique continuation
  argument that relies only on the null bifurcate geometry of the
  horizons. In each neighborhood, we also find some axisymmetric
  vectorfields $\KillPhi$ and $\underline{\KillPhi}$ that
  lie in $\Span\left( \KillT, \mathbf{K} \right)$ and
  $\Span\left( \KillT, \underline{\mathbf{K}} \right)$ respectively. 
\item The third step is to extend $\mathbf{K}$ and
  $\underline{\mathbf{K}}$ into the stationary region as far as
  possible by constructing a $\KillT$-pseudoconvex foliation. As
  previously discussed, it is not possible to construct a
  globally $\KillT$-pseudoconvex foliation like in the Kerr case. As a
  result, neither $\mathbf{K}$ nor $\underline{\mathbf{K}}$ can be
  directly extended all the way to the cosmological horizon or the
  event horizon respectively.  After extending $\mathbf{K}$ and
  $\underline{\mathbf{K}}$, we can show that $\KillPhi$ and
  $\underline{\KillPhi}$ also extend into the stationary region in the
  same regions as $\mathbf{K}$ and $\underline{\mathbf{K}}$
  respectively.
\item The fourth step is to use the subextremal assumption to show
  that there is nonetheless an open domain of $\mathcal{M}$ where both
  $\mathbf{K}$ and $\underline{\mathbf{K}}$ can be extended as Killing
  vectorfields. It will then follow that $\KillPhi$ and
  $\underline{\KillPhi}$ can both be extended into this region, from
  which one can show that in fact,
  $\KillPhi=\underline{\KillPhi}$. This then allows us to conclude
  that the stationary region of $\mathcal{M}$ possesses an
  vectorfield generating axisymmetry in its entirety. 
\end{enumerate}

\subsubsection{Rigidity using an adapted Mars-Simon tensor}

Despite the lack of an analogue of the Carter-Robinson theorem in
$\Lambda>0$, the Mars-Simon tensor $\mathcal{S}$ introduced in
\cite{marsSpacetimeCharacterizationKerr1999,marsUniquenessPropertiesKerr2000},
has been extended to the $\Lambda>0$ case by Mars and Senovilla in
\cite{marsSpacetimeCharacterizationKerrNUTAde2015}. The result proven
in \cite{marsSpacetimeCharacterizationKerrNUTAde2015} is actually
somewhat stronger than just that $\mathcal{S}=0$ characterizes \KdS,
since it suffices for the Weyl tensor to be proportional to a specific
four-tensor. The generalized Mars-Simon tensor not only characterizes
Kerr(-de Sitter), but it also satisfies a wave-type equation.
This suggests that the approach used in
\cite{ionescuUniquenessSmoothStationary2009} based on the Mars-Simon
tensor for Kerr could be generalized to \KdS. The main obstacle in
this case has to do with the nature of $\KillT$-pseudoconvexity in the
$\Lambda>0$ setting.

As previously mentioned, on \KdS, using our definition of
stationarity, there is no singular global $\KillT$-pseudoconvex
foliation.  To accommodate this behavior in \KdS, a rigidity result is
proven under a two-sided rigidity hypothesis that makes assumptions
both on the event horizon and the cosmological horizon.  A
rough version of our main theorem is stated below (for a more detailed version see
\Cref{thm:main:e1c1}, and in particular, the assumptions made in
\Cref{sec:assumptions}).
\begin{theorem}[Rough version of \Cref{thm:main:e1c1}]
  \label{thm:main:e1c1:rough}
  A stationary, regular cosmological black hole solution to
  ($\Lambda$-EVE) where $\Lambda>0$ with regular, connected null-bifurcate event and
  cosmological horizons that have compatible bifurcate spheres, has a
  stationary region isometric to that of a member of the \KdS{}
  family.
\end{theorem}
\Cref{thm:main:e1c1:rough} is proven using two unique continuation
arguments, based on the fact that the Mars-Simon tensor satisfies a
system of equations of the form \Cref{eq:intro:S-T-coupled-system}. 

The main steps in \Cref{thm:main:e1c1:rough} are described below (see
\Cref{fig:E1C1}).
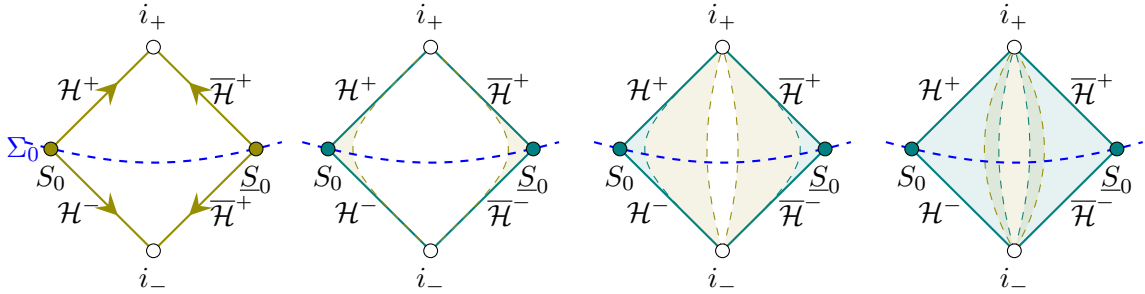
\begin{figure}[ht]
  \centering
  \begin{minipage}[t]{0.24\linewidth}
    \centering
    \begin{tikzpicture}[scale=0.27,every node/.style={scale=1.0}]


  \def \s{5} 
  \def \exts{0.1} 
  \def \t{0.3}
  \def \Tlen{1}
  \def \ulen{0.7}
  \def \offset{0}
  \def \initExt{-0.05}

  \coordinate (tInf) at (0,\s); 
  \coordinate (EventZero) at (-\s,0); 
  \coordinate (CosmoZero) at (\s,0); 
  \coordinate (tNegInf) at (0,-\s);

  \coordinate (tInfL) at (-\offset,\s); 
  \coordinate (EventZeroL) at (-\s-\offset,0); 
  \coordinate (CosmoZeroL) at (\s-\offset,0); 
  \coordinate (tNegInfL) at (-\offset,-\s);
  \coordinate (initExtControlL)  at (-\s-\offset+\initExt,0);

  \coordinate (tInfR) at (\offset,\s); 
  \coordinate (EventZeroR) at (-\s+\offset,0); 
  \coordinate (CosmoZeroR) at (\s+\offset,0); 
  \coordinate (tNegInfR) at (\offset,-\s);
    \coordinate (initExtControlR)  at (\s+\offset-\initExt,0);

    \begin{scope}[decoration={
        markings,
        mark=at position 0.5 with {\arrow[scale=2]{stealth[reversed]}}}]
      \draw[olive,thick,name path=EventFuture,postaction={decorate}] (tInfL) --
      node[pos=0.4,left]{\textcolor{black}{$\EventHorizonFuture$}} (EventZeroL) ;
      
      \draw[olive,thick,name path=EventPast,postaction={decorate}] (tNegInfL) --
      node[pos=0.4,left]{\textcolor{black}{$\EventHorizonPast$}} (EventZeroL) ;

      \draw[olive,thick,name path=CosmoFuture,postaction={decorate}] (tInfR) --
      node[pos=0.4,right]{\textcolor{black}{$\,\CosmologicalHorizonFuture$}} (CosmoZeroR) ;
      
      \draw[olive,thick,name path=CosmoPast,postaction={decorate}] (tNegInfR) --
      node[pos=0.4,right]{\textcolor{black}{$\,\CosmologicalHorizonFuture$}} (CosmoZeroR);
    \end{scope}




  

  

  \path[blue,thick,dashed,out=-15,in=-165,shorten <= -10,shorten >= -10]
  (EventZeroL) edge node[pos=0,left]{$\Sigma_0$} node[pos=0.6](TInit){} (CosmoZeroR) ;

  \node[scale=0.5,fill=white,draw,circle,label=above:$i_+$]at(tInf){};
  \node[scale=0.5,fill=white,draw,circle,label=below:$i_-$]at(tNegInf){};
  \node[scale=0.5,fill=olive,draw,circle,label=below:$S_0$]at(EventZeroL){};
  \node[scale=0.5,fill=olive,draw,circle,label=below:$\underline{S}_0$]at(CosmoZeroR){};

\end{tikzpicture}

  \end{minipage}%
  \begin{minipage}[t]{0.24\linewidth}
    \centering
    \begin{tikzpicture}[scale=0.27,every node/.style={scale=1.0}]


  \def \s{5} 
  \def \exts{0.1} 
  \def \t{0.3}
  \def \Tlen{1}
  \def \ulen{0.7}
  \def \offset{0}
  \def \initExt{-0.05}

  \coordinate (tInf) at (0,\s); 
  \coordinate (EventZero) at (-\s,0); 
  \coordinate (CosmoZero) at (\s,0); 
  \coordinate (tNegInf) at (0,-\s);

  \coordinate (tInfL) at (-\offset,\s); 
  \coordinate (EventZeroL) at (-\s-\offset,0); 
  \coordinate (CosmoZeroL) at (\s-\offset,0); 
  \coordinate (tNegInfL) at (-\offset,-\s);
  \coordinate (initExtControlL)  at (-\s-\offset+\initExt,0);

  \coordinate (tInfR) at (\offset,\s); 
  \coordinate (EventZeroR) at (-\s+\offset,0); 
  \coordinate (CosmoZeroR) at (\s+\offset,0); 
  \coordinate (tNegInfR) at (\offset,-\s);
    \coordinate (initExtControlR)  at (\s+\offset-\initExt,0);

  \draw[teal,thick,name path=EventFuture] (tInfL) --
  node[pos=0.4,left]{\textcolor{black}{$\EventHorizonFuture$}} (EventZeroL) ;  

  \draw[teal,thick,name path=EventPast] (tNegInfL) --
  node[pos=0.4,left]{\textcolor{black}{$\EventHorizonPast$}} (EventZeroL) ;

  \draw[teal,thick,name path=CosmoFuture] (tInfR) --
  node[pos=0.4,right]{\textcolor{black}{$\,\CosmologicalHorizonFuture$}} (CosmoZeroR) ;
  
  \draw[teal,thick,name path=CosmoPast] (tNegInfR) --
  node[pos=0.4,right]{\textcolor{black}{$\,\CosmologicalHorizonPast$}} (CosmoZeroR);

 \path[fill=olive,fill opacity=0.1,draw=none] (tInfL) ..controls(initExtControlL)..(tNegInfL)-- (EventZeroL) -- cycle;
 \draw[olive,dashed] (tInfL) .. controls(initExtControlL).. (tNegInfL);

  \path[fill=olive,fill opacity=0.1,draw=none] (tInfR) ..controls(initExtControlR)..(tNegInfR)-- (CosmoZeroR) -- cycle;
    \draw[olive,dashed] (tInfR) .. controls(initExtControlR).. (tNegInfR);

  

  \path[blue,thick,dashed,out=-15,in=-165,shorten <= -10,shorten >= -10]
  (EventZeroL) edge  (CosmoZeroR) ;

  \node[scale=0.5,fill=white,draw,circle,label=above:$i_+$]at(tInf){};
  \node[scale=0.5,fill=white,draw,circle,label=below:$i_-$]at(tNegInf){};
  \node[scale=0.5,fill=teal,draw,circle,label=below:$S_0$]at(EventZeroL){};
  \node[scale=0.5,fill=teal,draw,circle,label=below:$\underline{S}_0$]at(CosmoZeroR){};

\end{tikzpicture}

  \end{minipage}%
  \begin{minipage}[t]{0.24\linewidth}
    \centering
    \begin{tikzpicture}[scale=0.27,every node/.style={scale=1.0}]


  \def \s{5} 
  \def \exts{0.1} 
  \def \t{0.3}
  \def \Tlen{1}
  \def \ulen{0.7}
  \def \offset{0}
  \def \initExt{-0.05}

  \coordinate (tInf) at (0,\s); 
  \coordinate (EventZero) at (-\s,0); 
  \coordinate (CosmoZero) at (\s,0); 
  \coordinate (tNegInf) at (0,-\s);

  \coordinate (tInfL) at (-\offset,\s); 
  \coordinate (EventZeroL) at (-\s-\offset,0); 
  \coordinate (CosmoZeroL) at (\s-\offset,0); 
  \coordinate (tNegInfL) at (-\offset,-\s);
  \coordinate (initExtControlL)  at (-\s-\offset+\initExt,0);

  \coordinate (tInfR) at (\offset,\s); 
  \coordinate (EventZeroR) at (-\s+\offset,0); 
  \coordinate (CosmoZeroR) at (\s+\offset,0); 
  \coordinate (tNegInfR) at (\offset,-\s);
  \coordinate (initExtControlR)  at (\s+\offset-\initExt,0);

  \draw[teal,thick,name path=EventFuture] (tInfL) --
  node[pos=0.4,left]{\textcolor{black}{$\EventHorizonFuture$}} (EventZeroL) ;  

  \draw[teal,thick,name path=EventPast] (tNegInfL) --
  node[pos=0.4,left]{\textcolor{black}{$\EventHorizonPast$}} (EventZeroL) ;

  \draw[teal,thick,name path=CosmoFuture] (tInfR) --
  node[pos=0.4,right]{\textcolor{black}{$\,\CosmologicalHorizonFuture$}} (CosmoZeroR) ;
  
  \draw[teal,thick,name path=CosmoPast] (tNegInfR) --
  node[pos=0.4,right]{\textcolor{black}{$\,\CosmologicalHorizonPast$}} (CosmoZeroR);


  \path[fill=teal,fill opacity=0.1,draw=none] (tInfL) ..controls(initExtControlL)..(tNegInfL)-- (EventZeroL) -- cycle;
  \draw[teal,dashed] (tInfL) .. controls(initExtControlL).. (tNegInfL);
  \path[fill=teal,fill opacity=0.1,draw=none] (tInfR) ..controls(initExtControlR)..(tNegInfR)-- (CosmoZeroR) -- cycle;
  \draw[teal,dashed] (tInfR) .. controls(initExtControlR).. (tNegInfR);

  \path[fill=olive,fill opacity=0.1,draw=none] (tInfL)..controls(initExtControlL)..(tNegInfL) to [out=105,in=-105](tInf);
  \draw[olive, dashed] (tNegInfL) to [out=105,in=-105](tInf);
  \path[fill=olive,fill opacity=0.1,draw=none] (tInfR) ..controls(initExtControlR)..(tNegInfR) to [out=75,in=-75] (tInf);
  \draw[olive,dashed] (tNegInfR) to [out=75,in=-75] (tInf);

  

  \path[blue,thick,dashed,out=-15,in=-165,shorten <= -10, shorten >= -10]
  (EventZeroL) edge  (CosmoZeroR) ;

  \node[scale=0.5,fill=white,draw,circle,label=above:$i_+$]at(tInf){};
  \node[scale=0.5,fill=white,draw,circle,label=below:$i_-$]at(tNegInf){};
  \node[scale=0.5,fill=teal,draw,circle,label=below:$S_0$]at(EventZeroL){};
  \node[scale=0.5,fill=teal,draw,circle,label=below:$\underline{S}_0$]at(CosmoZeroR){};

\end{tikzpicture}

  \end{minipage}%
  \begin{minipage}[t]{0.24\linewidth}
    \centering
    \begin{tikzpicture}[scale=0.27,every node/.style={scale=1.0}]


  \def \s{5} 
  \def \exts{0.1} 
  \def \t{0.3}
  \def \Tlen{1}
  \def \ulen{0.7}
  \def \offset{0}
  \def \initExt{-0.05}

  \coordinate (tInf) at (0,\s); 
  \coordinate (EventZero) at (-\s,0); 
  \coordinate (CosmoZero) at (\s,0); 
  \coordinate (tNegInf) at (0,-\s);

  \coordinate (tInfL) at (-\offset,\s); 
  \coordinate (EventZeroL) at (-\s-\offset,0); 
  \coordinate (CosmoZeroL) at (\s-\offset,0); 
  \coordinate (tNegInfL) at (-\offset,-\s);
  \coordinate (initExtControlL)  at (-\s-\offset+\initExt,0);

  \coordinate (tInfR) at (\offset,\s); 
  \coordinate (EventZeroR) at (-\s+\offset,0); 
  \coordinate (CosmoZeroR) at (\s+\offset,0); 
  \coordinate (tNegInfR) at (\offset,-\s);
  \coordinate (initExtControlR)  at (\s+\offset-\initExt,0);

  \draw[teal,thick,name path=EventFuture] (tInfL) --
  node[pos=0.4,left]{\textcolor{black}{$\EventHorizonFuture$}} (EventZeroL) ;  

  \draw[teal,thick,name path=EventPast] (tNegInfL) --
  node[pos=0.4,left]{\textcolor{black}{$\EventHorizonPast$}} (EventZeroL) ;

  \draw[teal,thick,name path=CosmoFuture] (tInfR) --
  node[pos=0.4,right]{\textcolor{black}{$\,\CosmologicalHorizonFuture$}} (CosmoZeroR) ;
  
  \draw[teal,thick,name path=CosmoPast] (tNegInfR) --
  node[pos=0.4,right]{\textcolor{black}{$\,\CosmologicalHorizonPast$}} (CosmoZeroR);


  \path[fill=teal,fill opacity=0.1,draw=none] (tInfL) ..controls(initExtControlL)..(tNegInfL)-- (EventZeroL) -- cycle;
  \path[fill=teal,fill opacity=0.1,draw=none] (tInfR) ..controls(initExtControlR)..(tNegInfR)-- (CosmoZeroR) -- cycle;

  \path[fill=teal,fill opacity=0.1,draw=none] (tInfL)..controls(initExtControlL)..(tNegInfL) to [out=105,in=-105](tInf);
  \draw[teal, dashed] (tNegInfL) to [out=105,in=-105](tInf);
  \path[fill=teal,fill opacity=0.1,draw=none] (tInfR) ..controls(initExtControlR)..(tNegInfR) to [out=75,in=-75] (tInf);
  \draw[teal,dashed] (tNegInfR) to [out=75,in=-75] (tInf);

  \path[fill=olive,fill opacity=0.1,draw=none] (tInf) to [bend left=30] (tNegInf) to [bend left=30] (tInf);
  \draw[olive,dashed] (tInf) to [bend left=30] (tNegInf) to [bend left=30] (tInf);

  

  \path[blue,thick,dashed,out=-15,in=-165,shorten <= -10, shorten >= -10]
  (EventZeroL) edge  (CosmoZeroR) ;

  \node[scale=0.5,fill=white,draw,circle,label=above:$i_+$]at(tInf){};
  \node[scale=0.5,fill=white,draw,circle,label=below:$i_-$]at(tNegInf){};
  \node[scale=0.5,fill=teal,draw,circle,label=below:$S_0$]at(EventZeroL){};
  \node[scale=0.5,fill=teal,draw,circle,label=below:$\underline{S}_0$]at(CosmoZeroR){};

\end{tikzpicture}

  \end{minipage}%
  \caption[]{Penrose diagrams describing the main steps in proving
    \Cref{thm:main:e1c1:rough}.
    The olive green regions depict the regions of
    interest in each step and the teal regions are regions where
    previous steps have already shown that $\mathcal{S}=0$.}
  \label{fig:E1C1}
\end{figure}
\begin{enumerate}
\item The first step is to show that $\mathcal{S}=0$ on the
  horizons. This is done similarly but independently on the event and
  cosmological horizons. On both horizons, the compatibility
  conditions are used to show that $\mathcal{S}=0$ on the bifurcate
  sphere of the horizon ($S_0$ and $\underline{S}_0$ for the event
  horizon and the cosmological horizon respectively). The
  compatibility conditions embed information about the mass and
  angular momentum of the spacetime. Therefore, one is not free to
  prescribe the compatibility conditions freely on the two
  horizons. That is, the compatibility conditions prescribed on $S_0$
  and $\underline{S}_0$ must be consistent with each other. 
  After showing that $\mathcal{S}=0$ on the bifurcate spheres,
  $\mathcal{S}$ is shown to vanish along the entire horizon by using
  the transport equations in the null Bianchi equations. 
\item The second step is to show that $\mathcal{S}=0$ in a
  neighborhood of the horizons. This is done using a unique
  continuation argument where the necessary pseudoconvexity is a
  consequence of the null bifurcate geometry of the horizon. 
\item The third step is to extend the vanishing of $\mathcal{S}$ into
  the stationary region. This is done by constructing a
  $\KillT$-pseudoconvex foliation that mimics the $r$-foliation on
  \KdS. A unique continuation argument in the same style as
  \cite{ionescuUniquenessSmoothStationary2009} then shows that
  $\mathcal{S}$ vanishes in a large neighborhood of the event horizon
  and a large neighborhood of the cosmological horizon. 
\item Finally, the subextremality assumption is used to show that the
  two regions where $\mathcal{S}=0$ intersect and cover the entire
  stationary region of $\mathcal{M}$. 
\end{enumerate}

\subsubsection{A mixed rigidity result}

One can also combine the ideas behind \Cref{thm:main:e2c2:rough} and
\Cref{thm:main:e1c1:rough} to prove a mixed rigidity result. For this,
one assumes compatibility conditions at one bifurcate sphere, which
without loss of generality can be taken to be $S_0$, the bifurcate
sphere of the event horizon, and that the ``cosmological'' part of
$\mathcal{M}$ is a perturbation of \KdS.  Even though there is no
Carter-Robinson equivalent in the $\Lambda>0$ setting, one can
nonetheless make use of the presence a second Killing vectorfield by making
the following observation on subextremal \KdS: $\partial_t$ and
$\partial_{\varphi}$ span a timelike Killing vectorfield. This
observation can be used to show that on exact \KdS, with
$\KillT=\partial_t$ and $\KillPhi=\partial_{\phi}$, $r$ does define a
globally $\{\KillT, \KillPhi\}$-pseudoconvex foliation (for a precise
definition of $\{\KillT, \KillPhi\}$-pseudoconvexity, see
\Cref{def:strict-T-Z-null-convexity} and the surrounding discussion).

The goal is then to construct a second Killing vectorfield
$\underline{\mathbf{K}}$, and a foliation that is globally
$\curlyBrace*{\KillT, \underline{\mathbf{K}}}$-pseudoconvex that
mimics the $\{\KillT, \KillPhi\}$-pseudoconvexity of $r$ on exact
\KdS. This will enable our unique continuation arguments to extend the
vanishing of $\mathcal{S}$ from the event horizon all the way to the
cosmological horizon, proving that the underlying manifold is \KdS.
\begin{theorem}[Rough version of \Cref{thm:main}]
  \label{thm:main:rough}
  A regular stationary cosmological black hole solution
  $(\mathcal{M}, \Metric)$ to ($\Lambda$-EVE) where $\Lambda>0$ with
  regular, connected, null-bifurcate event and cosmological horizons
  with a suitably compatible bifurcate sphere of the event horizon
  $S_0$, which moreover possesses a Mars-Simon tensor $\mathcal{S}$
  that is globally small\footnote{The global smallness assumption on
    $\mathcal{S}$ is not necessary. It suffices for the proof to
    assume that $\mathcal{S}$ is small in a specific region of
    $\mathcal{M}$. For a more precise condition, see \ref{ass:E2}.},
  i.e. $\abs*{\mathcal{S}}\ll1$, contains a stationary region that is
  isometric to that belonging to a member of the \KdS{} family.
\end{theorem}
The proof of \Cref{thm:main:rough} will make use of three main unique
continuation arguments. We outline the main steps below (see
\Cref{fig:E1C2}).
\begin{figure}[ht]
  \centering
  \begin{minipage}[t]{0.24\linewidth}
    \centering
    \begin{tikzpicture}[scale=0.27,every node/.style={scale=1.0}]


  \def \s{5} 
  \def \exts{0.1} 
  \def \t{0.3}
  \def \Tlen{1}
  \def \ulen{0.7}
  \def \offset{0}
  \def \initExt{-0.05}

  \coordinate (tInf) at (0,\s); 
  \coordinate (EventZero) at (-\s,0); 
  \coordinate (CosmoZero) at (\s,0); 
  \coordinate (tNegInf) at (0,-\s);

  \coordinate (tInfL) at (-\offset,\s); 
  \coordinate (EventZeroL) at (-\s-\offset,0); 
  \coordinate (CosmoZeroL) at (\s-\offset,0); 
  \coordinate (tNegInfL) at (-\offset,-\s);
  \coordinate (initExtControlL)  at (-\s-\offset+\initExt,0);

  \coordinate (tInfR) at (\offset,\s); 
  \coordinate (EventZeroR) at (-\s+\offset,0); 
  \coordinate (CosmoZeroR) at (\s+\offset,0); 
  \coordinate (tNegInfR) at (\offset,-\s);
    \coordinate (initExtControlR)  at (\s+\offset-\initExt,0);

    \begin{scope}[decoration={
        markings,
        mark=at position 0.5 with {\arrow[scale=2]{stealth[reversed]}}}]
      \draw[olive,thick,name path=EventFuture,postaction={decorate}] (tInfL) --
      node[pos=0.4,left]{\textcolor{black}{$\EventHorizonFuture$}} (EventZeroL) ;
      
      \draw[olive,thick,name path=EventPast,postaction={decorate}] (tNegInfL) --
      node[pos=0.4,left]{\textcolor{black}{$\EventHorizonPast$}} (EventZeroL) ;

      \draw[purple,thick,name path=CosmoFuture] (tInfR) --
      node[pos=0.4,right]{\textcolor{black}{$\,\CosmologicalHorizonFuture$}} (CosmoZeroR) ;
      
      \draw[purple,thick,name path=CosmoPast] (tNegInfR) --
      node[pos=0.4,right]{\textcolor{black}{$\,\CosmologicalHorizonFuture$}} (CosmoZeroR);
    \end{scope}

  \path[blue,thick,dashed,out=-15,in=-165,shorten <= -10,shorten >= -10]
  (EventZeroL) edge node[pos=0,left]{$\Sigma_0$} node[pos=0.6](TInit){} (CosmoZeroR) ;

  \node[scale=0.5,fill=white,draw,circle,label=above:$i_+$]at(tInf){};
  \node[scale=0.5,fill=white,draw,circle,label=below:$i_-$]at(tNegInf){};
  \node[scale=0.5,fill=olive,draw,circle,label=below:$S_0$]at(EventZeroL){};
  \node[scale=0.5,fill=purple,draw,circle,label=below:$\underline{S}_0$]at(CosmoZeroR){};

\end{tikzpicture}

  \end{minipage}%
  \begin{minipage}[t]{0.24\linewidth}
    \centering
    \begin{tikzpicture}[scale=0.27,every node/.style={scale=1.0}]


  \def \s{5} 
  \def \exts{0.1} 
  \def \t{0.3}
  \def \Tlen{1}
  \def \ulen{0.7}
  \def \offset{0}
  \def \initExt{-0.05}

  \coordinate (tInf) at (0,\s); 
  \coordinate (EventZero) at (-\s,0); 
  \coordinate (CosmoZero) at (\s,0); 
  \coordinate (tNegInf) at (0,-\s);

  \coordinate (tInfL) at (-\offset,\s); 
  \coordinate (EventZeroL) at (-\s-\offset,0); 
  \coordinate (CosmoZeroL) at (\s-\offset,0); 
  \coordinate (tNegInfL) at (-\offset,-\s);
  \coordinate (initExtControlL)  at (-\s-\offset+\initExt,0);

  \coordinate (tInfR) at (\offset,\s); 
  \coordinate (EventZeroR) at (-\s+\offset,0); 
  \coordinate (CosmoZeroR) at (\s+\offset,0); 
  \coordinate (tNegInfR) at (\offset,-\s);
    \coordinate (initExtControlR)  at (\s+\offset-\initExt,0);

  \draw[teal,thick,name path=EventFuture] (tInfL) --
  node[pos=0.4,left]{\textcolor{black}{$\EventHorizonFuture$}} (EventZeroL) ;  

  \draw[teal,thick,name path=EventPast] (tNegInfL) --
  node[pos=0.4,left]{\textcolor{black}{$\EventHorizonPast$}} (EventZeroL) ;

  \draw[magenta,thick,name path=CosmoFuture] (tInfR) --
  node[pos=0.4,right]{\textcolor{black}{$\,\CosmologicalHorizonFuture$}} (CosmoZeroR) ;
  
  \draw[magenta,thick,name path=CosmoPast] (tNegInfR) --
  node[pos=0.4,right]{\textcolor{black}{$\,\CosmologicalHorizonPast$}} (CosmoZeroR);

 \path[fill=olive,fill opacity=0.1,draw=none] (tInfL) ..controls(initExtControlL)..(tNegInfL)-- (EventZeroL) -- cycle;
 \draw[olive,dashed] (tInfL) .. controls(initExtControlL).. (tNegInfL);

  \path[fill=purple,fill opacity=0.2,draw=none] (tInfR) ..controls(initExtControlR)..(tNegInfR)-- (CosmoZeroR) -- cycle;
    \draw[purple,dashed] (tInfR) .. controls(initExtControlR).. (tNegInfR);

  

  \path[blue,thick,dashed,out=-15,in=-165,shorten <= -10,shorten >= -10]
  (EventZeroL) edge  (CosmoZeroR) ;

  \node[scale=0.5,fill=white,draw,circle,label=above:$i_+$]at(tInf){};
  \node[scale=0.5,fill=white,draw,circle,label=below:$i_-$]at(tNegInf){};
  \node[scale=0.5,fill=teal,draw,circle,label=below:$S_0$]at(EventZeroL){};
  \node[scale=0.5,fill=magenta,draw,circle,label=below:$\underline{S}_0$]at(CosmoZeroR){};

\end{tikzpicture}

  \end{minipage}%
  \begin{minipage}[t]{0.24\linewidth}
    \centering
    \begin{tikzpicture}[scale=0.27,every node/.style={scale=1.0}]


  \def \s{5} 
  \def \exts{0.1} 
  \def \t{0.3}
  \def \Tlen{1}
  \def \ulen{0.7}
  \def \offset{0}
  \def \initExt{-0.05}

  \coordinate (tInf) at (0,\s); 
  \coordinate (EventZero) at (-\s,0); 
  \coordinate (CosmoZero) at (\s,0); 
  \coordinate (tNegInf) at (0,-\s);

  \coordinate (tInfL) at (-\offset,\s); 
  \coordinate (EventZeroL) at (-\s-\offset,0); 
  \coordinate (CosmoZeroL) at (\s-\offset,0); 
  \coordinate (tNegInfL) at (-\offset,-\s);
  \coordinate (initExtControlL)  at (-\s-\offset+\initExt,0);

  \coordinate (tInfR) at (\offset,\s); 
  \coordinate (EventZeroR) at (-\s+\offset,0); 
  \coordinate (CosmoZeroR) at (\s+\offset,0); 
  \coordinate (tNegInfR) at (\offset,-\s);
  \coordinate (initExtControlR)  at (\s+\offset-\initExt,0);

  \draw[teal,thick,name path=EventFuture] (tInfL) --
  node[pos=0.4,left]{\textcolor{black}{$\EventHorizonFuture$}} (EventZeroL) ;  

  \draw[teal,thick,name path=EventPast] (tNegInfL) --
  node[pos=0.4,left]{\textcolor{black}{$\EventHorizonPast$}} (EventZeroL) ;

  \draw[magenta,thick,name path=CosmoFuture] (tInfR) --
  node[pos=0.4,right]{\textcolor{black}{$\,\CosmologicalHorizonFuture$}} (CosmoZeroR) ;
  
  \draw[magenta,thick,name path=CosmoPast] (tNegInfR) --
  node[pos=0.4,right]{\textcolor{black}{$\,\CosmologicalHorizonPast$}} (CosmoZeroR);


  \path[fill=teal,fill opacity=0.1,draw=none] (tInfL) ..controls(initExtControlL)..(tNegInfL)-- (EventZeroL) -- cycle;
  \draw[teal,dashed] (tInfL) .. controls(initExtControlL).. (tNegInfL);
  \path[fill=magenta,fill opacity=0.1,draw=none] (tInfR) ..controls(initExtControlR)..(tNegInfR)-- (CosmoZeroR) -- cycle;
  \draw[magenta,dashed] (tInfR) .. controls(initExtControlR).. (tNegInfR);

  \path[fill=olive,fill opacity=0.1,draw=none] (tInfL)..controls(initExtControlL)..(tNegInfL) to [out=105,in=-105](tInf);
  \draw[olive, dashed] (tNegInfL) to [out=105,in=-105](tInf);
  \path[fill=purple,fill opacity=0.2,draw=none] (tInfR) ..controls(initExtControlR)..(tNegInfR) to [out=75,in=-75] (tInf);
  \draw[purple,dashed] (tNegInfR) to [out=75,in=-75] (tInf);

  

  \path[blue,thick,dashed,out=-15,in=-165,shorten <= -10, shorten >= -10]
  (EventZeroL) edge  (CosmoZeroR) ;

  \node[scale=0.5,fill=white,draw,circle,label=above:$i_+$]at(tInf){};
  \node[scale=0.5,fill=white,draw,circle,label=below:$i_-$]at(tNegInf){};
  \node[scale=0.5,fill=teal,draw,circle,label=below:$S_0$]at(EventZeroL){};
  \node[scale=0.5,fill=magenta,draw,circle,label=below:$\underline{S}_0$]at(CosmoZeroR){};

\end{tikzpicture}

  \end{minipage}%
  \begin{minipage}[t]{0.24\linewidth}
    \centering
    \begin{tikzpicture}[scale=0.27,every node/.style={scale=1.0}]


  \def \s{5} 
  \def \exts{0.1} 
  \def \t{0.3}
  \def \Tlen{1}
  \def \ulen{0.7}
  \def \offset{0}
  \def \initExt{-0.05}

  \coordinate (tInf) at (0,\s); 
  \coordinate (EventZero) at (-\s,0); 
  \coordinate (CosmoZero) at (\s,0); 
  \coordinate (tNegInf) at (0,-\s);

  \coordinate (tInfL) at (-\offset,\s); 
  \coordinate (EventZeroL) at (-\s-\offset,0); 
  \coordinate (CosmoZeroL) at (\s-\offset,0); 
  \coordinate (tNegInfL) at (-\offset,-\s);
  \coordinate (initExtControlL)  at (-\s-\offset+\initExt,0);

  \coordinate (tInfR) at (\offset,\s); 
  \coordinate (EventZeroR) at (-\s+\offset,0); 
  \coordinate (CosmoZeroR) at (\s+\offset,0); 
  \coordinate (tNegInfR) at (\offset,-\s);
  \coordinate (initExtControlR)  at (\s+\offset-\initExt,0);

  \draw[teal,thick,name path=EventFuture] (tInfL) --
  node[pos=0.4,left]{\textcolor{black}{$\EventHorizonFuture$}} (EventZeroL) ;  

  \draw[teal,thick,name path=EventPast] (tNegInfL) --
  node[pos=0.4,left]{\textcolor{black}{$\EventHorizonPast$}} (EventZeroL) ;

  \draw[magenta,thick,name path=CosmoFuture] (tInfR) --
  node[pos=0.4,right]{\textcolor{black}{$\,\CosmologicalHorizonFuture$}} (CosmoZeroR) ;
  
  \draw[magenta,thick,name path=CosmoPast] (tNegInfR) --
  node[pos=0.4,right]{\textcolor{black}{$\,\CosmologicalHorizonPast$}} (CosmoZeroR);


  \path[fill=teal,fill opacity=0.1,draw=none] (tInfL) ..controls(initExtControlL)..(tNegInfL)-- (EventZeroL) -- cycle;
  \path[fill=magenta,fill opacity=0.1,draw=none] (tInfR) ..controls(initExtControlR)..(tNegInfR)-- (CosmoZeroR) -- cycle;

  \path[fill=teal,fill opacity=0.1,draw=none] (tInfL)..controls(initExtControlL)..(tNegInfL) to [out=105,in=-105](tInf);
  \draw[teal, dashed] (tNegInfL) to [out=105,in=-105](tInf);
  
  \path[fill=magenta,fill opacity=0.1,draw=none] (tInfR) ..controls(initExtControlR)..(tNegInfR) to [out=75,in=-75] (tInf);
  \draw[magenta,dashed] (tNegInfR) to [out=75,in=-75] (tInf);

  \path[fill=olive,fill opacity=0.15,draw=none] (tInf) to [bend left=30] (tNegInf) to [bend left = 15] (tInf);
  \draw[olive,dashed] (tInf) to [bend left=30] (tNegInf);
    \path[fill=purple,fill opacity=0.15,draw=none] (tInf) to [bend right=30] (tNegInf) to [bend right = 15] (tInf);
  \draw[purple,dashed] (tInf) to [bend right=30] (tNegInf);

  

  \path[blue,thick,dashed,out=-15,in=-165,shorten <= -10, shorten >= -10]
  (EventZeroL) edge  (CosmoZeroR) ;

  \node[scale=0.5,fill=white,draw,circle,label=above:$i_+$]at(tInf){};
  \node[scale=0.5,fill=white,draw,circle,label=below:$i_-$]at(tNegInf){};
  \node[scale=0.5,fill=teal,draw,circle,label=below:$S_0$]at(EventZeroL){};
  \node[scale=0.5,fill=teal,draw,circle,label=below:$\underline{S}_0$]at(CosmoZeroR){};

\end{tikzpicture}

  \end{minipage}%
  \caption[]{Penrose diagrams describing the main steps in proving
    \Cref{thm:main:rough}. 
    The olive
    green and purple depict the regions of interest in each step for
    extending the vanishing of $\mathcal{S}$ and the Killing
    vectorfield $\underline{\mathbf{K}}$ respectively, and the teal
    and magenta regions are regions where previous steps have already
    shown that $\mathcal{S}=0$ and $\underline{\mathbf{K}}$ extends as
    a Killing vectorfield respectively.}
  \label{fig:E1C2}
\end{figure}
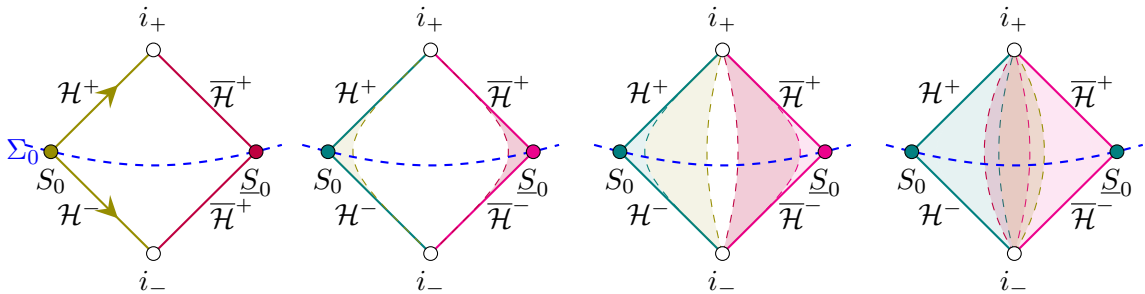

\begin{enumerate}
\item First, $\mathcal{S}=0$ is shown using the compatibility conditions on
  the bifurcate sphere $S_0$ and then extended to the entire event
  horizon using the null Bianchi equations. Independently, a second Killing vectorfield $\underline{\mathbf{K}}$ is constructed on the cosmological horizon. 
\item Next, $\mathcal{S}=0$ is proven to hold in a neighborhood of the
  event horizon and $\underline{\mathbf{K}}$ is extended as a Killing
  vectorfield in a neighborhood of the cosmological horizon. Both of
  these unique continuation arguments only rely on the classical
  pseudoconvexity of the bifurcate spheres that is a consequence of
  the assumed null bifurcate geometry of the horizons, and are done
  independently.
\item Next, $\mathcal{S}=0$ is extended towards the cosmological
  horizon, on a subset $\mathcal{M}_{\mathcal{S}}\subset \mathcal{M}$,
  and $\underline{\mathbf{K}}$ is extended as a Killing vectorfield
  towards the event horizon on a subset
  $\mathcal{M}_{\underline{\mathbf{K}}}\subset \mathcal{M}$.  This
  step relies on constructing a $\KillT$-pseudoconvex foliation on
  $\mathcal{M}$. The \textit{a priori} assumption that $\mathcal{S}$
  is small allows for the construction of a $\KillT$-pseudoconvex
  foliation which is used for the unique continuation argument on
  $\mathcal{M}_{\underline{\mathbf{K}}}$.
\item From the subextremality assumptions on $\mathcal{M}$, it will be
  shown that
  $\mathcal{M}_{\mathcal{S}}\bigcap
  \mathcal{M}_{\underline{\mathbf{K}}}\neq \emptyset$. A third and
  final unique continuation argument can then be used to show that
  $\mathcal{S}=0$ on the entirety of the stationary region
  by using $\{\KillT, \underline{\mathbf{K}}\}$-pseudoconvexity, and
  the fact that the existence of a second Killing vectorfield
  $\underline{\mathbf{K}}$ implies that the Mars-Simon tensor
  satisfies a system of equations of the form
  \begin{equation}
    \label{eq:intro:S-T-K-coupled-system}
    \begin{split}
      \Box_{\Metric}\mathcal{S} &= \AdmissibleRHS(\mathcal{S}, \CovariantDeriv \mathcal{S}),\\
      \LieDerivative_{\KillT}\mathcal{S} &= 0,\\
      \LieDerivative_{\underline{\mathbf{K}}}\mathcal{S} &= 0.
    \end{split}
  \end{equation}
\end{enumerate}

\subsection{Remarks on the literature} 
\label{sec:other-literature}

We briefly review other rigidity results in the literature.
Perturbative rigidity results similar to the main theorem of
\cite{alexakisUniquenessSmoothStationary2010} have been proven in
other settings. The equivalent perturbative $C^{\infty}$ rigidity of
Kerr-Newman as a stationary asymptotically flat black hole solution to the
Einstein-Maxwell equations was proven by Wong and Yu in
\cite{wongNonExistenceMultipleBlackHoleSolutions2014}.  It should also
be remarked that the perturbative black hole rigidity problem is
related to the black hole stability problem. In fact, perturbative
black hole rigidity is essentially a corollary to black hole
stability. After all, if a stationary black hole is stable, then it
can not have any nearby stationary solutions. This link has been
explicitly made by Hintz in \cite{hintzUniquenessKerrNewman2018},
which used the stability of slowly-rotating Kerr-Newman-de Sitter
proven in \cite{hintzNonlinearStabilityKerr2018} to prove the
perturbative rigidity of slowly-rotating Kerr-Newman-de Sitter. While
black hole stability has seen remarkable progress in recent years
within the slowly-rotating and non-rotating categories for both
$\Lambda=0$ as well as $\Lambda>0$
\cite{dafermosNonlinearStabilitySchwarzschild2021,
  klainermanGlobalNonlinearStability2020,
  klainermanKerrStabilitySmall2023,
  fangNonlinearStabilitySlowlyRotating2022,
  fangLinearStabilitySlowlyRotating2022,hintzGlobalNonLinearStability2018},
no nonlinear stability results are currently known outside of the
slowly-rotating regime. Note that the rigidity results discussed thus
far do not feature any restriction on the rotation parameter.

There are fewer results in the literature using the Mars-Simon tensor
to characterize a black hole spacetime.  A notable example is the work
by Wong in \cite{wongUniquenessKerrNewmanBlack2009}, where the author
adapted the Mars-Simon tensor to Kerr-Newman and used an approach like
that pioneered in \cite{ionescuUniquenessSmoothStationary2009} to
prove a similar uniqueness result for Kerr-Newman as a stationary and
asymptotically flat solution to the Einstein-Maxwell equations.

\subsection{Comparing Kerr and \KdS{} and the role of asymptotic flatness}
\label{sec:intro-AF}

A key difference between the settings of Kerr rigidity in
\Cref{conjecture:Kerr} and \KdS{} rigidity in \Cref{conjecture:KdS} is
the role that asymptotic flatness plays in the conjectures. As
discussed in \Cref{sec:intro:stationarity}, asymptotic flatness plays
a key role in defining the appropriate notion of stationarity in the
setting of Kerr rigidity. Its absence thus opens up some subtlety in
how stationarity should be defined in the setting of \KdS{} rigidity.

Asymptotic flatness also however plays a role similar to those of the
compatibility conditions used in \Cref{thm:main:e1c1:rough}. To see
this, observe that upon initial evaluation, it appears as if the Kerr
rigidity results discussed are one-sided rigidity results, where
information from the event horizon, whether that be the vanishing of
the Mars-Simon tensor $\mathcal{S}$, or the existence of an
axisymmetry-generating vectorfield, is propagated into the domain of
exterior communication, whereas the \KdS{} rigidity results are
two-sided, where information is propagated inwards into the stationary
region from both the event horizon and the cosmological horizon. It is
true that in the Kerr rigidity results, all the unique continuation
results are only in one direction (generally in the $r$ increasing
direction), while in the results on \KdS{} below, the results are all proved
with at least two unique continuation results in both the
$r$-increasing and the $r$-decreasing direction. This initial
observation masks the fact that a substantial amount of information
is hidden in the asymptotic flatness assumption. In particular, the
asymptotic flatness assumption gives both a notion of the mass and the
angular momentum of the spacetime. This is notable in that even in the
Kerr uniqueness result in
\cite{ionescuUniquenessSmoothStationary2009}, the compatibility
condition on the bifurcate sphere $S_0$ cannot be freely
prescribed. Rather, the bifurcate sphere $S_0$ must be compatible with
the ADM mass inferred from the asymptotic flatness assumption. In the
case of \KdS{} uniqueness, since there is no asymptotic flatness
assumption, rather than requiring the compatibility condition on $S_0$
to be consistent with the asymptotic flatness assumption, it is
instead necessary to be consistent with the compatibility conditions
on $\underline{S}_0$.

\subsection{Overview}

A brief overview of the structure of the paper is provided
below. \Cref{sec:main theorem} lists in detail the main assumptions
and the main results of this paper. \Cref{sec:geometric-prelim}
details the geometric setup and an overview of the main definitions
regarding pseudoconvexity and Carleman estimates that will be used in
what follows, and introduces the \KdS{} spacetime
itself. \Cref{sec:Lambda-stationary-spacetimes} introduces the main
machinery of stationary solutions to ($\Lambda$-EVE) and in particular
the Mars-Simon tensor $\mathcal{S}$. \Cref{sec:S-small-consequences}
presents some results that are consequences of the smallness of
$\mathcal{S}$, including some critical properties of the function $y$
(which will play the role of the radial coordinate on $\mathcal{M}$)
are proven. These are used in \Cref{sec:pseudoconvexity}, where the
pseudoconvexity properties of $y$ are proven.

\Cref{sec:MST:first} is then concerned with showing that
$\mathcal{S}=0$ and proving \Cref{thm:main:e1c1}, \Cref{sec:HawkingVF}
with the construction and extension of a second Killing vectorfield
and proving \Cref{thm:main:e2c2}, and \Cref{sec:MST:global} with
proving \Cref{thm:main}.

\subsection{Acknowledgments}

The author acknowledges support from the Deutsche
Forschungsgemeinschaft (DFG, German Research Foundation) through
Germany’s Excellence Strategy EXC 2044 390685587, Mathematics
M\"{u}nster: Dynamics–Geometry–Structure, from the Alexander von
Humboldt Foundation in the framework of the Alexander von Humboldt
Professorship endowed by the Federal Ministry of Education and
Research, and from NSF award DMS-2303241. The author would also like
to thank Gustav Holzegel for many insightful discussions.

\section{The main theorem}
\label{sec:main theorem}

We introduce the main assumptions before detailing the main theorems
in the paper. 

\subsection{Main assumptions}
\label{sec:assumptions}

\paragraph{Stationarity}

The main objects of study in this paper will be stationary,
cosmological black holes that are $3+1$ dimensional space-times
$(\Manifold, \Metric)$ that are smooth, strongly causal, orientable,
time-oriented, connected, paracompact $C^{\infty}$ solutions without
boundary of the Einstein vacuum equations with cosmological constant
$\Lambda>0$, which are crucially also \emph{stationary}.  More
precisely, we assume the existence of a smooth Killing vectorfield
$\KillT\in T\Manifold$ which generates a 1-parameter group of
isometries $\Psi_t$. We also assume that the stationary-generating
vectorfield $\KillT$ generates a \emph{regular} complex self-dual
$2$-form
\begin{equation*}
  \mathcal{F}\vcentcolon= F_{\alpha\beta} + \ImagUnit F^{*}_{\alpha\beta}, \qquad
  F_{\alpha\beta} \vcentcolon= \CovariantDeriv_{\alpha}\KillT_{\beta},
\end{equation*}
in the sense that $\mathcal{F}^2\neq 0$ at all points
$p\in \mathcal{M}$.

Moreover, we assume the existence of a smooth space-like slice
$\Sigma_0$ such that $\Manifold$ is the maximal globally hyperbolic
extension of $\Sigma_0$, there exists a diffeomorphism
\begin{equation}
  \label{eq:Phi0:def}
  \Phi_0:E_{\frac{1}{2},\frac{5}{2}}\to \Sigma_0,
\end{equation}
where
\begin{equation}
  \label{eq:E-notation:def}
  \begin{split}
    E_{r_1,r_2}\vcentcolon= \curlyBrace*{x\in\Real^3:r_1<\abs*{x}<r_2},    
  \end{split}  
\end{equation}
so that $\Sigma_0$ is diffeomorphic to
$\{x\in \Real^3: \frac{1}{2}<\abs*{x}<\frac{5}{2}\}$, and that every
orbit of $\KillT$ is complete and intersects
$\Sigma_0$. Moreover, denoting by $T_0$ the future-directed unit
vector normal to $\Sigma_0$, we assume that
\begin{equation}
  \label{eq:assumptions:Kill-T-T0}
  \abs*{\Metric(\KillT, T_0)} > 0,\qquad \text{ on } \Sigma_0\bigcap \mathbf{E}.
\end{equation}
We also denote
\begin{equation}
  \label{eq:Sigma-notation:def}
  \begin{split}
    \Sigma_r = \Phi_0(E_{r,\frac{5}{2}}),\qquad
    \underline{\Sigma}_r = \Phi_0(E_{\frac{1}{2},r}).
  \end{split}
\end{equation}

\paragraph{Smooth bifurcate sphere}

We assume the existence of two pairs of null,
nonexpanding\footnote{Recall that a null hypersurface $\mathcal{H}$ is
  non-expanding if the null second fundamental form of $\mathcal{H}$
  is trace-free.}  hypersurfaces,
$\EventHorizon = \EventHorizonFuture\bigcup \EventHorizonPast$ , and
$\CosmologicalHorizon= \CosmologicalHorizonFuture\bigcup
\CosmologicalHorizonPast$ , where $\EventHorizonFuture$ and
$\EventHorizonPast$ (resp. $\CosmologicalHorizonFuture$ and
$\CosmologicalHorizonPast$) are smooth imbedded null hypersurfaces
smoothly intersecting transversely on a $2$-surface $S_0$,
(resp. $\underline{S}_0$) with the topology of the standard sphere,
that are diffeomorphic to $S_0\times (-1,1)$ (respectively
$\underline{S}_0\times (-1,1)$), and where
\begin{equation*}
  \EventHorizon = \partial(\mathcal{I}^+(\CosmologicalHorizonPast)\cap \mathcal{I}^-(\CosmologicalHorizonFuture)),\qquad
    \CosmologicalHorizon = \partial(\mathcal{I}^+(\EventHorizonPast)\cap \mathcal{I}^-(\EventHorizonFuture)).
\end{equation*}
In particular, we assume that $S_0$ agrees with the sphere of radius 1
in $\Real^3$ under the identification of $\Sigma_0$ with
$\{x\in\Real^3:\frac{1}{2}<\abs*{x}<\frac{5}{2}\}$, and similarly,
that $\underline{S}_0$ agrees with the sphere of radius 2 in $\Real^3$
under the same identification of $\Sigma_0$.

Moreover, we assume that both $\EventHorizon$ and
$\CosmologicalHorizon$ are non-degenerate, in the sense that
they have nonzero surface gravity. We also assume that the
vectorfield $\KillT$ is tangent to
$\EventHorizonPast,\EventHorizonFuture,\CosmologicalHorizonPast$, and
$\CosmologicalHorizonFuture$, and does not vanish identically
on\footnote{$\KillT$ must either vanish identically on $S_0$ and
  $\underline{S}_0$, or only at a finite number of isolated points.}
$S_0$, and on $\underline{S}_0$.

We also make some regularity assumptions needed to ensure we do not
have a degenerate spacetime. These are given by the assumption that
$Q \mathcal{F}^2$ and $Q\mathcal{F}^2- 4\Lambda$ are not uniformly
zero on $S_0$ and $\underline{S}_0$, where $Q$ is as defined in
\Cref{eq:MST:Q:def}.

Finally, we define the stationary region $\mathbf{E}$, which serves as
the analogue of the domain of exterior communication of Kerr by
\begin{equation}
  \label{eq:domain-of-ext-communication:def}
  \mathbf{E} \vcentcolon= \mathcal{I}^{+}(\EventHorizonPast)\bigcap\mathcal{I}^-(\EventHorizonFuture) \bigcap\mathcal{I}^{+}(\CosmologicalHorizonPast)\bigcap\mathcal{I}^-(\CosmologicalHorizonFuture),
\end{equation}
and we assume that $\KillT$ is nonvanishing on $\mathbf{E}$.

\paragraph{Subextremality condition}

We first define the triplet
$(b,c,k)$ by the relations
\begin{equation}
  \label{eq:b-c-k:def}
  \begin{split}
    b \vcentcolon={}& -\ImagUnit \frac{36 Q \left( \mathcal{F}^2 \right)^{\frac{5}{2}}}{\left( Q \mathcal{F}^2 - 4\Lambda \right)^3},\\
    c \vcentcolon={}& - \Metric(\KillT,\KillT) - \Re\left(\frac{6 \mathcal{F}^2(Q\mathcal{F}^2+2\Lambda)}{(Q\mathcal{F}^2-4\Lambda)^2}\right),\\
    k \vcentcolon={}& \abs*{\frac{36\mathcal{F}^2}{(Q\mathcal{F}^2-4\Lambda)^2}}\CovariantDeriv_{\alpha}z\CovariantDeriv^{\alpha}z
                      + cz^2 + \frac{\Lambda}{3}z^4,
  \end{split}
\end{equation}
where we recall the definition of $Q$ in \Cref{eq:MST:Q:def}, and the
definition of $R$ in \Cref{eq:R:def}, and where $y, z\in \Real$ such
that
$y+\ImagUnit z = \frac{6 \ImagUnit
  \sqrt{\mathcal{F}^2}}{Q\mathcal{F}^2-4\Lambda}$ (see
\Cref{eq:y-z:def} for a more precise definition of $y$ and $z$). For two distinguished points $p_0\in S_0, \underline{p}_0\in \underline{S}_0$, define
\begin{equation}
  \label{eq:b-c-k-S0-SBar0:def}
  (b_{S_0}, c_{S_0}, k_{S_0}) =  (b(p_0), c(p_0), k(p_0)),\qquad
  (b_{\underline{S}_0}, c_{\underline{S}_0}, k_{\underline{S}_0}) =  (b(\underline{p}_0), c(\underline{p}_0), k(\underline{p}_0)).
\end{equation}

We then make the following subextremality conditions. These
conditions guarantee that the region under consideration is actually
the stationary region of a subextremal \KdS{} black hole. To this end,
we assume that
$(b_{S_0}, c_{S_0}, k_{S_0})\in \Real^+\times \Real^+\times \Real^+$
and that
\begin{equation}
  \label{eq:S0-y-Delta-assumptions}
  y(p_0)>0,  \qquad
  \Delta(p_0) \le  0,\qquad
  \partial_y\Delta(p_0) > 0
  ,
\end{equation}
where
\begin{equation}
  \label{eq:Delta:def}
  \Delta \vcentcolon= k_{S_0} - b_{S_0} y + c_{S_0}y^2 - \frac{\Lambda}{3}y^4,
\end{equation}
is assumed to have four distinct roots, the largest two of which we
will coin $0<y_{S_0}<y_{\underline{S}_0}$.

We also assume that $(b_{\underline{S}_0}, c_{\underline{S}_0}, k_{\underline{S}_0})\in \Real^+\times \Real^+\times \Real^+$
and that
\begin{equation}
  \label{eq:SBar0-y-DeltaBar-assumptions}
  y(\underline{p}_0)>y(p_0),  \qquad
  \underline{\Delta}(\underline{p}_0) \le  0,\qquad
  \partial_y\underline{\Delta}(\underline{p}_0) < 0
  ,
\end{equation}
where
\begin{equation}
  \label{eq:DeltaBar:def}
  \underline{\Delta} \vcentcolon= k_{\underline{S}_0} - b_{\underline{S}_0} y + c_{\underline{S}_0}y^2 - \frac{\Lambda}{3}y^4,
\end{equation}
is also assumed to have four distinct roots.

Let us observe that
there then exists some $y_{*}$ such that
\begin{equation}
  \label{eq:yStar:def}
  \Delta(y_{*}) = \max_{y(p_0)<y<y(\underline{p}_0)}\Delta(y).
\end{equation}
We then also make the subextremality assumption that $\KillT$ is
timelike at any point $p\in \mathcal{M}$ where $y(p)=y_{*}$.

\paragraph{Technical assumption}

We consider two sets of technical assumptions. The
first are made on the bifurcate sphere of the event horizon $S_0$, or
in a neighborhood of it.
\begin{enumerate}[start=1,label={(\bfseries E\arabic*)}]
\item \label{ass:E1}This is a compatibility condition of $S_0$ with
  some member of the \KdS{} family. There exists a triplet of positive
  real numbers $(M,a,\gamma = \frac{\Lambda a^2}{3})\in \Real^+\times\Real^+\times (0,1)$ such
  that at each point on $p\in S_0$
  \begin{equation}
    \label{eq:compatibility-conditions}
    b(p) = b_{S_0} 
    = \frac{2M}{(1+\gamma)^3}, \qquad
    c(p) = c_{S_0} = 
    \frac{1-\gamma}{(1+\gamma)^2},\qquad
    k(p) = k_{S_0} 
    = \frac{a^2}{(1+\gamma)^4}.
  \end{equation}
  
\item \label{ass:E2}This is a smallness condition on the Mars-Simon
  tensor $\mathcal{S}$ (i.e, an assumption that the underlying
  manifold $\mathcal{M}$ is somehow close to a member of the \KdS{}
  family). For the complete definition of $\mathcal{S}$, see
  \Cref{def:MST}. There exists some $y_{+}'>y_{*}$ such that 
  \begin{equation}
    \label{eq:S-smallness-assumption:E}
    \abs*{\frac{1}{R-J\sigma_0}\KillT^{\sigma}\mathcal{S}_{\mu\nu\sigma\rho}}\lesssim \varepsilon_{\mathcal{S}}\qquad
    \text{on }\underline{\Sigma}_{y_+'}\bigcap \closure \mathbf{E}
    \quad \text{for any }\mu,\nu,\rho\in \{0,1,2,3\}
    ,
  \end{equation}
  for some sufficiently small $\varepsilon_{\mathcal{S}}$, where
  $R, J, \sigma_0$ are as defined in
  \Cref{sec:Lambda-stationary-spacetimes:preliminaries}.
\end{enumerate}

The second set of assumptions are instead made on the bifurcate sphere
of the cosmological horizon $\underline{S}_0$, or in a neighborhood of
it.
\begin{enumerate}[start=1,label={(\bfseries C\arabic*)}]
\item \label{ass:C1}This is a compatibility condition of
  $\underline{S}_0$ with some member of the \KdS{} family.  There
  exists a triplet of positive real numbers
  $(M,a,\gamma= \frac{\Lambda a^2}{3})\in \Real^+\times\Real^+\times
  (0,1)$ such that at each point on $p\in \underline{S}_0$
  \begin{equation}
    \label{eq:compatibility-conditions:C}
    b(p) = b_{\underline{S}_0} 
    = \frac{2M}{(1+\gamma)^3}, \qquad
    c(p) = c_{\underline{S}_0} =
    \frac{1-\gamma}{(1+\gamma)^2},\qquad
    k(p) = k_{\underline{S}_0} = 
    \frac{a^2}{(1+\gamma)^4}.
  \end{equation}
\item \label{ass:C2}This is a smallness condition on the Mars-Simon
  tensor $\mathcal{S}$ (i.e, an assumption that the underlying
  manifold $\mathcal{M}$ is somehow close to a member of the \KdS{}
  family). For the complete definition of $\mathcal{S}$, see
  \Cref{def:MST}. There exists some $y_-'<y_{*}$ such that
  \begin{equation}
    \label{eq:S-smallness-assumption}
    \abs*{\frac{1}{R-J\sigma_0}\KillT^{\sigma}\mathcal{S}_{\mu\nu\sigma\rho}}\lesssim \varepsilon_{\mathcal{S}}\qquad
    \text{on }\Sigma_{y_-'}\bigcap \closure \mathbf{E}
    \quad \text{for any }\mu,\nu,\rho\in \{0,1,2,3\}
    ,
  \end{equation}
  for some sufficiently small $\varepsilon_{\mathcal{S}}$, 
  where $R, J, \sigma_0$ are as defined in
  \Cref{sec:Lambda-stationary-spacetimes:preliminaries}.
\end{enumerate}

\begin{remark}
  We make some remarks comparing our technical assumptions to those
  made in \cite{ionescuUniquenessSmoothStationary2009}. The
  compatibility assumptions in \Cref{eq:compatibility-conditions} are
  in the same spirit as the technical assumption made in (1.6) of
  \cite{ionescuUniquenessSmoothStationary2009} in that they are
  assumptions made on the whole of $S_0$ which relate the relevant
  black hole parameters to geometric quantities related to the
  stationarity of the spacetime (such as $Q$, which is
  $\frac{6}{1-\sigma}$ in
  \cite{ionescuUniquenessSmoothStationary2009}, and
  $\mathcal{F}^2$). 

  The subextremality assumptions in \Cref{eq:S0-y-Delta-assumptions}
  on the other hand are in the same spirit as (1.7) in
  \cite{ionescuUniquenessSmoothStationary2009} in that they are
  consistent with the natural subextremality assumption for Kerr(-de
  Sitter) black holes. The fact that our subextremality assumptions
  are more complicated than those appearing in
  \cite{ionescuUniquenessSmoothStationary2009} should not be
  surprising given the more complicated nature of the characterization
  of subextremality in the $\Lambda>0$ setting. We emphasize that like
  (1.7) in \cite{ionescuUniquenessSmoothStationary2009}, the
  subextremality assumptions are only assumed to hold
  at a point on the bifurcate spheres. The condition that $\KillT$ is timelike
  along $\{y=y_{*}\}$ should be likened to the condition that $\KillT$
  is timelike at spacelike infinity in the asymptotically flat
  case. Indeed, under the assumption that $c_{S_0}>0$, as in the
  compatibility conditions in \ref{ass:E1}, it is clear that
  $\lim_{\Lambda\to 0}y_{*}=\infty$. 
\end{remark}

\subsection{Statement of the main theorem}
\label{sec:main theorem:statement}

We list the main theorems.
\begin{theorem}
  \label{thm:main:e1c1}
  If $(\mathcal{M}, \Metric)$ is a solution to ($\Lambda$-EVE)
  satisfying the assumptions (as described in \Cref{sec:assumptions})
  of stationarity, of a smooth bifurcate sphere, of subextremality,
  and the technical assumptions \ref{ass:E1} and \ref{ass:C1} (for the
  same $M,a,\gamma$ between them), then the stationary region
  $\mathbf{E}$ of $(\mathcal{M}, \Metric)$ is isometrically
  diffeomorphic to the stationary region of a \KdS{} spacetime with
  black hole parameters $(M,a)$.
\end{theorem}
\begin{remark}
  \label{remark:thm:e1c1}
  \Cref{thm:main:e1c1} represents the equivalent of
  \cite{ionescuUniquenessSmoothStationary2009} in the $\Lambda=0$
  setting, the main difference being the appearance of an extra
  condition on $\underline{S}_0$ which is not present in the
  $\Lambda=0$ setting. We remark though that in the $\Lambda=0$
  setting, there is implicitly some compatibility condition being
  given by the asymptotic flatness assumption itself. In particular,
  the ADM mass appearing in the asymptotic flatness assumption must
  match the $M$ used to define the compatibility condition on $S_0$
  (which in \cite{ionescuUniquenessSmoothStationary2009} is coined the
  ``technical condition''). See also the discussion in
  \Cref{sec:intro-AF}.
\end{remark}

\begin{theorem}
  \label{thm:main:e2c2}
  If $(\mathcal{M}, \Metric)$ is a solution to ($\Lambda$-EVE)
  satisfying the assumptions (as described in \Cref{sec:assumptions})
  of stationarity, of a smooth bifurcate sphere, of subextremality,
  and the technical assumptions \ref{ass:E2} and \ref{ass:C2} then
  there exists a second, rotational, Killing vectorfield on
  $\mathbf{E}$.
\end{theorem}

\begin{remark}
  Observe that \Cref{thm:main:e2c2} in particular does not yield the
  conclusion that the underlying manifold of interest is
  isometrically diffeomorphic to a member of the \KdS{} family, in contrast to
  \cite{alexakisUniquenessSmoothStationary2010}, where equivalent
  assumptions in the asymptotically flat case did imply that the
  underlying manifold was isometrically diffeomorphic to a member of the Kerr
  family. This distinction is primarily due to the lack of a
  Carter-Robinson type theorem in the $\Lambda\neq0$ setting.
\end{remark}

We also prove a rigidity result with mixed assumptions.
\begin{theorem}
  \label{thm:main}
  If $(\mathcal{M}, \Metric)$ is a solution to ($\Lambda$-EVE)
  satisfying the assumptions (as described in \Cref{sec:assumptions})
  of stationarity, of a smooth bifurcate sphere, of subextremality,
  and the technical assumptions \ref{ass:E1} and \ref{ass:C2},
  the stationary region $\mathbf{E}$ of
  $\left( \mathcal{M}, \Metric \right)$ is isometrically diffeomorphic
  to the stationary region of a \KdS{} spacetime with black hole
  parameters $(M,a)$ as specified in the assumed compatibility
  assumptions in \ref{ass:E1}.
\end{theorem}
\begin{remark}
  We remark that \Cref{thm:main} remains true if assumptions
  \ref{ass:E2} and \ref{ass:C1} are made in place of \ref{ass:E1} and
  \ref{ass:C2}. The main difference between the proof of
  \Cref{thm:main} and the proofs of \Cref{thm:main:e1c1} and
  \Cref{thm:main:e2c2} is that to prove \Cref{thm:main}, we need to
  use $\{\KillT, \underline{\mathbf{K}}\}$-pseudoconvexity, whereas
  only $\KillT$-pseudoconvexity is needed to prove
  \Cref{thm:main:e1c1} and \Cref{thm:main:e2c2}.
\end{remark}

We also make some final remarks regarding where the assumptions are
used.
\begin{enumerate}
\item Similar to the case of
  \cite{ionescuUniquenessSmoothStationary2009}, the compatibility
  condition on $S_0$ (and $\underline{S}_0$) is only used to show that
  $\rho(\mathcal{S})=0$ on $S_0$ (and $\underline{S}_0$
  respectively). This point is crucial for showing that
  $\mathcal{S}=0$ along the event horizon (which is then extended to
  the rest of the spacetime using unique continuation arguments).
\item The proofs of each of the main theorems involves at least two
  unique continuation arguments.  The subextremality condition plays a
  critical role in ensuring that the unique continuation arguments we
  use overlap and cover the entirety of the stationary region. We
  recall from the discussion in \Cref{sec:intro:stationarity} that
  with our definition of subextremality and stationarity, $\KillT$
  should be thought of as $\KillT=\partial_t + \frac{a}{r_{*}^2+a^2}\partial_{\varphi}$
  on \KdS, where $r_{*}$ is the unique maximizer of
  $\Delta = (r^2+a^2)\left(1-\frac{\Lambda}{3}r^2\right)-2Mr$. Observe
  that
  \begin{equation*}
    \lim_{\Lambda\to 0}r_{*}=\infty,\qquad \lim_{\Lambda\to 0}\KillT = \partial_t.
  \end{equation*}
  This also turns out to be convenient for our proof, and is related
  to the absence of $\KillT$-trapped null geodesics in subextremal
  \KdS{} (see
  \cite{petersenStationarityFredholmTheory2024,petersenWaveEquationsKerr2024}). However,
  it should be said that, unlike in the asymptotically flat case,
  where the choice of $\KillT$ is canonically given by the asymptotic
  flatness assumption, there is nothing canonical about the present
  choice of $\KillT$.
\end{enumerate}

\section{Geometric preliminaries}
\label{sec:geometric-prelim}

\subsection{Notation}

Throughout the article, we will use Einstein summation notation. Latin
$a,b,c,d$ indices will be used to indicate indices over the spheres,
while other Latin indices will be used to indicate spatial indices.
Greek indices will be used to indicate spacetime indices. We will use
$\CovariantDeriv$ to refer to the Levi-Civita connection,
$\nabla_{\mu}$ to refer to the horizontal covariant derivative (we
recall the definition of the horizontal covariant derivative in
\Cref{def:horizontal-covariant-derivative}), and $\nabla$ without
indices to refer to derivatives in the horizontal directions. We will
generally use underlined quantities (such as $\underline{S}_0$) to
refer to quantities linked to the cosmological horizon, and
nonunderlined quantities (such as $S_0$) to refer to objects connected
to the event horizon. We will use $\cdot$ to denote contraction with
respect to $\Metric$. 

\subsection{Double-null foliation near \texorpdfstring{$S_0$, $\underline{S}_0$}{the bifurcate spheres}}

We define two pairs of optical functions $(u_-, u_+)$, and
$(\underline{u}_+,\underline{u}_-)$ in a neighborhood $\mathbf{O}$ of
the bifurcate sphere $S_0$ and in a neighborhood $\underline{\mathbf{O}}$ of $\underline{S}_0$
respectively.  To this end, we choose a future directed null pair
$(L_+,L_-)$ along $S_0$ (resp. $(\underline{L}_-,\underline{L}_+)$
along $\underline{S}_0$) so that $L_+, \underline{L}_-$ are future
oriented while $L_-, \underline{L}_+$ are past oriented, and on $S_0$
(respectively $\underline{S}_0$ for the underlined quantities),
\begin{equation}
  \label{eq:L+L-:normalization-on-horizon}
  \begin{gathered}
    \Metric(L_-,L_-) = \Metric(L_+,L_+) =0,\qquad
    \Metric(L_+,T_0)=-1,\qquad \Metric(L_+,L_-)=2,\\
    \Metric(\underline{L}_-,\underline{L}_-) = \Metric(\underline{L}_+,\underline{L}_+) =0,\qquad
    \Metric(\underline{L}_+,T_0)=-1,\qquad \Metric(\underline{L}_+,\underline{L}_-)=2,
  \end{gathered}
\end{equation}
see \Cref{fig:penrose}.
We extend $L_+$ (resp. $L_-$ , $\underline{L}_-$, $\underline{L}_+$)
along the null geodesic generators of $\EventHorizonFuture$
(resp. $\EventHorizonPast$ , $\CosmologicalHorizonFuture$,
 $\CosmologicalHorizonPast$) by parallel transport. We then define
$u_-$ (resp $u_+, \underline{u}_+, \underline{u}_-$) along
$\EventHorizonFuture$ (resp. $\EventHorizonPast, \CosmologicalHorizonFuture, \CosmologicalHorizonPast$) as the solutions to
the system
\begin{equation*}
  \begin{cases}
    L_+(u_-)=1,\\
    u_-(S_0)=0,
  \end{cases}\qquad
  \begin{cases}
    L_-(u_+)=1,\\
    u_+(S_0)=0.
  \end{cases} \qquad
  \begin{cases}
    \underline{L}_-(\underline{u}_+)=1,\\
    \underline{u}_-(S_0)=0,
  \end{cases}\qquad
  \begin{cases}
    \underline{L}_+(\underline{u}_-)=1,\\
    \underline{u}_+(S_0)=0.
  \end{cases}
\end{equation*}
We now define $S_{u_-}$ (resp. $S_{u_+}$,
$\underline{S}_{\underline{u}_-}$, $\underline{S}_{\underline{u}_+}$)
be the level surfaces of $u_-$ (resp. $u_+$ , $\underline{u}_-$,
$\underline{u}_+$) along $\EventHorizonFuture$
(resp. $\EventHorizonPast$ , $\CosmologicalHorizonFuture$,
$\CosmologicalHorizonPast$). We define $L_-$ at every point of
$\EventHorizonFuture$ (resp. $L_+$ at every point of
$\EventHorizonPast$ , $\underline{L}_-$ at every point of
$\CosmologicalHorizonPast$, and $\underline{L}_{+}$ at every point of
$\CosmologicalHorizonFuture$) as the unique past-directed
(resp. future-directed, past-directed, and future-directed)
vector-field orthogonal to the surface $S_{u_-}$ (resp $S_{u_+}$ ,
$\underline{S}_{\underline{u}_-}$, and
$\underline{S}_{\underline{u}_+}$) passing through that point and such
that $\Metric(L_+,L_-)=2$
(resp. $\Metric(\underline{L}_-,\underline{L}_+)=2$).

\begin{definition}[Level sets of $u_-,u_+$, $u_-,u_+,\underline{u}_-,\underline{u}_+$]
  We define the null hypersurface $\EventHorizon_{u_-}$ to be the
  congruence of null geodesics emanating from
  $S_{u_-}\subset \EventHorizonFuture$ in the direction of
  $L_-$. Similarly, we define $\EventHorizon_{u_+}$ to be the
  congruence of null geodesics emanating from
  $S_{u_+}\subset \EventHorizonPast$.  Similarly, we define
  $\CosmologicalHorizon_{\underline{u}_-}$ to be the congruence of null
  geodesics emanating from
  $\underline{S}_{\underline{u}_-}\subset
  \CosmologicalHorizonPast$. Similarly, we define
  $\CosmologicalHorizon_{\underline{u}_+}$ to be the congruence of null
  geodesics emanating from
  $\underline{S}_{\underline{u}_+}\subset \CosmologicalHorizonFuture$.
  $\EventHorizon_{u_-}$,
  $\EventHorizon_{u_+} ,\CosmologicalHorizon_{\underline{u}_-}$, and
   $\CosmologicalHorizon_{\underline{u}_+}$ are the level sets of
  $u_-, u_+ , \underline{u}_-,\underline{u}_+
  $ respectively. 
\end{definition}
\begin{remark}
  The congruences defined above are well-defined in a neighborhood
  $\mathbf{O}\subset \Manifold$ of $S_0$  and
   $\underline{\mathbf{O}}\subset \Manifold$ of $\underline{S}_0$
  . 
\end{remark}
We moreover arrange that $u_-,u_+ ,\underline{u}_-,\underline{u}_+
$ are positive in the stationary region $\mathbf{E}$. By
construction, $u_-,u_+ ,\underline{u}_-,\underline{u}_+
$ satisfy
\begin{equation}
  \label{eq:geodesic-u-uBar-qtys}
  \Metric^{\alpha\beta}\partial_{\alpha}u_{\pm}\partial_{\beta}u_{\pm}
   = \Metric^{\alpha\beta}\partial_{\alpha}\underline{u}_{\pm}\partial_{\beta}\underline{u}_{\pm}
  = 0. 
\end{equation}
\begin{definition}
  Define
  \begin{align*}
    \Omega&\vcentcolon= \Metric^{\alpha\beta}\partial_{\alpha}u_+\partial_{\beta}u_-,\\
    \underline{\Omega}&\vcentcolon= \Metric^{\alpha\beta}\partial_{\alpha}\underline{u}_+\partial_{\beta}\underline{u}_-.
  \end{align*}
  In view of the preceding choices, we have that
  \begin{equation}
    \label{eq:u-Omega-horizon}
    \begin{gathered}    
      \evalAt*{u_+}_{\EventHorizonFuture} = \evalAt*{u_-}_{\EventHorizonPast} = 0,\qquad
      \evalAt*{\Omega}_{\EventHorizonPast\bigcup \EventHorizonFuture} = \frac{1}{2},\\
      \evalAt*{\underline{u}_+}_{\CosmologicalHorizonFuture} = \evalAt*{\underline{u}_-}_{\CosmologicalHorizonPast} = 0,\qquad
      \evalAt*{\underline{\Omega}}_{\CosmologicalHorizonPast\bigcup \CosmologicalHorizonFuture} = \frac{1}{2}.
    \end{gathered}
  \end{equation}
\end{definition}
\begin{definition}
  We can then extend $(L_+, L_-)$ off the event horizon
  $\EventHorizon$ and $(\underline{L}_+, \underline{L}_-)$ off the
  cosmological horizon $\CosmologicalHorizon$ by
  \begin{align*}
    L_+ = \Metric^{\alpha\beta}\partial_{\alpha}u_+\partial_{\beta},\qquad
    L_- = \Metric^{\alpha\beta}\partial_{\alpha}u_-\partial_{\beta},\qquad
    \underline{L}_+ = \Metric^{\alpha\beta}\partial_{\alpha}\underline{u}_+\partial_{\beta},\qquad
    \underline{L}_- = \Metric^{\alpha\beta}\partial_{\alpha}\underline{u}_-\partial_{\beta}.
  \end{align*}
  In particular, this extension yields that
  \begin{gather*}
    \Metric(L_+,L_+) = \Metric(L_-,L_-)=0,\qquad \Metric(L_+,L_-) = \Omega,\\
    \Metric(\underline{L}_+,\underline{L}_+) = \Metric(\underline{L}_-,\underline{L}_-)=0,\qquad \Metric(\underline{L}_+,\underline{L}_-) = \underline{\Omega}.
  \end{gather*}
\end{definition}
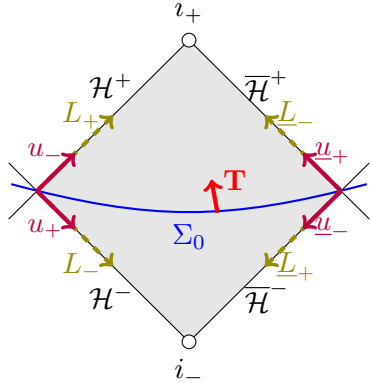
\begin{figure}[ht]
  \centering
  \begin{tikzpicture}[scale=0.4,every node/.style={scale=1.0}]


  \def \s{5} 
  \def \exts{0.1} 
  \def \t{0.3}
  \def \Tlen{1}
  \def \ulen{1.25}

  \coordinate (tInf) at (0,\s); 
  \coordinate (EventZero) at (-\s,0); 
  \coordinate (CosmoZero) at (\s,0); 
  \coordinate (tNegInf) at (0,-\s);

  \draw[shorten >= -15,name path=EventFuture] (tInf) --
  node[pos=0.3,left]{$\EventHorizonFuture$} (EventZero) ;  
  \draw[shorten >= -15,name path=CosmoFuture] (tInf) --
  node[pos=0.3,right]{$\CosmologicalHorizonFuture$} (CosmoZero) ; 
  \draw[shorten >= -15,name path=EventPast] (tNegInf) --
  node[pos=0.3,left]{$\EventHorizonPast$} (EventZero) ; 
  \draw[shorten >= -15,name path=CosmoPast] (tNegInf) --
  node[pos=0.3,right]{$\CosmologicalHorizonPast$} (CosmoZero);



  \path[fill=black,opacity=0.1] (EventZero) --(tInf) --(CosmoZero) --(tNegInf);
  \path[name path=EventLowerBound] (-1.6*\s, 0.6*\s)-- (-\s, 0)--
  (tInf);
  \path[name path=CosmoLowerBound] (tInf) -- (\s, 0) -- (1.6*\s, 0.6*\s);

  \path[blue,thick,out=-15,in=-165,shorten >= -10, shorten <= -10]
  (EventZero) edge node[pos=0.5,below]{$\Sigma_0$} node[pos=0.6](TInit){} (CosmoZero) ;

    \draw[->,ultra thick,red] (TInit.center) -- ++ (-0.2 * \Tlen,1.1 *\Tlen)
    node[right]{$\KillT$};
         \draw[->,ultra thick,dashed,olive] (EventZero) -- ++ (2*\ulen,2*\ulen)
  node[left]{$L_+$};
   \draw[->,ultra thick,dashed,olive] (EventZero) -- ++ (2*\ulen,-2*\ulen)
   node[left]{$L_-$};
  \draw[->,ultra thick,dashed,olive] (CosmoZero) -- ++ (-2*\ulen,-2*\ulen)
  node[right]{$\underline{L}_+$};
   \draw[->,ultra thick,dashed,olive] (CosmoZero) -- ++ (-2*\ulen,2*\ulen)
   node[right]{$\underline{L}_-$};
  \draw[->,ultra thick,purple] (EventZero) -- ++ (\ulen,\ulen)
  node[left]{$u_-$};
   \draw[->,ultra thick,purple] (EventZero) -- ++ (\ulen,-\ulen)
   node[left]{$u_+$};
  \draw[->,ultra thick,purple] (CosmoZero) -- ++ (-\ulen,-\ulen)
  node[right]{$\underline{u}_-$};
   \draw[->,ultra thick,purple] (CosmoZero) -- ++ (-\ulen,\ulen)
   node[right]{$\underline{u}_+$};

  \node[scale=0.5,fill=white,draw,circle,label=above:$i_+$]at(tInf){};
  \node[scale=0.5,fill=white,draw,circle,label=below:$i_-$]at(tNegInf){};

\end{tikzpicture}

  \caption{A Penrose diagram describing the spacetime and the double null frames attached to each of the horizons.}
  \label{fig:penrose}
\end{figure}
\begin{definition}
  \label{def:O-set}
  We define the sets
  \begin{align*}
    \mathbf{O}_{\varepsilon} &\vcentcolon= \curlyBrace*{x\in \mathbf{O}: \abs*{u_-}<\varepsilon, \abs*{u_+} < \varepsilon}. \\
    \underline{\mathbf{O}}_{\varepsilon} &\vcentcolon= \curlyBrace*{x\in \underline{\mathbf{O}}: \abs*{\underline{u}_-}<\varepsilon, \abs*{\underline{u}_+} < \varepsilon}.
  \end{align*}
\end{definition}
Observe that for $\varepsilon_0 > 0$ sufficiently small, we have that
\begin{equation*}
  \begin{gathered}
      \Omega > \frac{1}{4}\text{ in }\mathbf{O}_{\varepsilon_0},\qquad
  \closure\mathbf{O}_{\varepsilon_0}\subset \mathbf{O},\\
  \underline{\Omega} > \frac{1}{4}\text{ in }\underline{\mathbf{O}}_{\varepsilon_0},\qquad
  \closure\underline{\mathbf{O}}_{\varepsilon_0}\subset \underline{\mathbf{O}}.
  \end{gathered}
\end{equation*}
We also have, for $\varepsilon\le \varepsilon_0$,
\begin{align*}
  \mathbf{O}_{\varepsilon}\bigcap \closure\mathbf{E}
  &= \curlyBrace*{0\le u_-<\varepsilon, 0\le u_+<\varepsilon},\\
  \underline{\mathbf{O}}_{\varepsilon}\bigcap \closure\mathbf{E}
  &= \curlyBrace*{0\le
  \underline{u}_-<\varepsilon, 0\le \underline{u}_+<\varepsilon}.
\end{align*}
If $\phi$ is a smooth function in
$\underline{\mathbf{O}}_{\varepsilon}$ and vanishes on
$\CosmologicalHorizonFuture\bigcap
\underline{\mathbf{O}}_{\varepsilon}$, one can show that there exists
a smooth function $\phi'$ defined on
$\underline{\mathbf{O}}_{\varepsilon}$ such that
\begin{equation}
  \label{eq:smooth-extension:cosmological-future}
  \phi = \underline{u}_-\phi'\qquad \text{on }\underline{\mathbf{O}}_{\varepsilon}.
\end{equation}
Similarly, if $\phi \in C^{\infty}(\underline{\mathbf{O}}_{\varepsilon})$ and $\phi$
vanishes on $\CosmologicalHorizonPast\bigcap \underline{\mathbf{O}}_{\varepsilon}$, then
there exists another smooth function $\phi'$ defined on
$\underline{\mathbf{O}}_{\varepsilon}$ such that
\begin{equation}
  \label{eq:smooth-extension:cosmological-past}
  \phi = \underline{u}_+\phi'\qquad \text{on }\underline{\mathbf{O}}_{\varepsilon}.
\end{equation}
Similarly, if $\phi$ is a smooth function in
$\underline{\mathbf{O}}_{\varepsilon}$ and vanishes on
$\EventHorizonFuture\bigcap \mathbf{O}_{\varepsilon}$, one can show
that there exists a smooth function $\phi'$ defined on
$\mathbf{O}_{\varepsilon}$ such that
\begin{equation}
  \label{eq:smooth-extension:event-future}
  \phi = u_+\phi'\qquad \text{on }\mathbf{O}_{\varepsilon}.
\end{equation}
Similarly, if $\phi \in C^{\infty}(\mathbf{O}_{\varepsilon})$ and $\phi$
vanishes on $\EventHorizonPast\bigcap \mathbf{O}_{\varepsilon}$, then
there exists another smooth function $\phi'$ defined on
$\mathbf{O}_{\varepsilon}$ such that
\begin{equation}
  \label{eq:smooth-extension:event-past}
  \phi = u_-\phi'\qquad \text{on }\mathbf{O}_{\varepsilon}.
\end{equation}

\subsection{Quantitative bounds}

We construct a system of coordinates in a small neighborhood $\widetilde{\Manifold}$ of $\Sigma_0\bigcap \closure\mathbf{E}$. To this end, first define a vectorfield interpolating between $T_0$ and $\KillT$ via a smooth double-bump function
\begin{equation*}
  T' = \chi_T(r)T_0 + (1-\chi_T(r))\KillT,\qquad
   \chi_T(r) =
   \begin{cases}
     1,& r<1+\delta_0,\\
     0,& 1+2\delta_0<r<2-2\delta_0,\\
     1,& r>2-\delta_0.
   \end{cases}
 \end{equation*}
 Then using the flow induced by $T'$, we extend the diffeomorphism
 $\Phi_0:E_{\frac{1}{2},\frac{5}{2}}\to \Sigma_0$ to cover a full
 neighborhood of $\Sigma_{1,2}$. In particular, for $\varepsilon_0$ sufficiently small, there exists a diffeomorphism
 \begin{equation}
   \label{eq:Phi1-diffeo:def}
   \Phi_1:(-\varepsilon_0,\varepsilon_0)\times E_{1-\varepsilon_0,2+\varepsilon_0} \to \widetilde{\Manifold},
 \end{equation}
 which agrees with $\Phi_0$ on $\{0\}\times E_{1-\varepsilon_0,2+\varepsilon_0}
 $ and such that $\partial_0 = \partial_{x^0} = T'$.

We now construct a system of coordinates which cover a neighborhood of
the spacelike hypersurface $\Sigma_0$. For any $0< R\le 1$, let
$B_R = \{x\in \Real^4: \abs*{x}<R\}$ denote the open ball of radius
$R$ in $\Real^4$.

Using the assumption in \Cref{eq:assumptions:Kill-T-T0}, we have
that for every $0<\varepsilon<\varepsilon_0$, there exists a sufficiently
large constant $\widetilde{A}_{\varepsilon}\in \Real$ such that
\begin{equation}
  \label{eq:A-tilde-def}
  \abs*{\Metric(\KillT,T_0)}>\frac{1}{\widetilde{A}_{\varepsilon}},\qquad
  x\in \left(\Sigma_0\bigcap \mathbf{E}\right)\backslash \left( \mathbf{O}_{\varepsilon}\bigcup \underline{\mathbf{O}}_{\varepsilon} \right). 
\end{equation}
In view of the normalization in \Cref{eq:L+L-:normalization-on-horizon},
we then have that (after potentially decreasing the value of
$\varepsilon_0$) there exists some constant $A_0\in\Real$ such that
\begin{equation}
  \label{eq:A0-bounding-uPlus-uMinus}
  \begin{gathered}
    \frac{u_+}{u_-} + \frac{u_-}{u_+} \le A_0,\qquad \text{on }\mathbf{O}_{\varepsilon_0}\bigcap \mathbf{E}\bigcap \Sigma_0,\\
    \frac{\underline{u}_+}{\underline{u}_-} + \frac{\underline{u}_-}{\underline{u}_+} \le A_0,\qquad \text{on }\underline{\mathbf{O}}_{\varepsilon_0}\bigcap \mathbf{E}\bigcap \Sigma_0
    .
  \end{gathered}  
\end{equation}
Moreover, (potentially increasing $A_0$), with
$\varepsilon_0^{-1}< A_0 < \infty$, for any
$x_0\in \Sigma_0\bigcap \closure\mathbf{E}$, there exists an open
neighborhood of $B(x_0)\subset \Manifold$ of $x_0$ and a smooth
coordinate chart
\begin{equation}
  \label{eq:Phi-x0:def}
  \Phi^{x_0}: B_1\to B_1(x_0),\qquad \Phi^{x_0}(0) = x_0,
\end{equation}
satisfying the property that
\begin{equation}
  \label{eq:coordinate-system-bound}
  \begin{gathered}
    \sup_{x_0\in \Sigma_0\bigcap \closure \mathbf{E}}\sup_{x\in B_1(x_0)}
    \sum_{j=0}^6\sum_{\alpha_1,\cdots,\alpha_j,\beta,\psi=1}^4 \left(
    \abs*{\partial_{\alpha_1}\cdots\partial_{\alpha_j}\Metric_{\beta\gamma}(x)}
    + \abs*{\partial_{\alpha_1}\cdots\partial_{\alpha_j}\Metric^{\beta\gamma}(x)}
    \right)
    \le A_0,\\
    \sup_{x_0\in \Sigma_0\bigcap \closure \mathbf{E}}\sup_{x\in B_1(x_0)}
    \sum_{j=0}^6\sum_{\alpha_1,\cdots,\alpha_j,\beta=1}^4
    \abs*{\partial_{\alpha_1}\cdots\partial_{\alpha_j}\KillT^{\beta}(x)}\le A_0.
  \end{gathered}
\end{equation}
We also define the function
\begin{equation}
  \label{eq:N-x0:def}
  \begin{split}
    N^{x_0}: \Phi^{x_0}(B_1) = B_1(x_0)&\to [0,\infty),\\
    x&\mapsto \abs*{(\Phi^{x_0})^{-1}(x)}^2. 
  \end{split}  
\end{equation}
In addition, we may assume that
$B_1(x_0)\subset \mathbf{O}_{\varepsilon_0}$ if $x_0\in S_0$ and
$B_1(x_0)\subset \underline{\mathbf{O}}_{\varepsilon_0}$ if
$x_0\in \underline{S}_0$. We define $\widetilde{\Manifold}$ to be the
union of balls $B_1(x_0)$ over all points
$x_0\in \Sigma_0\bigcap \closure \mathbf{E}$. Note that we can arrange
that $\widetilde{\Manifold}$ is simply connected.

Since $S_0$ and $\underline{S}_0$ are compact, we may assume (after
possibly increasing the value of $A_0$), that
\begin{equation}
  \label{eq:u:bound}
  \begin{split}
    \sup_{x\in S_0}\sup_{x\in B_1(x_0)}\left(
    \sum_{\abs*{\alpha}\le 6}
    \abs*{\partial^{\alpha}\mathfrak{u}(x)} + \abs*{D^1\mathfrak{u}(x)}^{-1}
  \right)\le A_0,\qquad
  \mathfrak{u}\in \{u_-,u_+\},\\
  \sup_{x\in \underline{S}_0}\sup_{x\in B_1(x_0)}\left(
    \sum_{\abs*{\alpha}\le 6}
    \abs*{\partial^{\alpha}\mathfrak{u}(x)} + \abs*{D^1 \mathfrak{u}(x)}^{-1}
  \right)\le A_0,\qquad
  \mathfrak{u}\in \{\underline{u}_-,\underline{u}_+\}.
  \end{split}  
\end{equation}

\subsection{Pseudoconvexity}
\label{sec:null-convexity}

To prove our \KdS{} uniqueness theorems, we will use several unique
continuation arguments, for which notions of pseudoconvexity (also
known in the literature as null-convexity) play a crucial role.  The
most basic notion of pseudoconvexity is simply the classical notion of
pseudoconvexity of Hormander
\cite{hormanderLinearPartialDifferential1964} applied to our context
with operators which are principally the scalar wave operator.
\begin{definition}
  \label{def:pseudoconvexity}
  A domain $O\subset \Manifold$, is said to be \emph{strictly
    pseudoconvex} at a boundary point $p\in \partial O$ if there exists
  a smooth function $h$, defined in a small neighborhood $U$ of $p$
  such that $dh(p)\neq 0$ and $O\bigcap U = \{x\in U: h(x)<0\}$ which
  verifies the following pseudoconvexity property at $p$
  \begin{equation}
    \label{eq:null-convexity:def}
    \CovariantDeriv^2h(X,X)(p) < 0
  \end{equation}
  for all null vectors $X$ that are also tangent to $\partial O$ at $p$.
\end{definition}
The pseudoconvexity assumption can be made quantitative.
\begin{lemma}
  Let $h$ be a strictly pseudoconvex boundary defining function at
  $p\in O\subset \Manifold$. Then, $O$ is strictly pseudoconvex if
  there exists a constant $C>0$, depending only on bounds for the
  metric $\Metric$ and its derivatives with respect to a fixed
  coordinate system at $p$ in a fixed coordinate neighborhood $V$ or
  $p$, there exists some $\mu\in [-C, C]$, and there exists a
  neighborhood $U\in V$ of $p$ such that for any vectorfield
  $X\in T\Manifold$,
  \begin{equation*}
    \begin{cases}
      \abs*{dh}&\ge C^{-1}\\
      X^{\alpha}X^{\beta}\left(\mu\Metric_{\alpha\beta} - \CovariantDeriv_{\alpha}\CovariantDeriv_{\beta}h\right)
      + C\abs*{X(h)}^2&\ge C^{-1}\abs*{X}^2,
    \end{cases}
  \end{equation*}
  where $\abs*{X}^2 = \sum_{i=1}^4(X^i)^2$ denotes the Euclidean
  absolute value. 
\end{lemma}

Under the assumption that $\Manifold$ contains a Killing vectorfield
$\mathbf{V}$, we also have the following variant of the
pseudoconvexity condition, which played a critical role in understanding Kerr rigidity \cite{alexakisRigidityStationaryBlack2014,alexakisUniquenessSmoothStationary2010,ionescuUniquenessSmoothStationary2009}.
\begin{definition}[Strict $\mathbf{V}$-pseudoconvexity]
  \label{def:strict-T-null-convexity}
  A family of weights
  $h_{\varepsilon}: B_{\varepsilon^{10}}\to \Real_+,$
  $\varepsilon\in (0,\varepsilon_1)$, $\varepsilon_1\le A_0^{-1}$, is
  called $\mathbf{V}$-pseudoconvex if for any
  $\varepsilon\in (0, \varepsilon_1)$,
  \begin{equation}
    \label{eq:strict-T-null-convexity:cond1}
    h_{\varepsilon}(x_0) = \varepsilon, \qquad
    \sup_{x\in B_{\varepsilon^{10}}}\sum_{j=1}^4\varepsilon^j\abs*{D^jh_{\varepsilon}(x)}\le \frac{\varepsilon}{\varepsilon_1},\qquad
    \abs*{\mathbf{V}(h_{\varepsilon})(x_0)}\le \varepsilon^{10},
  \end{equation}
  and
  \begin{equation}
    \label{eq:strict-T-null-convexity:cond2}
    \CovariantDeriv^{\alpha}h_{\varepsilon}(x_0)\CovariantDeriv^{\beta}h_{\varepsilon}(x_0)
    \left(\CovariantDeriv_{\alpha}h_{\varepsilon}\CovariantDeriv_{\beta}h_{\varepsilon} - \varepsilon\CovariantDeriv_{\alpha}\CovariantDeriv_{\beta}h_{\varepsilon}\right)(x_0)\ge \varepsilon_1^2,
  \end{equation}
  and there is $\mu\in [-\varepsilon_1^{-1},\varepsilon_1^{-1}]$ such that for all vectors $X = X^{\alpha}\partial_{\alpha}\in T_{x_0}\mathcal{M}$,
  \begin{equation}
    \label{eq:strict-T-null-convexity:main-condition}
      X^{\alpha}X^{\beta}\left(\mu\Metric_{\alpha\beta} - \CovariantDeriv_{\alpha}\CovariantDeriv_{\beta}h_{\varepsilon}\right)(x_0)
      + \varepsilon^{-2}\left( \abs*{X(h_{\varepsilon})}^2 + \abs*{\Metric(\mathbf{V}, X)}^2 \right)\ge \varepsilon_1^2\abs*{X}^2,
  \end{equation}
\end{definition}

This can be further generalized to
$\curlyBrace*{\mathbf{V}_i}_{i=1}^k$-pseudoconvexity, which will be
useful if the manifold in question contains multiple Killing
vectorfields.
\begin{definition}[Strict $\curlyBrace*{\mathbf{V}_i}_{i=1}^k$-pseudoconvexity]
  \label{def:strict-T-Z-null-convexity}
  Let $\curlyBrace*{\mathbf{V}_i}_{i=1}^k$ be a family of
  smooth vectorfields.  A family of weights
  $h_{\varepsilon}: B_{\varepsilon^{10}}\to \Real_+,$
  $\varepsilon\in (0,\varepsilon_1)$, $\varepsilon_1\le A_0^{-1}$, is
  called \emph{$\curlyBrace*{\mathbf{V}_i}_{i=1}^k$-pseudoconvex} if for any
  $\varepsilon\in (0, \varepsilon_1)$,
  \begin{equation}
    \label{eq:strict-T-Z-null-convexity:cond1}
    h_{\varepsilon}(x_0) = \varepsilon, \qquad
    \sup_{x\in B_{\varepsilon^{10}}}\sum_{j=1}^4\varepsilon^j\abs*{D^jh_{\varepsilon}(x)}\le \frac{\varepsilon}{\varepsilon_1},\qquad
    \sum_{i=1}^k\abs*{\mathbf{V}_i(h_{\varepsilon})(x_0)}\le \varepsilon^{10},
  \end{equation}
  and
  \begin{equation}
    \label{eq:strict-T-Z-null-convexity:cond2}
    \CovariantDeriv^{\alpha}h_{\varepsilon}(x_0)\CovariantDeriv^{\beta}h_{\varepsilon}(x_0)
    \left(\CovariantDeriv_{\alpha}h_{\varepsilon}\CovariantDeriv_{\beta}h_{\varepsilon} - \varepsilon\CovariantDeriv_{\alpha}\CovariantDeriv_{\beta}h_{\varepsilon}\right)(x_0)\ge \varepsilon_1^2,
  \end{equation}
  and there is $\mu\in [-\varepsilon_1^{-1},\varepsilon_1^{-1}]$ such that for all vectors $X = X^{\alpha}\partial_{\alpha}\in T_{x_0}\mathcal{M}$,
  \begin{equation}
    \label{eq:strict-T-Z-null-convexity:main-condition}
      X^{\alpha}X^{\beta}\left(\mu\Metric_{\alpha\beta} - \CovariantDeriv_{\alpha}\CovariantDeriv_{\beta}h_{\varepsilon}\right)(x_0)
      + \varepsilon^{-2}\left( \abs*{X(h_{\varepsilon})}^2 + \sum_{i=1}^k\abs*{\Metric(\mathbf{V}_i, X)}^2 \right)\ge \varepsilon_1^2\abs*{X}^2,
  \end{equation}
\end{definition}
\begin{definition}
  \label{def:negligible perturbation}
  We say that a function
  $e_{\varepsilon}: B_{\varepsilon^{10}}\to \Real$ is a
  \emph{negligible perturbation} if
  \begin{equation*}
    \sup_{x\in B_{\varepsilon^{10}}}\abs*{D^je_{\varepsilon}(x)}\le \varepsilon^{10},\qquad j\in \{0,1,2,3\}.  
  \end{equation*}
\end{definition}
Negligible perturbations of a $\{\mathbf{V}_i\}_{i=1}^k$-pseudoconvex
foliation satisfy a few convenient properties.
\begin{lemma}
  \label{lemma:negligible-perturbation:basic-props}
  Let the family of weights
  $h_{\varepsilon}: B_{\varepsilon^{10}}\to \Real_+,\varepsilon\in (0,
  \varepsilon_1)$ be $\curlyBrace*{\mathbf{V}_i}_{i=1}^k$-pseudoconvex
  as defined in \Cref{def:strict-T-Z-null-convexity}, and let
  $e_{\varepsilon}$ be a negligible perturbation as defined
  in~\Cref{def:negligible perturbation}. Then
  \begin{equation*}
    h_{\varepsilon}+ e_{\varepsilon} \in \left[ \frac{\varepsilon}{2}, 2\varepsilon \right],\qquad \text{in }B_{\varepsilon^{10}},
  \end{equation*}
  and
  \begin{equation*}
    \sum_{i=1}^k\abs*{\mathbf{V}_i(h_{\varepsilon} + e_{\varepsilon})(x)}\lesssim \varepsilon^8\qquad \text{in }B_{\varepsilon^{10}}.
  \end{equation*}
\end{lemma}
\begin{proof}
  The statements are simple consequences of the definitions in
  \Cref{def:strict-T-Z-null-convexity} and \Cref{def:negligible
    perturbation}.
\end{proof}

\subsection{Carleman estimates}

In this section, we review the Carleman estimates we will
use. Carleman estimates are often used to prove unique continuation
properties given some pseudoconvexity condition. We note that since
our main equations will be a system of wave-transport equations (or
just a system of wave equations), we will prove two general Carleman
estimates, one adapted to wave equations and the other adapted to
transport equations.  The proofs in this section are straight-forward
modifications of the proof of Proposition 3.3 in
\cite{ionescuUniquenessSmoothStationary2009} and Lemma 3.4 of
\cite{alexakisHawkingsLocalRigidity2010}, but have been included for
completeness.
\begin{prop}[Carleman estimate]
  \label{prop:carleman-estimate}
  Assume that $x_0\in \mathbf{E}$, $\{\mathbf{V}_i\}_{i=1}^k$ are a family of
  smooth, bounded vectorfields on $\mathcal{M}$,
  $\varepsilon_1\le A_0^{-1}$,
  $\{h_{\varepsilon}\}_{\varepsilon\in(0,\varepsilon_1)}$, is a
  $\{\mathbf{V}_i\}_{i=1}^k$-pseudoconvex family, and
  $e_{\varepsilon}$ is a negligible perturbation for any
  $0<\varepsilon\le \varepsilon_1$. Then there exists
  $\varepsilon\in (0, \varepsilon_1)$ sufficiently small and
  $C(\varepsilon)$ sufficiently large such that for any
  $\lambda\ge C(\varepsilon)$ and any
  $\varphi\in C^{\infty}_0(B_{\varepsilon^{10}})$,
  \begin{equation}
    \label{eq:carleman-estimate}
    \begin{split}
      \lambda\norm*{e^{-\lambda f_{\varepsilon}}\phi}_{L^2}
        + \norm*{e^{-\lambda f_{\varepsilon}}D^1\phi}_{L^2}
      \le{}& C(\varepsilon)\lambda^{-\frac{1}{2}}\norm*{e^{-\lambda f_{\varepsilon}}\Box_{\Metric}\phi}_{L^2}
             + \varepsilon^{-6}\sum_{i=1}^k\norm*{\mathbf{V}_i(\phi)}_{L^2}.
    \end{split}
  \end{equation}
  where $f_{\varepsilon} = \ln (h_{\varepsilon} + e_{\varepsilon})$.
\end{prop}

We first show a key technical lemma for the proof of the
Carleman estimates.
\begin{lemma}
  \label{lemma:carleman-estimates:pseudoconvexity-key}
  Under the assumptions of \Cref{prop:carleman-estimate}, 
  there exist constants $\varepsilon\ll1$ and
  $\mu_1\in [-\varepsilon^{-\frac{3}{2}},\varepsilon^{-\frac{3}{2}}]$ such that 
  the following pointwise bounds hold on $B_{\varepsilon^{10}}$ for any
  $\psi\in C^{\infty}_0(B_{\varepsilon^{10}})$.
  \begin{align}
    \label{eq:carleman-estimates:pseudoconvexity-key:1}
    \abs*{D^1\psi}^2
    \le{}& \sum_{i=1}^k\varepsilon^{-8}\abs*{ \mathbf{V}_i\psi}^2
    + \varepsilon^{-8}\abs*{\CovariantDeriv_{\alpha}\widetilde{h}_{\varepsilon}\cdot\CovariantDeriv^{\alpha}\psi}^2
    + (\CovariantDeriv^{\alpha}\psi\CovariantDeriv^{\beta}\psi)\left(
    \mu_1\Metric_{\alpha\beta} - \widetilde{h}_{\varepsilon}^{-1} \CovariantDeriv_{\alpha}\CovariantDeriv_{\beta}\widetilde{h}_{\varepsilon}
    \right),\\
    \label{eq:carleman-estimates:pseudoconvexity-key:2}
    2\le{}& \widetilde{h}_{\varepsilon}^{-4}H_{\varepsilon}^2
    - \widetilde{h}_{\varepsilon}^{-3}\CovariantDeriv^{\alpha}\widetilde{h}_{\varepsilon}\CovariantDeriv^{\beta}\widetilde{h}_{\varepsilon}\CovariantDeriv_{\alpha}\CovariantDeriv_{\beta}\widetilde{h}_{\varepsilon}
    - \widetilde{h}_{\varepsilon}^{-2}\mu_1\widetilde{H}_{\varepsilon},
  \end{align}
  where
  \begin{equation*}
    \widetilde{H_{\varepsilon}} \vcentcolon= \CovariantDeriv^{\alpha}\widetilde{h}_{\varepsilon}\CovariantDeriv_{\alpha}\widetilde{h}_{\varepsilon}.
  \end{equation*}
\end{lemma}
\begin{proof}
  We start with proving
  \Cref{eq:carleman-estimates:pseudoconvexity-key:2}. We begin by
  recalling \Cref{eq:strict-T-Z-null-convexity:cond2}, namely, that
  there exists some $\mu\in [-\varepsilon_1^{-1},\varepsilon_1^{-1}]$ such
  that for all $X\in T_{x_0}(\Manifold)$,
  \begin{equation*}
        \CovariantDeriv^{\alpha}h_{\varepsilon}(x_0)\CovariantDeriv^{\beta}h_{\varepsilon}(x_0)
    \left(\CovariantDeriv_{\alpha}h_{\varepsilon}\CovariantDeriv_{\beta}h_{\varepsilon} - \varepsilon\CovariantDeriv_{\alpha}\CovariantDeriv_{\beta}h_{\varepsilon}\right)(x_0)\ge \varepsilon_1^2. 
  \end{equation*}
  for $x\in B_{\varepsilon^{10}}$, let
  \begin{equation}
    K(x) \vcentcolon= \CovariantDeriv^{\alpha}h_{\varepsilon}(x)\CovariantDeriv^{\beta}h_{\varepsilon}(x)
    \left(\CovariantDeriv_{\alpha}h_{\varepsilon}\CovariantDeriv_{\beta}h_{\varepsilon} - h_\varepsilon\CovariantDeriv_{\alpha}\CovariantDeriv_{\beta}h_{\varepsilon}\right)(x).
  \end{equation}
  Then, we have that
  \begin{equation*}
    \abs*{D^1K(x)}\lesssim \varepsilon^{-1}, \qquad K(x) \ge \frac{\varepsilon_1^2}{2},
  \end{equation*}
  where the first bound follows from $\varepsilon<\varepsilon_1\ll 1$
  and \Cref{eq:strict-T-Z-null-convexity:cond1}, and the second bound
  follows from recalling that $\varepsilon = h_{\varepsilon}(x_0)$.

  Now, let
  \begin{equation*}
  \begin{split}    
    \widetilde{K}(x)&\vcentcolon= \CovariantDeriv^{\alpha}\widetilde{h}_{\varepsilon}(x)\CovariantDeriv^{\beta}\widetilde{h}_{\varepsilon}(x)(\CovariantDeriv_{\alpha}\widetilde{h}_{\varepsilon}\CovariantDeriv_{\beta}\widetilde{h}
                      - \widetilde{h}_{\varepsilon}\CovariantDeriv_{\alpha}\CovariantDeriv_{\beta}\widetilde{h}_{\varepsilon})(x)\\
                    &= \widetilde{H}_{\varepsilon}^2
                      - \widetilde{h}_{\varepsilon}(x)\CovariantDeriv^{\alpha}\widetilde{h}_{\varepsilon}(x)\CovariantDeriv^{\beta}\widetilde{h}_{\varepsilon}(x)\CovariantDeriv_{\alpha}\CovariantDeriv_{\beta}\widetilde{h}_{\varepsilon}(x)
                      . 
  \end{split}  
  \end{equation*}
    From \Cref{eq:strict-T-Z-null-convexity:cond1} and the fact that
  $e_{\varepsilon}$ is an admissible perturbation, we have that
  \begin{equation*}
    \abs*{\widetilde{K}(x) - K(x)}\lesssim \varepsilon,
  \end{equation*}
  which implies that if $\varepsilon$ is sufficiently small, then
  \begin{equation*}
    \widetilde{K}(x) \ge \frac{\varepsilon_1^2}{4}
  \end{equation*}
  on $B_{\varepsilon^{10}}$. Multiplying by
  $\widetilde{h}_{\varepsilon}^{-4}$, we then have that on $B_{\varepsilon^{10}}$, 
  \begin{equation*}
    \frac{1}{4}\widetilde{h}_{\varepsilon}^{-4}\varepsilon_1^2
    \le \widetilde{h}_{\varepsilon}^{-4}\widetilde{K}(x)
    = \widetilde{h}_{\varepsilon}^{-4}\widetilde{H}_{\varepsilon}^2
    - \widetilde{h}_{\varepsilon}^{-3}\CovariantDeriv^{\alpha}\widetilde{h}_{\varepsilon}(x)\CovariantDeriv^{\beta}\widetilde{h}_{\varepsilon}(x)\CovariantDeriv_{\alpha}\CovariantDeriv_{\beta}\widetilde{h}_{\varepsilon}(x). 
  \end{equation*}
  Then, recall that
  \begin{equation*}
    \widetilde{h}_{\varepsilon}(x)\in \left[ \frac{\varepsilon}{2}, 2\varepsilon \right],\qquad
    \abs*{\widetilde{h}_{\varepsilon}^{-2}\mu_1\widetilde{H}_{\varepsilon}}\lesssim \abs*{\mu_1\varepsilon^{-2}}
    \lesssim \varepsilon^{-\frac{7}{2}}, 
  \end{equation*}
  so \Cref{eq:carleman-estimates:pseudoconvexity-key:2} follows for $\varepsilon$ sufficiently small.

  We now move on to proving
  \Cref{eq:carleman-estimates:pseudoconvexity-key:1}. We begin with
  the assumption in \ref{eq:strict-T-Z-null-convexity:main-condition}
  that for any $x\in B_{\varepsilon^{10}}$ and
  vector $X\in \KillT_{x_0}(\mathcal{M})$ for $\varepsilon$ sufficiently
  small,
  \begin{equation*}
    \frac{1}{2}\varepsilon_1^2\abs*{X}^2
    \le X^{\alpha}X^{\beta}(\mu \Metric_{\alpha\beta} - \CovariantDeriv_{\alpha}\CovariantDeriv_{\beta}h_{\varepsilon})(x_0)
    + \varepsilon^{-2}\left(\abs*{X^{\alpha}(\mathbf{V}_i)_{\alpha}(x_0)}^2 + \abs*{X^{\alpha}\CovariantDeriv_{\alpha}h_{\varepsilon}(x_0)}^2\right)    
  \end{equation*}
  for some
  $\mu\in \left[ -\frac{1}{ \varepsilon_1}, \frac{1}{\varepsilon_1}
  \right]$. Then, let
  \begin{equation*}
    K_{\alpha\beta} \vcentcolon= \mu \varepsilon^{-1}h_{\varepsilon}\Metric_{\alpha\beta}
    - \CovariantDeriv_{\alpha}\CovariantDeriv_{\beta}h_{\varepsilon}
    + \sum_{i=1}^k\varepsilon^{-2}(\mathbf{V}_i)_{\alpha}(\mathbf{V}_i)_{\beta}
    + \varepsilon^{-2}\CovariantDeriv_{\alpha}h_{\varepsilon}\CovariantDeriv_{\beta}h_{\varepsilon}.
  \end{equation*}
  From \Cref{eq:strict-T-Z-null-convexity:cond1}, we have that in
  $B_{\varepsilon^{10}}$
  \begin{equation*}
    \abs*{D^1K_{\alpha\beta}(x)}\lesssim \varepsilon^{-3}.
  \end{equation*}
  Then, since $h_{\varepsilon}(x_0)=\varepsilon$, we have that 
  \begin{equation*}
    X^{\alpha}X^{\beta}K_{\alpha\beta}(x_0) = X^{\alpha}X^{\beta}(\mu \Metric_{\alpha\beta} - \CovariantDeriv_{\alpha}\CovariantDeriv_{\beta}h_{\varepsilon})(x_0)
    + \varepsilon^{-2}\left(\abs*{X^{\alpha}(\mathbf{V}_i)_{\alpha}(x_0)}^2 + \abs*{X^{\alpha}\CovariantDeriv_{\alpha}h_{\varepsilon}(x_0)}^2\right).
  \end{equation*}
  As a result, we in fact have that for $\varepsilon$ sufficiently
  small, for any $x\in B_{\varepsilon^{10}}$ and any $X\in \Real^4$
  \begin{equation}
    \label{eq:carleman-estimates:pseudoconvexity-key:aux1}
    X^{\alpha}X^{\beta}K_{\alpha\beta}(x) \ge \frac{\varepsilon_1^2}{2}\abs*{X}^2.
  \end{equation}  
  Then let
  \begin{equation*}
    \widetilde{K}_{\alpha\beta}\vcentcolon= \mu \varepsilon^{-1}\widetilde{h}_{\varepsilon}\Metric_{\alpha\beta}
    - \CovariantDeriv_{\alpha}\CovariantDeriv_{\beta}\widetilde{h}
    + \sum_{i=1}^k\varepsilon^{-2}(\mathbf{V}_i)_{\alpha}(\mathbf{V}_i)_{\beta}
    + \varepsilon^{-2}\CovariantDeriv_{\alpha}\widetilde{h}_{\varepsilon}\CovariantDeriv_{\beta}\widetilde{h}_{\varepsilon}.
  \end{equation*}
  From \Cref{eq:strict-T-Z-null-convexity:cond1} and the fact that
  $e_{\varepsilon}$ is an admissible perturbation, we have that for
  $x\in B_{\varepsilon^{10}}$,
  \begin{equation*}
    \abs*{\widetilde{K}_{\alpha\beta}(x) - K_{\alpha\beta}(x)}\lesssim \varepsilon^5.
  \end{equation*}
  Thus, from \Cref{eq:carleman-estimates:pseudoconvexity-key:aux1} we
  have that if $\varepsilon$ is sufficiently small, then for any
  $X\in \Real^4$ and $x\in B_{\varepsilon^{10}}$,
  \begin{equation*}
    X^{\alpha}X^{\beta}\widetilde{K}_{\alpha\beta}(x)\ge\frac{\varepsilon_1^2}{4}\abs*{X}^2. 
  \end{equation*}
  Multiplying by
  $\widetilde{h}_{\varepsilon}^{-1}\in \left[ \frac{1}{2\varepsilon},
    \frac{2}{\varepsilon} \right]$, we then have that
  \begin{equation*}
    \begin{split}
      X^{\alpha}X^{\beta}\left(\mu \varepsilon^{-1}\Metric_{\alpha\beta} - \widetilde{h}_{\varepsilon}^{-1}\CovariantDeriv_{\alpha}\CovariantDeriv_{\beta}\widetilde{h}_{\varepsilon}\right)
      + 2\varepsilon^{-3}\sum_{i=1}^k\abs*{X\cdot \mathbf{V}_i}^2
      + 2\varepsilon^{-3}\abs*{\CovariantDeriv_Xh_{\varepsilon}}^2
      \ge \frac{\varepsilon_1^2}{4\widetilde{h}_{\varepsilon}}\abs*{X}^2.
    \end{split}
  \end{equation*}
  The bound in \Cref{eq:carleman-estimates:pseudoconvexity-key:1}
  then follows for $\varepsilon$ sufficiently small, with
  $\mu_1 = \mu \varepsilon^{-1}\in \left[
    -\frac{1}{\varepsilon\varepsilon_1},\frac{1}{\varepsilon\varepsilon_1}
    \right]$.
\end{proof}

We can now prove the Carleman estimates in
\Cref{prop:carleman-estimate}.
\begin{proof}[Proof of \Cref{prop:carleman-estimate}]
  Without loss of regularity, we assume that $\psi$ is real-valued.
  We see that if $\lambda\ge 2\varepsilon^{-7}$, then
  \Cref{eq:carleman-estimates:pseudoconvexity-key:1} implies that
  \begin{equation}
    \label{eq:carleman-estimates:step-6:1}
    \abs*{D^1\psi}^2 \le \varepsilon^{-8}\sum_{i=1}^k\abs*{\mathbf{V}_i\psi}^2
    + \lambda \abs*{\CovariantDeriv_{\alpha} f_{\varepsilon} \CovariantDeriv^{\alpha} \psi}^2
    + \left(\CovariantDeriv^{\alpha}\psi \CovariantDeriv^{\beta}\psi\right)\left(\mu_1\Metric_{\alpha\beta} - \CovariantDeriv_{\alpha}\CovariantDeriv_{\beta}f_{\varepsilon}\right).
  \end{equation}
  In addition, if $\lambda \gtrsim \varepsilon^{-2}$, then
  $\abs*{\lambda^{-2}\Box_{\Metric}^2(f_{\varepsilon})}\le 1$, and
  thus,
  \begin{equation}
    \label{eq:carleman-estimates:step-6:2}
    1\le \mu_1 G - \CovariantDeriv^{\alpha}f_{\varepsilon}\CovariantDeriv^{\beta}f_{\varepsilon} \CovariantDeriv_{\alpha}\CovariantDeriv_{\beta}f_{\varepsilon}
    + \frac{1}{4}\lambda^{-2}\Box_{\Metric}^2(f_{\varepsilon})
  \end{equation}
  on $B_{\varepsilon^{10}}$, where
  \begin{equation*}
    G\vcentcolon= \CovariantDeriv_{\alpha}f_{\varepsilon}\CovariantDeriv^{\alpha}f_{\varepsilon}.
  \end{equation*}
  Then, define
  \begin{equation*}
    w\vcentcolon= \mu_1 - \frac{1}{2}\Box_{\Metric}f_{\varepsilon},\qquad W^{\alpha} \vcentcolon= \CovariantDeriv^{\alpha}f_{\varepsilon},
  \end{equation*}
  so that
  \begin{equation*}
    \begin{split}
      w \CovariantDeriv^{\alpha}\psi\CovariantDeriv_{\alpha}\psi
      - \CovariantDeriv^{\alpha}W^{\beta} \EMT[\psi]_{\alpha\beta}
      ={}& \left( \CovariantDeriv^{\alpha}\psi\CovariantDeriv^{\beta}\psi \right)\left(
           \left(w + \frac{1}{2}\Box_{\Metric}f_{\varepsilon}\right)\Metric_{\alpha\beta}
           - \CovariantDeriv_{\alpha}\CovariantDeriv_{\beta}f_{\varepsilon}
           \right)\\
      -2wG - W(G) - G\Divergence W 
      ={}& - G\left( 2w + \Box_{\Metric}f_{\varepsilon}\right)
           - 2\CovariantDeriv^{\alpha}f_{\varepsilon}\CovariantDeriv^{\beta}f_{\varepsilon}\CovariantDeriv_{\alpha}\CovariantDeriv_{\beta}f_{\varepsilon}.
    \end{split}
  \end{equation*}
  Then, we can rewrite \Cref{eq:carleman-estimates:step-6:1} and
  \Cref{eq:carleman-estimates:step-6:2} as
  \begin{equation}
    \label{eq:carleman-estimates:step-5}
    \begin{split}
      \abs*{D^1\psi}^2
    \le{}& \varepsilon^{-8}\sum_{i=1}^k\abs*{\mathbf{V}_i\psi}^2
    + \lambda \abs*{W\psi}^2
           + \left(w\CovariantDeriv^{\alpha} \psi \CovariantDeriv_{\alpha} \psi
           - \CovariantDeriv^{\alpha} W^{\beta} \EMT_{\alpha\beta}[\psi]\right),\\
    2 \le{}& -2wG - W(G) - G\Divergence W - \lambda^{-2}\Box_{\Metric}w,
    \end{split}    
  \end{equation}
  which hold on $B_{\varepsilon}^{10}$.

  Now we can compute that
  \begin{align*}
    \Box_{\Metric}\psi(W-w)\psi
    ={}& \CovariantDeriv^{\alpha}\left(W^{\beta}\EMT_{\alpha\beta}[\psi] - w\psi\CovariantDeriv_{\alpha}\psi + \psi^2\CovariantDeriv_{\alpha}w\right)\\
       & -\CovariantDeriv^{\alpha}W^{\beta}\EMT_{\alpha\beta}[\psi]
         + w\CovariantDeriv^{\alpha}\psi\CovariantDeriv_{\alpha}\psi
         - \psi^2\Box_{\Metric}w,\\
    G\psi(W-w)\psi
    ={}& \frac{1}{2}\CovariantDeriv^{\alpha}\left(\psi^2 G W_{\alpha}\right)
         - \psi^2\left(wG + \frac{1}{2}W(G) + \frac{1}{2}G\Divergence W\right).
  \end{align*}
  Then, since $\psi\in C^{\infty}_0(B_{\varepsilon^{10}})$, we have
  via integration by parts that
  \begin{equation*}
    \begin{split}
      \int_{B_{\varepsilon^{10}}}\left(\Box_{\Metric}\psi + \lambda^2 G\psi\right)(W\psi - w\psi)
      ={}& \int_{B_{\varepsilon^{10}}}w\CovariantDeriv^{\alpha}\psi\CovariantDeriv_{\alpha}\psi
           - \CovariantDeriv^{\alpha} W^{\beta} \EMT_{\alpha\beta}[\psi]\\
         &- \lambda^2\int_{B_{\varepsilon^{10}}}\psi^2\left(wG + \frac{1}{2}W(G) + \frac{1}{2}G\Divergence W + \frac{1}{2\lambda^2}\Box_{\Metric}w\right).
    \end{split}    
  \end{equation*}
  Then, \Cref{eq:carleman-estimates:step-5} implies that
  \begin{equation}
    \label{eq:carleman-estimates:step-4}
    \lambda \varepsilon^{-8}\sum_{i=1}^k\norm*{\mathbf{V}_i\psi}_{L^2}^2
    + \int_{B_{\varepsilon^{10}}}\left(\Box_{\Metric }\psi + \lambda^2 G \psi\right)(W-w)\psi
    + \lambda^2\norm*{W\psi}_{L^2}^2
    \ge \lambda^3\norm*{\psi}_{L^2}^2 + \lambda\norm*{\abs*{D^1\psi}}_{L^2}^2. 
  \end{equation}
  Now define
  \begin{equation}
    \label{eq:carleman-estimates:L-epsilon:def}
    L_{\varepsilon}\vcentcolon= \Box_{\Metric} + 2\lambda\CovariantDeriv^{\alpha}(f_{\varepsilon}) \CovariantDeriv_{\alpha} + \lambda^2\CovariantDeriv_{\alpha}(f_{\varepsilon})\CovariantDeriv^{\alpha} (f_{\varepsilon})
    = \Box_{\Metric} + 2\lambda W + \lambda^2G,
  \end{equation}
  which clearly satisfies
  \begin{equation}
    \label{eq:carleman-estimates:L-epsilon:L2-relation-to-wave}
    \norm*{e^{-\lambda f_{\varepsilon}}\Box_{\Metric}(e^{\lambda f_{\varepsilon}}\psi)}_{L^2}
    \ge \norm*{L_{\varepsilon}\psi}_{L^2}
    - \lambda\norm*{\Box_{\Metric}(f_{\varepsilon})\psi}_{L^2}.
  \end{equation}
  Then we can trivially write that
  \begin{equation*}
    L_{\varepsilon}\vcentcolon= \Box_{\Metric}
    + \lambda^2 G \psi
    + \lambda(W-w)\psi
    + \lambda(W+w)\psi,
  \end{equation*}
  to see that \Cref{eq:carleman-estimates:step-4} implies that
  \begin{align*}
    &\lambda\varepsilon^{-8}\sum_{i=1}^k\norm*{\mathbf{V}_i\psi}_{L^2}^2
    + \int_{B_{\varepsilon^{10}}}L_{\varepsilon}\psi(\lambda W - \lambda w)\psi\\
    ={}& \lambda\varepsilon^{-8}\sum_{i=1}^k\norm*{\mathbf{V}_i\psi}_{L^2}^2
         + \int_{B_{\varepsilon^{10}}}(\Box_{\Metric}\psi + \lambda^2G\psi)(\lambda W - \lambda w)\psi
         + \lambda^2\norm*{(W-w)\psi}_{L^2}^2\\
       & + \lambda^2\left(\norm*{W\psi}_{L^2}^2 - \norm*{w\psi}_{L^2}^2\right)\\
    \ge{}& \lambda^3\norm*{\psi}_{L^2}^2
           + \lambda\norm*{\abs*{D^1\psi}}_{L^2}^2
           + \norm*{\lambda (W-w)\psi}_{L^2}^2
           - \lambda^2\norm*{w\psi}_{L^2}^2\\
    \ge{}& \norm*{\lambda (W-w)\psi}_{L^2}^2
           + \frac{1}{2}\lambda^3\norm*{\psi}_{L^2}^2
           + \lambda\norm*{\abs*{D^1\psi}}_{L^2}^2
  \end{align*}
  if $\lambda$ is sufficiently large, where the last equality follows
  from the fact that $w\lesssim \varepsilon^{-2}$ and thus, for
  $\lambda \gtrsim \varepsilon^{-4}$ sufficiently large,
  \begin{equation*}
    \lambda^3\norm*{\psi}^2_{L^2}
    -\lambda^2\norm*{w\psi}^2_{L^2} \ge \frac{1}{2}\lambda^3\norm*{\psi}^2_{L^2}.
  \end{equation*}
  Then for $C(\varepsilon)$ sufficiently large, 
  \begin{equation*}
    2\lambda \varepsilon^{-8}\sum_{i=1}^k\norm*{\mathbf{V}_i\psi}_{L^2}^2
    + 2\int_{B_{\varepsilon^{10}}}L_{\varepsilon}\psi(\lambda W - \lambda w)\psi
    \ge C(\varepsilon)^{-1}\norm*{\lambda(W-w)\psi}_{L^2}^2
    + \lambda^3\norm*{\psi}_{L^2}^2
    + \lambda \norm*{\abs*{D^1\psi}}_{L^2}^2
  \end{equation*}
  for any $\lambda\ge C(\varepsilon)$ and any
  $\psi\in C^{\infty}_0(B_{\varepsilon^{10}})$. It then follows
  directly by Cauchy-Schwarz that
  \begin{equation*}
    \begin{split}
      &C(\varepsilon)^{-1}\norm*{\lambda(W-w)\psi}_{L^2}^2
        + \lambda^3\norm*{\psi}_{L^2}^2
        + \lambda \norm*{\abs*{D^1\psi}}_{L^2}^2\\
      \le{}& 2\lambda \varepsilon^{-8}\sum_{i=1}^k\norm*{\mathbf{V}_i\psi}_{L^2}^2
             + C(\varepsilon)^{-1}\norm*{\lambda(W -  w)\psi}_{L^2}^2
             + C(\varepsilon)\norm*{L_{\varepsilon}\psi}_{L^2}^2,
    \end{split}    
  \end{equation*}
  so in fact,
  \begin{equation*}
    \lambda^3\norm*{\psi}_{L^2}^2
    + \lambda \norm*{\abs*{D^1\psi}}_{L^2}^2
    \le 2\lambda \varepsilon^{-8}\sum_{i=1}^k\norm*{\mathbf{V}_i\psi}_{L^2}^2
    + C(\varepsilon)\norm*{L_{\varepsilon}\psi}_{L^2}^2.
  \end{equation*}
  Then it immediately follows that 
  \begin{equation}
    \label{eq:carleman-estimate:step2:reduction}
    \lambda\norm*{\psi}_{L^2} + \norm*{\abs*{D^1\psi}}_{L^2}
    \le C(\varepsilon)\lambda^{-\frac{1}{2}}\norm*{L_{\varepsilon}\psi}_{L^2} + 4\varepsilon^{-4}\sum_{i=1}^{k}\norm*{\mathbf{V}_i(\psi)}_{L^2}.
  \end{equation}
  Now recall from \Cref{eq:strict-T-Z-null-convexity:cond1} that on $B_{\varepsilon^{10}}$
  \begin{equation*}
    \abs*{\Box_{\Metric}(f_{\varepsilon})}\le C(\varepsilon).
  \end{equation*}
  Thus, using
  \Cref{eq:carleman-estimates:L-epsilon:L2-relation-to-wave}, we have
  from \Cref{eq:carleman-estimate:step2:reduction} that
  \begin{align*}
    \lambda\norm*{\psi}_{L^2}
    + \norm*{\abs*{D^1\psi}}_{L^2}
    \le{}& C(\varepsilon)\lambda^{-\frac{1}{2}}\left(
           \norm*{e^{-\lambda f_{\varepsilon}}\Box_{\Metric}(e^{\lambda f_{\varepsilon}}\psi)}_{L^2}
           + \lambda \norm*{\Box_{\Metric}(f_{\varepsilon})\psi}_{L^2}
           \right)
           + 4\varepsilon^{-4}\sum_{i=1}^k\norm*{\mathbf{V}_i(\psi)}_{L^2}\\
      \le{}& C(\varepsilon)\lambda^{-\frac{1}{2}}
           \norm*{e^{-\lambda f_{\varepsilon}}\Box_{\Metric}(e^{\lambda f_{\varepsilon}}\psi)}_{L^2}
           + C(\varepsilon)\lambda^{\frac{1}{2}} \norm*{\psi}_{L^2}
           + 4\varepsilon^{-4}\sum_{i=1}^k\norm*{\mathbf{V}_i(\psi)}_{L^2}
           .
  \end{align*}
  For $\lambda$ sufficiently large relative to $C(\varepsilon)$,
  $\frac{\lambda}{2} > C(\varepsilon)\lambda^{\frac{1}{2}}$, and thus
  \begin{equation}
    \label{eq:carleman-estimate:step1:renormalized-conclusion:reduction}
    \lambda\norm*{\psi}_{L^2}
    + \norm*{\abs*{D^1\psi}}_{L^2}
    \le C(\varepsilon)\lambda^{-\frac{1}{2}}\norm*{e^{-\lambda f_{\varepsilon}}\Box_{\Metric}(e^{\lambda f_{\varepsilon}}\psi)}_{L^2}
    + 8 \varepsilon^{-4}\sum_{i=1}^k\norm*{\mathbf{V}_i(\psi)}_{L^2}. 
  \end{equation}
  
  Using this, we can compute that
  \begin{align*}
    &\lambda\norm*{\psi}_{L^2}
    + \norm*{e^{-\lambda f_{\varepsilon}}\abs*{D^1(e^{\lambda f_{\varepsilon}}\psi)}}_{L^2}\\
    \lesssim{}& \lambda\norm*{\psi}_{L^2}
                + \norm*{\abs*{D^1\psi}}_{L^2}
                + \varepsilon^{-1}\lambda\norm*{\psi}_{L^2}\\
    \lesssim{}& \varepsilon^{-1}\left(
                C(\varepsilon)\lambda^{-\frac{1}{2}}\norm*{e^{-\lambda f_{\varepsilon}}\Box_\Metric (e^{\lambda f_{\varepsilon}}\psi)}_{L^2}
                + 8 \varepsilon^{-4}\sum_{i=1}^k\norm*{\mathbf{V}_i(\psi)}_{L^2}
                \right)\\
    \lesssim{}& \varepsilon^{-1}\left(
                C(\varepsilon)\lambda^{-\frac{1}{2}}\norm*{e^{-\lambda f_{\varepsilon}}\Box_\Metric (e^{\lambda f_{\varepsilon}}\psi)}_{L^2}
                + 8 \varepsilon^{-4}\sum_{i=1}^k\norm*{e^{-\lambda f_{\varepsilon}}\mathbf{V}_i(e^{\lambda f_{\varepsilon}}\psi)}_{L^2}
                + 8 \widetilde{C}\varepsilon^3\lambda \norm*{\psi}_{L^2}
                \right)\\
    \lesssim{}& \lambda^{-\frac{1}{2}} \norm*{e^{-\lambda f_{\varepsilon}}\Box_\Metric (e^{\lambda f_{\varepsilon}}\psi)}_{L^2}
                + \varepsilon^{-5}\sum_{i=1}^k\norm*{e^{-\lambda f_{\varepsilon}}\mathbf{V}_i(e^{\lambda f_{\varepsilon}}\psi)}_{L^2}
                + \varepsilon^2\lambda \norm*{\psi}_{L^2}.
  \end{align*}
  Then define
  $\psi = e^{-\lambda f_{\varepsilon}}\phi \in
  C^{\infty}_0(B_{\varepsilon^{10}})$. Then
  \Cref{eq:carleman-estimate} rewritten in terms of $\psi$ is exactly
  \begin{equation}
    \label{eq:carleman-estimate:step1:renormalized-conclusion}
    \lambda \norm*{\psi}_{L^2}
    + \norm*{e^{-\lambda f_{\varepsilon}}\abs*{D^1(e^{\lambda f_{\varepsilon}}\psi)}}_{L^2}
    \le C(\varepsilon)\lambda^{-\frac{1}{2}}\norm*{e^{-\lambda f_{\varepsilon}}\Box_{\Metric}(e^{\lambda f_{\varepsilon}}\psi)}_{L^2}
    + \sum_{i=1}^k\varepsilon^{-6}\norm*{e^{-\lambda f_{\varepsilon}}\mathbf{V}_i(e^{\lambda f_{\varepsilon}}\psi)}_{L^2},
  \end{equation}
  so we are done for $\varepsilon$ sufficiently small, concluding the
  proof of \Cref{prop:carleman-estimate}.
\end{proof}

The second is a Carleman estimate that we will need to exploit the
transport equations in our system of equations.
\begin{lemma}
  \label{lemma:Carleman-estimate:transport}
  Assume that $\varepsilon\le A^{-1}$ is sufficiently small,
  $e_{\varepsilon}$ is a negligible perturbation, and that
  $h_{\varepsilon}:B_{\varepsilon^{10}}\to \Real^{+}$ satisfies
  \begin{equation}
    \label{eq:Carleman-estimate:transport:conditions}
    h_{\varepsilon}(x_0) = \varepsilon, \qquad
    \sup_{x\in B_{\varepsilon^{10}}}\sum_{j=1}^2\varepsilon^j\abs*{D^jh_{\varepsilon}(x)}\le 1,\qquad
    \abs*{\GeodesicVF(h_{\varepsilon})(x_0)}\ge c,
  \end{equation}
  where $\GeodesicVF$ is a smooth vectorfield that is bounded near
  $x_0$, and $c\in \Real^+$.

  Then there is some $C(\varepsilon)$ sufficiently large such that for
  any $\lambda\ge C(\varepsilon)$, $\phi\in C_0^{\infty}(B_{\varepsilon^{10}})$,
  \begin{equation}
    \label{eq:Carleman-estimate:transport}
    \norm*{e^{-\lambda f_{\varepsilon}}\phi}_{L^2}
    \lesssim  \lambda^{-1}\norm*{e^{-\lambda f_{\varepsilon}}\GeodesicVF(\phi)}_{L^2},
  \end{equation}
  where $f_{\varepsilon}=\ln(h_{\varepsilon}+ e_{\varepsilon})$. 
\end{lemma}
\begin{proof}
  Without loss of generality, we assume that $\phi$ is real-valued and
  define
  $\psi = e^{-\lambda f_{\varepsilon}}\phi\in
  C_0^{\infty}(B_{\varepsilon^{10}})$. Then, proving
  \Cref{eq:Carleman-estimate:transport} is equivalent to proving the
  following estimate for $\psi$.
  \begin{equation}
    \label{eq:Carleman-estimate:transport:conjugated-version}
    \norm*{\psi}_{L^2}
    \le 4\norm*{\lambda^{-1}\GeodesicVF(\psi) + \GeodesicVF(f_{\varepsilon})\psi}_{L^2}.
  \end{equation}
  Integrating by parts, we find that
  \begin{equation*}
    \begin{split}
      \int_{B_{\varepsilon^{10}}}
    \left( \lambda^{-1}\GeodesicVF(\psi) + \GeodesicVF(f_{\varepsilon})\psi \right) \GeodesicVF(f_{\varepsilon})\psi
      ={}& \int_{B_{\varepsilon^{10}}} \left( \GeodesicVF(f_{\varepsilon})\psi \right)^2
           - 2\lambda^{-1}\int_{B_{\varepsilon^{10}}}\psi^2\CovariantDeriv_{\alpha}(\GeodesicVF(f_{\varepsilon})\GeodesicVF^{\alpha}).
    \end{split}    
  \end{equation*}
  Using the fact that $\GeodesicVF$ is bounded near $x_0$ and
  \Cref{eq:Carleman-estimate:transport:conditions}, we have that
  in $B_{\varepsilon^{10}}$, for $\varepsilon$ sufficiently small,
  \begin{equation*}
    \abs*{\GeodesicVF(f_{\varepsilon})}\ge c,\qquad \abs*{\CovariantDeriv_{\alpha}(\GeodesicVF(f_{\varepsilon})\GeodesicVF^{\alpha})}\le C(\varepsilon).
  \end{equation*}
  Thus, for $\lambda$ sufficiently large,
  \begin{equation*}
    \int_{B_{\varepsilon^{10}}}\left(\lambda^{-1}\GeodesicVF(\psi) + \GeodesicVF(f_{\varepsilon})\psi\right) \GeodesicVF(f_{\varepsilon})\psi
    \ge \frac{c}{2}\int_{B_{\varepsilon^{10}}}(\GeodesicVF(f_{\varepsilon})\psi)^2,
  \end{equation*}
  and the bound in
  \Cref{eq:Carleman-estimate:transport:conjugated-version} and thus
  \Cref{eq:Carleman-estimate:transport} follow immediately. 
\end{proof}

\subsection{The Kerr-de Sitter spacetime}
\label{sec:KdS}

In this section, we introduce the Kerr-de Sitter family. Kerr-de
Sitter is a two-parameter family of black hole solutions to
($\Lambda$-EVE) with positive cosmological constant $\Lambda$
parametrized by the black hole mass and angular momentum $(M,a)$. In
Boyer-Lindquist coordinates, the \KdS{} metric takes the form
\begin{equation}
  \label{eq:KdS-metric:metric:BL}
  \begin{split}
    \Metric_{M,a,\Lambda}
    ={}& -\frac{\Delta}{\left(1+\gamma\right)^2\abs*{q}^2}(dt - a \sin^2\theta d\phi)^2
         + \abs*{q}^2 \left(
         \Delta^{-1}dr^2 +  \kappa^{-1}d\theta^2
         \right)\\
       & + \frac{\kappa\sin^2\theta}{(1+\gamma)^2\abs*{q}^2}\left(a\,dt - (r^2+a^2)d\phi\right)^2,
  \end{split}  
\end{equation}
and the inverse metric takes the form
\begin{equation}
  \label{eq:KdS-metric:inverse:BL}
  \Metric_{M,a,\Lambda}^{-1} = \frac{1}{\abs*{q}^2}\left(
    \Delta \,\partial_r^2
    +  \frac{(1+\gamma)^2}{\kappa \sin^2\theta}\left(a\sin^2\theta\,\partial_t + \partial_\phi\right)^2
    + \kappa\,\partial_\theta^2
    - \frac{(1+\gamma)^2}{\Delta}\left(\left(r^2+a^2\right)\partial_t + a\partial_\phi \right)^2
  \right),
\end{equation}
where
\begin{equation}
  \label{eq:KdS-metric:coefficients-def}
  \begin{gathered}
    q \vcentcolon= r+ \ImagUnit a\cos\theta,\qquad
  \Delta_{M,a,\Lambda} \vcentcolon= (r^2+a^2)\left(1-\frac{\Lambda}{3}r^2\right)-2Mr,\\
  \kappa \vcentcolon= 1+ \gamma\cos^2\theta,\qquad
  \gamma \vcentcolon= \frac{\Lambda a^2}{3}.
  \end{gathered}  
\end{equation}

We briefly remark that the main geometric property of \KdS{} that we
will make use of throughout the ensuing proof is the absence of
$\KillT$-trapped null geodesics on \KdS. There is some nuance in the
statement of the absence of $\KillT$-trapped null geodesics however
since the definition of $\KillT$-trapped null geodesics depends on the
choice of $\KillT$. In Kerr, the choice $\KillT=\partial_t$, suffices
to rule out the existence of a $\KillT$-trapped null geodesics for all
subextremal members of the Kerr family. However, for \KdS, this choice
does not work for all subextremal members. We refer the interested
reader to \cite{petersenWaveEquationsKerr2024,
  petersenStationarityFredholmTheory2024} for a detailed discussion.

\section{Basic properties of stationary \texorpdfstring{$\Lambda$-}{}vacuum spacetimes}
\label{sec:Lambda-stationary-spacetimes}

In this section, we review basic properties of stationary
$\Lambda$-vacuum spacetimes. The main object introduced will be a
Mars-Simon tensor well-adapted to our subsequent proof. For a more
thorough introduction to the Mars-Simon tensor for $\Lambda\neq 0$
spacetimes, we refer the reader to
\cite{marsSpacetimeCharacterizationKerrNUTAde2015,
  marsCharacterizationAsymptoticallyKerr2016}.

\subsection{Preliminaries}
\label{sec:Lambda-stationary-spacetimes:preliminaries}

Let $\KillT$ be a Killing vectorfield of $(\Manifold, \Metric)$, with
\begin{equation}
  \label{eq:N-def}
  N\vcentcolon= -\Metric(\KillT, \KillT).
\end{equation}
\begin{definition}
  We define the real (resp. complex) \emph{Killing two-form} $F$
  (resp. $\mathcal{F}$) by
  \begin{equation*}
    F_{\mu\nu} \vcentcolon= \CovariantDeriv_{\mu}\KillT_{\nu},\qquad
    \mathcal{F}_{\mu\nu}\vcentcolon= F_{\mu\nu} + \ImagUnit F_{\mu\nu}^{*}.
  \end{equation*}
  The complex Killing two-form $\mathcal{F}$ is self-dual satisfying
  \begin{equation*}
    \mathcal{F}_{\mu\nu}^{*} = \ImagUnit \mathcal{F}_{\mu\nu},
  \end{equation*}
  where $*$ is the Hodge dual operator.
\end{definition}
\begin{lemma}[Basic properties of $\mathcal{F}$]
  $\mathcal{F}$ satisfies the following basic properties.
  \begin{equation}
    \label{eq:F-cal:basic-props}
    \mathcal{F}_{\mu\rho}\tensor[]{\mathcal{F}}{_{\nu}^{\rho}} = \frac{1}{4}\mathcal{F}^2\Metric_{\mu\nu},\qquad
    \mathcal{F}^2 \vcentcolon= \mathcal{F}_{\mu\nu}\mathcal{F}^{\mu\nu}, \qquad
    \mathcal{F}_{\mu\rho}\overline{\mathcal{F}}^{\mu\rho} = 0.
  \end{equation}
\end{lemma}
\begin{proof}
  Direct computation.
\end{proof}

\begin{definition}[Ernst one-form]
  \label{def:Ernst-one-form}
  We define the \emph{Ernst one-form} $\bm{\sigma}_{\mu}$ by
  \begin{equation*}
    \bm{\sigma}_{\mu} \vcentcolon= 2 \KillT^{\nu}\mathcal{F}_{\nu\mu}.    
  \end{equation*}
\end{definition}
\begin{definition}[Ernst potential]
  \label{def:Ernst-potential}
  Since the Ernst one-form is closed (see
  \Cref{lemma:sigma-basic-props}), we infer that there exists a
  function $\sigma: \Manifold \to \Complex$ such that
  \begin{equation*}
    \bm{\sigma}_{\mu} = \CovariantDeriv_{\mu}\sigma
  \end{equation*}
  We call the function $\sigma$ the \emph{Ernst potential}. Observe that the
  Ernst potential is in general only fixed up to a constant. To fix
  our choice of Ernst potential, we will require that
  \begin{equation*}
    \Re\sigma = -\Metric(\KillT, \KillT),\qquad
    \Im \sigma(S_0^+) = 0,
  \end{equation*}
  where by $S_0^+$ we denote the point on $S_0$ that is identified to
  the north pole of the unit sphere under the identification of
  $\Sigma_0$ with $\curlyBrace*{x\in \Real^3:\frac{1}{2}<\abs*{x}<\frac{5}{2}}$.
\end{definition}

The Ernst one-form $\bm{\sigma}_{\mu}$ satisfies some basic properties.
\begin{lemma}[Basic properties of the Ernst one-form]
  \label{lemma:sigma-basic-props}
  The Ernst one-form $\bm{\sigma}_{\nu}$ as defined in
  \cref{def:Ernst-one-form} satisfies the following properties.
  \begin{equation}
    \label{eq:deriv-sigma}
    \CovariantDeriv_{\mu}\bm{\sigma}_{\nu}
    = - \frac{\mathcal{F}^2}{4}\Metric_{\mu\nu}
    - 2 \mathcal{T}_{\mu\nu}
    + 2\KillT^{\rho}\KillT^{\sigma}\mathcal{W}_{\rho\mu\sigma\nu}
    - \frac{2\Lambda}{3}\left(N\Metric_{\mu\nu} + \KillT_{\mu}\KillT_{\nu}\right),
  \end{equation}
  and
  \begin{equation}    
    \label{eq:sigma-basic-props}
    \CovariantDeriv_{[\mu}\bm{\sigma}_{\nu]}=0, \qquad 
    \KillT\cdot\bm{\sigma} = 0, \qquad
    \Metric(\bm{\sigma}^{\sharp},\bm{\sigma}^{\sharp}) = \Metric(\KillT, \KillT)\mathcal{F}^2.
  \end{equation}
  We can also observe that
  \begin{equation}
    \label{eq:sigma-basic-props:2} 
    \bm{\sigma}^{\mu}\mathcal{F}_{\mu\nu} = -\frac{1}{2}\mathcal{F}^2\KillT_{\nu}, \qquad
    \mathcal{F}_{\mu\nu}\bm{\sigma}^{\mu}\overline{\bm{\sigma}}^{\nu}= 0.
  \end{equation}
\end{lemma}
\begin{proof}
  We can compute that 
  \begin{align*}
    \CovariantDeriv_{\mu}\bm{\sigma}_{\nu}
    ={}& 2\CovariantDeriv_{\mu}\KillT^{\rho}\mathcal{F}_{\rho\nu}
         + 2\KillT^{\rho}\CovariantDeriv_{\mu}\mathcal{F}_{\rho\nu}
    \\
    ={}&\left(\tensor[]{\mathcal{F}}{_{\mu}^{\rho}} + \tensor[]{\overline{\mathcal{F}}}{_{\mu}^{\rho}}\right)\mathcal{F}_{\rho\nu}
         + 2\KillT^{\rho}\KillT^{\alpha}\left(\mathcal{W}_{\alpha\mu\rho\nu} + \frac{4\Lambda}{3}\mathcal{I}_{\alpha\mu\rho\nu}\right)\\
    ={}&\left(\tensor[]{\mathcal{F}}{_{\mu}^{\rho}} + \tensor[]{\overline{\mathcal{F}}}{_{\mu}^{\rho}}\right)\mathcal{F}_{\rho\nu}
         + 2\KillT^{\rho}\KillT^{\alpha}\mathcal{W}_{\alpha\mu\rho\nu}
         - \frac{2\Lambda}{3}\left(- \Metric\left( \KillT,\KillT \right)\Metric_{\mu\nu}+ \KillT_{\mu}\KillT_{\nu} \right),
  \end{align*}
  so that it is clear that the Ernst one-form is closed.  The
  properties in \Cref{eq:sigma-basic-props:2} follow directly from
  the relations in \Cref{eq:F-cal:basic-props}.
\end{proof}

\begin{definition}[Energy-momentum tensor of $\mathcal{F}$]
  \label{def:EMT-F}
  We define the \emph{energy-momentum tensor of $\mathcal{F}$} to be the symmetric 2-tensor defined by
  \begin{equation*}
    \mathcal{T}_{\mu\nu} \vcentcolon= \frac{1}{2}\mathcal{F}_{\mu\rho}\tensor[]{\overline{\mathcal{F}}}{_{\nu}^{\rho}}
    .
  \end{equation*}
\end{definition}

\begin{lemma}
  \label{lemma:EMT-F:basic-props}
  The energy-momentum tensor of $\mathcal{F}$, $\mathcal{T}_{\mu\nu}$, satisfies the following
  basic properties.
  \begin{equation*}
    \mathcal{T}_{\mu\nu} = \mathcal{T}_{\nu\mu}, \qquad
    \Trace \mathcal{T}=0, \qquad
    \mathcal{T}_{\mu\rho}\tensor[]{\mathcal{T}}{_{\nu}^{\rho}}
    = \frac{1}{4}\mathcal{T}_{\rho\sigma}\mathcal{T}^{\rho\sigma}\Metric_{\mu\nu},\qquad
    \mathcal{T}_{\rho\sigma}\cdot \mathcal{T}^{\rho\sigma}
    = \frac{1}{16}\mathcal{F}^2\overline{\mathcal{F}}^2.
  \end{equation*}
\end{lemma}
\begin{proof}
  Basic computation.
\end{proof}

\begin{definition}
  \label{def:eta:def}
  We define the real one-form
  \begin{align}
    \label{eq:eta:def}
    \bm{\eta}_{\mu} \vcentcolon= \overline{\bm{\sigma}}^{\rho}\mathcal{F}_{\rho\mu}
    = \bm{\sigma}^{\rho}\overline{\mathcal{F}}_{\rho\mu}.
  \end{align}
  It is a direct computation to show that
  \begin{equation}
    \label{eq:g-eta-T}
    \bm{\eta}\cdot\KillT = \frac{1}{2}\bm{\sigma}\cdot\overline{\bm{\sigma}}.
  \end{equation}
\end{definition}
Since $\KillT$ is a Killing vector of $\Metric$, we have that
\begin{equation}
  \label{eq:deriv-FCal}
  \CovariantDeriv_{\mu}\mathcal{F}_{\alpha\beta}
  = \KillT^{\nu}\left(\mathcal{W}_{\nu\mu\alpha\beta} + \frac{4\Lambda}{3}\mathcal{I}_{\nu\mu\alpha\beta}\right),
\end{equation}
where $\mathcal{W}$ is the complexified Weyl tensor (recall the
complex notation defined in \Cref{sec:complex-notation}), and
$\mathcal{I}$ is the metric on complex self-dual two-forms,
\begin{equation}
  \label{eq:ICal:def}
  \mathcal{I}_{\alpha\beta\mu\nu}\vcentcolon= \frac{1}{4}\left(
    \Metric_{\alpha\mu}\Metric_{\beta\nu}
    - \Metric_{\alpha\nu}\Metric_{\beta\mu}
    + \ImagUnit \varepsilon_{\alpha\beta\mu\nu}
  \right).
\end{equation}

We thus have the following relations
\begin{align}
  \label{eq:deriv-FCalSquared}
  \CovariantDeriv_{\mu}\mathcal{F}^2 ={}& 2\KillT^{\nu}\mathcal{F}^{\alpha\beta}\mathcal{W}_{\nu\mu\alpha\beta} + \frac{4\Lambda}{3}\bm{\sigma}_{\mu},\\
  \label{eq:divergence-sigma}
  \CovariantDeriv\cdot\bm{\sigma} ={}& -\mathcal{F}^2 - 2\Lambda N,\\
  \label{eq:wave-FCalSquared}
  \CovariantDeriv_{\mu}\CovariantDeriv^{\mu}\mathcal{F}^2
  ={}& -\mathcal{W}_{\alpha\beta\mu\nu}\mathcal{F}^{\alpha\beta}\mathcal{F}^{\mu\nu}
       - \frac{4}{3}\Lambda \mathcal{F}^2
       - \frac{N}{2}\left( \mathcal{W}_{\alpha\beta\mu\nu}\mathcal{W}^{\alpha\beta\mu\nu} + \frac{16\Lambda^2}{3} \right).
\end{align}

\begin{definition}
  \label{def:R}
  If the self-dual Killing form $\mathcal{F}$ is regular everywhere,
  then there exists a smooth complex function $R:\Manifold\to \Complex$ such that
  \begin{equation}
    \label{eq:R:def}
    \mathcal{F}^2=-4R^2.
  \end{equation}
\end{definition}
\begin{remark}
  There are two possible choices at each point, but we assume that one
  choice has been made globally. From this point forward, we will use
  the choice
  \begin{equation}    
    R = -\frac{\ImagUnit}{2}\sqrt{\mathcal{F}^2}.
  \end{equation}
\end{remark}
We now construct a null tetrad. Since $\mathcal{F}$ is regular, there
exist two different real null eigenvectors $\bm{\ell}_+, \bm{\ell}_-$ of $\mathcal{F}$
at $p$ with opposite eigenvalues $\pm R$
\begin{equation}
  \label{eq:L-LBar:eigenvector-properties} 
  \bm{\ell}_+\cdot \mathcal{F} = R \bm{\ell}_+, \qquad 
  \bm{\ell}_-\cdot \mathcal{F} = -R \bm{\ell}_-, \qquad
  \Metric(\bm{\ell}_+,\bm{\ell}_+) = \Metric(\bm{\ell}_-, \bm{\ell}_-) = 0. 
\end{equation}
We can now establish an outgoing-ingoing null frame
$(e_4 ,e_3)=(\bm{\ell}_+, \bm{\ell}_-)$ by normalizing
$(\bm{\ell}_+, \bm{\ell}_-)$ so that $\Metric(e_4, e_3)=-2$.  We then
define $(e_1,e_2)$ horizontal vectors (i.e. orthogonal to $e_4,e_3$)
such that
\begin{equation*}
  \Metric(e_1,e_2) = \Metric(e_2,e_2)=1,\qquad \Metric(e_1,e_2)=0. 
\end{equation*}
\begin{remark}
  There is some freedom in whether we choose which of the
  outgoing-ingoing pair of null vectors have positive/negative
  eigenvalue. We also note that this only fixes the
  choice of null frame $(e_4,e_3)$ up to scaling.
\end{remark}
We can observe that $(\bm{\ell}_-,\bm{\ell}_+)$ are also eigenvectors of
$\mathcal{T}_{\mu\nu}$ in the sense that
\begin{equation*}
  (\bm{\ell}_+)^{\mu}\mathcal{T}_{\mu\nu} = -\frac{1}{2}R \overline{R}(\bm{\ell}_+)_{\nu},\qquad
  (\bm{\ell}_-)^{\mu}\mathcal{T}_{\mu\nu} = -\frac{1}{2}R \overline{R}(\bm{\ell}_-)_{\nu}, 
\end{equation*}
and now fixing for the remainder of this section the null frame
$(e_4,e_3)=(\bm{\ell}_+,\bm{\ell}_-)$,
\begin{equation*}
  \begin{gathered}
    \Metric \left(e_4, \bm{\sigma}\right)
  = 2R\Metric(\KillT, e_4), \qquad
  \Metric \left(e_3, \bm{\sigma}\right)
  = -2R\Metric(\KillT, e_3), \\
  \Metric(e_4, \bm{\eta})
  = 2R \overline{R}\Metric(\KillT, e_4),\qquad
  \Metric(e_3, \bm{\eta})
  = 2R \overline{R}\Metric(\KillT, e_3). 
  \end{gathered}  
\end{equation*}
We can also now define and compute the null components of
$\mathcal{F}$.
\begin{definition}
  \label{def:F-null-components:def}
  We define the null components of $\mathcal{F}$ as follows.
  \begin{equation}
    \label{eq:F-null-components:def}
    B(\mathcal{F}) \vcentcolon= \mathcal{F}(e_a,e_4),\qquad
    \underline{B}(\mathcal{F}) \vcentcolon= \mathcal{F}(e_a,e_3),\qquad
    P(\mathcal{F}) \vcentcolon= \mathcal{F}(e_3,e_4).
  \end{equation}
\end{definition}
\begin{remark}
  Observe that the null components
  $B(\mathcal{F}), \underline{B}(\mathcal{F}), P(\mathcal{F})$
  completely determine $\mathcal{F}$. Moreover,
  \begin{equation}
    \label{eq:Fcal-34-ab-components-relationship}
    -\ImagUnit \mathcal{F}_{34} = \LeftDual{\mathcal{F}}_{34}
    = \frac{1}{2}\in_{34ab}\mathcal{F}^{ab}
    = \frac{1}{2}\in_{ab}\mathcal{F}^{ab}.
  \end{equation}
\end{remark}

\begin{corollary} 
  \label{coro:FCal-null-components}
  $\mathcal{F}$ satisfies
  \begin{equation}
    \label{eq:FCal-null-decomposition}
    \mathcal{F}_{\alpha\beta}
    = \frac{R}{2}\left(
      -(e_4)_{\alpha}(e_3)_{\beta}
      + (e_4)_{\beta}(e_3)_{\alpha}
      - \ImagUnit\in_{\alpha\beta\mu\nu}(e_4)^{\mu}(e_3)^{\nu}
    \right),
  \end{equation}
  and moreover, the only non-zero null
  components of $\mathcal{F}$ are
  \begin{equation}
    \label{eq:FCal-null-components}
    \mathcal{F}_{34} = \frac{\ImagUnit}{2}\in^{ab}\mathcal{F}_{ab} = R.
  \end{equation}
\end{corollary}
\begin{proof}
  Direct computation using \Cref{eq:L-LBar:eigenvector-properties}.
\end{proof}

\begin{corollary}
  \label{coro:TCal-null-components}
  $\mathcal{T}$ satisfies
  \begin{equation}
    \label{eq:TCal-null-components}
    \mathcal{T}_{\alpha\beta}
    = \frac{1}{2}R \overline{R} \left(
      (e_3)_{(\alpha}(e_4)_{\beta)}
      + \slashed{\Metric}_{\alpha\beta}\right),    
  \end{equation}
  where
  $\slashed{\Metric}_{\alpha\beta} = (e_1)_{\alpha}(e_1)_{\beta} +
  (e_2)_{\alpha}(e_2)_{\beta}$.
\end{corollary}
\begin{proof}
  Direct computation. 
\end{proof}

We also have the following computations.
\begin{lemma}
  \label{lemma:basic-computations:eta-T-sigma-null-decomp}
  We have the following relations.
  \begin{equation}
    \label{eq:basic-computations:eta-T-sigma-null-decomp}
    \begin{split}
      \bm{\eta} + 2R \overline{R}\KillT
      ={}& -4R \overline{R} \left( \Metric(e_3, \KillT)e_4 + \Metric(e_4,\KillT)e_3 \right)\\
      R\overline{\bm{\sigma}} + \overline{R}\bm{\sigma}
      ={}& 4R \overline{R}\left(  \Metric(e_3,\KillT)e_4 - \Metric(e_4,\KillT)e_3 \right).
    \end{split}
  \end{equation}
\end{lemma}
\begin{proof}
  See \Cref{appendix:lemma:basic-computations:eta-T-sigma-null-decomp}.
\end{proof}

For what follows, it will be convenient to rescale many of the
previously defined quantities by $R$. 
\begin{definition}
  \label{def:rescaled-tensor-qtys:def}
  We define the \emph{rescaled Ernst one-form} $\mathbf{P}$, the
  \emph{rescaled Killing 2-form} $\mathcal{X}$, the \emph{rescaled
    auxiliary one-form} $\mathbf{p}$, and the \emph{rescaled
    energy-momentum tensor} $\bm{\tau}$ via the following relations.
  \begin{equation}
    \label{eq:rescaled-tensor-qtys:def}
    \bm{\sigma}_{\alpha} = 2R \mathbf{P}_{\alpha},\qquad
    \mathcal{X}_{\alpha\beta} = \frac{1}{R}\mathcal{F}_{\alpha\beta},\qquad        
    \mathbf{p} \vcentcolon= \frac{1}{2R\overline{R}}\bm{\eta},\qquad
    \bm{\tau} \vcentcolon= \frac{2}{R \overline{R}}\mathcal{T}
    . 
  \end{equation}
\end{definition}

We list below some properties that follow immediately from the
preceding definitions.
\begin{lemma}
  \label{lemma:renormalized-qtys:basic-props}
  The following hold.
  \begin{equation}
    \label{eq:renormalized-qtys:basic-props}
    \begin{gathered}
      \mathcal{X}_{\mu\rho}\tensor[]{\mathcal{X}}{_{\nu}^{\rho}} = -\Metric_{\mu\nu},\qquad
      \mathcal{X}_{\mu\nu}\mathcal{X}^{\mu\nu}=-4,\qquad
      \mathcal{X}_{\mu\nu}\overline{\mathcal{X}}^{\mu\nu}=0,\\
      \mathbf{P}_{\mu}=\KillT^{\rho}\mathcal{X}_{\rho\mu},\qquad
      \mathbf{P}^{\rho}\mathcal{X}_{\rho\mu}=\KillT_{\mu},\qquad
      \mathbf{P}\cdot \mathbf{P} = -\mathbf{p}\cdot \mathbf{p} = N,\qquad
      \mathbf{p}\cdot\KillT = -\mathbf{P}\cdot \overline{\mathbf{P}},\\
      \KillT(R)=  0,\qquad
      \LieDerivative_{\KillT}\mathcal{X}_{\mu\nu}=0,\qquad
      [\KillT, \mathbf{P}] = [\KillT, \mathbf{p}] = 0.
    \end{gathered}
  \end{equation}
\end{lemma}
\begin{proof}
  Direct computation.
\end{proof}

We define the following vectorfield which essentially extracts the
components of $\KillT$ along the $e_3$ and $e_4$ directions. 
\begin{definition}
  \label{def:THat:def}
  Define
  \begin{equation}
    \label{eq:THat:def}
    \widehat{T} \vcentcolon = \frac{\abs*{P}^2}{2}\left(\mathbf{p} + \KillT \right). 
  \end{equation}
  Observe that from \Cref{lemma:basic-computations:eta-T-sigma-null-decomp} and the definition of $\mathbf{p}$ in \Cref{def:rescaled-tensor-qtys:def}, we have that
  \begin{equation}
    \label{eq:THat:null-decomp}
    \frac{1}{\abs*{P}^2}\widehat{T} = - \frac{1}{2}\left( \Metric(e_3, \KillT)e_4 + \Metric(e_4, \KillT)e_3 \right)
    = \KillT - \horProj{\KillT}
    .
  \end{equation}
\end{definition}

We now define the quantities leading up the definition of $y$, which
will be the central function in our construction of a pseudoconvex
foliation in $\mathbf{E}$.
\begin{definition}
  We define
  \begin{equation*}
    \sigma_0 = \sigma - c_{S_0}.
  \end{equation*}
  Motivated by Section 4.2 of
  \cite{marsCharacterizationAsymptoticallyKerr2016}, we also define\footnote{Note that in general, $J\in \mathbb{C}$, so it is \emph{not} necessary to assume that $R^2-\Lambda \sigma_0 \in \Real^+$. }
  \begin{equation}
    \label{eq:J:def}
    J \vcentcolon= \frac{R + \sqrt{R^2-\Lambda \sigma_0}}{\sigma_0}.
  \end{equation}
  Observe that $J$ then satisfies
  \begin{equation}
    \label{eq:J:quadratic-eqn}
    \sigma_0 J^2 - 2JR + \Lambda = 0. 
  \end{equation}
\end{definition}

Following Section 4.2
\cite{marsCharacterizationAsymptoticallyKerr2016}, we will choose $Q$
such that
\begin{equation}
  \label{eq:MST:Q:def}
  Q = \frac{3J}{R} - \frac{\Lambda}{R^2}. 
\end{equation}
\begin{remark}
  We remark that from the assumption that $N\le 0$ on $S_0$,
  $\Re\sigma=N$, and $c_{S_0}>0$, we have that $\sigma_0\neq 0$ in
  some neighborhood $\mathbf{O}_{\varepsilon_{\Horizon}}$ of $S_0$ and
  $\underline{\mathbf{O}}_{\varepsilon_{\Horizon}}$ for some
  $\varepsilon_{\Horizon}< \varepsilon_0$. As a result, we have that
  $J, Q$ are well-defined on these neighborhoods.
\end{remark}

\begin{lemma}
  \label{lemma:P-sigma-alt-expressions}
  $\sigma_0$ can be expressed in terms of $Q$ and $\mathcal{F}$ as
  \begin{equation}
    \label{eq:sigma-0-in-terms-of-Q-F}
    \sigma_0 = 6 \mathcal{F}^2 \frac{Q\mathcal{F}^2 + 2\Lambda}{(Q\mathcal{F}^2-4\Lambda)^2}. 
  \end{equation}
  Moreover, we also have that
  \begin{equation}
    \label{eq:P-in-terms-of-Q-F}
    P = -\frac{6\ImagUnit\sqrt{\mathcal{F}^2}}{Q\mathcal{F}^2-4\Lambda}. 
  \end{equation}
  In particular, it then follows that
  \begin{equation}
    \label{eq:P-sigma0-relation}
    \sigma_0 = -\frac{1}{6}P^2\left(Q \mathcal{F}^2 + 2\Lambda\right).
  \end{equation}
\end{lemma}
\begin{proof}
  The relation in \Cref{eq:P-sigma0-relation} follows immediately
  from the other relations.

  To prove \Cref{eq:sigma-0-in-terms-of-Q-F}, observe that
  \begin{align*}
    6 \mathcal{F}^2 \frac{Q\mathcal{F}^2 + 2\Lambda}{(Q\mathcal{F}^2-4\Lambda)^2}
    ={}& \frac{\mathcal{F}^2}{4R^2}\frac{(-2JR+\Lambda)}{J^2}\\
    ={}& \frac{2JR - \Lambda}{J^2}\\
    ={}& \sigma_0,
  \end{align*}
  where we used the fact that $J$ solves the equation in
  \Cref{eq:J:quadratic-eqn}. This proves
  \Cref{eq:sigma-0-in-terms-of-Q-F}.

  To prove \Cref{eq:P-in-terms-of-Q-F}, we observe that
  \begin{align*}
    \frac{6\ImagUnit\sqrt{\mathcal{F}^2}}{Q\mathcal{F}^2-4\Lambda}
    ={}& \frac{12 R}{12 JR}\\
    ={}& \frac{1}{J} = -P,
  \end{align*}
  as desired.
\end{proof}
We then have the following reformulation of \Cref{eq:b-c-k:def}.
\begin{corollary}
  \label{cor:b-c-k:def:renorm}
  The definitions of $(b,c,k)$ in \Cref{eq:b-c-k:def} are equivalent
  to the definitions
  \begin{equation}
    \label{eq:b-c-k:def:renorm}
    \begin{split}
      b &= \frac{2}{3}\left(3JR - \Lambda\right)P^3,\\
      c &= N - \Re\sigma_0,\\
      k &= \abs*{P}^2\CovariantDeriv_{\alpha}z\CovariantDeriv^{\alpha}z + cz^2 + \frac{\Lambda}{3}z^4.
    \end{split}
  \end{equation}
\end{corollary}
\begin{proof}
  Simple computation using \Cref{lemma:P-sigma-alt-expressions} and
  \Cref{eq:b-c-k:def}.
\end{proof}

\begin{definition}
  \label{def:P}
  We define $P \vcentcolon= \frac{1}{J}$, where $J$ is as defined in
  \Cref{eq:J:def}. $P$ will play an important role in constructing
  the subsequent pseudoconvex foliation. In particular, we also
  define $y$, $z$ to be real numbers defined by the relation
  \begin{equation}
    \label{eq:y-z:def}
    P = y + \ImagUnit z.
  \end{equation}
  Because of its importance, we will also define now
  \begin{equation}
    \label{eq:Upsilon:def}
    \Upsilon \vcentcolon= \frac{\Delta}{\abs*{P}^2}.
  \end{equation}
\end{definition}

\begin{lemma}
  We have the following two equations.
  \begin{equation}
    \label{eq:deriv-R}
    \CovariantDeriv_{\rho}R
    ={}\left(J - \frac{\Lambda}{2R}\right)\bm{\sigma}_{\rho}
    - \frac{1}{4}\KillT^{\sigma}\mathcal{X}^{\mu\nu}\mathcal{S}_{\mu\nu\sigma\rho}.  
  \end{equation}
  and
  \begin{equation}
    \label{eq:deriv-P}
    \begin{split}
      \CovariantDeriv_{\alpha}\mathbf{P}_{\beta}
      ={}& J \mathbf{P}_{\alpha}\mathbf{P}_{\beta}
           + \frac{1}{4R}\KillT^{\sigma}\mathcal{X}_{\mu\nu}\mathcal{S}_{\mu\nu\sigma\alpha}\mathbf{P}_{\beta}
           + \left(\frac{R}{2} - NJ\right)\Metric_{\alpha\beta}
           - \frac{1}{R}\mathcal{T}_{\alpha\beta}\\
         & + \frac{1}{R}\KillT^{\rho}\KillT^{\sigma}\mathcal{S}_{\rho\alpha\sigma\beta}
           - J \KillT_{\alpha}\KillT_{\beta}.
    \end{split}
  \end{equation}

\end{lemma}
\begin{proof}
  We can easily compute that
  \begin{align*}
    \CovariantDeriv_{\rho}R
    ={}& -\frac{\ImagUnit}{2}\CovariantDeriv_{\rho}\sqrt{\mathcal{F}^2}\notag\\
    ={}&  -\frac{1}{8R}\CovariantDeriv_{\rho}\mathcal{F}^2\notag\\
    ={}& \frac{1}{6}\left(2QR - \frac{\Lambda}{R}\right)\bm{\sigma}_{\rho}
         - \frac{1}{4}\KillT^{\sigma}\mathcal{X}^{\mu\nu}\mathcal{S}_{\mu\nu\sigma\rho}\notag\\
    ={}& \left(J - \frac{\Lambda}{2R}\right)\bm{\sigma}_{\rho}
         - \frac{1}{4}\KillT^{\sigma}\mathcal{X}^{\mu\nu}\mathcal{S}_{\mu\nu\sigma\rho}.  
  \end{align*}
  \Cref{eq:deriv-P} also follows immediately using the definition of $\mathbf{P}$. 
\end{proof}

\subsection{Mars-Simon tensor}
\label{sec:Mars-Simon}

We now introduce the Mars-Simon tensor that is well-adapted for
proving \KdS{} rigidity. The Mars-Simon tensor has played a crucial
role in proofs of Kerr black hole rigidity in the asymptotically flat
setting \cite{ionescuUniquenessSmoothStationary2009,
  ionescuRigidityResultsGeneral2015} because of two key
properties. First, the Mars-Simon tensor characterizes Kerr in the
sense that if the Mars-Simon tensor of a stationary spacetime
vanishes, then the stationary spacetime in question is locally
isometric to a Kerr black hole; and second, the Mars-Simon tensor
satisfies a wave equation suitable for the application of Carleman
estimates and their associated unique continuation results. A similar
Mars-Simon tensor was uncovered in
\cite{marsSpacetimeCharacterizationKerrNUTAde2015} for \KdS. We review
both the way this Mars-Simon tensor characterizes \KdS{} as
well as the wave equation it satisfies.

\begin{definition}[Mars-Simon tensor]
  \label{def:MST}
  We define the \emph{Mars-Simon tensor} by
  \begin{equation}
    \label{eq:MST:def}
    \mathcal{S} = \mathcal{W} - Q\mathcal{U},
  \end{equation}
  where we recall the definition of $Q$ in \cref{eq:MST:Q:def}, where
  \begin{equation}
    \label{eq:UCal:def}
    \mathcal{U}_{\alpha\beta\mu\nu}
    = \mathcal{F}_{\alpha\beta}\mathcal{F}_{\mu\nu}
    - \frac{1}{3} \mathcal{F}^2\mathcal{I}_{\alpha\beta\mu\nu},
  \end{equation}
  and where we recall the definition of $\mathcal{I}$ from
  \Cref{eq:ICal:def}.  We note that $\mathcal{I}$ is the metric in the
  space of complex self-dual 2-forms in the sense that
  \begin{equation*}
    \mathcal{I}_{\alpha\beta\mu\nu}\mathcal{G}^{\mu\nu} = \mathcal{G}_{\alpha\beta}
  \end{equation*}
  for any complex self-dual 2-form $\mathcal{G}$.
\end{definition} 

We recall below the a key result of
\cite{marsSpacetimeCharacterizationKerrNUTAde2015} which describes how
the Mars-Simon tensor $\mathcal{S}$ characterizes
\KdS\footnote{Technically, the Mars-Simon tensor as defined in
  \Cref{def:MST} does not exactly characterize \KdS. As shown in
  \cite{marsSpacetimeCharacterizationKerrNUTAde2015}, to show a
  spacetime is locally isometric to \KdS, it suffices for there to
  exist some $Q$ such that the associated Mars-Simon tensor vanishes.}.
\begin{theorem}[Theorem 1 of \cite{marsSpacetimeCharacterizationKerrNUTAde2015}]
  \label{thm:MST-characterization}
  Let $(\Manifold, \Metric)$ satisfy $\Lambda$-EVE and be stationary,
  admitting a global Killing vectorfield $\KillT$ with self-dual
  two-form $\mathcal{F}\in S^2T^{*}\Manifold$. Moreover, assume there
  exists some $Q\in C^{\infty}(\Manifold, \Complex)$ such that
  \begin{equation}
    \label{eq:Mars-Senovilla:condition}
    \mathcal{W}_{\alpha\beta\mu\nu}
    = Q\left(\mathcal{F}_{\alpha\beta}\mathcal{F}_{\mu\nu} - \frac{1}{3}\mathcal{F}^2\mathcal{I}_{\alpha\beta\mu\nu}\right), 
  \end{equation}
  and assume that $Q\mathcal{F}^2$ and $Q\mathcal{F}^2-4\Lambda$ are
  not uniformly zero. Then $\mathcal{F}^2$ and
  $Q\mathcal{F}^2-4\Lambda$ are nonzero everywhere and there exist
  constants $b_1, b_2, c, k\in \Real$ such that
  \begin{align*}
    36 Q (\mathcal{F}^2)^{\frac{5}{2}} + (b_2-\ImagUnit b_1)(Q\mathcal{F}^2-4\Lambda)^3 &= 0,\\
    \Metric(\KillT, \KillT) + \Re\left(\frac{6 \mathcal{F}^2(Q\mathcal{F}^2+2\Lambda)}{(Q\mathcal{F}^2-4\Lambda)^2}\right)
    +c &= 0,\\
    -k + \abs*{\frac{36\mathcal{F}^2}{(Q\mathcal{F}^2-4\Lambda)^2}}\CovariantDeriv_{\alpha}z\CovariantDeriv^{\alpha}z - b_2z + cz^2 + \frac{\Lambda}{3}z^4 &= 0,
  \end{align*}
  where
  $z = \Im \left(\frac{6\ImagUnit\sqrt{\mathcal{F}^2}}{Q\mathcal{F}^2
      - 4\Lambda}\right)$.

  If moreover, the polynomial
  $V(\zeta)\vcentcolon= k+ b_2\zeta - c\zeta^2 - \frac{\Lambda}{3}\zeta^4$ satisfies that 
  \begin{equation*}
    \frac{V(\zeta)}{(\zeta-\zeta_0)(\zeta+\zeta_0)} < 0
  \end{equation*}
  on $[-\zeta_0,\zeta_0]$, and $z:\mathcal{M}\to [-\zeta_0,\zeta_0]$,
  then $\left( \Manifold, \Metric \right)$ is locally isometric to
  Kerr-(a)dS with parameters $\{\Lambda, M, a\}$ where
  \begin{equation*}
    M = \frac{b_1}{2v_0\sqrt{v_0}},\qquad
    a = \frac{\zeta_0}{\sqrt{v_0}},
  \end{equation*}
  where $v_0\vcentcolon= \frac{V(0)}{\zeta_0^2}$.
\end{theorem}
\begin{remark}
  In the case of \KdS, we have that
  \begin{equation*}
    V(\zeta) = -\left(\frac{\Lambda}{3}\zeta^2 + \frac{1}{(1+\gamma)^2}\right)\left(\zeta+\frac{a}{1+\gamma}\right)\left(\zeta-\frac{a}{1+\gamma}\right),
  \end{equation*}
  which satisfies the hypothesis automatically for $\Lambda\ge 0$. 
\end{remark}

\begin{remark}
  In \KdS, we have that
  \begin{equation*}
    b = \frac{2M}{(1+\gamma)^3}, \qquad
    c = \frac{1-\gamma}{\left( 1+\gamma \right)^2},\qquad
    k = \frac{a^2}{\left( 1+\gamma \right)^4},
  \end{equation*}
  which is what informs our choice of the compatibility assumptions in
  \Cref{eq:compatibility-conditions}. 
\end{remark}

\subsection{A wave equation for \texorpdfstring{$\mathcal{S}$}{S}}
\label{sec:MST-wave-eqn}

As was shown in \cite{marsCharacterizationAsymptoticallyKerr2016} (see
for example (4.57) and (4.58)), $\mathcal{S}$ as defined in
\Cref{eq:MST:def} solves a wave equation. In particular, we have the
following proposition.
\begin{prop}[Divergence property of $\mathcal{S}$, (4.57), (4.58) of
  \cite{marsCharacterizationAsymptoticallyKerr2016}]
  \label{prop:divergence-of-MST}
  $\mathcal{S}$ satisfies the following equation.
  \begin{equation}
    \label{eq:divergence-of-MST}
    \CovariantDeriv^{\rho}\mathcal{S}_{\alpha\beta\mu\rho}
    = \mathcal{J}_{\alpha\beta\mu},
  \end{equation}
  where
  \begin{equation}
    \label{eq:J-cal:def}
    \begin{split}
      \mathcal{J}(\mathcal{S})_{\mu\alpha\beta}
      \vcentcolon={}& -\frac{\Lambda - 3JR}{R}\tensor[]{\mathcal{X}}{_{\mu}^{\rho}}\KillT^{\sigma}\mathcal{S}_{\alpha\beta\sigma\rho}
                      + \frac{J^2\sigma_0}{R-J\sigma_0}\KillT^{\sigma}\tensor[]{\mathcal{I}}{_{\alpha\beta\mu}^{\rho}}\mathcal{X}^{\gamma\delta}\mathcal{S}_{\gamma\delta\sigma\rho}\\
                    & + \frac{\Lambda}{4R}\frac{R+2J\sigma_0}{R-J\sigma_0}\KillT^{\sigma}\mathcal{X}_{\alpha\beta}\mathcal{X}^{\gamma\delta}\tensor[]{\mathcal{X}}{_{\mu}^{\rho}}\mathcal{S}_{\gamma\delta\sigma\rho}\\
      ={}&  4\Lambda\frac{5Q\mathcal{F}^2+ 4\Lambda}{Q\mathcal{F}^2+8\Lambda}\mathcal{U}_{\alpha\beta\mu\rho}\mathcal{F}^{-4}\KillT^{\sigma}\mathcal{F}^{\gamma\delta}\tensor[]{\mathcal{S}}{_{\gamma\delta\sigma}^{\rho}}\\
         & - Q\KillT^{\sigma}\left(\frac{2}{3}\mathcal{I}_{\alpha\beta\mu\rho}\mathcal{F}^{\gamma\delta}\tensor[]{\mathcal{S}}{_{\gamma\delta\sigma}^{\rho}}
           - \mathcal{F}_{\mu\rho}\tensor[]{\mathcal{S}}{_{\alpha\beta\sigma}^{\rho}}
           \right)
           . 
    \end{split}  
  \end{equation}
\end{prop}
\Cref{prop:divergence-of-MST} is proven in Section 4.4.1 of
\cite{marsCharacterizationAsymptoticallyKerr2016}. However, for the
sake of completeness, we have included a proof in the appendix.

\begin{proof}
  See \Cref{appendix:prop:divergence-of-MST}.
\end{proof}

To derive the wave equation for $\mathcal{S}$, we use the following
general lemma on Weyl fields (Proposition 4.1 in
\cite{ionescuUniquenessSmoothStationary2009}).
\begin{lemma}
  \label{lemma:Bianchi-on-Weyl-field-with-divergence}
  If $W$ is a Weyl field such that
  \begin{equation*}
    \CovariantDeriv^{\alpha}W_{\alpha\beta\gamma\delta} = J_{\beta\gamma\delta},
  \end{equation*}
  then
  \begin{equation}
    \label{eq:Bianchi-on-Weyl-field-with-divergence}
    \CovariantDeriv_{[\sigma}W_{\gamma\delta]\alpha\beta} = \in_{\mu\sigma\gamma\delta}\tensor[^{*}]{J}{^{\mu}_{\alpha\beta}}.
  \end{equation}
\end{lemma}

\begin{definition}
  \label{def:smooth-multiple-notation}
  Let $\{B^i\}_{i=1}^k$ be a collection of tensors.  We denote by
  $\AdmissibleRHS(B^1,..., B^k)$ a combination of
  contractions of $\{B^i\}_{i=1}^k$ allowing coefficients depending on
  $\Metric, \Riem, \Ric$, or their covariant derivatives i.e, a tensor of the form
  \begin{equation*}
    \AdmissibleRHS(B^1,\cdots, B^k)
    = \sum \CovariantDeriv^{\le 1}(\Metric + \Riem + \Ric)^i\cdot B^j,
  \end{equation*}
  where $\cdot$ denotes a contraction of tensors.

  Since $\Metric, \Riem$, and $\Ric$ are allowed as coefficients, in
  the notation $\AdmissibleRHS(B^1,\cdots, B^k)$, we only list
  $B^i$ which are distinct from $\Metric, \Riem$, and $\Ric$. 
\end{definition}

Combining \Cref{prop:divergence-of-MST} and
\Cref{lemma:Bianchi-on-Weyl-field-with-divergence}, we then have the
following result.
\begin{theorem}
  \label{theorem:MST-wave-equation}
  $\mathcal{S}$ satisfies a wave equation of the following form.
  \begin{equation}
    \label{eq:MST-wave equation}
    \Box_{\Metric}\mathcal{S} = \AdmissibleRHS[R-J\sigma_0](\mathcal{S}, \CovariantDeriv \mathcal{S}),
  \end{equation}
  where
  $\AdmissibleRHS[R-J\sigma_0](\mathcal{S}, \CovariantDeriv
  \mathcal{S})$ consists of terms of the form
  $\AdmissibleRHS(\mathcal{S}, \CovariantDeriv \mathcal{S})$ (i.e. the
  coefficients of $\mathcal{S}$, $\CovariantDeriv \mathcal{S}$ are
  smooth) in regions where $R-J\sigma_0\neq 0$.
\end{theorem}
\begin{proof}
  The proof follows directly from using \Cref{prop:divergence-of-MST}
  and \Cref{lemma:Bianchi-on-Weyl-field-with-divergence},
  \begin{align*}
    \CovariantDeriv^{\alpha}\mathcal{S}_{\alpha\beta\gamma\delta} ={}& \mathcal{J}(\mathcal{S})_{\beta\gamma\delta},\\
    \CovariantDeriv_{[\sigma}\mathcal{S}_{\alpha\beta]\gamma\delta}
    ={}& -\ImagUnit \in_{\rho\sigma\alpha\beta}\tensor[]{\mathcal{J}(\mathcal{S})}{^{\rho}_{\gamma\delta}}.
  \end{align*}
  Differentiating the second equation above again, we have that
  \begin{equation*}
    \CovariantDeriv^{\sigma}\CovariantDeriv_{[\sigma}\mathcal{S}_{\alpha\beta]\gamma\delta}
    ={} -\ImagUnit \in_{\rho\sigma\alpha\beta}\CovariantDeriv^{\sigma}\tensor[]{\mathcal{J}(\mathcal{S})}{^{\rho}_{\gamma\delta}}.
  \end{equation*}
  Commuting covariant derivatives and using the first equation and the
  exact form of $\mathcal{J}(\mathcal{S})$ from \Cref{eq:J-cal:def}
  then yields that
  \begin{equation*}
    \Box_{\Metric}\mathcal{S} = \AdmissibleRHS[R-J\sigma_0](\mathcal{S}, \CovariantDeriv \mathcal{S}),
  \end{equation*}
  as desired. 
\end{proof}

\subsection{Relationship between \texorpdfstring{$\mathbf{P}_{\alpha}$}{the renormalized Ernst one-form} and \texorpdfstring{$P$}{P}}
\label{sec:P-POneForm-relationship}

We first show two lemmas that show that $P = - \frac{1}{J}$ as defined
in \Cref{def:P} is related to the rescaled Ernst one-form
$\mathbf{P}_{\alpha}$ defined in
\Cref{def:rescaled-tensor-qtys:def}. In particular, when
$\mathcal{S}=0$, $\mathbf{P}$ is closed and $P$ is its potential.
\begin{lemma}
  \label{lemma:vanishing-S:deriv-J}
  Let $J$ be as defined in \Cref{eq:J:def}
  . Then,
  \begin{equation*}
    \label{eq:vanishing-S:deriv-J}
    \CovariantDeriv_{\rho}J
    = \frac{J^2}{2R}\bm{\sigma}_{\rho}
         + \frac{J}{4R(R-J\sigma_0)}\KillT^{\sigma}\mathcal{F}^{\mu\nu}\mathcal{S}_{\mu\nu\sigma\rho}
         . 
  \end{equation*}
\end{lemma}
\begin{proof}
  Recall from \Cref{eq:deriv-JR} that
  \begin{equation*}
     \CovariantDeriv_{\rho}\left(JR\right)
    ={} \frac{J}{2R}(3JR-\Lambda)\bm{\sigma}_{\rho}
         + \frac{J^2\sigma_0}{4(R-J\sigma_0)}\KillT^{\sigma}\mathcal{X}^{\mu\nu}\mathcal{S}_{\mu\nu\sigma\rho},
  \end{equation*}
  and from \Cref{eq:deriv-R} that
  \begin{equation*}
    \CovariantDeriv_{\rho}R    
    ={} \left(J - \frac{\Lambda}{2R}\right)\bm{\sigma}_{\rho}
    - \frac{1}{4}\KillT^{\sigma}\mathcal{X}^{\mu\nu}\mathcal{S}_{\mu\nu\sigma\rho}.
  \end{equation*}
  Thus, we have that
  \begin{align*}
    \CovariantDeriv_{\rho}J
    ={}& \frac{1}{R}\left(
         \CovariantDeriv_{\rho}\left(JR\right)
         - J\CovariantDeriv_{\rho}R
         \right)\\
    ={}& \left(\frac{J}{2R}\left(3JR-\Lambda\right)
         - J\left(J - \frac{\Lambda}{2R}\right)\right)\frac{\bm{\sigma}_{\rho}}{R}
         + \frac{1}{4R^2}\KillT^{\sigma}\mathcal{F}^{\mu\nu}\mathcal{S}_{\mu\nu\sigma\rho}\left(
         \frac{J^2\sigma_0}{R-J\sigma_0} + J
         \right)\\
    ={}& \frac{J^2}{2R}\bm{\sigma}_{\rho}
         + \frac{1}{4R}\KillT^{\sigma}\mathcal{F}^{\mu\nu}\mathcal{S}_{\mu\nu\sigma\rho}\left(
         \frac{J}{R-J\sigma_0}
         \right)
         , 
  \end{align*}
  as desired.
\end{proof}

We also have the following lemma computing the derivative of $P$. 
\begin{lemma}
  \label{lemma:P-one-form-P-J-inverse-relation}
  We have that
  \begin{equation}
    \label{eq:P-one-form-P-J-inverse-relation}
    \CovariantDeriv_{\alpha}P = \mathbf{P}_{\alpha}
    + \frac{1}{4J(R-J\sigma_0)}\KillT^{\sigma}\mathcal{X}^{\mu\nu}\mathcal{S}_{\mu\nu\sigma\alpha}.
  \end{equation}
  In particular, we observe that if $\mathcal{S}=0$, then
  $\CovariantDeriv_{\alpha} P = \mathbf{P}_{\alpha}$, where $\mathbf{P}_{\alpha}$ is the
  rescaled Ernst one-form defined in
  \Cref{def:rescaled-tensor-qtys:def}. 
\end{lemma}
\begin{proof}
  Recall from \Cref{eq:deriv-JR} that
  \begin{equation*}
     \CovariantDeriv_{\rho}\left(JR\right)
    ={} \frac{J}{2R}(3JR-\Lambda)\bm{\sigma}_{\rho}
         + \frac{J^2\sigma_0}{4(R-J\sigma_0)}\KillT^{\sigma}\mathcal{X}^{\mu\nu}\mathcal{S}_{\mu\nu\sigma\rho},
  \end{equation*}
  and from \Cref{eq:deriv-R} that
  \begin{equation}
    \CovariantDeriv_{\rho}R    
    ={} \left(J - \frac{\Lambda}{2R}\right)\bm{\sigma}_{\rho}
    - \frac{1}{4}\KillT^{\sigma}\mathcal{X}^{\mu\nu}\mathcal{S}_{\mu\nu\sigma\rho}.
  \end{equation}
  Thus, we have that
  \begin{align*}
    \CovariantDeriv_{\rho}J
    ={}& \frac{1}{R}\left(
         \CovariantDeriv_{\rho}\left(JR\right)
         - J\CovariantDeriv_{\rho}R
         \right)\\
    ={}& \left(\frac{J}{2R}\left(3JR-\Lambda\right)
         - J\left(J - \frac{\Lambda}{2R}\right)\right)\frac{\bm{\sigma}_{\rho}}{R}
         + \frac{1}{4R^2}\KillT^{\sigma}\mathcal{F}^{\mu\nu}\mathcal{S}_{\mu\nu\sigma\rho}\left(
         \frac{J^2\sigma_0}{R-J\sigma_0} + J
         \right)\\
    ={}& \frac{J^2}{2R}\bm{\sigma}_{\rho}
         + \frac{1}{4R}\KillT^{\sigma}\mathcal{F}^{\mu\nu}\mathcal{S}_{\mu\nu\sigma\rho}\left(
         \frac{J}{R-J\sigma_0}
         \right)
         . 
  \end{align*}
  Then, we can compute that
  \begin{align*}
    \CovariantDeriv_{\rho}J
    ={}& \frac{1}{R}\left(
         \CovariantDeriv_{\rho}\left(JR\right)
         - J\CovariantDeriv_{\rho}R
         \right)         
    \\
    ={}& \left(\frac{J}{2R}\left(3JR-\Lambda\right)
         - J\left(J - \frac{\Lambda}{2R}\right)\right)\frac{\bm{\sigma}_{\rho}}{R}
         + \frac{J^2\sigma_0}{4R(R-J\sigma_0)}\KillT^{\sigma}\mathcal{X}^{\mu\nu}\mathcal{S}_{\mu\nu\sigma\rho}
         + \frac{J}{4R}\KillT^{\sigma}\mathcal{X}^{\mu\nu}\mathcal{S}_{\mu\nu\sigma\rho}
    \\
    ={}& \frac{J^2}{2R}\bm{\sigma}_{\rho}
         + \frac{J}{4(R-J\sigma_0)}\KillT^{\sigma}\mathcal{X}^{\mu\nu}\mathcal{S}_{\mu\nu\sigma\rho}.
  \end{align*}
  Then we see that 
  \begin{align*}
    \CovariantDeriv_{\rho}P
    ={}& \frac{1}{J^2}\CovariantDeriv_{\rho}J\\
    ={}& \frac{1}{2R}\bm{\sigma}_{\rho}
         + \frac{1}{4J(R-J\sigma_0)}\KillT^{\sigma}\mathcal{X}^{\mu\nu}\mathcal{S}_{\mu\nu\sigma\rho}
         ,
  \end{align*}
  as desired. 
\end{proof}

\subsection{Basic consequences of null-bifurcate geometry}
\label{sec:null-bifurcate-geometry}

In this section, we list some basic consequences of the null-bifurcate
geometry. These will prove critical in the arguments at the horizons
where we will initialize our extension arguments.
\begin{prop}
  \label{prop:non-expanding-null-hypersurface:basic-props}
  Let $(e_3,e_4)$ be a null pair in an open set $\mathbf{N}$ such that
  $e_4$ is orthogonal to a non-expanding null hypersurface in
  $\mathcal{N}\subset \mathbf{N}$. Then $\xi$ and $\chi$, as defined
  in \Cref{eq:ricci-components:def}, vanish identically on
  $\mathcal{N}$. Moreover, the null components $A(\mathcal{W})$ and
  $B(\mathcal{W})$ of the complexified Weyl tensor $\mathcal{W}$
  vanish along $\mathcal{N}$, and the component $P(\mathcal{W})$ is
  constant along the null generators.
\end{prop}
\begin{proof}
  First observe that if $X, Y\in T \mathcal{N}$ are vectorfields
  tangent to $\mathcal{N}$, then we have that
  \begin{align*}
    \Metric(\CovariantDeriv_Xe_4, Y) - \Metric(\CovariantDeriv_Ye_4, X)
    &= -\Metric(e_4, [X,Y])=0,\\
    \Metric(\CovariantDeriv_{e_4}e_4, X) &=  - \Metric(L, \CovariantDeriv_{e_4}X)=0.
  \end{align*}
  where we used the fact that the Lie bracket of two vectorfields in
  $T \mathcal{N}$ is also itself in $T \mathcal{N}$.
  In particular, we see that $\xi$ vanishes identically and $\chi$ is
  symmetric.

  Then, using \Cref{eq:null-structure:D4-tr-X}, using that the null
  hypersurface is non-expanding, that $\xi$ vanishes identically, and
  that $\chi$ is symmetric, we have that
  \begin{equation*}
    \aTrace{\chi} = 0.
  \end{equation*}
  As a result, we have that $\chi = 0$. Then, using
  \Cref{eq:null-structure:D4-XHat}, we have that $A(\mathcal{W})=0$
  along $\mathcal{N}$. Similarly, using
  \Cref{eq:null-structure:D-tr-X}, we have that $B(\mathcal{W})$
  along $\mathcal{N}$. The Bianchi equation in
  \Cref{eq:Bianchi:nabla4-P} then shows that
  $\CovariantDeriv_4P(\mathcal{W})$ vanishes along $\mathcal{N}$, and
  thus that $P(\mathcal{W})$ is a constant along $\mathcal{N}$.

  The result is in particular independent of the choice of null frame
  provided $e_4$ remains orthogonal to $\mathcal{N}$. Given any two
  null frames $(e_1,e_2,e_3,e_4)$ and $(e_1',e_2',e_3',e_4')$ where
  $e_4 = \lambda e_4'$ so that both $e_4$, $e_4'$ are orthogonal to
  $\mathcal{N}$, we have from
  \Cref{lemma:frame-transformation:general-form} that there exists
  transition coefficients $(\lambda, 0, \underline{f})$ that transform
  from $(e_1,e_2,e_3,e_4)$ to $(e_1',e_2',e_3',e_4')$. But then from
  \Cref{eq:frame-transformation:rho} and
  \Cref{eq:frame-transformation:rho-dual} of
  \Cref{prop:frame-transformation} we see that if $(e_1,e_2,e_3,e_4)$
  and $(e_1',e_2',e_3',e_4')$ are related by the transition
  coefficients $(\lambda, 0, \underline{f})$ (i.e. $e_4$ and $e_4'$
  are related by rescaling), we have that
  $P'(\mathcal{W}) - P(\mathcal{W})$ is a linear combination of
  $A(\mathcal{W})$ and $B(\mathcal{W})$.
\end{proof}
\begin{corollary}
  \label{coro:non-expanding-null-hypersurface:horizon-qtys}
  Recall that
  $\chi, \xi, \zeta, \eta, \underline{\eta}, \underline{\chi},
  \underline{\xi}$ are the Ricci coefficients defined in
  \Cref{eq:ricci-components:def} such that $e_4$ is tangent to
  $\EventHorizonFuture$ and $\CosmologicalHorizonPast$, and $e_3$ is
  tangent to $\EventHorizonPast$ and $\CosmologicalHorizonFuture$.
  The following hold along $\EventHorizon$ and $\CosmologicalHorizon$.
  \begin{enumerate}
  \item On $\EventHorizonFuture$ and $\CosmologicalHorizonPast$, we have that
    \begin{equation*}
      A(\mathcal{W}) = B(\mathcal{W}) = \chi = \xi = \zeta+\underline{\eta}=0.
    \end{equation*}
  \item On $\EventHorizonPast$ and $\CosmologicalHorizonFuture$, we have that
    \begin{equation*}
      \underline{A}(\mathcal{W}) = \underline{B}(\mathcal{W}) = \underline{\chi} = \underline{\xi} = \zeta-\eta=0,
    \end{equation*}
  \end{enumerate}
  where here $\mathcal{W}$ denotes the complexified Weyl tensor. 
\end{corollary}
\begin{proof}
  We only show the proof for the result on $\EventHorizonFuture$. The results on the other null horizons follow similarly. The fact that on $\EventHorizonFuture$
  \begin{equation*}
    A(\mathcal{W}) = B(\mathcal{W}) = \chi = \xi =0
  \end{equation*}
  follows directly from \Cref{prop:non-expanding-null-hypersurface:basic-props}.
  We show now that $\zeta-\underline{\eta}$ also vanishes along $\EventHorizonFuture$.
  From the Ricci formula \Cref{eq:Ricci-formulas}, we see that
  \begin{equation*}
    [e_4,e_a] = \chi_{ab}e_b - \nabla_4e_a - (\zeta + \underline{\eta})_ae_4,
  \end{equation*}
  which applying to $\underline{u}_-$ along $\EventHorizonFuture$ yields that 
  \begin{equation*}
    \zeta = -\underline{\eta},
  \end{equation*}
  as desired.
\end{proof}

\begin{lemma}
  \label{lemma:horizon:Fcal-null-components}
  For $B(\mathcal{F}), \underline{B}(\mathcal{F}), P(\mathcal{F})$
  defined as in \Cref{def:F-null-components:def}, we have the following:
  \begin{enumerate}
  \item $B(\mathcal{F})$ vanishes along $\EventHorizonFuture$  and
     $\CosmologicalHorizonPast$
    ,
  \item $\underline{B}(\mathcal{F})$ vanishes along
    $\EventHorizonPast$  and $\CosmologicalHorizonFuture$
    ,
  \item $P(\mathcal{F})$ does not vanish on $S_0$ and $\underline{S}_0$
    . 
  \end{enumerate}
\end{lemma}
\begin{proof}
  We prove the claims in \Cref{lemma:horizon:Fcal-null-components} on
  $\EventHorizon$. The proofs on $\CosmologicalHorizon$ follow similar
  arguments. Let us now fix $(e_4,e_3) = (L_+,-L_-)$.
  We first compute that on $\EventHorizonFuture$,
  \begin{align*}
    F(e_a,e_4) &= - \Metric(\KillT, \CovariantDeriv_{e_a}e_4)\\
               &= - \chi(\KillT, e_a)\\
               &=0,
  \end{align*}
  where we used the fact from
  \Cref{prop:non-expanding-null-hypersurface:basic-props} that
  $\chi=0$ on $\EventHorizonFuture$.
  On the other hand, we can also compute that
  \begin{align*}
    \LeftDual{F}(e_a,e_4)
    &= \frac{1}{2}\in_{a4\mu\nu}F^{\mu\nu}\\
    &= \in_{a4b3}F^{b3}\\
    &= 0.       
  \end{align*}
  Then, by definition
  \begin{equation*}
    B(\mathcal{F}) = F(e_a,e_4) + \ImagUnit\LeftDual{F}(e_a,e_4) = 0\qquad \text{on }\EventHorizonFuture. 
  \end{equation*}

  Similarly, we can compute that on $\EventHorizonPast$
  \begin{align*}
    F(e_a,e_3) &= - \Metric(\KillT, \CovariantDeriv_{e_a}e_3)\\
               &= - \underline{\chi}(\KillT, e_a)\\
               &=0,
  \end{align*}
  where we used the fact from
  \Cref{prop:non-expanding-null-hypersurface:basic-props} that
  $\underline{\chi}=0$ on $\EventHorizonPast$.
  On the other hand, we can also compute that
  \begin{align*}
    \LeftDual{F}(e_a,e_3)
    &= \frac{1}{2}\in_{a3\mu\nu}F^{\mu\nu}\\
    &= \in_{a3b4}F^{b4}\\
    &= 0.       
  \end{align*}
  Then, by definition
  \begin{equation*}
    \underline{B}(\mathcal{F}) = F(e_a,e_3) + \ImagUnit\LeftDual{F}(e_a,e_3) = 0\qquad \text{on }\EventHorizonPast. 
  \end{equation*}
  It then naturally follows that both
  $B(\mathcal{F})=\underline{B}(\mathcal{F})$ on $S_0$.

  Finally, we can observe that on $S_0$,
  \begin{align*}
    \mathcal{F}^2 &= 2\mathcal{F}^{34}\mathcal{F}_{34} + \mathcal{F}^{ab}\mathcal{F}_{ab}\\
                  &= -4\mathcal{F}_{34}^2\\
                  &= -4P(\mathcal{F})^2,
  \end{align*}
  where we used
  \Cref{eq:Fcal-34-ab-components-relationship}. Moreover, by our
  technical regularity assumption we have that $\mathcal{F}^2$ cannot
  vanish on $S_0$. Thus, $P(\mathcal{F})$ also does not vanish on
  $S_0$.
\end{proof}

We will also make use of the following lemma in what follows.
\begin{lemma}
  \label{lemma:horizon-generator-hor-deriv-commutator}
  Let
  $(L, \Horizon)\in \curlyBrace*{(\underline{L}_-,
    \CosmologicalHorizonFuture), (-\underline{L}_+,
    \CosmologicalHorizonPast), (L_+, \EventHorizonFuture), (-L_-,
    \EventHorizonPast)}$ be a paired horizon generator and its
  corresponding horizon, and let $X\in \mathbf{O}(\mathcal{M})$ be a
  horizontal vectorfield as defined in
  \Cref{def:horizontal-vectorfield}). Then
  \begin{equation}
    \label{eq:horizon-generator-hor-deriv-commutator}
    \left[\nabla_L, \nabla\right]X = 0.
  \end{equation}
\end{lemma}
\begin{proof}
  Without loss of generality, let us consider the case of
  $(L, \Horizon)=(L_+, \EventHorizonFuture)$. We will use the null
  frame $(e_3,e_4) = (-L_-, L_+)$.  We can compute using the Ricci
  formula in \Cref{eq:Ricci-formulas} that
  \begin{align*}
    \CovariantDeriv_4\CovariantDeriv_aX_b
    ={}& e_4(\CovariantDeriv_aX_b)
         - \CovariantDeriv_{\CovariantDeriv_4e_a}X_b
         - \CovariantDeriv_aX_{\CovariantDeriv_4e_b}\\
    ={}& e_4(\nabla_bX_a)
         - \CovariantDeriv_{\nabla_4e_a}X_b
         + \zeta_a\CovariantDeriv_4X_b
         - \CovariantDeriv_aX_{\nabla_4e_a}
         + \zeta_b\CovariantDeriv_aX_4\\
    ={}& \nabla_4\nabla_aX_b
         + \zeta_b\nabla_4X_b.
  \end{align*}
  Similarly, we also have that
  \begin{align*}
    \CovariantDeriv_a\CovariantDeriv_4X_b
    ={}& e_a(\CovariantDeriv_4X_b)
         -\CovariantDeriv_{\CovariantDeriv_ae_4}X_b
         - \CovariantDeriv_4X_{\CovariantDeriv_ae_b}\\
    ={}& e_a(\CovariantDeriv_4X_b)
         - \CovariantDeriv_{\nabla_ae_4}X_b
         + \zeta_a\CovariantDeriv_4X_b
         - \CovariantDeriv_4X_{\nabla_ae_b}\\
    ={}& \nabla_a\nabla_4X_b
         + \zeta_a\nabla_4X_b.
  \end{align*}
  Thus,
  \begin{equation*}
    [\CovariantDeriv_a, \CovariantDeriv_4] X_b
    =[\nabla_a, \nabla_4] X_b.
  \end{equation*}
  Then, using that from
  \Cref{coro:non-expanding-null-hypersurface:horizon-qtys}, we have
  that
  \begin{equation*}
    [\CovariantDeriv_4,\CovariantDeriv_a] X_b
    ={} \Riem_{a4cb}X^c = 0,
  \end{equation*}
  so in fact, $[\nabla_4,\nabla_a]=0$, as desired. 
\end{proof}

\section{Consequences of the smallness of \texorpdfstring{$\mathcal{S}$}{S}}
\label{sec:S-small-consequences}

In this section we go through the consequences of our smallness
assumptions on $\mathcal{S}$ in
\Cref{eq:S-smallness-assumption}. These will be critical in
constructing a pseudoconvex foliation using $y$.  Throughout this
section, let us fix $(e_4 ,e_3)=(\bm{\ell}_+, \bm{\ell}_-)$ and assume
that we are on an open set $\mathbf{N}\subset \mathbf{E}$ such that
$S_0\in \closure \mathbf{N}$ and such that there is some
$0<\varepsilon_{\mathcal{S}}\ll 1$ such that
\begin{equation}
  \label{eq:smallness-assumption}
  \abs*{\frac{1}{R-J\sigma_0}\KillT^{\sigma}\mathcal{X}^{\mu\nu}\mathcal{S}_{\mu\nu\sigma\rho}}
  \le \varepsilon_{\mathcal{S}}
  \quad \text{ on }\Sigma_0\bigcap \closure\mathbf{N}.
\end{equation}
We remark that the restriction to such an $\mathbf{N}$ is not an
obstacle in the settings we consider. In the setting of
\Cref{thm:main:e1c1}, we will effectively have that $\mathcal{S}=0$ on
the regions where we apply the results in this section. In the setting
of \Cref{thm:main:e2c2}, \Cref{eq:smallness-assumption} is already
assumed to be globally true. In the setting of \Cref{thm:main}, we
will have that $\mathcal{S}=0$ when we apply the results of this
section on $\{y<y_{*}\}$, and that \Cref{eq:smallness-assumption} is
true by assumption in $\{y>y_{*}\}$.

We first look at some properties of the triple $(b,c,k)$ used to
define the compatibility conditions before moving on to consider some
properties of $y$.

\subsection{Properties of the compatibility triplet \texorpdfstring{$(b,c,k)$}{(b,c,k)}}

We now show that the triplet $(b,c,k)$ we used to define our
compatibility assumptions in \Cref{eq:compatibility-conditions} are
themselves almost constant if we assume
\Cref{eq:smallness-assumption}. In particular, we see that they are
constant if $\mathcal{S}=0$ on $\mathbf{N}$, which is consistent with
\Cref{thm:MST-characterization}. We start with $(b,c)$, which are easier to handle than $k$. 
\begin{lemma}
  \label{lemma:b-c-k-almost-constant:S-small}
  On $\mathbf{N}$, we have that
  \begin{equation*}
    \abs*{D b} + \abs*{D c} 
    = O(\varepsilon_{\mathcal{S}}),
  \end{equation*}
  so in particular, there exist functions $(\widetilde{b}, \widetilde{c}, \widetilde{k})$ such that
  \begin{equation*}
    \abs*{\widetilde{b}} + \abs*{\widetilde{c}}
    + \abs*{D\widetilde{b}} + \abs*{D\widetilde{c}}
    \lesssim \varepsilon_{\mathcal{S}}, \qquad
    (b,c) = (b_{S_0} + \widetilde{b}, c_{S_0} + \widetilde{c}).
  \end{equation*}
\end{lemma}
\begin{proof}
  We first remark that it suffices to show that the derivatives
  $\CovariantDeriv b$ and $\CovariantDeriv c$ are
  $O(\varepsilon_{\mathcal{S}})$. Taking a derivative of the
  expressions for $b$ and $c$ from \Cref{eq:b-c-k:def:renorm} and
  using \Cref{eq:deriv-R} and
  \Cref{eq:P-one-form-P-J-inverse-relation}, we see that
  \begin{equation}
    \label{eq:b-c-k-almost-constant:S-small:T-b-c-vanishes}
    \KillT (b) = \KillT (c) = 0. 
  \end{equation}
  Thus, if we had
  \begin{equation*}
    \abs*{D\widetilde{b}} + \abs*{D\widetilde{c}}
    \lesssim \varepsilon_{\mathcal{S}},
  \end{equation*}
  we could just integrate $b$, $c$ from a point $p'\in S_0$ to
  $p\in \Phi_1\left( (-\varepsilon_0, \varepsilon_0)\times E_{r_0} \right)$, using
  the fact that $\Sigma_0$ is compact, to see that
  \begin{equation}
    \label{eq:b-c-k-almost-constant:S-small:aux1}
    \abs*{b(p)-b_{S_0}} + \abs*{c(p)-c_{S_0}} = O(\varepsilon_{\mathcal{S}}). 
  \end{equation}
  Then, \Cref{eq:b-c-k-almost-constant:S-small:T-b-c-vanishes} implies
  that \Cref{eq:b-c-k-almost-constant:S-small:aux1} holds on all of $\mathbf{E}$. 

  Now, we first observe that we can rewrite the definition of $b$ in
  \Cref{eq:b-c-k:def} as
  \begin{equation*}
    b = \frac{2}{3J^3}\left(3JR-\Lambda\right).
  \end{equation*}
  We can compute that
  \begin{align*}
    \CovariantDeriv_{\rho}(J^3)
    ={}& 3J^2\CovariantDeriv_{\rho}J\\
    ={}& \frac{3J^4}{2R}\bm{\sigma}_{\rho}
         + \frac{3J^3}{4(R-J\sigma_0)}\KillT^{\sigma}\mathcal{X}^{\mu\nu}\mathcal{S}_{\mu\nu\sigma\rho},
  \end{align*}
  and that
  \begin{equation*}
    \CovariantDeriv_{\rho}(3JR-\Lambda)
    ={} \frac{3J}{2R}(3JR-\Lambda)\bm{\sigma}_{\rho}
    + \frac{3J^2\sigma_0}{4(R-J\sigma_0)}\KillT^{\sigma}\mathcal{X}^{\mu\nu}\mathcal{S}_{\mu\nu\sigma\rho}.
  \end{equation*}
  We can then compute that
  \begin{align}
    \CovariantDeriv_{\rho}\left(\frac{3JR-\Lambda}{J^3}\right)
    ={}& \frac{3\sigma_0}{4J(R-J\sigma_0)}\KillT^{\sigma}\mathcal{X}^{\mu\nu}\mathcal{S}_{\mu\nu\sigma\rho}
         - \frac{3(3JR-\Lambda)}{4J^3(R-J\sigma_0)}\KillT^{\sigma}\mathcal{X}^{\mu\nu}\mathcal{S}_{\mu\nu\sigma\rho}\notag\\
    ={}& -\frac{3R}{4J^2(R-J\sigma_0)}\KillT^{\sigma}\mathcal{X}^{\mu\nu}\mathcal{S}_{\mu\nu\sigma\rho}\notag\\
    ={}& O(\varepsilon_{\mathcal{S}})
         . \label{eq:y-z:small-derivatives:S-small:aux1}
  \end{align}
  It is easy to verify that the definition of $b$ in
  \Cref{eq:b-c-k:def} implies that
  \begin{equation*}
    \left( 3JR-\Lambda \right)J^{-3} = \frac{3}{2}b.
  \end{equation*}
  As a result,
  \begin{equation*}
    \abs*{\CovariantDeriv b} = O(\varepsilon_{\mathcal{S}}).
  \end{equation*}
  We now show that $\CovariantDeriv c = O(\varepsilon_{\mathcal{S}})$. To this end, observe that the definition of $c$ can be written as
  \begin{equation*}
    c = -\Metric(\KillT, \KillT) +\Re \left(2PR + \Lambda P^2\right).
  \end{equation*}
  Then, we can calculate using \Cref{eq:deriv-R} and
  \Cref{lemma:P-one-form-P-J-inverse-relation} that
  \begin{align*}
    \CovariantDeriv_{\beta}\left(2PR + \Lambda P^2\right)
    ={}& 2R\CovariantDeriv_{\beta}P
         +2P\CovariantDeriv_{\beta} R
         + 2\Lambda P \CovariantDeriv_{\beta}P\\
    ={}& 2R\mathbf{P}_{\beta} - \frac{PR}{2(R-J\sigma_0)}\KillT^{\sigma}\mathcal{X}^{\mu\nu}\mathcal{S}_{\mu\nu\sigma\beta}
         -(4R+2\Lambda P)\mathbf{P}_{\beta}
         - \frac{P}{2}\KillT^{\sigma}X^{\mu\nu}\mathcal{S}_{\mu\nu\sigma\beta}\\
       &  + 2\Lambda P \mathbf{P}_{\beta}
         -\frac{P^2\Lambda}{2(R-J\sigma_0)}\KillT^{\sigma}\mathcal{X}^{\mu\nu}\mathcal{S}_{\mu\nu\sigma\beta}\\
    ={}& -2R\mathbf{P}_{\beta} + O(\varepsilon_{\mathcal{S}}).
  \end{align*}
  Then observe that
  \begin{equation*}
    -\CovariantDeriv_{\beta}\Metric(\KillT,\KillT)
    = \Re \bm{\sigma}_{\beta}
    = 2R\mathbf{P}_{\beta}.
  \end{equation*}
  Thus, we have that
  \begin{equation*}
    \CovariantDeriv_{\beta} c = O(\varepsilon_{\mathcal{S}}),
  \end{equation*}
  as desired.
\end{proof}

Using the fact that $(b, c)$ are
$O(\varepsilon_{\mathcal{S}})$-constant we can find equations for
$R$ and $\sigma$ in terms of $P$.
\begin{lemma}
  \label{lemma:R-sigma-in-terms-of-P:S-small}
  Assuming \Cref{eq:smallness-assumption}, we have that in
  $\mathbf{N}$,
  \begin{equation}
    \label{eq:R-sigma-in-terms-of-P:S-small}
    R = \frac{b_{S_0}}{2P^2}(1+\widetilde{b})-\frac{\Lambda}{3}P,\qquad
    \sigma = c_{S_0} - \frac{b_{S_0}(1+\widetilde{b})}{P} - \frac{\Lambda}{3}P^2.
  \end{equation}
\end{lemma}
\begin{proof}
  We first observe that we can rewrite the definition of $b$
  in \Cref{eq:compatibility-conditions} as
  \begin{equation*}
    3JR-\Lambda = \frac{3}{2}bJ^3.
  \end{equation*}
  Then, using \Cref{lemma:b-c-k-almost-constant:S-small}, we can write
  that
  \begin{equation*}
    \left( 3JR-\Lambda \right)J^{-3} = \frac{3}{2}b_{S_0}(1+\widetilde{b}),    
  \end{equation*}
  which implies that
  \begin{equation*}
    R = \frac{b_{S_0}}{2P^2}(1+\widetilde{b}) - \frac{\Lambda}{3}P.
  \end{equation*}
  Using that $J$ by definition is a root of \Cref{eq:J:quadratic-eqn},
  we have that
  \begin{equation*}
    \sigma_0 = -2RP - \Lambda P^2 = -\frac{b_{S_0}(1+\widetilde{b})}{P} - \frac{\Lambda}{3}P^2.
  \end{equation*}
  Recalling that $\sigma = \sigma_0 + c_{S_0}$ then yields the desired conclusion. 
\end{proof}

Next, observe that $\CovariantDeriv y, \CovariantDeriv z$ are almost
orthogonal to each other. 
\begin{lemma}
  \label{lemma:y-z-derivatives:S-small}
  Assuming \Cref{eq:smallness-assumption}, we have that in
  $\mathbf{N}$,
  \begin{equation}
    \label{eq:y-z-derivatives:basic:S-small}
    \begin{split}
      \CovariantDeriv_{\beta} y
      ={}& -\frac{1}{2}\KillT\cdot e_4 (e_3)_{\beta}
           + \frac{1}{2}\KillT\cdot e_3(e_4)_{\beta}
           + \widetilde{y}_{\beta}
           ,\\
      \CovariantDeriv_{\beta}z
      ={}&-\frac{1}{2}\KillT^{\alpha}\in_{\alpha\beta\mu\nu}(e_4)^{\mu}(e_3)^{\nu}
           + \widetilde{z}_{\beta}.
    \end{split}    
  \end{equation}
  In particular,
  \begin{equation}
    \label{eq:y-z:small-derivatives:S-small}
    \abs*{\widetilde{y}_{\beta}} + \abs*{\widetilde{z}_{\beta}}
    = O(\varepsilon_{\mathcal{S}}).
  \end{equation}
\end{lemma}
\begin{proof}  
  Observe that using \Cref{eq:P-one-form-P-J-inverse-relation}, we have that 
  \begin{align}
    \CovariantDeriv_{\beta}P
    ={}& \frac{1}{2R}\bm{\sigma}_{\beta}
         + \frac{1}{4J(R-J\sigma_0)}\KillT^{\sigma}\mathcal{X}^{\mu\nu}\mathcal{S}_{\mu\nu\sigma\beta}\notag\\
    ={}& \frac{1}{2}\KillT_{\alpha}\left(
         -(e_4)^{\alpha}(e_3)_{\beta}
         + (e_4)_{\beta}(e_3)^{\alpha}
         - \ImagUnit\in_{\alpha\beta\mu\nu}(e_4)^{\mu}(e_3)^{\nu}
         \right)
         - \frac{P}{4(R-J\sigma_0)}\KillT^{\sigma}\mathcal{X}^{\mu\nu}\mathcal{S}_{\mu\nu\sigma\beta},
         \label{eq:y-z:small-derivatives:S-small:DP-aux}
  \end{align}
  where the second equality follows from contracting
  \Cref{eq:FCal-null-decomposition} against $\KillT$.

  As a result, we can easily find that
  \begin{equation*}   
    \begin{split}
      \CovariantDeriv_{\beta} y
      ={}& -\frac{1}{2}\KillT\cdot e_4 (e_3)_{\beta}
           + \frac{1}{2}\KillT\cdot e_3(e_4)_{\beta}
           + \widetilde{y}_{\beta}
           ,\\
      \CovariantDeriv_{\beta}z
      ={}&-\frac{1}{2}\KillT^{\alpha}\in_{\alpha\beta\mu\nu}(e_4)^{\mu}(e_3)^{\nu}
           + \widetilde{z}_{\beta}.
    \end{split}    
  \end{equation*}
  In particular,
  \begin{equation*}    
    \abs*{\widetilde{y}_{\beta}} + \abs*{\widetilde{z}_{\beta}}
    = O(\varepsilon_{\mathcal{S}}),
  \end{equation*}
  as desired.
\end{proof}

\begin{lemma}
  \label{lemma:nabla-y-nabla-z-smallness-qtys}
  At any point in $\mathbf{N}$, we have that
  \begin{equation}
    \label{eq:nabla-y-nabla-z-smallness-qtys}
    \abs*{
      \CovariantDeriv_{\alpha}y\CovariantDeriv^{\alpha}y
      - \CovariantDeriv_{\alpha}z\CovariantDeriv^{\alpha}z
      - \left(
        c_{S_0} - \frac{b_{S_0}y}{y^2+z^2} - \frac{\Lambda}{3}(y^2-z^2)
      \right)
    }
    + \abs*{\CovariantDeriv_{\alpha}y\CovariantDeriv^{\alpha}z}
    = O(\varepsilon_{\mathcal{S}}).
  \end{equation}
\end{lemma}

\begin{proof}
  Using \Cref{lemma:P-one-form-P-J-inverse-relation}, we have that
  \begin{equation*}
    \begin{split}
      \Metric(\CovariantDeriv P,\CovariantDeriv P)
    ={}& \left( \mathbf{P}_{\alpha} - \frac{P}{4(R-J\sigma_0)}\KillT^{\sigma}\mathcal{X}^{\mu\nu}\mathcal{S}_{\mu\nu\sigma\alpha} \right)
    \left( \mathbf{P}^{\alpha} - \frac{P}{4(R-J\sigma_0)}\KillT^{\sigma'}\mathcal{X}^{\mu'\nu'}\tensor[]{\mathcal{S}}{_{\mu'\nu'\sigma'}^{\alpha}} \right).    
    \end{split}    
  \end{equation*}
  Recalling from \Cref{lemma:renormalized-qtys:basic-props} that
  \begin{equation*}
    \Metric(\mathbf{P}, \mathbf{P}) = -\Metric(\KillT, \KillT),
  \end{equation*}
  we find that
  \begin{equation}
    \label{eq:CovP-N:relation}
    \Metric(\CovariantDeriv P,\CovariantDeriv P)
    ={}  -\Metric(\KillT,\KillT) + O(\varepsilon_{\mathcal{S}}). 
  \end{equation}
  Taking the real part of both sides and using the fact that
  $-\KillT^{\alpha}\KillT_{\alpha} = \Re\sigma$ gives that
  \begin{align*}
    \CovariantDeriv_{\alpha}y\CovariantDeriv^{\alpha}y
    - \CovariantDeriv_{\alpha}z\CovariantDeriv^{\alpha}z
    ={}& -\KillT^{\alpha}\KillT_{\alpha} + O(\varepsilon_{\mathcal{S}})\\
    ={}& \Re \sigma + O(\varepsilon_{\mathcal{S}})\\
    ={}& c_{S_0} + \Re\left(-\frac{b_{S_0}}{P}(1+\widetilde{b}) - \frac{\Lambda}{3}P^2\right) + O(\varepsilon_{\mathcal{S}}),    
  \end{align*}
  where the last equality follows from \Cref{lemma:R-sigma-in-terms-of-P:S-small}.
  Then, using \Cref{lemma:b-c-k-almost-constant:S-small} and decomposing $P = y+\ImagUnit z$, we have that 
  \begin{equation*}
    \abs*{
      \CovariantDeriv_{\alpha}y\CovariantDeriv^{\alpha}y
      - \CovariantDeriv_{\alpha}z\CovariantDeriv^{\alpha}z
      - \left(
        c_{S_0} - \frac{b_{S_0}y}{y^2+z^2} - \frac{\Lambda}{3}(y^2-z^2)
      \right)
    } = O(\varepsilon_{\mathcal{S}}).
  \end{equation*}
  We can similarly calculate that
  \begin{equation*}
    \CovariantDeriv_{\alpha}y\CovariantDeriv^{\alpha}z = O(\varepsilon_{\mathcal{S}})
  \end{equation*}
  directly from \Cref{lemma:y-z-derivatives:S-small}.
\end{proof}

Finally, we can show that $k$ as defined in
\Cref{eq:b-c-k:def} is also almost constant.
\begin{lemma}
  \label{lemma:k-almost-constant}
  Then,
  \begin{equation}
    \label{eq:k-almost-constant}
    k = k_{S_0}(1+\widetilde{k}),\qquad \abs*{\widetilde{k}} + \abs*{\CovariantDeriv_{\alpha}k} = O(\varepsilon_{\mathcal{S}}).
  \end{equation}
\end{lemma}
\begin{proof}
  Similar to the proof of \Cref{lemma:b-c-k-almost-constant:S-small},
  we see that since $\KillT(k)=0$ it suffices to show that
  $\abs*{\CovariantDeriv k}=O(\varepsilon_{\mathcal{S}})$.

  We first observe from \Cref{eq:b-c-k:def:renorm} that
  \begin{equation}
    \label{eq:k-almost-constant:k-deriv-preliminary}
    \begin{split}
      \CovariantDeriv_{\alpha}k
      ={}& 2(y^2+z^2)\CovariantDeriv_{\alpha}\CovariantDeriv_{\beta}z\CovariantDeriv^{\beta}z
           + \left(2y\CovariantDeriv_{\alpha} y + 2z\CovariantDeriv_{\alpha}z\right)\CovariantDeriv_{\beta}z\CovariantDeriv^{\beta}z\\
         & + 2cz\CovariantDeriv_{\alpha}z
           + \frac{4\Lambda}{3}z^3\CovariantDeriv_{\alpha}z
           + 2\CovariantDeriv_{\alpha}c z^2
           .
    \end{split}
  \end{equation}
  We can then compute using the definition of $\mathbf{P}$ that
  \begin{align*}
    \CovariantDeriv_{\alpha}\mathbf{P}_{\beta}
    ={}& - \frac{1}{2R^2}\bm{\sigma}_{\beta}\CovariantDeriv_{\alpha}R + \frac{1}{2R}\CovariantDeriv_{\alpha}\bm{\sigma}_{\beta}\\
    ={}& -\frac{2JR-\Lambda}{R}\mathbf{P}_{\alpha}\mathbf{P}_{\beta}
         + \frac{1}{4R}\mathbf{P}_{\beta}\KillT^{\sigma}\mathcal{X}^{\mu\nu}\mathcal{S}_{\mu\nu\sigma\alpha}
         + \frac{1}{2}R \Metric_{\alpha\beta}
         - \frac{1}{R}\mathcal{T}_{\alpha\beta}\\
      &  + \frac{1}{R}\KillT^{\rho}\KillT^{\sigma}\mathcal{S}_{\rho\alpha\sigma\beta}
         + \frac{Q}{R} \KillT^{\rho}\KillT^{\sigma}\mathcal{U}_{\rho\alpha\sigma\beta}
         - \frac{\Lambda}{3R}N \Metric_{\alpha\beta}
         - \frac{\Lambda}{3R}\KillT_{\alpha}\KillT_{\beta}.
  \end{align*}
  Then observe that
  \begin{align*}
    \KillT^{\rho}\KillT^{\sigma}\mathcal{U}_{\rho\alpha\sigma\beta}
    ={}& \frac{1}{4}\bm{\sigma}_{\alpha}\bm{\sigma}_{\beta}
         + \frac{R^2}{3}\left(-N \Metric_{\alpha\beta} - \KillT_{\beta}\KillT_{\alpha} + \ImagUnit \in_{\rho\alpha\sigma\beta}\KillT^{\rho}\KillT^{\sigma}\right),
  \end{align*}
  so that we have that
  \begin{align*}
    \frac{Q}{R}\KillT^{\rho}\KillT^{\sigma}\mathcal{U}_{\rho\alpha\sigma\beta}
    ={}& \left( \frac{3J}{R^2} - \frac{\Lambda}{R^3} \right)
         \left(
         R^2\mathbf{P}_{\alpha}\mathbf{P}_{\beta}
         + \frac{R^2}{3}\left(-N \Metric_{\alpha\beta} - \KillT_{\beta}\KillT_{\alpha} + \ImagUnit \in_{\rho\alpha\sigma\beta}\KillT^{\rho}\KillT^{\sigma}\right)
         \right)\\
    ={}& \left(3J - \frac{\Lambda}{R}\right)
         \left(
         \mathbf{P}_{\alpha}\mathbf{P}_{\beta}
         -\frac{N}{3} \Metric_{\alpha\beta}
         - \frac{1}{3}\KillT_{\beta}\KillT_{\alpha}
         + \frac{1}{3}\ImagUnit \in_{\rho\alpha\sigma\beta}\KillT^{\rho}\KillT^{\sigma}
         \right).
  \end{align*}
  As a result, we can compute that
  \begin{equation}
    \label{eq:D-P}
    \begin{split}
          \CovariantDeriv_{\alpha}\mathbf{P}_{\beta}
    ={}& J \mathbf{P}_{\alpha}\mathbf{P}_{\beta}
         + \frac{1}{4R}\KillT^{\sigma}\mathcal{X}_{\mu\nu}\mathcal{S}_{\mu\nu\sigma\alpha}\mathbf{P}_{\beta}
         + \left(\frac{R}{2} - NJ\right)\Metric_{\alpha\beta}
         - \frac{1}{R}\mathcal{T}_{\alpha\beta}\\
       & + \frac{1}{R}\KillT^{\rho}\KillT^{\sigma}\mathcal{S}_{\rho\alpha\sigma\beta}
         - J \KillT_{\alpha}\KillT_{\beta}.
    \end{split}
  \end{equation}
  Then, since we have from
  \Cref{lemma:P-one-form-P-J-inverse-relation} that $\KillT(z)=0$, we
  have that
  \begin{align*}
    \CovariantDeriv_{\alpha}\mathbf{P}_{\beta}\CovariantDeriv^{\beta}z
    ={}& -P^{-1}\mathbf{P}_{\alpha}\mathbf{P}_{\beta}\CovariantDeriv^{\beta}z
         + \frac{1}{4R}\KillT^{\sigma}\mathcal{X}^{\mu\nu}\mathcal{S}_{\mu\nu\sigma\alpha}\mathbf{P}_{\beta}\CovariantDeriv^{\beta}z
         + \left(\frac{R}{2}-NJ\right)\CovariantDeriv_{\alpha}z
         - \frac{1}{R}\mathcal{T}_{\alpha\beta}\CovariantDeriv^{\beta}z\\
       & + \frac{1}{R}\KillT^{\rho}\KillT^{\sigma}\mathcal{S}_{\rho\alpha\sigma\beta}\CovariantDeriv^{\beta}z
         .
  \end{align*}
  Then we find that
  \begin{align*}
    (y^2+z^2)\Im \left[P^{-1}N\right]\CovariantDeriv_{\alpha}z
    ={}& \Im\left[(y-\ImagUnit z)\mathbf{P}_{\beta}\cdot \mathbf{P}^{\beta}\right]\CovariantDeriv_{\alpha}z\\
    ={}& \Im \left[(y-\ImagUnit z)\CovariantDeriv_{\beta}(y+\ImagUnit z)\cdot \CovariantDeriv^{\beta}(y+\ImagUnit z) + O(\varepsilon_{\mathcal{S}})\right]\CovariantDeriv_{\alpha}z\\
    ={}& 2y \CovariantDeriv_{\alpha}z \CovariantDeriv_{\beta}y\CovariantDeriv^{\beta}z
         - z\CovariantDeriv_{\alpha}z \left(\CovariantDeriv_{\beta}y\CovariantDeriv^{\beta}y - \CovariantDeriv_{\beta}z\CovariantDeriv^{\beta}z\right)
         + O(\varepsilon_{\mathcal{S}}).
  \end{align*}
  We can also compute that
  \begin{align*}
    (y^2+z^2)\Im\left[P^{-1}\mathbf{P}_{\alpha}\mathbf{P}_{\beta}\CovariantDeriv^{\beta}z\right]
    ={}& \CovariantDeriv^{\beta}z
         \Im\left[(y-\ImagUnit z)\CovariantDeriv_{\alpha}(y+\ImagUnit z)\CovariantDeriv_{\beta}(y+\ImagUnit z) + O(\varepsilon_{\mathcal{S}})\right]\\
    ={}& (y\CovariantDeriv_{\alpha}y + z \CovariantDeriv_{\alpha}z ) \CovariantDeriv^{\beta}z\CovariantDeriv_{\beta}z
         + (y \CovariantDeriv_{\alpha} z - z \CovariantDeriv_{\alpha}y)\CovariantDeriv^{\beta}y \CovariantDeriv_{\beta}z.
  \end{align*}
  We also have from our smallness assumption
  \Cref{eq:smallness-assumption} that
  \begin{equation*}
    \frac{1}{4R}\KillT^{\sigma}\mathcal{X}^{\mu\nu}\mathcal{S}_{\mu\nu\sigma\alpha}\mathbf{P}_{\beta}\CovariantDeriv^{\beta}z
    + \frac{1}{R}\KillT^{\rho}\KillT^{\sigma}\mathcal{S}_{\rho\alpha\sigma\beta}\CovariantDeriv^{\beta}z
    = O(\varepsilon_{\mathcal{S}}). 
  \end{equation*}
  As a result, we have that
  \begin{equation*}
    \begin{split}
      (y^2+z^2)\CovariantDeriv_{\alpha}\CovariantDeriv_{\beta}z\CovariantDeriv^{\beta}z
      ={}& (y^2+z^2)\Im\left[ \CovariantDeriv_{\alpha}\mathbf{P}_{\beta} \right]\CovariantDeriv^{\beta}z + O(\varepsilon_{\mathcal{S}})\\
      ={}& (y^2+z^2)\Im\left[ \left(\frac{R}{2}\right)\CovariantDeriv_{\alpha}z - \frac{1}{R}\mathcal{T}_{\alpha\beta}\CovariantDeriv^{\beta}z \right]
           - z\CovariantDeriv_{\alpha}z \CovariantDeriv_{\beta}y\CovariantDeriv^{\beta}y
         -  y\CovariantDeriv_{\alpha}y  \CovariantDeriv^{\beta}z\CovariantDeriv_{\beta}z \\
         & + (y \CovariantDeriv_{\alpha} z + z \CovariantDeriv_{\alpha}y)\CovariantDeriv^{\beta}y \CovariantDeriv_{\beta}z
           + O(\varepsilon_{\mathcal{S}}).
    \end{split}
  \end{equation*}  
  Returning to \Cref{eq:k-almost-constant:k-deriv-preliminary}, and
  using \Cref{lemma:b-c-k-almost-constant:S-small}, we have that
  \begin{align}    
      \CovariantDeriv_{\alpha}k
      ={}& O(\varepsilon_{\mathcal{S}}) + 2(y^2+z^2)\Im\left[ \left(\frac{R}{2}\right)\CovariantDeriv_{\alpha}z - \frac{1}{R}\mathcal{T}_{\alpha\beta}\CovariantDeriv^{\beta}z \right]
           - 2z\CovariantDeriv_{\alpha}z \CovariantDeriv_{\beta}y\CovariantDeriv^{\beta}y
           -  2y\CovariantDeriv_{\alpha}y  \CovariantDeriv^{\beta}z\CovariantDeriv_{\beta}z \notag\\
         & + 2(y \CovariantDeriv_{\alpha} z + z \CovariantDeriv_{\alpha}y)\CovariantDeriv^{\beta}y \CovariantDeriv_{\beta}z
           + 2(y\CovariantDeriv_{\alpha}y + z\CovariantDeriv_{\alpha}z)\CovariantDeriv_{\beta}z\CovariantDeriv^{\beta}z
           + 2cz\CovariantDeriv_{\alpha}z
           + \frac{4\Lambda}{3}z^3\CovariantDeriv_{\alpha}z\notag\\
      ={}& 2(y^2+z^2)\Im\left[ \left(\frac{R}{2}\right)\CovariantDeriv_{\alpha}z - \frac{1}{R}\mathcal{T}_{\alpha\beta}\CovariantDeriv^{\beta}z \right]
           - 2z\CovariantDeriv_{\alpha}z \left(
           \CovariantDeriv_{\beta}y\CovariantDeriv^{\beta}y
           - \CovariantDeriv_{\beta}z\CovariantDeriv^{\beta}z
           \right)\notag \\
         & + 2(y \CovariantDeriv_{\alpha} z + z \CovariantDeriv_{\alpha}y)\CovariantDeriv^{\beta}y \CovariantDeriv_{\beta}z
           + 2c_{S_0}z\CovariantDeriv_{\alpha}z
           + \frac{4\Lambda}{3}z^3\CovariantDeriv_{\alpha}z
           + O(\varepsilon_{\mathcal{S}}).
           \label{eq:k-almost-constant:k-deriv:2}
  \end{align}
  Then, using \Cref{eq:nabla-y-nabla-z-smallness-qtys}, we have that
  \begin{equation*}
    \begin{split}
      2z\CovariantDeriv_{\alpha}z\left(
      \CovariantDeriv_{\beta}z\CovariantDeriv^{\beta}z
      - \CovariantDeriv_{\beta}y\CovariantDeriv^{\beta}y
      + c_{S_0}
      + \frac{2\Lambda}{3}z^2
      \right)
      ={}& 2z\CovariantDeriv_{\alpha}z\left(
           \frac{b_{S_0}y}{y^2+z^2}
           - \frac{\Lambda}{3}\left( y^2 + z^2 \right) \right)
           + O(\varepsilon_{\mathcal{S}}).
    \end{split}
  \end{equation*}
  We can also compute that for $\alpha = a \in \{1,2\}$, 
  \begin{align*}
    2(y^2+z^2)\Im \left[ \frac{R}{2}\CovariantDeriv_az - \frac{1}{R}\mathcal{T}_{a\beta}\CovariantDeriv^{\beta}z \right]
    ={}& 2(y^2+z^2)\Im(R\CovariantDeriv_az)\\
    ={}& 2\left(\frac{b_{S_0}y}{y^2+z^2} - \frac{\Lambda}{3}(y^2+z^2)\right)z\CovariantDeriv_az
         + O(\varepsilon_{\mathcal{S}}),
  \end{align*}
  where the second equality follows from the exact form of $R$ in
  \Cref{lemma:R-sigma-in-terms-of-P:S-small}.  As a result, we have
  that
  \begin{equation*}
    2(y^2+z^2)\Im \left[ \frac{R}{2}\CovariantDeriv_az - \frac{1}{R}\mathcal{T}_{a\beta}\CovariantDeriv^{\beta}z \right]
    + 2z\CovariantDeriv_az\left(
      \CovariantDeriv_{\beta}z\CovariantDeriv^{\beta}z
      - \CovariantDeriv_{\beta}y\CovariantDeriv^{\beta}y
      + c_{S_0}
      + \frac{2\Lambda}{3}z^2
    \right)
    = O(\varepsilon_{\mathcal{S}}).
  \end{equation*}
  Thus, using \Cref{eq:nabla-y-nabla-z-smallness-qtys} to control
  $\CovariantDeriv_{\beta}y\CovariantDeriv^{\beta}z$, in
  \Cref{eq:k-almost-constant:k-deriv:2},  we have that
  \begin{equation*}
    \CovariantDeriv_ak = O(\varepsilon_{\mathcal{S}}).
  \end{equation*}
  We also have directly from
  \Cref{eq:y-z-derivatives:basic:S-small} that
  \begin{equation*}
    \abs*{\CovariantDeriv_4k} + \abs*{\CovariantDeriv_3k} = O(\varepsilon_{\mathcal{S}}),
  \end{equation*}
  as desired. 
\end{proof}

\begin{lemma}
  \label{lemma:nabla-y-contracted:S-small}
   We have that 
  \begin{equation}
    \label{eq:nabla-y-contracted:S-small}
    \CovariantDeriv^{\alpha}y\CovariantDeriv_{\alpha}y
    = \frac{\Delta}{y^2+z^2}
    + O(\varepsilon_{\mathcal{S}})
    .
  \end{equation}
  Observe that in particular, this implies that away from the points
  where $\Delta(y)=0$, the vectorfield
  $\mathbf{Y} = \Metric^{\alpha\beta}\partial_{\alpha}y\partial_{\beta}$ is spacelike.
\end{lemma}
\begin{proof}
  Using \Cref{lemma:nabla-y-nabla-z-smallness-qtys} and the definition
  of $k$ in \Cref{eq:b-c-k:def}, we can compute that
  \begin{align*}
    \CovariantDeriv^{\alpha}y\CovariantDeriv_{\alpha}y
    ={}& \CovariantDeriv^{\alpha}y\CovariantDeriv_{\alpha}y
         - \CovariantDeriv^{\alpha}z\CovariantDeriv_{\alpha}z
         + \CovariantDeriv^{\alpha}z\CovariantDeriv_{\alpha}z\\
    ={}& c_{S_0} - \frac{b_{S_0}y}{y^2+z^2} - \frac{\Lambda}{3}(y^2-z^2)
         + \frac{k}{y^2+z^2}
         - \frac{cz^2}{y^2+z^2}
         - \frac{\Lambda z^4}{3(y^2+z^2)}
         + O(\varepsilon_{\mathcal{S}}).
  \end{align*}
  Then using \Cref{lemma:k-almost-constant} and \Cref{lemma:b-c-k-almost-constant:S-small}, we have that
  \begin{align*}
    \CovariantDeriv^{\alpha}y\CovariantDeriv_{\alpha}y
    ={}& \frac{k_{S_0} - b_{S_0}y +c_{S_0}y^2 - \frac{\Lambda}{3}y^4}{y^2+z^2} + O(\varepsilon_{\mathcal{S}})\\
    ={}& \frac{\Delta}{y^2+z^2} + O(\varepsilon_{\mathcal{S}})
    ,
  \end{align*}
  as desired.
\end{proof}

\Cref{lemma:nabla-y-contracted:S-small} allows us to write the following helpful lemma.
\begin{lemma}
  \label{lemma:metric-decomp-into-THat-RHat}
  Let us define the vectorfield
  $\mathbf{Y} = \CovariantDeriv^{\alpha}y$. Then we have that
  \begin{equation}
    \label{eq:metric-decomp-into-THat-RHat}
    \Metric = \left( -\abs*{P}^2\Delta + O(\varepsilon_{\mathcal{S}}) \right) \widehat{T}^{\flat}\otimes \widehat{T}^{\flat}
    + \left( \frac{\Delta}{y^2+z^2} + O(\varepsilon_{\mathcal{S}}) \right) \mathbf{Y}^{\flat}\otimes \mathbf{Y}^{\flat}
    + O(\varepsilon_{\mathcal{S}})\mathbf{Y}^{\flat}\otimes e_a^{\flat}
    + r^2\slashed{\Metric}_{\Sphere^2},
  \end{equation}
  where we recall the definition of $\widehat{T}$ from \Cref{def:THat:def}.
\end{lemma}
\begin{proof}
  Using \Cref{lemma:renormalized-qtys:basic-props} and
  \Cref{eq:CovP-N:relation}, we find that 
  \begin{equation*}
    \Metric(\mathbf{p}, \mathbf{p})
    = \Metric(\KillT, \KillT)
    = -\CovariantDeriv^{\alpha}y\CovariantDeriv_{\alpha}y
    + \CovariantDeriv^{\alpha}z\CovariantDeriv_{\alpha}z
    + O(\varepsilon_{\mathcal{S}}),
  \end{equation*}
  and that \Cref{eq:g-eta-T}
  \begin{equation*}
    -\Metric(\mathbf{p},\KillT)
    = \CovariantDeriv^{\alpha}y\CovariantDeriv_{\alpha}y
    + \CovariantDeriv^{\alpha}z\CovariantDeriv_{\alpha}z.
  \end{equation*}
  As a result, we have that
  \begin{align*}
    \frac{1}{\abs*{P}^4}\Metric(\widehat{T}, \widehat{T})
    ={}& \frac{1}{4}\left( \Metric(\KillT, \KillT)
         + 2\Metric(\KillT, \mathbf{p})
         + \Metric(\mathbf{p}, \mathbf{p}) \right)\\
    ={}& \frac{1}{2}\left( - \CovariantDeriv^{\alpha}y\CovariantDeriv_{\alpha}y
         + \CovariantDeriv^{\alpha}z\CovariantDeriv_{\alpha}z
         - \CovariantDeriv^{\alpha}y\CovariantDeriv_{\alpha}y
         - \CovariantDeriv^{\alpha}z\CovariantDeriv_{\alpha}z \right)\\
    ={}& - \CovariantDeriv^{\alpha}y\CovariantDeriv_{\alpha}y
         .
  \end{align*}
  Then from \Cref{lemma:nabla-y-contracted:S-small}, we have that
  \begin{equation*}
    \Metric(\widehat{T}, \widehat{T}) = -\abs*{P}^2\Delta + O(\varepsilon_{\mathcal{S}}).
  \end{equation*}
  Moreover,
  \begin{equation*}
    \Metric(\widehat{T}, \mathbf{Y})
    = \frac{1}{2\abs*{P}^2}\left( \CovariantDeriv_{\KillT} y + \CovariantDeriv_{\mathbf{p}}y \right) = O(\varepsilon_{\mathcal{S}}),
  \end{equation*}
  and
  \begin{equation*}
    \Metric(\mathbf{Y}, e_a) = O(\varepsilon_{\mathcal{S}}).
  \end{equation*}
  The lemma then follows from \Cref{lemma:nabla-y-contracted:S-small}
  and the definitions of $(e_1,e_2)$. 
\end{proof}

\subsection{Properties of the function \texorpdfstring{$y$}{y}}
\label{sec:y-prop}

In this section, we analyze the function $y$ defined in
\Cref{eq:y-z:def}. $y$ will play a significant role in what follows as
it will be function that gives us the appropriately pseudo-convex
foliation of $\mathbf{E}$ which we will use in order to extend both
$\mathcal{S}$ and $\underline{\mathbf{K}}$.

\subsubsection{Control of \texorpdfstring{$y$}{y} in a neighborhood of the bifurcation spheres}
\label{sec:y0 control bifurcation spheres}

In this section, we will analyze the function $y$ in a neighborhood of
the bifurcation spheres $S_0$ and $\underline{S}_0$. Throughout this
section we will assume that the smallness assumption in
\Cref{eq:S-smallness-assumption} holds. We recall from the discussion
at the beginning of \Cref{sec:S-small-consequences} that this
assumption is satisfied in practice everywhere where the results in
this section will be applied.

\begin{lemma}
  \label{lemma:y:S0-properties}
  On the bifurcation sphere $\underline{S}_0$,
  \begin{equation}
    \label{eq:y:S0-properties:y-almost-constant}
    \abs*{y - y_{\underline{S}_0}}
    + \sum_{\alpha=0}^3\abs*{\partial_{\alpha}y} \lesssim \varepsilon_{\mathcal{S}}^{\frac{1}{40}}. 
  \end{equation}
  Moreover, there are constants $\underline{r}_1<2, C\gg 1$ such that on $\underline{\Sigma}_2\backslash \underline{\Sigma}_{\underline{r}_1}$
  \begin{equation}
    \label{eq:y:S0-properties:uniform-behavior-near-S0}
    C\left( (\underline{r}_1-2)^2 + \varepsilon_{\mathcal{S}}^{\frac{1}{40}} \right)
    \ge y_{\underline{S}_0} - y
    \ge \frac{1}{C} (\underline{r}_1-2)^2 - C\varepsilon_{\mathcal{S}}^{\frac{1}{40}}.
  \end{equation}
  Similarly, we have that on the bifurcation sphere $S_0$,
  \begin{equation}
    \label{eq:y:S0-properties:y-almost-constant:S0}
    \abs*{y - y_{S_0}}
    + \sum_{\alpha=0}^3\abs*{\partial_{\alpha}y} \lesssim \varepsilon_{\mathcal{S}}^{\frac{1}{40}}. 
  \end{equation}
  Moreover, there are constants $r_1>2, C\gg 1$ such that on $\Sigma_1\backslash \Sigma_{r_1}$,
  \begin{equation}
    \label{eq:y:S0-properties:uniform-behavior-near-S0:S0}
    C\left( (r_1-1)^2 + \varepsilon_{\mathcal{S}}^{\frac{1}{40}} \right)
    \ge y - y_{S_0}
    \ge \frac{1}{C} (r_1-1)^2 - C\varepsilon_{\mathcal{S}}^{\frac{1}{40}}.
  \end{equation}
\end{lemma}

\begin{proof}
  We only prove \Cref{eq:y:S0-properties:y-almost-constant} and
  \Cref{eq:y:S0-properties:uniform-behavior-near-S0}, as the proofs of
  \Cref{eq:y:S0-properties:y-almost-constant:S0} and
  \Cref{eq:y:S0-properties:uniform-behavior-near-S0:S0} are similar.
  Working on the cosmological horizon, we may now assume that
  $(e_3,e_4) = (\underline{L}_-,-\underline{L}_+)$. From
  \Cref{lemma:horizon:Fcal-null-components}, we see that ,
  \begin{equation*}
    \mathcal{F}\cdot e_4 = -P(\mathcal{F})e_4 \text{ on } \CosmologicalHorizonPast,
    \qquad
    \mathcal{F}\cdot e_3 = P(\mathcal{F})e_3 \text{ on } \CosmologicalHorizonFuture,
  \end{equation*}
  so on $\CosmologicalHorizonFuture$, we can identify
  $\underline{L}_-= \bm{\ell}_-$ and on $\CosmologicalHorizonPast$, we
  can identify $-\underline{L}_+=\bm{\ell}_+$.  In particular we may
  also assume that $e_1,e_2$ are tangent to $\underline{S}_0$. Since
  $\KillT$ is tangent to $\underline{S}_0$, we have that
  \begin{equation*}
     e_4\cdot\bm{\sigma}
    = 2\mathcal{F}(\KillT, e_4)
    = 0.
  \end{equation*}
  Similarly, $e_3\cdot\bm{\sigma}=0$.  Then, we see from
  \Cref{coro:non-expanding-null-hypersurface:horizon-qtys} and
  \Cref{lemma:horizon:Fcal-null-components} that on $\underline{S}_0$,
  \begin{equation*}
    \mathcal{X}^{\mu\nu}\KillT^{\sigma}\mathcal{S}_{\mu\nu\sigma 3}=0 \text{ on } \CosmologicalHorizonFuture
    ,\qquad
    \mathcal{X}^{\mu\nu}\KillT^{\sigma}\mathcal{S}_{\mu\nu\sigma 4}
    = 0 \text{ on } \CosmologicalHorizonPast
    .
  \end{equation*}
  As a result, using \Cref{eq:y-z:small-derivatives:S-small:DP-aux},
  \begin{equation}
     \label{eq:S0:y-null-derivatives-along-horizons}
     e_3(y) = 0 \text{ on } \CosmologicalHorizonFuture
     ,\qquad
     e_4(y) = 0 \text{ on } \CosmologicalHorizonPast
     .
  \end{equation}
  Using \Cref{eq:y-z-derivatives:basic:S-small}, we see that on $\underline{S}_0$,
  \begin{equation}
    \label{eq:S0:y-null-derivatives}
    e_4(y) = e_3(y) = 0,\qquad 
    \abs*{e_1(y)} + \abs*{e_2(y)} = O(\varepsilon_{\mathcal{S}}).
  \end{equation}
  We thus immediately have that on $\underline{S}_0$,
  \begin{equation}
    \label{eq:y:dy-small-bound}
    \sum_{\alpha=0}^3\abs*{\partial_{\alpha}y} \lesssim \varepsilon_{\mathcal{S}}^{\frac{1}{40}}. 
  \end{equation}
  To show the first inequality in
  \Cref{eq:y:S0-properties:y-almost-constant}, we first observe from
  \Cref{eq:y:dy-small-bound} that on $\underline{S}_0$,
  \begin{equation*}
    \abs*{\CovariantDeriv_{\alpha}y\CovariantDeriv^{\alpha}y} \lesssim \varepsilon_{\mathcal{S}}^{\frac{1}{40}}. 
  \end{equation*}
  Thus from \Cref{eq:nabla-y-contracted:S-small} we have that on
  $\underline{S}_0$,
  \begin{equation*}
    \abs*{\Delta} \lesssim \varepsilon_{\mathcal{S}}^{\frac{1}{40}}.
  \end{equation*}
  Recall from our assumptions in \Cref{eq:S0-y-Delta-assumptions} that
  $\partial_y\Delta(y_{S_0})>0$. However, we know that $\Delta(y)$ can
  only have three positive roots, only one of which satisfies
  $\partial_y\Delta>0$. Thus, we must have that for 
  $\varepsilon_{\mathcal{S}}$ sufficiently small,
  $\evalAt*{\partial_y\Delta}_{p\in \underline{S}_0} <0$, we thus have
  that
  \begin{equation*}
    \abs*{y-y_{\underline{S}_0}} \lesssim \varepsilon_{\mathcal{S}}^{\frac{1}{40}}
  \end{equation*}
  on $\underline{S}_0$. This proves \Cref{eq:y:S0-properties:y-almost-constant}.

  To prove \Cref{eq:y:S0-properties:uniform-behavior-near-S0}, we
  begin by deriving a wave equation for $y$. Recall from
  \Cref{eq:deriv-P} that
  \begin{align*}
    \CovariantDeriv_{\alpha}\mathbf{P}_{\beta}
    ={}& J \mathbf{P}_{\alpha}\mathbf{P}_{\beta}
         + \frac{1}{4R}\KillT^{\sigma}\mathcal{X}_{\mu\nu}\mathcal{S}_{\mu\nu\sigma\alpha}\mathbf{P}_{\beta}
         + \left(\frac{R}{2} - NJ\right)\Metric_{\alpha\beta}
         - \frac{1}{R}\mathcal{T}_{\alpha\beta}\\
       & + \frac{1}{R}\KillT^{\rho}\KillT^{\sigma}\mathcal{S}_{\rho\alpha\sigma\beta}
         - J \KillT_{\alpha}\KillT_{\beta}.
  \end{align*}
  This then implies that
  \begin{equation}
    \label{eq:y:S0:aux1}
    \CovariantDeriv^{\alpha}\mathbf{P}_{\alpha}
    ={} - \frac{1}{P}\mathbf{P}^{\alpha}\mathbf{P}_{\alpha}
    + 2\left(R + 2\frac{N}{P}\right)
    - \frac{N}{P}
    + \frac{1}{4R}\KillT^{\sigma}\mathcal{X}^{\mu\nu}\mathcal{S}_{\mu\nu\sigma\alpha}\mathbf{P}^{\alpha}.
  \end{equation}
  Observe that using \Cref{eq:renormalized-qtys:basic-props} we also have that 
  \begin{equation}
    \label{eq:y:S0:aux2}
    - \frac{1}{P}\mathbf{P}^{\alpha}\mathbf{P}_{\alpha}
    + 2\left(R + 2\frac{N}{P}\right)
    - \frac{N}{P}
    ={} 2\left(R + \frac{N}{P}\right).        
  \end{equation}
  whereupon we observe that by using \Cref{lemma:R-sigma-in-terms-of-P:S-small},
  \begin{equation}
    \label{eq:y:S0:aux3}
    2\Re\left(R + \frac{N}{P}\right)
    ={}2\left(\Re R  +\Re \frac{N}{P}\right)
    ={} \frac{-b_{S_0}+2c_{S_0}y}{2(y^2+z^2)}
    - \frac{2\Lambda}{3}\frac{y^3}{y^2+z^2}
    + O(\varepsilon_{\mathcal{S}})
    .
  \end{equation}
  Thus, combining \Cref{eq:y:S0:aux1}, \Cref{eq:y:S0:aux2}, and
  \Cref{eq:y:S0:aux3}, and taking the real part of each equation, we
  find that on $\underline{\Sigma}_{2+\varepsilon_\mathcal{S}}$,
  \begin{equation}
    \label{eq:y:S0:wave-eq-y}
    \CovariantDeriv^{\alpha}\CovariantDeriv_{\alpha}y
    = \frac{-b_{S_0}+2c_{S_0}y}{2(y^2+z^2)}
    - \frac{2\Lambda}{3}\frac{y^3}{y^2+z^2}
    + O(\varepsilon_{\mathcal{S}}^{\frac{1}{5}}). 
  \end{equation}
  We will now compare $y$ with a function $y'$ which coincides with
  $y$ on $\CosmologicalHorizonFuture$ and verifies
  $\underline{L}_+(y')=0$. For
  $\varepsilon_1 = \varepsilon_1(A_0)\in (0, c_0]$ sufficiently small we
  define the function
  \begin{equation}
    \label{eq:y:S0:aux-yprime-def}
    y':\mathbf\underline{\mathbf{O}}_{\varepsilon_1}\to \Real,\qquad
    y'= y \quad \text{on }\CosmologicalHorizonFuture\bigcap \underline{\mathbf{O}}_{\varepsilon_1},\qquad
    \underline{L}_{+}(y') = 0\qquad \text{ in }\underline{\mathbf{O}}_{\varepsilon_1}. 
  \end{equation}
  The functions $y$ and $y'$ are smooth on
  $\underline{\mathbf{O}}_{\varepsilon_1}$. Then from
  \Cref{eq:S0:y-null-derivatives-along-horizons}, 
  \Cref{eq:S0:y-null-derivatives} and the definition of $y'$ in
  \Cref{eq:y:S0:aux-yprime-def}, we have that
  \begin{equation}
    \label{eq:y:S0:y-yprime-difference-horizons}
    y-y' = 0\quad\text{on }\left( \CosmologicalHorizonFuture\bigcup \CosmologicalHorizonPast \right)\bigcap \underline{\mathbf{O}}_{\varepsilon_1}.
  \end{equation}
  In addition, again using \Cref{eq:S0:y-null-derivatives}, we have that
  \begin{equation*}
    \abs*{\nabla y'} + \abs*{e_3y'} +\abs*{e_4y'}\lesssim \varepsilon_{\mathcal{S}}
    \quad \text{on } \underline{S}_0.
  \end{equation*}
  Then using \Cref{eq:y-z-derivatives:basic:S-small}, we have that
  \begin{equation}
    \label{eq:y:S0:yprime-wave-eq:eaeb-bounds}
    \nabla^a\nabla_a y' = O(\varepsilon_{\mathcal{S}})\quad \text{on }\underline{S}_0.
  \end{equation}
  On the other hand, using $\underline{L}_+y' = 0$, we have that
  \begin{equation}
    \label{eq:y:S0:yprime-wave-eq:e3e4-bounds}
    \abs*{\CovariantDeriv_3\CovariantDeriv_4y'}
    +\abs*{\CovariantDeriv_4\CovariantDeriv_3y'}
    \lesssim \varepsilon_{\mathcal{S}}
    \quad \text{on }\underline{S}_0.
  \end{equation}
  Combining \Cref{eq:y:S0:yprime-wave-eq:eaeb-bounds} and
  \Cref{eq:y:S0:yprime-wave-eq:e3e4-bounds}, we have that
  \begin{equation}
    \label{eq:y:S0:yprime-box-bound}
    \abs*{\Box_{\Metric}y'}\lesssim \varepsilon_{\mathcal{S}}^{\frac{1}{5}}\quad \text{on }\underline{S}_0. 
  \end{equation}
  Given, \Cref{eq:y:S0:y-yprime-difference-horizons}, we have that
  there must exist some $\varepsilon_2 \in (0, \varepsilon_1)$ such
  that there exists some function
  $f\in C^{\infty}(\underline{\mathbf{O}}_{\varepsilon_2}; \Real)$ such that
  \begin{equation*}
    y - y' = \underline{u}_+ \underline{u}_- f\quad \text{in }\underline{\mathbf{O}}_{\varepsilon_2}.
  \end{equation*}
  On $\underline{S}_0$, we have that
  $\underline{u}_+=\underline{u}_-=0$, and
  $\CovariantDeriv \underline{u}_-\cdot\CovariantDeriv \underline{u}_+
  = -1$, so we can write using \Cref{eq:y:S0:wave-eq-y} and
   \Cref{eq:y:S0:yprime-box-bound}, that for some $\widetilde{C}$ sufficiently large,
   \begin{equation}
     \label{eq:y:S0:f-S0-bound}
     f = - \frac{1}{2}\CovariantDeriv^{\alpha}\CovariantDeriv_{\alpha}(y-y')
     = \frac{b_{S_0}-2c_{S_0}y}{4(y^2+z^2)} + \frac{\Lambda}{3}\frac{y^3}{y^2+z^2}
     + \widetilde{C} \varepsilon_{\mathcal{S}}^{\frac{1}{5}}.
   \end{equation}
   Now, since $y$ is close to $y_{\underline{S}_0}$ on $\underline{S}_0$, and
   $\partial_y\Delta(y_{\underline{S}_0}) = -b_{S_0} + 2c_{S_0}y_{\underline{S}_0} - \frac{4\Lambda}{3}y_{\underline{S}_0}^3<0$, we have that there must exist some
   $\varepsilon_3\in (0, \varepsilon_2)$ such that
   $f\in [-\widetilde{C}, \widetilde{C}^{-1}]$, where $\widetilde{C}$
   is the implicit constant in \Cref{eq:y:S0:f-S0-bound}. We also have that
   \begin{equation*}
     \abs*{y'-y_{\underline{S}_0}}\lesssim \varepsilon_{\mathcal{S}}^{\frac{1}{40}},\quad
     \text{in }\mathbf{O}_{\varepsilon_3},
   \end{equation*}
   and
   \begin{equation*}
     \frac{\underline{u}_+\underline{u}_-}{(r-2)^2}\in [-\widetilde{C}, \widetilde{C}^{-1}]. 
   \end{equation*}
   The inequalities in
   \Cref{eq:y:S0-properties:uniform-behavior-near-S0} then follow
   from the fact that $y = y' + f \underline{u}_+\underline{u}_-$.
\end{proof}

\subsubsection{Regularity properties of \texorpdfstring{$y$}{y} away from the bifurcate spheres}

We now derive some properties of the level sets of the function
$y$. Define
\begin{equation}
  \label{eq:y0-yBar0:def}
  y_0 \vcentcolon=  y_{S_0} + C_0^{-1},\qquad
  \underline{y}_0 \vcentcolon=  y_{\underline{S}_0} - C_0^{-1},
\end{equation}
where we will subsequently fix $C_0$ to be a sufficiently large
constant depending on $C$ and $\underline{r}_1$ from \Cref{lemma:y:S0-properties}.
Then, for $R\le \underline{y}_0$, we define

\begin{definition}
  \label{def:U-cal-R}    
  For $R > y_{S_0}$, we define
  \begin{equation}
    \label{eq:y-near-S0:V-U-def}
    \begin{split}
      \mathcal{V}_R&\vcentcolon= \curlyBrace*{x\in \Sigma_1: y(x)< R},\\
      \mathcal{U}_{R} &\vcentcolon= \{\text{the connected component of $\mathcal{V}_R$ such that its closure contains $S_0$}\}.
    \end{split}    
  \end{equation}
  We also define for $\underline{R}\le \underline{y}_0$, we define
  \begin{equation}
    \label{eq:y-near-Sbar0:VBar-UBar-def}
    \begin{split}
      \underline{\mathcal{V}}_{\underline{R}} &\vcentcolon= \{p\in \underline{\Sigma}_2: y(p) > \underline{R}\};\\
      \underline{\mathcal{U}}_{\underline{R}} &\vcentcolon= \{\text{the connected component of $\underline{\mathcal{V}}_{\underline{R}}$ such that its closure contains $\underline{S}_0$}\}.
    \end{split}
  \end{equation}
\end{definition}

\begin{definition}
  \label{def:E-EBar:def}
  Denoting by \emph{$\Psi_{t,\KillT}$ the flow associated to $\KillT$},
  we define for $R\ge y_0$,
  \begin{equation}
    \label{eq:E-R:def}
    \mathbf{E}_R \vcentcolon= \curlyBrace*{p\in \mathbf{E}: y(p)<R}
    = \bigcup_{t\in\Real}\Psi_{t,\KillT}(\mathcal{U}_R)\subset \mathbf{E}.
  \end{equation}
  Similarly, we define for $R\le \underline{y}_0$,
  \begin{equation}
    \label{eq:E-bar-R:def}
    \underline{\mathbf{E}}_R \vcentcolon= \curlyBrace*{p\in \mathbf{E}: y(p)>R}
    = \bigcup_{t\in\Real}\Psi_{t,\KillT}(\underline{\mathcal{U}}_R)\subset \mathbf{E}.
  \end{equation}
\end{definition}
Observe that in view of
\Cref{eq:y:S0-properties:uniform-behavior-near-S0}, we have that if
$C_0$ is sufficiently large and $\varepsilon_{\mathcal{S}}$ is
sufficiently small, then
\begin{equation}
  \label{eq:basic-inclusions-VBar}
  \underline{\Sigma}_2\backslash \underline{\Sigma}_{2-(4CC_0)^{-\frac{1}{2}}}
  \subset \underline{\mathcal{V}}_{\underline{y}_0}
  \subset \underline{\mathcal{V}}_{\underline{y}_0-C_0^{-1}}\bigcap (\underline{\Sigma}_{2}\backslash \underline{\Sigma}_{\underline{r}_1})
  \subset\underline{\Sigma}_2\backslash \underline{\Sigma}_{2-(4CC_0^{-1})^{\frac{1}{2}}}.
\end{equation}
In particular, this implies that 
\begin{equation}
  \label{eq:basic-inclusions-UBar}
  \underline{\Sigma}_{2}\backslash \underline{\Sigma}_{2-(4CC_0)^{-\frac{1}{2}}}
  \subset \underline{\mathcal{U}}_{\underline{y}_0}
  \subset \underline{\mathcal{U}}_{\underline{y}_0-C_0^{-1}}\bigcap (\underline{\Sigma}_2\backslash \underline{\Sigma}_{\underline{r}_1})
  \subset\underline{\Sigma}_2\backslash \underline{\Sigma}_{2-(4CC_0^{-1})^{\frac{1}{2}}}.
\end{equation}
For
$p = \Phi_1(0,q) \in \Sigma_0\bigcap
E_{1-\varepsilon_0,2+\varepsilon_0}$ and $r\le \varepsilon_0$, we
define
\begin{equation*}
  B_r(p) = \Phi_1\left(\curlyBrace*{(t,q'}\in (-\varepsilon_0,\varepsilon_0)\times E_{1-\varepsilon_0,2+\varepsilon_0}: t^2 + \abs*{q-q'}^2 < r^2\right).
\end{equation*}
Recall that for any set $U\subset \underline{\Sigma}_2$, we denote by
$\partial_{\underline{\Sigma}_2}(U)$ its boundary in
$\underline{\Sigma}_2$. Moreover, it is clear that for some
$R\le y_{\underline{S}_0}$, if
$p\in \partial_{\underline{\Sigma}_2}(\underline{\mathcal{U}}_R)$,
then $y(p) = R$. Similarly, for $R\ge y_{S_0}$, if
$p\in \partial_{\Sigma_1}(\mathcal{U}_R)$, then $y(p) = R$.

We then define $\widetilde{\mathbf{Y}}$ to be the projection of
$\mathbf{Y} = \CovariantDeriv^{\alpha}y\CovariantDeriv_{\alpha}$ in
$\Phi_1\left[ (-\varepsilon_0,\varepsilon_0)\times
  E_{1-\varepsilon_0,2+\varepsilon_0} \right]$ along the hypersurface
$\Sigma_{1-\varepsilon_0,2+\varepsilon_0}$,
\begin{equation}
  \label{eq:Yprime:def}
  \widetilde{\mathbf{Y}} = \mathbf{Y} + \Metric(\mathbf{Y}, T_0)T_0. 
\end{equation}
It is clear then that $\widetilde{\mathbf{Y}}$ is a smooth vectorfield
tangent to $\Sigma_{1-\varepsilon_0,2+\varepsilon_0}$, and
\begin{equation*}
  \sum_{\alpha=1}^3\abs*{(\widetilde{\mathbf{Y}})^{\alpha}}\lesssim 1\quad \text{on }\Sigma_{1-\varepsilon_0,2+\varepsilon_0}. 
\end{equation*}
In addition, we can check that
\begin{equation}
  \label{eq:Yprime-spacelike}
  \widetilde{\mathbf{Y}}(y)
  = \Metric(\widetilde{\mathbf{Y}}, \mathbf{Y})
  = \Metric (\mathbf{Y}, \mathbf{Y})
  + \Metric(\mathbf{Y}, T_0)^2
  \ge \Metric (\mathbf{Y}, \mathbf{Y})
  = \CovariantDeriv_{\alpha}y \CovariantDeriv^{\alpha}y. 
\end{equation}
In particular, for
$p\in \partial_{\underline{\Sigma}_2}\underline{\mathcal{U}}_R$,
where $y_0<R<\underline{y}_0$ we have from the fact that $y(p)=R$ and
\Cref{eq:nabla-y-contracted:S-small} that
$\widetilde{\mathbf{Y}}(y) \gtrsim 1$. Therefore,
\begin{equation}
  \label{eq:UCal-boundary-well-behaved}
  p\in \partial_{\underline{\Sigma}_2}\underline{\mathcal{U}}_R\implies
  \curlyBrace*{x\in B_{\delta}(p)\bigcap \underline{\Sigma}_2: y(x)>R} = B_{\delta}(p)\bigcap\underline{\mathcal{U}}_R,
\end{equation}
for any $\delta\le \delta_2 = \delta_2(A) > 0$. Similarly, for 
$p\in \partial_{\Sigma_1}\mathcal{U}_R$, we have from
the fact that $y(p)=R$ and \Cref{eq:nabla-y-contracted:S-small} that
$\widetilde{\mathbf{Y}}(y) \gtrsim 1$. Therefore, 
\begin{equation}
  \label{eq:UCal-boundary-well-behaved:U}
  p\in \partial_{\Sigma_1}\mathcal{U}_R\implies
  \curlyBrace*{x\in B_{\delta}(p)\bigcap \Sigma_1: y(x)<R} = B_{\delta}(p)\bigcap\mathcal{U}_R,
\end{equation}

We now prove the following lemma, showing that the regions
$\underline{\mathcal{U}}_R$ and $\underline{\mathcal{U}}_R$ grow in a
controlled way.

\begin{lemma}
  \label{lemma:UBar-VBar-behavior}
  The following hold.
  \begin{enumerate}
  \item There exists some $\widetilde{\delta}_2\in (0, \delta_2)$ such that
    for any $\delta\le \widetilde{\delta}_2$ and
    $R \in (y_{S_0}, y_{\underline{S}_0})$,
    \begin{equation}
      \label{eq:UBar-VBar-behavior:growth}
      \begin{gathered}
        \bigcup_{p\in \underline{\mathcal{U}}_R}(B_{\delta^3}(p)\bigcap \underline{\Sigma}_2)
      \subset \underline{\mathcal{U}}_{R-\delta^2}
      \subset \bigcup_{p\in\underline{\mathcal{U}}_R}\left( B_{\delta}(p)\bigcap \underline{\Sigma}_2 \right),\\
      \bigcup_{p\in \mathcal{U}_R}(B_{\delta^3}(p)\bigcap \Sigma_1)
      \subset \mathcal{U}_{R+\delta^2}
      \subset \bigcup_{p\in \mathcal{U}_R}\left( B_{\delta}(p)\bigcap \Sigma_1 \right)
      .
      \end{gathered}      
    \end{equation}
  \item  We have that
    \begin{equation}
      \label{eq:UBar-VBar-behavior:completeness}
      \bigcup_{R\le \underline{y}_0}\underline{\mathcal{U}}_R = \underline{\Sigma}_2,\qquad
      \bigcup_{R\ge y_0}\mathcal{U}_R = \Sigma_1.
    \end{equation}
  \item For any $R\in \left( y_0, \underline{y}_0 \right)$,
    \begin{equation}
      \label{eq:UBar-VBar-behavior:U-V-equal}
      \underline{\mathcal{U}}_R = \underline{\mathcal{V}}_R,\qquad
      \mathcal{U}_R = \mathcal{V}_R
      .
    \end{equation}
  \end{enumerate}  
\end{lemma}
\begin{proof}
  We only prove the first statement in
  \Cref{eq:UBar-VBar-behavior:growth},
  \Cref{eq:UBar-VBar-behavior:completeness}, and
  \Cref{eq:UBar-VBar-behavior:U-V-equal}. The second statements all
  follow similarly.  Since $y$ is a smooth function in a neighborhood
  of $\closure\underline{\Sigma}_2$, it follows that for $\delta$
  sufficiently small, for any
  $q\in B_{\delta^3}(p)\bigcap \underline{\Sigma}_2$, where
  $p\in \underline{\mathcal{U}}_R$, we have that $y(q)>
  R-\delta^2$. As a result, we have that
  \begin{equation*}
    \bigcup_{p\in \underline{\mathcal{U}}_R}\left(B_{\delta^3}(p)\bigcap \underline{\Sigma}_2\right)
    \subset \underline{\mathcal{U}}_{R-\delta^2}.
  \end{equation*}
  To prove that
  \begin{equation*}
    \underline{\mathcal{U}}_{R-\delta^2}
    \subset \bigcup_{p\in \underline{\mathcal{U}}_R}\left(B_{\delta}(p)\bigcap \underline{\Sigma}_2\right),
  \end{equation*}
  we see that it suffices to prove that for $R\le \underline{y}_0$, $\delta$ sufficiently small,
  \begin{equation}
    \label{eq:U-bar-aux-inclusion}
    \underline{\mathcal{U}}_{R-\delta^2}
    \subset \underline{\mathcal{U}}_R\bigcup
    \left[ \bigcup_{p\in \partial_{\underline{\Sigma}_2}\underline{\mathcal{U}}_R}\left(B_{\frac{\delta}{4}}(p)\bigcap \underline{\Sigma}_2\right) \right].
  \end{equation}
  To this end, we assume for the sake of contradiction that there exists some point
  \begin{equation*}
    q\in
    \underline{\mathcal{U}}_{R-\delta^2}\backslash
    \left(\underline{\mathcal{U}}_R\bigcup
  \left[ \bigcup_{p\in
      \partial_{\underline{\Sigma}_2}\underline{\mathcal{U}}_R}\left(B_{\frac{\delta}{4}}(p)\bigcap
      \underline{\Sigma}_2\right) \right]\right).
  \end{equation*}
  Then, define
  \begin{equation*}
    \begin{gathered}
      \gamma:[0,1]\to \underline{\mathcal{U}}_{R-\delta^2}\bigcup
      \underline{S}_0,\qquad \gamma\in C^0([0,1]; \mathcal{M}),\\
      \gamma(0) \in \underline{S}_0,\qquad
      \gamma(1) = q.
    \end{gathered}
  \end{equation*}
  Now, let $q'=\gamma(t')$, where
  \begin{equation*}
    t' = \inf_{(0,1]}\curlyBrace*{s\in (0,1]: q\not\in \underline{\mathcal{U}}_R\bigcup
      \left[ \bigcup_{p\in
          \partial_{\underline{\Sigma}_2}\underline{\mathcal{U}}_R}\left(B_{\frac{\delta}{4}}(p)\bigcap
          \underline{\Sigma}_2\right) \right]}.
  \end{equation*}
  It is clear that
  $q'\not\in \closure_{\underline{\Sigma}_2}\underline{\mathcal{U}}_R$, which
  implies that
  \begin{equation*}
    q'\in \bigcup_{p\in \partial_{\underline{\Sigma}_2}(\underline{\mathcal{U}}_R)}\left( B_{\frac{\delta}{4}}(p)\bigcap \underline{\Sigma}_2 \right).
  \end{equation*}
  Moreover, since $\partial_{\underline{\Sigma}_2}\underline{\mathcal{U}}_R$ is
  a compact set, we have that
  \begin{equation*}
    q'\in B_{\frac{\delta}{2}}(p_0)\bigcap \underline{\Sigma}_2,\quad \text{for some }p_0\in\partial_{\underline{\Sigma}_2}(\underline{\mathcal{U}}_R).
  \end{equation*}
  Now, for $\delta'$ sufficiently small, let us denote by
  $\gamma_p:(0, \delta')\to \underline{\Sigma}_2$ the integral curves of
  $\widetilde{\mathbf{Y}}$ such that $\gamma_p(0)=p$, for any
  $p\in \underline{\Sigma}_{2-(4CC_0)^{-\frac{1}{2}}}$. Using
  \Cref{eq:Yprime-spacelike}, the fact that $y(p_0)=R\le \underline{y}_0$, and
  \Cref{eq:nabla-y-contracted:S-small}, it then follows that in
  $B_{\delta}(p_0)\bigcap \underline{\Sigma}_2$ for $\delta$ sufficiently
  small,
  \begin{equation}
    \label{eq:Yprime-of-y-positive}
    \widetilde{\mathbf{Y}}(y) \gtrsim 1.
  \end{equation}
  Thus, since $q'\in B_{\frac{\delta}{2}}(p_0)\bigcap \underline{\Sigma}_2$, we
  have that for $\delta$ sufficiently small,
  \begin{equation*}
    \gamma_{q'}(t)\in B_{\delta}(p_0)\bigcap \underline{\Sigma}_{2} \forall t\in [-\delta^{\frac{3}{2}},\delta^{\frac{3}{2}}].
  \end{equation*}
  Then, recall that $y(q') > R-\delta^2$ since
  $q'\in \underline{\mathcal{U}}_{R_{\delta^2}}$, so
  \Cref{eq:Yprime-of-y-positive} actually implies that there must be
  some $t''\in [-\delta^{\frac{3}{2}},\delta^{\frac{3}{2}}]$ such that
  $y(q'') > R$, where $q''=\gamma_{q'}(t'')$. Thus, from
  \Cref{eq:UCal-boundary-well-behaved}, we have that
  $q''\in \underline{\mathcal{U}}_R$. But then, since
  $q'\not\in \underline{\mathcal{U}}_R$, there must be a
  $t'''\in [-\delta^{\frac{3}{2}},\delta^{\frac{3}{2}}]$ such that
  $q'''=\gamma_{q'}(t''')\in
  \partial_{\underline{\Sigma}_2}\underline{\mathcal{U}}_R$. But then it
  follows that for $\delta$ sufficiently small,
  \begin{equation*}
    q'\in B_{\frac{\delta}{8}}(q'''),
  \end{equation*}
  which contradicts that assumption that $q'\not\in \underline{\mathcal{U}}_R\bigcup
  \left[ \bigcup_{p\in \partial_{\underline{\Sigma}_2}\underline{\mathcal{U}}_R}\left(B_{\frac{\delta}{4}}(p)\bigcap \underline{\Sigma}_2\right) \right]$. This completes the proof of \Cref{eq:UBar-VBar-behavior:growth}.

  The completeness property in
  \Cref{eq:UBar-VBar-behavior:completeness} follows directly from the
  fact that $\underline{\Sigma}_2$ is compact and $y$ is a smooth
  function on $\underline{\Sigma}_2$.

  To prove \Cref{eq:UBar-VBar-behavior:U-V-equal}, we use
  \Cref{eq:basic-inclusions-VBar} to see that it suffices to
  prove that for any $R\le \underline{y}_0$, 
  \begin{equation*}
    \underline{\mathcal{V}}_R\bigcap \underline{\Sigma}_{2 - (4CC_0)^{-\frac{1}{2}}}\subset \underline{\mathcal{U}}_R.
  \end{equation*}
  To this end, assume for the sake of contradiction that there exists
  some $R_0<\underline{y}_0$ and some
  $q\in \underline{\Sigma}_{2 - (4CC_0)^{-\frac{1}{2}}}$ such that
  $y(q) > R_0$ and $q\not \in \underline{\mathcal{U}}_{R_0}$.  Now let
  $I\vcentcolon= \curlyBrace*{R\in (y_{S_0}, R_0): q\not\in
    \underline{\mathcal{U}}_{R_0}}$, and let $R' = \inf I$ (which
  exists since $I$ is bounded from
  \Cref{eq:UBar-VBar-behavior:completeness}). We now separately
  consider the cases $q\in \underline{\mathcal{U}}_{R'}$ and
  $q\not\in \underline{\mathcal{U}}_{R'}$.

  If $q\in \underline{\mathcal{U}}_{R'}$, then we must have that
  $R'<R_0$. Now from \Cref{eq:UBar-VBar-behavior:growth}, it follows
  for $\delta$ sufficiently small that there must exist some
  $R'' = R'+\delta^2 \le R_0 - \delta^{\frac{1}{2}}$ and some
  $q'\in \underline{\mathcal{U}}_{R''}$ such that
  $\abs*{q-q'}<\delta$.
  However, since by assumption we had that $y(q) > R_0$, it follows that
  \begin{equation*}
    y(x) > R_0 - \delta^{\frac{1}{2}} \ge R'', \forall x\in B_{\delta}(q).
  \end{equation*}
  Since
  $B_{\delta}(q)\bigcap \underline{\mathcal{U}}_{R''}\neq \emptyset$,
  it follows that $q\in \underline{\mathcal{U}}_{R''}$, which
  contradicts the definition of $R'$.

  Now consider the case where $q\not\in
  \underline{\mathcal{U}}_{R'}$. In this case we must have that
  $q\not\in
  \partial_{\underline{\Sigma}_2}\underline{\mathcal{U}}_{R'}$ from
  \Cref{eq:UCal-boundary-well-behaved}. For $\delta$ sufficiently
  small, it follows from \Cref{eq:UBar-VBar-behavior:growth} that
  $q\not\in \underline{\mathcal{U}}_{R'-\delta^2}$, which contradicts
  the definition of $R'$. This concludes the proof of
  \Cref{lemma:UBar-VBar-behavior}.
\end{proof}

\section{Pseudoconvexity}
\label{sec:pseudoconvexity}

In this section, we list the pseudoconvex properties we will make use
of in what follows.

\subsection{Pseudoconvexity at the horizons}
\label{sec:pseudoconvexity:horizons}

We first prove the following basic lemma showing the preliminary
computations we need for the upcoming Carleman estimate.
\begin{lemma}
  \label{lemma:nbhd-horizon:hessian-h}
  At any point $p\in S_{u_-,u_+}$, we
  choose an orthonormal frame $(e_a)_{a=1,2}$ tangent to
  $S_{u_-,u_+}$. Then denote the null
  frame
  $(e_1,e_2,e_3,e_4) =
  (e_1,e_2,L_-,L_+)$ and for
  $0<\varepsilon\le \varepsilon_0$, define
  \begin{equation}
    \label{eq:nbhd-horizon:h:def}
    h_{\varepsilon} = \varepsilon^{-1}(u_++\varepsilon)(u_-+\varepsilon)
  \end{equation}
  in $\mathbf{O}_{\varepsilon^2}$.
  Then, 
  \begin{equation}
    \label{eq:nbhd-horizon:deriv-h}
    L_+(h_{\varepsilon})= \varepsilon^{-1}(u_+ + \varepsilon)\Omega,\qquad
    L_-(h_{\varepsilon}) = \varepsilon^{-1}(u_- + \varepsilon)\Omega,\qquad
    e_a(h_{\varepsilon}) = 0,
  \end{equation}
  and
  \begin{equation}
    \label{eq:nbhd-horizon:hessian-h}
    \begin{gathered}
      (\CovariantDeriv^2h_{\varepsilon})_{33} = (\CovariantDeriv^2h_{\varepsilon})=O(1),\\
      (\CovariantDeriv^2h_{\varepsilon})_{ab} = (\CovariantDeriv^2h_{\varepsilon})_{3a}= (\CovariantDeriv^2h_{\varepsilon})_{4a} = O(1),\\
      (\CovariantDeriv^2h_{\varepsilon})_{34} = \varepsilon^{-1}\Omega^2 + O(1). 
    \end{gathered}
  \end{equation}
  Similarly, at any point $p\in \underline{S}_{u_-,u_+}$, we
  choose an orthonormal frame $(e_a)_{a=1,2}$ tangent to
  $\underline{S}_{u_-,u_+}$. Then denote the null
  frame
  $(e_1,e_2,e_3,e_4) =
  (e_1,e_2,\underline{L}_-,\underline{L}_+)$ and for
  $0<\varepsilon\le \varepsilon_0$, define
  \begin{equation}
    \label{eq:nbhd-horizon:hBar:def}
    \underline{h}_{\varepsilon} = \varepsilon^{-1}(\underline{u}_++\varepsilon)(\underline{u}_-+\varepsilon)
  \end{equation}
  in $\underline{\mathbf{O}}_{\varepsilon^2}$.
  Then
  \begin{equation}
    \label{eq:nbhd-horizon:deriv-hBar}
    \underline{L}_+(\underline{h}_{\varepsilon})= \varepsilon^{-1}(\underline{u}_+ + \varepsilon)\underline{\Omega},\qquad
    \underline{L}_-(\underline{h}_{\varepsilon}) = \varepsilon^{-1}(\underline{u}_- + \varepsilon)\underline{\Omega},\qquad
    e_a(\underline{h}_{\varepsilon}) = 0,
  \end{equation}
  and
  \begin{equation}
    \label{eq:nbhd-horizon:hessian-hBar}
    \begin{gathered}
      (\CovariantDeriv^2\underline{h}_{\varepsilon})_{33} = (\CovariantDeriv^2\underline{h}_{\varepsilon})=O(1),\\
      (\CovariantDeriv^2\underline{h}_{\varepsilon})_{ab} = (\CovariantDeriv^2\underline{h}_{\varepsilon})_{3a}= (\CovariantDeriv^2\underline{h}_{\varepsilon})_{4a} = O(1),\\
      (\CovariantDeriv^2\underline{h}_{\varepsilon})_{34} = \varepsilon^{-1}\underline{\Omega}^2 + O(1). 
    \end{gathered}
  \end{equation}
\end{lemma}
\begin{proof}
  The identities in \Cref{eq:nbhd-horizon:deriv-h} and
  \Cref{eq:nbhd-horizon:deriv-hBar} follow easily from the definition
  and construction of the double null foliation.  In what follows, we
  will only prove \Cref{eq:nbhd-horizon:hessian-h}. The proof of
  \Cref{eq:nbhd-horizon:hessian-hBar} follows very similarly.
  Recall that
  \begin{equation*}
    \Metric(L_{\pm}, L_{\pm}) = 0, \qquad
    \Metric(L_+,L_-) = \Omega > \frac{1}{4}\qquad \text{ in } \mathbf{O}_{\varepsilon_0}.
  \end{equation*}
  Relative to the null frame
  $(e_1,e_2,e_3,e_4)$, the metric $\Metric$ takes the form
  \begin{equation*}
    \begin{gathered}
      \Metric_{ab} = \delta_{ab}, \qquad
      \Metric_{a3}=\Metric_{a4}=0, \\
      \Metric_{33} = \Metric_{44} =0, \qquad
      \Metric_{34}=\Omega.
    \end{gathered}
  \end{equation*}
  The inverse metric takes the form
  \begin{equation*}
    \begin{gathered}
      \Metric^{ab} = \delta^{ab} ,\qquad
      \Metric^{a3} = \Metric^{a4} = 0,\\
      \Metric^{33} = \Metric^{44} = 0, \qquad
      \Metric^{34} = \Omega^{-1}.
    \end{gathered}  
  \end{equation*}
  Moreover, we have that all the Ricci coefficients satisfy
  \begin{equation}
    \label{eq:nbhd-horizon:Ricci-aux}
    \abs*{\Metric(e_{\beta},\CovariantDeriv_{e_{\mu}}e_{\nu})} = O(1). 
  \end{equation}
  Then observe that
  \begin{equation}
    \label{eq:horizon-nbhd:null-deriv-on-u-pm}
    e_a(u_{\pm})
    = e_3(u_-)
    = e_4(u_+)
    =0, \qquad
    e_3(u_+)
    = e_4(u_-)
    = \Omega. 
  \end{equation}
  We observe that
  \begin{equation*}
    e_4(h_{\varepsilon})= \varepsilon^{-1}(u_+ + \varepsilon)\Omega,\qquad
    e_3(h_{\varepsilon}) = \varepsilon^{-1}(u_- + \varepsilon)\Omega,\qquad
    e_a(h_{\varepsilon}) = 0,
  \end{equation*}
  so in particular, $e_\mu(h_{\varepsilon})=O(1)$, $\{\mu=1,2,3,4\}$.  
  Then we have that
  \begin{equation*}
    (\CovariantDeriv^2h_{\varepsilon})_{33}
    = \varepsilon^{-1}(u_- + \varepsilon)e_3(\Omega)  + O(1)
    = O(1),
  \end{equation*}
  where we used the fact that we are working in
  $\mathbf{O}_{\varepsilon^2}$, and \Cref{eq:nbhd-horizon:Ricci-aux}.
  Similarly, we have that
  \begin{align*}
    (\CovariantDeriv^2h_{\varepsilon})_{44}
    ={}& \varepsilon^{-1}(u_+ + \varepsilon)e_4(\Omega)
         + O(1)
         = O(1),\\
    (\CovariantDeriv^2h_{\varepsilon})_{3a}={}& O(1),\\
    (\CovariantDeriv^2h_{\varepsilon})_{4a}={}& O(1),\\
    (\CovariantDeriv^2h_{\varepsilon})_{ab}={}& O(1),\\
    (\CovariantDeriv^2h_{\varepsilon})_{34}={}& \varepsilon^{-1}e_3(u_++\varepsilon)\Omega +O(1)
                                             = \varepsilon^{-1} \Omega^2 +O(1),
  \end{align*}
  as desired.
\end{proof}
The can now prove the classical pseudoconvexity of $S_0$,
$\underline{S}_0$, which is a consequence of the bifurcate null
geometry of the horizons.
\begin{lemma}
  \label{lemma:nbhd-horizon:pseudo-convexity}
  For $h_{\varepsilon}$ as defined in \Cref{eq:nbhd-horizon:h:def},
  the family of weights $\curlyBrace*{h}_{\varepsilon}$ is
  pseudoconvex at any point $x_0\in S_0$.

  Similarly, for $\underline{h}_{\varepsilon}$ as defined in
  \Cref{eq:nbhd-horizon:hBar:def}, the family of weights
  $\curlyBrace*{\underline{h}}_{\varepsilon}$ is pseudoconvex at any point
  $x_0\in \underline{S}_0$.
\end{lemma}
\begin{proof}
  We again only prove the pseudoconvexity of $h_{\varepsilon}$ on
  $S_0$. The pseudoconvexity of $\underline{h}_{\varepsilon}$ on
  $\underline{S}_0$ follows similarly.  We first prove
  \Cref{eq:strict-T-null-convexity:cond1}. From
  \Cref{eq:nbhd-horizon:h:def}, we have that
  \begin{equation*}
    h_{\varepsilon}(x_0)=\varepsilon,
  \end{equation*}
  and moreover, on $B_{\varepsilon^{10}}(x_0)$,
  \begin{equation*}
    \abs*{D^1h_{\varepsilon}}\lesssim 1, \qquad \abs*{D^jh_{\varepsilon}}\lesssim \varepsilon^{-1}, \quad j=2,3,4,
  \end{equation*}
  where we emphasize that the implicit constants depend only on
  $A_0$. Then for $\varepsilon_1$ sufficiently small,
  \Cref{eq:strict-T-null-convexity:cond1} holds.

  We now prove \Cref{eq:strict-T-null-convexity:cond2}. Using
  \Cref{eq:nbhd-horizon:hessian-h} and the fact that
  $\Omega(x_0)=\frac{1}{2}$ and $u_+(x_0) = u_-(x_0)=0$, we compute
  that
  \begin{align*}
    & \CovariantDeriv^{\alpha}h_{\varepsilon}(x_0)\CovariantDeriv^{\beta}h_{\varepsilon}(x_0)\left(
      \CovariantDeriv_{\alpha}h_{\varepsilon}\CovariantDeriv_{\beta}h_{\varepsilon}
      - \varepsilon \CovariantDeriv_{\alpha}\CovariantDeriv_{\beta}h_{\varepsilon}(x_0)
    \right)\\
    ={}& 4\varepsilon^{-4}(u_++\varepsilon)^2(u_-+\varepsilon)^2\Omega^2
         - 2\varepsilon^{-1}(u_++\varepsilon)(u_-+\varepsilon)\left(\varepsilon^{-1}\Omega^2+O(1)\right)\\
    ={}& \frac{1}{2} + \varepsilon O(1) \ge \frac{1}{4}.
  \end{align*}
  for $\varepsilon<\varepsilon_1$ sufficiently small, proving
  \Cref{eq:strict-T-null-convexity:cond2}.

  We now move onto proving
  \Cref{eq:strict-T-null-convexity:main-condition}. Let
  $X^{\alpha}e_{\alpha}\in T_{x_0}\mathcal{M}$. Then, using
  \Cref{eq:nbhd-horizon:hessian-h}, we can compute that
  \begin{align*}
    \mu\abs*{X}_{\Metric}^2 - X^{\alpha}X^{\beta}\CovariantDeriv_{\alpha}\CovariantDeriv_{\beta}h_{\varepsilon}
    (x_0)
    ={}& \mu\left(\abs*{\horProj{X}}^2 + X^3X^4\right)
         - \frac{1}{2}\varepsilon^{-1}X^3X^4
         + O(1)\abs*{X}^2
         ,\\
    \varepsilon^{-2}\abs*{X^{\alpha}\CovariantDeriv_{\alpha}h_{\varepsilon}}^2(x_0)
    ={}& \frac{1}{4}\varepsilon^{-2}\left(X^3+X^4\right)^2. 
  \end{align*}
  As a result, we have that for $\mu = \varepsilon_1^{-\frac{1}{2}}$,
  we have that
  \begin{align*}
    \mu\abs*{X}_{\Metric}^2 - X^{\alpha}X^{\beta}\CovariantDeriv_{\alpha}\CovariantDeriv_{\beta}h_{\varepsilon}
    (x_0)
    +  \varepsilon^{-2}\abs*{X^{\alpha}\CovariantDeriv_{\alpha}h_{\varepsilon}}^2(x_0)
    \ge{}& \frac{\mu}{2}\abs*{\horProj{X}}^2 + \frac{1}{4\varepsilon}\left((X^3)^2 + (X^4)^2\right)\\
    \ge{}& \abs*{X}^2,
  \end{align*}
  so \Cref{eq:strict-T-Z-null-convexity:main-condition} is true for
  $\varepsilon_1$ sufficiently small.
\end{proof}

\subsection{Pseudoconvexity in the interior of \texorpdfstring{$\mathbf{E}$}{E}}

In this section, we discuss various pseudoconvexity lemmas in the
interior of $\mathbf{E}$ we will need throughout the rest of the
paper. We recall that pseudoconvexity will be the key to all of the
extensions we perform in $\mathbf{E}$.

We first relate the $\mathbf{V}$-pseudoconvexity property to the
time-like nature of the vectorfield $\mathbf{V}$. 
\begin{lemma}
  \label{lemma:pseudo-convexity:outside-ergoregion}
  Let $p\in \mathbf{E}$ be a point at which $\mathbf{K}$ is
  a time-like vectorfield. Then, there exists some
  $\mu \lesssim A_0^{-1}$ such that
  \begin{equation}
    \label{eq:pseudo-convexity:outside-ergoregion}
    \begin{split}
      \mathbf{X}^{\alpha}\mathbf{X}^{\beta}\mu\Metric_{\alpha\beta} - \abs*{\mathbf{X}^{\alpha}\mathbf{X}^{\beta}\CovariantDeriv_{\alpha}\CovariantDeriv_{\beta}y}
      + C\left( \abs*{\mathbf{X}(y)}^2+ \abs*{\Metric(\mathbf{K}, \mathbf{X})}^2 \right)&\ge C^{-1}\abs*{\mathbf{X}}^2.
    \end{split}    
  \end{equation}
\end{lemma}
\begin{proof}
  Observe that since
    \begin{equation*}
    \abs*{\mathbf{X}^{\alpha}\mathbf{X}^{\beta}\CovariantDeriv_{\alpha}\CovariantDeriv_{\beta}y} \lesssim \frac{C^{-1}}{2}\abs*{\mathbf{X}}^2,
  \end{equation*}
  it suffices to show that there exists some $\mu'$ such that
  \begin{equation*}
    \abs*{\Metric(\mathbf{K}, \mathbf{X})}^2
    + \mu' \Metric(\mathbf{X}, \mathbf{X})
    \gtrsim \abs*{\mathbf{X}}^2.
  \end{equation*}
  To this end, let us fix some adapted null frame $(e_1,e_2,e_3,e_4)$.
  Let us define
  \begin{equation*}
    e_- = \frac{1}{2}\left( e_4+e_3 \right),\qquad
    e_+ = \frac{1}{2} \left( e_4 - e_3 \right),
  \end{equation*}
  so that
  \begin{equation*}
    e_+\cdot e_+ = - e_-\cdot e_- = 1,\qquad e_{\pm}\cdot e_{\mp} = e_{\pm} \cdot e_a = 0. 
  \end{equation*}
  Then let us decompose
  \begin{equation*}
    \begin{split}
      \mathbf{K} = Ke_- + \overline{\mathbf{K}},\qquad \overline{\mathbf{K}} = \horProj{\mathbf{K}}
      + Ke_- + \underline{K}e_+,\qquad
      \mathbf{X} = Xe_- + \overline{\mathbf{X}},\qquad
      \overline{\mathbf{X}} = \horProj{\mathbf{X}}
      +  \underline{X}e_+.
    \end{split}
  \end{equation*}
  We can compute then that
  \begin{equation*}
    \begin{split}
      \Metric(\mathbf{X}, \mathbf{K})
      ={}& -XK + \overline{\mathbf{X}}\cdot\overline{\mathbf{K}},\\
      \Metric(\mathbf{X}, \mathbf{X})
      ={}& - X^2 + \overline{\mathbf{X}}\cdot\overline{\mathbf{X}}
           ={} - X^2 + \underline{X}^2 + \abs*{\horProj{\mathbf{X}}}^2.
    \end{split}    
  \end{equation*}
  Now, we have that
  \begin{align*}
    \abs*{\Metric(\mathbf{K}, \mathbf{X})}^2
    + \mu' \Metric(\mathbf{X}, \mathbf{X})
    ={}& (XK)^2         
         - 2(XK)\overline{\mathbf{X}}\cdot\overline{\mathbf{K}}
         + \left( \overline{\mathbf{X}}\cdot\overline{\mathbf{K}} \right)^2
         - \mu' X^2
         +  \mu' \overline{\mathbf{X}}\cdot\overline{\mathbf{X}}\\
    ={}& (K^2-\mu')X^2
         - 2(XK)\overline{\mathbf{X}}\cdot\overline{\mathbf{K}}
         + \left( \overline{\mathbf{X}}\cdot\overline{\mathbf{K}} \right)^2
         +  \mu' \overline{\mathbf{X}}\cdot\overline{\mathbf{X}}\\
    \ge{}& \left(
           K^2 - \left(\frac{2-\delta}{2}\right)^2K^2
           - \mu' - \frac{\delta}{2}K\abs{\overline{\mathbf{K}}}
           \right)X^2
           +  \left( \mu' - \frac{\delta}{2}K\abs{\overline{\mathbf{K}}} \right) \abs*{\overline{\mathbf{X}}}^2
    \\
    ={}& \left(
           \frac{4\delta-\delta^2}{4}K^2
           - \mu'
           - \frac{\delta}{2}K\abs*{\overline{\mathbf{K}}}
           \right)X^2
           +  \left( \mu' - \frac{\delta}{2}K\abs{\overline{\mathbf{K}}} \right) \abs*{\overline{\mathbf{X}}}^2.
  \end{align*}
  So we see that for
  $\mu' = \frac{\delta}{2}K\abs*{\overline{\mathbf{K}}} + \frac{\delta^2}{4}K^2$,
  we have that
  \begin{align*}
    \abs*{\Metric(\mathbf{K}, \mathbf{X})}^2
    + \mu' \Metric(\mathbf{X}, \mathbf{X})
    \ge \left(
    \delta\left( K^2-K\abs*{\overline{\mathbf{K}}} \right)
    - \frac{\delta^2}{2}K^2
    \right)X^2
    + \frac{\delta^2}{4}K^2\abs*{\overline{\mathbf{X}}}^2.
  \end{align*}
  Then, since $\mathbf{K}$ is timelike, we have that
  $K^2>\abs*{\overline{\mathbf{K}}}^2$ and thus for $\delta$
  sufficiently small,
  \begin{equation*}
    \abs*{\Metric(\mathbf{K}, \mathbf{X})}^2
    + \mu' \Metric(\mathbf{X}, \mathbf{X})
    \gtrsim \abs*{\mathbf{X}}^2,
  \end{equation*}
  as desired.
\end{proof}
We have the following simple corollary of \Cref{lemma:pseudo-convexity:outside-ergoregion}.
\begin{corollary}
  \label{coro:pseudo-convexity:outside-ergoregion:multiple-vectors}
  Let $p\in \Manifold$ be a point at which $\{\mathbf{V}_i\}_{i=1}^d$
  span a time-like vectorfield. Then, there exists some
  $\mu \lesssim A_0^{-1}$ such that
  \begin{equation}
    \label{eq:pseudo-convexity:outside-ergoregion:multiple-vectors}
    \begin{split}
      \mathbf{X}^{\alpha}\mathbf{X}^{\beta}\mu\Metric_{\alpha\beta} - \abs*{\mathbf{X}^{\alpha}\mathbf{X}^{\beta}\CovariantDeriv_{\alpha}\CovariantDeriv_{\beta}y}
      + C\left( \abs*{\mathbf{X}(y)}^2+ \sum_i\abs*{\Metric(\mathbf{V}_i, \mathbf{X})}^2 \right)&\ge C^{-1}\abs*{\mathbf{X}}^2.
    \end{split}    
  \end{equation}
\end{corollary}
\begin{proof}
  Let $\mathbf{K} = \sum_i c_i\mathbf{V}_i$ be a timelike vectorfield
  spanned by $\{\mathbf{V}_i\}_{i=1}^d$. Then \Cref{eq:pseudo-convexity:outside-ergoregion:multiple-vectors} follows directly from \Cref{eq:pseudo-convexity:outside-ergoregion}.
\end{proof}

We now prove that $y$ induces a $\KillT$-pseudoconvex foliation.
\begin{prop}
  \label{prop:y-pseudoconvexity}
  Let $\mathbf{N}$ be an open set $\mathbf{N}\subset \mathbf{E}$ such
  that $S_0\in \closure \mathbf{N}$ and such that there is some
  $0<\varepsilon_{\mathcal{S}}\ll 1$ such that the smallness condition
  in \Cref{eq:smallness-assumption} holds.  Let
  $p\in \Sigma\bigcap \mathbf{N}$ such that
  $y(p)\in (y_0, \underline{y}_0)$. Then for
  $\varepsilon_{\mathcal{S}}$ sufficiently small, there exist
  constants $y_+ > y_{*}$ and $y_- < y_{*}$ such that if $X$ is any
  real vector such that
  \begin{equation*}
    \abs*{X\cdot\KillT(p)}
    + \abs*{\CovariantDeriv_Xy(p)}
    \le c \abs*{X}
  \end{equation*}
  for some $c$ sufficiently small, there exists some $\mu\lesssim A_0^{-1}$ such that if $y(p) < y_+$,
  \begin{equation}
    \label{eq:y-pseudoconvexity:MST}
    X^{\alpha}X^{\beta}(\mu \Metric_{\alpha\beta}(p) - \CovariantDeriv_{\alpha}\CovariantDeriv_{\beta}y(p)) > c \abs*{X}^2,
  \end{equation}
  and if $y(p) > y_-$,
  \begin{equation}
    \label{eq:y-pseudoconvexity:HawkingVF}
    X^{\alpha}X^{\beta}(\mu \Metric_{\alpha\beta}(p) + \CovariantDeriv_{\alpha}\CovariantDeriv_{\beta}y(p)) > c \abs*{X}^2.
  \end{equation}  
\end{prop}

\begin{proof} 
  We start from recalling from \Cref{eq:D-P} that
  \begin{equation*}
    \begin{split}
      \CovariantDeriv_{\alpha}\mathbf{P}_{\beta}
      ={}& J \mathbf{P}_{\alpha}\mathbf{P}_{\beta}
           + \frac{1}{4R}\KillT^{\sigma}\mathcal{X}^{\mu\nu}\mathcal{S}_{\mu\nu\sigma\alpha}\mathbf{P}_{\beta}
           + \left(\frac{R}{2} - NJ\right)\Metric_{\alpha\beta}
           - \frac{1}{R}\mathcal{T}_{\alpha\beta}\\
         & + \frac{1}{R}\KillT^{\rho}\KillT^{\sigma}\mathcal{S}_{\rho\alpha\sigma\beta}
           - J \KillT_{\alpha}\KillT_{\beta}.
    \end{split}
  \end{equation*}
  Then, we use \Cref{lemma:P-one-form-P-J-inverse-relation} to find that
  \begin{align*}
    J \mathbf{P}_{\alpha}\mathbf{P}_{\beta}
    ={}&- \frac{y-\ImagUnit z}{y^2+z^2}\left(
         \CovariantDeriv_{\alpha}y\CovariantDeriv_{\beta}y
         - \CovariantDeriv_{\alpha}z\CovariantDeriv_{\beta}z
         + 2\ImagUnit \CovariantDeriv_{(\alpha}y\CovariantDeriv_{\beta)}z
         \right)
         + O(\varepsilon_{\mathcal{S}})
    \\
    ={}& O(\varepsilon_{\mathcal{S}})
         -\frac{y}{y^2+z^2}\left(\CovariantDeriv_{\alpha}y\CovariantDeriv_{\beta}y
         - \CovariantDeriv_{\alpha}z\CovariantDeriv_{\beta}z\right)
         -  \frac{2z}{y^2+z^2}\CovariantDeriv_{(\alpha}y\CovariantDeriv_{\beta)}z\\
       & + \ImagUnit\left[  \frac{z}{y^2+z^2}\left(\CovariantDeriv_{\alpha}y\CovariantDeriv_{\beta}y
         - \CovariantDeriv_{\alpha}z\CovariantDeriv_{\beta}z\right)
         + \frac{2y}{y^2+z^2}\CovariantDeriv_{(\alpha}y\CovariantDeriv_{\beta)}z \right] .
  \end{align*}
  Similarly, we can also compute using the definition of $\bm{\tau}$
  in \Cref{def:rescaled-tensor-qtys:def} that
  \begin{align*}
    -\Re\left( \frac{1}{R}\mathcal{T}_{\alpha\beta} \right)
    ={}& -\Re \frac{\overline{R}}{2}\bm{\tau}_{\alpha\beta}
         .
  \end{align*}
  Moreover, observe that
  \begin{equation*}
    \bm{\tau}_{\alpha\beta} + \Metric_{\alpha\beta}
    = 2\slashed{\Metric}_{\alpha\beta}.
  \end{equation*}
  Thus, we have that
  \begin{equation}
    \label{eq:y-pseudoconvexity:DP-main}
    \begin{split}
      \Re\CovariantDeriv_{\alpha}\mathbf{P}_{\beta}
      ={}& \frac{y}{y^2+z^2}\CovariantDeriv_{\alpha}z\CovariantDeriv_{\beta}z
           - \Re R\slashed{\Metric}_{\alpha\beta}
           + \Re\left(R + \frac{N}{P}\right)\Metric_{\alpha\beta}\\
         & - \frac{y}{y^2+z^2}\CovariantDeriv_{\alpha}y\CovariantDeriv_{\beta}y
           -  \frac{2z}{y^2+z^2}\CovariantDeriv_{(\alpha}y\CovariantDeriv_{\beta)}z
           + y \KillT_{\alpha}\KillT_{\beta}
           + O(\varepsilon_{\mathcal{S}}). 
    \end{split}
  \end{equation}
  Then, we realize from \Cref{eq:R-sigma-in-terms-of-P:S-small} that
  \begin{equation*}
    2\Re R = \partial_y\frac{\Delta}{y^2+z^2} + \frac{2y}{y^2+z^2}\CovariantDeriv_{\alpha}z\CovariantDeriv^{\alpha}z + O(\varepsilon_{\mathcal{S}}). 
  \end{equation*}
  As a result, we have that 
  \begin{equation*}
    \frac{y}{y^2+z^2}\CovariantDeriv_{\alpha}z\CovariantDeriv_{\beta}z
    - \Re R\slashed{\Metric}_{\alpha\beta}
    ={} \frac{y}{y^2+z^2}\left( \CovariantDeriv_{\alpha}z\CovariantDeriv_{\beta}z    
      -\CovariantDeriv_{\gamma}z\CovariantDeriv^{\gamma}z\slashed{\Metric}_{\alpha\beta}\right)
    - \frac{1}{2}\partial_y\frac{\Delta}{y^2+z^2}\slashed{\Metric}_{\alpha\beta}
    + O(\varepsilon_{\mathcal{S}})
    .
  \end{equation*}
  Let us now fix a choice of $(e_1,e_2)$ so that
  $\horProj{\KillT}\cdot e_1=0$ to be slightly more precise in the
  following estimates. Observe that from
  \Cref{lemma:y-z-derivatives:S-small}, this implies that
  $e_2(z)=O(\varepsilon_{\mathcal{S}})$.  We observe that
  \begin{align*}
    &\frac{y}{y^2+z^2}\left( \CovariantDeriv_{\alpha}z\CovariantDeriv_{\beta}z    
    -\CovariantDeriv_{\gamma}z\CovariantDeriv^{\gamma}z\slashed{\Metric}_{\alpha\beta}\right)X^{\alpha}X^{\beta}\\
    ={}& \frac{y}{y^2+z^2}\left( \CovariantDeriv_1z\CovariantDeriv_1z    
         -\CovariantDeriv_{\gamma}z\CovariantDeriv^{\gamma}z\right)X^1X^1
         -\frac{y}{y^2+z^2}\CovariantDeriv_{\gamma}z\CovariantDeriv^{\gamma}z X^2X^2
         + O(\varepsilon_{\mathcal{S}})\abs*{X}^2.
  \end{align*}
  Moreover, from \Cref{eq:y-z-derivatives:basic:S-small} we have that
  \begin{equation*}
    \frac{y}{y^2+z^2}\left( \CovariantDeriv_1z\CovariantDeriv_1z    
      -\CovariantDeriv_{\gamma}z\CovariantDeriv^{\gamma}z\right)
    = O(\varepsilon_{\mathcal{S}}).
  \end{equation*}
  As a result, we also have that
  \begin{align*}
    &\frac{y}{y^2+z^2}\left(
      \CovariantDeriv_{\gamma}z\CovariantDeriv^{\gamma}z\slashed{\Metric}_{\alpha\beta}
      -\CovariantDeriv_{\alpha}z\CovariantDeriv_{\beta}z          
      \right)X^{\alpha}X^{\beta}
      + \frac{1}{2}\partial_y\frac{\Delta}{y^2+z^2}\slashed{\Metric}_{\alpha\beta}X^{\alpha}X^{\beta}
    \\
    \ge{}& \frac{y}{y^2+z^2}\CovariantDeriv_{\gamma}z\CovariantDeriv^{\gamma}z X^2X^2
           + \frac{1}{2}\partial_y\frac{\Delta}{y^2+z^2}\left( X^1X^1 + X^2X^2 \right)
           - O(\varepsilon_{\mathcal{S}})\abs*{X}^2
           .
  \end{align*}
  Thus using \Cref{lemma:metric-decomp-into-THat-RHat}, 
  \begin{align}
    &\frac{y}{y^2+z^2}\left(
      \CovariantDeriv_{\gamma}z\CovariantDeriv^{\gamma}z\slashed{\Metric}_{\alpha\beta}
      -\CovariantDeriv_{\alpha}z\CovariantDeriv_{\beta}z          
      \right)X^{\alpha}X^{\beta}
      + \frac{1}{2}\partial_y\frac{\Delta}{y^2+z^2}\slashed{\Metric}_{\alpha\beta}X^{\alpha}X^{\beta}
      - \frac{1}{2}\partial_y\frac{\Delta}{y^2+z^2}\Metric(X, X)
    \notag\\
    \ge{}& \frac{y}{y^2+z^2}\CovariantDeriv_{\gamma}z\CovariantDeriv^{\gamma}z X^2X^2
           - \frac{1}{2}\partial_y\frac{\Delta}{y^2+z^2}\frac{y^2+z^2}{\Delta}\left( \CovariantDeriv_Xy \right)^2\notag\\
          & + \frac{1}{2}\partial_y\frac{\Delta}{y^2+z^2}\frac{1}{(y^2+z^2)\Delta} \left( \widehat{T}\cdot X \right)^2
           + O(\varepsilon_{\mathcal{S}})\abs*{X}^2
           . \label{eq:y-pseudoconvexity:aux1}
  \end{align}
  Then observe from \Cref{eq:THat:def} that 
  \begin{equation}
    \label{eq:That-dot-X-expansion}
    \left( \frac{1}{\abs*{P}^2}\widehat{T}\cdot X \right)^2
    ={} \left( \KillT\cdot X \right)^2
         -2\KillT\cdot X \horProj{\KillT}\cdot X
         + \left( \horProj{\KillT}\cdot X \right)^2.
  \end{equation}
  Moreover, observe from \Cref{lemma:y-z-derivatives:S-small} that
  \begin{equation*}
    \CovariantDeriv_{\alpha} z \CovariantDeriv^{\alpha}z
    = \horProj{\KillT}\cdot\horProj{\KillT} + O(\varepsilon_{\mathcal{S}}),\qquad
    \horProj{\KillT}(z) = O(\varepsilon_{\mathcal{S}}).
  \end{equation*}
  As a result, we have that
  \begin{equation}
    \label{eq:y-pseudoconvexity:aux2}
   \frac{1}{(y^2+z^2)^2} \left( \widehat{T}\cdot X \right)^2
    ={} \CovariantDeriv_{\alpha}z\CovariantDeriv^{\alpha}z X^2X^2
         + O(\varepsilon_{\mathcal{S}})\abs*{X}^2
         + O(\abs*{X}\abs*{\KillT\cdot X})
         + O(\abs*{\KillT\cdot X}^2)
         .
  \end{equation}
  Combining \Cref{eq:y-pseudoconvexity:aux1}
  and \Cref{eq:y-pseudoconvexity:aux2}, we now have that for some $C$
  sufficiently large,
  \begin{align*}
    &\frac{y}{y^2+z^2}\left(
      \CovariantDeriv_{\gamma}z\CovariantDeriv^{\gamma}z\slashed{\Metric}_{\alpha\beta}
      -\CovariantDeriv_{\alpha}z\CovariantDeriv_{\beta}z          
      \right)X^{\alpha}X^{\beta}
      + \frac{1}{2}\partial_y\frac{\Delta}{y^2+z^2}\slashed{\Metric}_{\alpha\beta}X^{\alpha}X^{\beta}
      - \frac{1}{2}\partial_y\frac{\Delta}{y^2+z^2}\Metric(X, X)\\
    \ge{}& \left(
           \frac{y}{y^2+z^2} 
           + \frac{y^2+z^2}{2\Delta}\partial_y\frac{\Delta}{y^2+z^2}
           \right)\CovariantDeriv_{\gamma}z\CovariantDeriv^{\gamma}z X^2X^2
           - \frac{1}{2}\partial_y\frac{\Delta}{y^2+z^2}\frac{y^2+z^2}{\Delta}\left( \CovariantDeriv_Xy \right)^2
    \\
    & - O(\abs*{\KillT\cdot X}^2)
      - O(\abs*{X}\abs*{\KillT\cdot X})
      - O(\varepsilon_{\mathcal{S}})\abs*{X}^2
    \\
    \ge{}& \left(
           \frac{y}{y^2+z^2}
           + \frac{1}{2}\frac{\partial_y\Delta}{\Delta}
           - \frac{y}{y^2+z^2}
           \right)\CovariantDeriv_{\gamma}z\CovariantDeriv^{\gamma}z X^2X^2
           - \frac{1}{2}\partial_y\frac{\Delta}{y^2+z^2}\frac{y^2+z^2}{\Delta}\left( \CovariantDeriv_Xy \right)^2
    \\
    & - O(\abs*{\KillT\cdot X}^2)
      - O(\abs*{X}\abs*{\KillT\cdot X})
      - O(\varepsilon_{\mathcal{S}})\abs*{X}^2
    \\
    \ge{}& \frac{1}{2}\frac{\partial_y\Delta}{\Delta}\CovariantDeriv_{\gamma}z\CovariantDeriv^{\gamma}z X^2X^2
           - \frac{1}{2}\partial_y\frac{\Delta}{y^2+z^2}\frac{y^2+z^2}{\Delta}\left( \CovariantDeriv_Xy \right)^2
    \\
    & - O(\abs*{\KillT\cdot X}^2)
      - O(\abs*{X}\abs*{\KillT\cdot X})    
      - O(\varepsilon_{\mathcal{S}})\abs*{X}^2
      .
  \end{align*}
  In particular, combined with \Cref{eq:y-pseudoconvexity:DP-main} we
  see then that for any $X\in T \mathcal{M}$, there exists some
  constant $C>0$ such that
  \begin{equation*}
    \begin{split}
      &-X^{\alpha}X^{\beta}\Re\CovariantDeriv_{\alpha}\mathbf{P}_\beta
      + \Re\left(R + \frac{N}{P}\right)X\cdot X
      - \frac{1}{2}\partial_y\frac{\Delta}{y^2+z^2}X\cdot X\\
      \ge{}& \frac{\partial_y\Delta}{2\Delta}\CovariantDeriv_{\gamma}z\CovariantDeriv^{\gamma}X^2X^2
             - C\left(
             \abs*{X\cdot \KillT}^2
             + \abs*{\CovariantDeriv_X y}^2
             \right)
             - O(\varepsilon_{\mathcal{S}})\abs*{X}^2
             .
    \end{split}
  \end{equation*}
  It remains to show that there exists some small constant $c>0$, and
  some real number $\delta_2$ such that
  \begin{equation*}
    \frac{1}{2}\frac{\partial_y\Delta}{\Delta}\CovariantDeriv_{\gamma}z\CovariantDeriv^{\gamma}X^2X^2
   + \delta_2X\cdot X
   > c\abs*{X}^2
   - C\abs*{X\cdot \KillT}^2.
  \end{equation*}
  From \Cref{lemma:pseudo-convexity:outside-ergoregion} we see that we
  would be done if $\KillT$ is timelike at $p$. Thus we may assume
  without loss of generality that $\KillT$ is null or spacelike at
  $p$. In particular, from \Cref{eq:CovP-N:relation}, we see that this
  implies that
  \begin{equation*}
    \CovariantDeriv_{\gamma}z\CovariantDeriv^{\gamma}z \ge \CovariantDeriv_{\gamma}y\CovariantDeriv^{\gamma}y + O(\varepsilon_{\mathcal{S}}). 
  \end{equation*}
  Thus, since $\Delta(y(p))>0$ by assumption, we see that for
  $\varepsilon_{\mathcal{S}}$ sufficiently small, we have in fact that
  there exists some small $\delta_1$ such that
  \begin{equation*}
    \CovariantDeriv_{\gamma}z\CovariantDeriv^{\gamma}z > \delta_1.
  \end{equation*}
  Next observe from \Cref{lemma:metric-decomp-into-THat-RHat} and
  \Cref{eq:That-dot-X-expansion} that there exists some small
  $\delta_2$ such that for $\varepsilon_{\mathcal{S}}$ sufficiently
  small,
  \begin{equation*}
    \CovariantDeriv_{\gamma}z\CovariantDeriv^{\gamma}z X^2X^2 + \delta_2\Metric(X, X)
    \gtrsim \abs*{X}^2 - C\abs*{X\cdot \KillT}^2.
  \end{equation*}
  We thus see using \Cref{lemma:P-one-form-P-J-inverse-relation} and
  potentially increasing $C$ that \Cref{eq:y-pseudoconvexity:MST}
  holds in the region $\mathcal{M}^+\vcentcolon= \{y < y_*\}$ with
  $\mu=\Re\left(R + \frac{N}{P}\right)-
  \frac{1}{2}\partial_y\frac{\Delta}{y^2+z^2}\Metric(X, X) - \delta_2$
  and that \Cref{eq:y-pseudoconvexity:HawkingVF} holds in the region
  $\mathcal{M}^-\vcentcolon= \{y > y_*\}$ with
  $\mu=\Re\left(R + \frac{N}{P}\right)-
  \frac{1}{2}\partial_y\frac{\Delta}{y^2+z^2}\Metric(X, X) -
  \delta_2$. To see that \Cref{eq:y-pseudoconvexity:MST} and
  \Cref{eq:y-pseudoconvexity:HawkingVF} extend to some regions
  $\{y<y_+\}$ and $\{y> y_-\}$ respectively, we observe that the
  subextremal assumption means that $\KillT$ must be timelike in a
  neighborhood of $y_*$. Applying
  \Cref{lemma:pseudo-convexity:outside-ergoregion} then yields the
  conclusion, concluding the proof of \Cref{prop:y-pseudoconvexity}.
\end{proof}

\section{Extension of \texorpdfstring{$\mathcal{S}$}{S} using \texorpdfstring{\ref{ass:E1}}{E1}}
\label{sec:MST:first}

In this section, we describe the extension of the vanishing of the
Mars-Simon tensor $\mathcal{S}$ from $S_0$ using \ref{ass:E1}. The
main proposition is as follows.
\begin{prop}
  \label{prop:main-prop}
  Under the stationary, smooth bifurcate sphere, subextremality, and
  \ref{ass:E1} assumptions of \Cref{sec:assumptions}, there exists some
  $y_{\mathcal{S}}>y_{*}$ such that for any
  $ y_{S_0}\le R< y_{\mathcal{S}}$, we have that $\mathcal{S} = 0$ in
  $\mathcal{U}_R$.
\end{prop}
We remark that one has the following corollary which would be proven
in an identical manner,
\begin{corollary}
  \label{coro:main-prop:c1}
    Under the stationary, smooth bifurcate sphere, subextremality, and
  \ref{ass:C1} assumptions of \Cref{sec:assumptions} there exists some
  $\underline{y_{\mathcal{S}}}<y_{*}$ such that for any
  $y_{\underline{S}_0}\ge R> \underline{y_{\mathcal{S}}}$, we have that $\mathcal{S} = 0$ in
  $\underline{\mathcal{U}}_R$.
\end{corollary}

We will prove \Cref{prop:main-prop} in a few key steps. 
\begin{enumerate}
\item First, we show that $\mathcal{S}$ vanishes on the event
  horizon. The fact that $A(\mathcal{S})$ and $B(\mathcal{S})$ vanish
  on $\EventHorizonFuture$, and that $\underline{A}(\mathcal{S})$ and
  $\underline{B}(\mathcal{S})$ vanish on $\EventHorizonPast$ is a
  simple consequence of fact that we assumed that the event horizon is
  non-expanding. On the other hand, the vanishing of $P(\mathcal{S})$
  anywhere on the event horizon is non-trivial. This is where we make
  use of the compatibility conditions in
  \Cref{eq:compatibility-conditions} to show that
  $\evalAt*{P(\mathcal{S})}_{S_0}=0$. Using the null Bianchi equations
  then allows us to show that all null components of $\mathcal{S}$ and
  thus $\mathcal{S}$ itself vanishes along the entirety of the event
  horizon $\EventHorizon$. 
\item Next, we first show that $\mathcal{S}=0$ in a neighborhood of the
  event horizon. This is a classic unique continuation argument which
  makes use of the pseudoconvexity of the event horizon in a
  neighborhood of $S_0$ which is a simple consequence of its
  null-bifurcate nature. 
\item Finally, we extend the vanishing of $\mathcal{S}$ up until some
  $y_{\mathcal{S}}$ such that $y_{\mathcal{S}}>y_{*}$, where we recall
  that $y_{*}$ is the unique maxima of $\Delta(y)$ in
  $\mathbf{E}$. This uses a bootstrap argument where we construct a
  $\KillT$-pseudoconvex foliation of $\mathbf{E}$ when
  $\mathcal{S}=0$. The foliation is constructed by considering the
  function $y$, which acts like the radial function $r$ on \KdS. The
  extension unfortunately cannot be extended to the entirety of
  $\mathbf{E}$ since even on \KdS, $r$ is not $\KillT$-pseudoconvex in
  the entirety of the stationary region. Crucially, the $\KillT$-null
  convexity breaks down in the ergoregion adjacent to the cosmological
  horizon. We also remark that the fact that the extension can be done
  to some $y_{\mathcal{S}}>y_{*}$ is a consequence of the assumed
  subextremality of the spacetime.
\end{enumerate}

\subsection{Vanishing of \texorpdfstring{$\mathcal{S}$}{S} on the horizon}
\label{sec:MST:horizon-vanishing}

In this section, we prove that the Mars-Simon tensor $\mathcal{S}$
vanishes on the event horizon
$\EventHorizonFuture\bigcup \EventHorizonPast$. This section follows
closely the approach taken in
\cite{ionescuUniquenessSmoothStationary2009}.  The main idea is that
$A(\mathcal{S})$ and $B(\mathcal{S})$ (respectively
$\underline{B}(\mathcal{S})$ and $\underline{A}(\mathcal{S})$) vanish
on $\EventHorizonFuture$ (respectively, $\EventHorizonPast$). To show
that they vanish on the opposite side of the horizon, it suffices to
use the null Bianchi equations induced by \Cref{prop:divergence-of-MST},
and transport the desired null component of $\mathcal{S}$ out from
$S_0$ along $\EventHorizonPast$ (respectively
$\EventHorizonFuture$). $P(\mathcal{S})$ follows a similar approach,
but unlike $A(\mathcal{S})$,
$B(\mathcal{S})$. $\underline{B}(\mathcal{S})$, and
$\underline{A}(\mathcal{S})$, we need to use the compatibility
condition in \ref{ass:E1} in order to show that $P(\mathcal{S})$
vanishes on $S_0$ before applying the null Bianchi equations. 

The main proposition of this section is the following. 
\begin{prop}
  \label{prop:horizon:main}
  The Mars-Simon tensor $\mathcal{S}$ vanishes on the event horizon
  $\EventHorizonFuture\bigcup \EventHorizonPast$.
\end{prop}
In the rest of this section, we consider an adapted null frame
$(e_1,e_2,e_3,e_4)$ in $\mathbf{O}$ where $e_3$ is tangent to the null
generators of $\EventHorizonPast$ and $e_4$ is tangent to the null
generators of $\EventHorizonFuture$.

We first prove an auxiliary lemma that will be needed in what follows.
\begin{lemma}
  Using the compatibility assumptions on $S_0$ given in
  \Cref{eq:compatibility-conditions}, we have that on $S_0$,
  \begin{equation}
    \label{eq:R-sigma-in-terms-of-P:S0}
    R = \frac{b_{S_0}}{2P^2} - \frac{\Lambda}{3}P, \qquad
    \sigma = c_{S_0} - \frac{b_{S_0}}{P} - \frac{\Lambda}{3}P^2.
  \end{equation}
  In other words, $R, \sigma$ have the same expressions in terms of
  $P$ as they have when $\mathcal{S}=0$ (see
  \Cref{eq:R-sigma-in-terms-of-P:S-small}).
\end{lemma}
\begin{proof}
  We first observe that from the definition of $(b,c,k)$ from
  \Cref{eq:b-c-k:def} and the
  assumption \Cref{eq:compatibility-conditions} that $(b,c,k)$ are
  constant over $S_0$, we have that on $S_0$,
  \begin{equation*}
    b_{S_0} = -\ImagUnit \frac{36 Q(\mathcal{F}^2)^{\frac{5}{2}}}{(Q \mathcal{F}^2 - 4\Lambda)^3}.
  \end{equation*}
  Then recalling the expression for $P$ in
  \Cref{eq:P-in-terms-of-Q-F}, we have that on $S_0$, 
  \begin{equation*}
    b_{S_0} = -\ImagUnit P^2 \frac{( \mathcal{F}^2 )^{\frac{3}{2}}}{Q \mathcal{F}^2 - 4\Lambda}
    = 2P^2\left(R + \frac{\Lambda}{3}P\right),
  \end{equation*}
  where we used the definition of $Q$ in \Cref{eq:MST:Q:def}, and the
  definition of $R$ in \Cref{eq:R:def}. Rearranging the equation then
  allows us to write that on $S_0$
  \begin{equation*}
    R = \frac{b_{S_0}}{2P^2} - \frac{\Lambda}{3}P.
  \end{equation*}

  To prove the expression for $\sigma$, we see that it suffices to
  show that $\sigma_0 = -\frac{b_{S_0}}{P} - \frac{\Lambda}{3}P^2$. To this
  end, we recall from \Cref{eq:P-sigma0-relation} that
  \begin{equation*}
    \sigma_0 = -\frac{1}{6}P^2\left(Q \mathcal{F}^2 + 2\Lambda\right). 
  \end{equation*}
  Then, using the definition of $Q$ in \Cref{eq:MST:Q:def}, and the
  definition of $R$ in \Cref{eq:R:def}, we see that
  \begin{equation*}
    \sigma_0 = P^2\left(2J R - \Lambda\right)
    = - 2PR - \Lambda P^2.
  \end{equation*}
  Then, using the expression for $R$ in
  \Cref{eq:R-sigma-in-terms-of-P:S0} which we have already proven, we have that on $S_0$
  \begin{equation*}
    \sigma_0 = - \frac{b_{S_0}}{P} - \frac{\Lambda}{3}P^2,
  \end{equation*}
  as desired.
\end{proof}

\begin{lemma}
  \label{lemma:yS0-constant}
  $y$ is constant on $S_0$. 
\end{lemma}

\begin{proof}
  Recall from \Cref{lemma:horizon:Fcal-null-components} that
  $B(\mathcal{F})$ vanishes along $\EventHorizon^+$, and
  $\underline{B}(\mathcal{F})$ vanishes along both
  $\EventHorizon^-$. As a result, we have that for any vectorfield
  $\widetilde{L}$ tangent to the null generators of $\EventHorizon^+$,
  there exists a constant $C\in \Real$ such that
  \begin{equation*}
    \mathcal{F}\cdot\widetilde{L} = C \widetilde{L}.
  \end{equation*}
  Thus, $\widetilde{L}$ is parallel to either $\bm{\ell}_-$ or
  $\bm{\ell}_+$ on $\EventHorizon^+$. Similarly, any vectorfield
  tangent to the null generators of $\EventHorizon^-$ must also be
  parallel to either $\bm{\ell}_-$ or $\bm{\ell}_+$. Thus, $(e_1,e_2)$
  must be tangent to the bifurcate sphere $S_0$.

  Next, observe that from the expressions for $R$ and $\sigma$ in
  \Cref{eq:R-sigma-in-terms-of-P:S0}, we have that on $S_0$,
  \begin{equation*}
    \frac{1}{2R}\nabla \sigma = \nabla P,  
  \end{equation*}
  which in turn implies that on $S_0$
  \begin{equation*}
    \mathbf{P}_a = \nabla_aP.
  \end{equation*}
  But then from
  \Cref{lemma:basic-computations:eta-T-sigma-null-decomp}, we have
  that, fixing the null frame $(e_4,e_3) = (L_+,-L_-)$, we have that
  \begin{equation*}
    \Re \mathbf{P} = \Metric(e_3,\KillT)e_4 - \Metric(e_4,\KillT)e_3,
  \end{equation*}
  so we have that 
  \begin{equation*}
    \evalAt*{\horProj{\Re \mathbf{P}}}_{S_0} = 0,
  \end{equation*}
  which implies that
  \begin{equation*}
    \nabla y =0. 
  \end{equation*}
  Then, since $(e_1,e_2)$ are tangent to  $S_0$,
  we have that $y$ is in fact constant on
  $S_0$. 
\end{proof}

\begin{corollary}
  \label{coro:regularity:S0}
  On $S_0$, $R-J\sigma_0\neq 0$. 
\end{corollary}
\begin{proof}
  Assume for the sake of contradiction that there is a point where
  $R-J\sigma_0= 0$. Recall from \Cref{lemma:yS0-constant} that
  $y=y_{S_0}$ is a constant on $S_0$. We also have from
  \Cref{lemma:b-c-k-almost-constant:S-small} and
  \Cref{lemma:k-almost-constant} that
  $(b,c,k)=(b_{S_0},c_{S_0},k_{S_{0}})$ are constants if
  $\mathcal{S}=0$. As a result, we have that the assumptions made in
  \Cref{eq:S0-y-Delta-assumptions} hold for any $p\in S_0$.  Then
  observe that on $S_0$,
  \begin{equation*}
    \Delta + \frac{b}{2}y =  k + \frac{y}{2}\partial_r\Delta + \frac{\Lambda}{3}y^4,\qquad
    cy^2 + \Delta = k + y \partial_y\Delta + \Lambda y^4. 
  \end{equation*}
  Since $\Delta=0$ on $S_0$, we have that
  \begin{equation}
    \label{eq:regularity:S0:contradiction}
    \frac{b}{2} > \frac{\Lambda}{3}y_{S_0}^3, \qquad
    cy_{S_0}^2 > k.
  \end{equation}
  Next, observe that using the formulas in
  \Cref{eq:R-sigma-in-terms-of-P:S0}, we can compute that
  \begin{equation*}
    R- J\sigma_0
    ={} -\frac{b}{2}P\left(\frac{1}{P^3} + \frac{4\Lambda}{3b}\right).
  \end{equation*}
  Since we assumed that $y>0$ in \Cref{eq:S0-y-Delta-assumptions}, we
  see that $R-J\sigma_0$ vanishes if and only if
  \begin{equation}
    \label{eq:regularity:S0:contradiction-assumption}
    P^3 = -\frac{3b}{4\Lambda}. 
  \end{equation}
  Since $y_{S_0}>0$, we then have that $z = \pm\sqrt{3}y_{S_0}$
  and $\abs*{P} = 2y_{S_0}$.  However, we observe that
  \Cref{eq:regularity:S0:contradiction-assumption} also implies that
  $\abs*{P}^3 = \frac{3b}{4\Lambda}$ which in turn implies that
  \begin{equation*}
    \frac{\Lambda}{3}y_{S_0}^3 = \frac{b}{32}.
  \end{equation*}
  On the other hand, using the formula for $\sigma$ in
  \Cref{eq:R-sigma-in-terms-of-P:S0}, we can compute that
  \begin{equation*}
    \Re\sigma_0 = -\frac{by}{y^2+z^2} - \frac{\Lambda}{3}(y^2-z^2).
  \end{equation*}
  As a result, we can calculate that if $z = \sqrt{3}y_{S_0}$,
  \begin{equation*}
    \Re\sigma_0 = -\frac{by_{S_0}}{4y_{S_0}^2} + \frac{2\Lambda}{3}y_{S_0}^2
    = -\frac{b}{4y_{S_0}} + \frac{b}{16y_{S_0}} = -\frac{3b}{16y_{S_0}},
  \end{equation*}
  where the second equality follows from
  \Cref{eq:regularity:S0:contradiction-assumption}.  Since we assumed
  that $\KillT$ is spacelike on $S_0$, and
  $\Re\sigma_0 + c = \Re\sigma = -\Metric(\KillT, \KillT)$, we thus
  have that
  \begin{equation*}
    cy_{S_0} < \frac{3}{16}b. 
  \end{equation*}
  Thus, using \Cref{eq:regularity:S0:contradiction}, we have that
  \begin{align*}
    \Delta(y_{S_0}) ={}& k - by_{S_0} + cy_{S_0}^2 - \frac{\Lambda}{3}y_{S_0}^4\\
    <{}& 2cy_{S_0}^2 - by_{S_0} - \frac{\Lambda}{3}y_{S_0}^4\\
    <{}& -\frac{1}{2}by_{S_0} - \frac{\Lambda}{3}y_{S_0}^4\\
    <{}&0,
  \end{align*}
  which is a contradiction.
\end{proof}

We are now ready to prove \Cref{prop:horizon:main}. 
\begin{proof}[Proof of \Cref{prop:horizon:main}]
  Let us fix $(e_4,e_3) = (L_+,-L_-)$.
  We prove \Cref{prop:horizon:main} in three steps.
  
  \textbf{Step 1.} We first show
  that
  \begin{equation}
    \label{eq:prop:horizon:main:step1}
    A(\mathcal{S}) = B(\mathcal{S}) = 0 ,\quad \text{on }\EventHorizonFuture,\qquad
    \underline{A}(\mathcal{S}) = \underline{B}(\mathcal{S}) = 0 ,\quad \text{on }\EventHorizonPast.
  \end{equation}
  First, observe that
  \begin{equation*}
    \mathcal{I}(e_4,e_a,e_4,e_b) = 0,\qquad
    \mathcal{I}(e_a,e_4,e_3,e_4) = 0,\qquad
    \mathcal{I}(e_4,e_3,e_4,e_3) = -\frac{3}{4}. 
  \end{equation*}
  Therefore, along $\EventHorizonFuture$, we have that
  \begin{align*}
    A(\mathcal{S})_{ab}
    ={}& \mathcal{S}(e_4,e_a,e_4,e_b)\\
    ={}& A(\mathcal{W})_{ab} - Q\left(\mathcal{F}(e_4,e_a)\mathcal{F}(e_4,e_b) - \frac{1}{3}\mathcal{F}^2\mathcal{I}(e_4,e_a,e_4,e_a)\right)\\
    ={}& 0,
  \end{align*}
  since from \Cref{lemma:horizon:Fcal-null-components} and
  \Cref{coro:non-expanding-null-hypersurface:horizon-qtys}, 
  we have that
  $A(\mathcal{W}), B(\mathcal{F})$ vanish along $\EventHorizonFuture$.

  Similarly, we have that
  \begin{align*}
    B(\mathcal{S})_a
    ={}& \mathcal{S}(e_a,e_4,e_3,e_4)\\
    ={}& B(\mathcal{W})_{ab}
         - Q\left(
         \mathcal{F}(e_a,e_4)\mathcal{F}(e_3,e_4)
         -\frac{1}{3}\mathcal{I}(e_a,e_4,e_3,e_4)
         \right)\\
    ={}0,
  \end{align*}
  since from \Cref{coro:non-expanding-null-hypersurface:horizon-qtys}
  and \Cref{lemma:horizon:Fcal-null-components} we have that
  $B(\mathcal{W})$, $B(\mathcal{F})$ vanish on
  $\EventHorizonFuture$. Similar computations show that
  $\underline{A}(\mathcal{S})$ and $\underline{B}(\mathcal{S})$ vanish
  on $\EventHorizonPast$. This proves
  \Cref{eq:prop:horizon:main:step1}.

  \textbf{Step 2.} We now show that $P(\mathcal{S})$ vanishes on $S_0$.
  First, observe that
  \begin{align*}
    \CovariantDeriv_{\rho}\mathcal{F}_{\gamma\delta}
    &= \KillT^{\nu}\mathcal{W}_{\nu\rho\gamma\delta}
      + \frac{4\Lambda}{3}\KillT^{\nu}\mathcal{I}_{\nu\rho\gamma\delta}
    \\
    &= \KillT^{\nu}\mathcal{S}_{\nu\rho\gamma\delta}
      + Q\KillT^{\nu}\mathcal{U}_{\nu\rho\gamma\delta}
      + \frac{4\Lambda}{3}\KillT^{\nu}\mathcal{I}_{\nu\rho\gamma\delta}
    \\
    &= \KillT^{\nu}\mathcal{S}_{\nu\rho\gamma\delta}
      + Q\KillT^{\nu}\left(\mathcal{F}_{\nu\rho}\mathcal{F}_{\gamma\delta} - \frac{1}{3}\mathcal{F}^2\mathcal{I}_{\nu\rho\gamma\delta}\right)
      + \frac{4\Lambda}{3}\KillT^{\nu}\mathcal{I}_{\nu\rho\gamma\delta}
    \\
    &= \KillT^{\nu}\mathcal{S}_{\nu\rho\gamma\delta} + Q\left(\frac{1}{2}\mathcal{F}_{\gamma\delta}\sigma_{\rho} - \frac{1}{3}\mathcal{F}^2\KillT^{\nu}\mathcal{I}_{\nu\rho\gamma\delta}\right)
      + \frac{4\Lambda}{3}\KillT^{\nu}\mathcal{I}_{\nu\rho\gamma\delta}
      .
  \end{align*}
  As a result, we have that
  \begin{align*}
    \CovariantDeriv_{\rho}\mathcal{F}^2
    &= 2\CovariantDeriv_{\rho}\mathcal{F}_{\gamma\delta}\mathcal{F}^{\gamma\delta}\\
    &= 2\left(
      \KillT^{\nu}\mathcal{S}_{\nu\rho\gamma\delta}\mathcal{F}^{\gamma\delta}
      + \frac{Q}{2}\mathcal{F}^2\bm{\sigma}_{\rho}
      - \frac{Q}{3}\mathcal{F}^2\KillT^{\nu}\mathcal{F}_{\nu\rho}
      + \frac{4\Lambda}{3}\KillT^{\nu}\mathcal{F}_{\nu\rho}
      \right)\\
    &= 2\KillT^{\nu}\mathcal{S}_{\nu\rho\gamma\delta}\mathcal{F}^{\gamma\delta}
      + \frac{2Q}{3}\mathcal{F}^2\bm{\sigma}_{\rho}
      + \frac{4\Lambda}{3}\bm{\sigma}_{\rho}
      .
  \end{align*}
  Then, using the definition of $Q$ in \Cref{eq:MST:Q:def}, we have that
  \begin{align*}
    2\KillT^{\nu}\mathcal{S}_{\nu\rho\gamma\delta}\mathcal{F}^{\gamma\delta}
    ={}& \CovariantDeriv_{\rho}\mathcal{F}^2
         - \frac{2Q}{3}\mathcal{F}^2\bm{\sigma}_{\rho}
         - \frac{4\Lambda}{3}\bm{\sigma}_{\rho}
    \\
    ={}&-4\CovariantDeriv_{\rho}R^2
         - \frac{2}{3}\left(-12JR + 4\Lambda + 2\Lambda\right)\bm{\sigma}_{\rho}
    \\
    ={}& -8R\CovariantDeriv_{\rho}R - 8R\left(-2JR + \Lambda\right)\CovariantDeriv_\rho P
    \\
    ={}& -8R\left(\CovariantDeriv_{\rho}R + \left( \Lambda-2JR \right)\CovariantDeriv_{\rho}P\right).    
  \end{align*}
  On the other hand, using the definition of $Q$ in
  \Cref{eq:MST:Q:def}, and the definition of $R$ in \Cref{eq:R:def},
  we have that
  \begin{align*}
    Q\mathcal{F}^2 =  -12JR + 4\Lambda,\qquad
    (\mathcal{F}^2)^{\frac{3}{2}} = 8\ImagUnit R^3.
  \end{align*}
  As a result, using the definition of $b$ in \Cref{eq:b-c-k:def}, we
  have that on $S_0$,
  \begin{align*}
    \ImagUnit b ={}& \frac{288\ImagUnit(-12JR + 4\Lambda)R^3}{(-12JR)^3}\\
    b={}& \frac{2R}{J^2} -\frac{2\Lambda}{3J^3}.
  \end{align*}
  Then we have that on $S_0$,
  \begin{align*}
    0 ={}& \CovariantDeriv_aR
           - bJ\CovariantDeriv_aJ
           - \frac{\Lambda}{3J^2}\CovariantDeriv_aJ\\
    ={}& \CovariantDeriv_aR
         - \left(\frac{2R}{J} - \frac{\Lambda}{3J^2}\right)\CovariantDeriv_aJ
  \end{align*}
  As a result, we have that
  \begin{align*}
    -\frac{1}{4R}\KillT^{\nu}\mathcal{S}_{\nu\rho\gamma\delta}\mathcal{F}^{\gamma\delta}
    ={}& \CovariantDeriv_{\rho}R
         + \left( \Lambda-2JR \right)\left(
         \frac{1}{J^2}\CovariantDeriv_{\rho}J
         - \frac{1}{4JR(R-J\sigma_0)}\KillT^{\nu}\mathcal{S}_{\nu\rho\gamma\delta}\mathcal{F}^{\gamma\delta}
         \right)\\
    ={}&-\frac{\Lambda-2JR}{4RJ(R-J\sigma_0)}\KillT^{\nu}\mathcal{S}_{\nu\rho\gamma\delta}\mathcal{F}^{\gamma\delta}.
  \end{align*}
  Thus, on $S_0$,
  \begin{equation*}
    \left(1 - \frac{\Lambda-2JR}{J(R-J\sigma_0)}\right)\KillT^{\nu}\mathcal{S}_{\nu\rho\gamma\delta}\mathcal{F}^{\gamma\delta}=0. 
  \end{equation*}
  But then we observe that 
  \begin{align*}  
    1 - \frac{\Lambda - 2JR}{J(R-J\sigma_0)}
    ={}& \frac{-J^2\sigma_0 + 2JR -\Lambda + JR}{J(R-J\sigma_0)}\\
    = \frac{R}{R-J\sigma_0},
  \end{align*}
  where we used the fact that $J$ by definition solves
  $J^2\sigma_0 - 2JR + \Lambda = 0$.

  Then, since by \Cref{coro:regularity:S0}, $R-J\sigma_0\neq 0$ on
  $S_0$, and we have assumed that $R\neq 0$ on $S_0$, we have in fact that
  \begin{equation*}
    \KillT^{\nu}\mathcal{S}_{\nu\rho\gamma\delta}\mathcal{F}^{\gamma\delta} = 0,\qquad \text{on }S_0.
  \end{equation*}    
  We recall from our assumptions on the smooth bifurcate spheres 
  that $\KillT$ is tangent to $S_0$ and can only vanish at a discrete
  set of points. Let us now consider a point $p$ where
  $\sqrt{\Metric(\KillT, \KillT)}$ does not vanish. We introduce an
  orthogonal frame $(e_1,e_2)$ at $p$ such that
  \begin{equation*}
    \KillT = \sqrt{\Metric(\KillT,\KillT)}e_1. 
  \end{equation*}
  Then we have that
  \begin{align*}
    0={}& \KillT^{\nu}\mathcal{S}_{\nu2\gamma\delta}\mathcal{F}^{\gamma\delta}\\
    ={}& 2\KillT^{\nu}\mathcal{S}_{\nu234}\mathcal{F}^{34}
         + \KillT^{\nu}\mathcal{S}_{\nu2cd}\mathcal{F}^{cd}\\
    ={}& -8\KillT^{\nu}\mathcal{S}_{\nu234}\mathcal{F}^{34}
         - 4\ImagUnit\KillT^{\nu}\mathcal{S}_{\nu2cd}\in^{cd}P(\mathcal{F})\\
    ={}& -\sqrt{\Metric(\KillT, \KillT)}P(\mathcal{F})\left(
         8\mathcal{S}_{1234}
         + 4\ImagUnit \mathcal{S}_{12cd}\in^{cd}
         \right)\\
    ={}& 16\ImagUnit\sqrt{\Metric(\KillT,\KillT)}P(\mathcal{F})P(\mathcal{S}),
  \end{align*}
  where the last equality follows from
  \begin{equation*}
    \mathcal{S}_{1234} = \ImagUnit P(\mathcal{S}),\qquad
    \mathcal{S}_{12cd} = -\in_{cd}P(\mathcal{S}).
  \end{equation*}
  Recall that $p\in S_0$ was chosen so that $\KillT$ does not vanish,
  and from \Cref{lemma:horizon:Fcal-null-components}, we have that
  $P(\mathcal{F})$ does not vanish. Thus, $P(\mathcal{S})=0$ at
  $p$. Since $p$ was an arbitrary point on $S_0$ where $\KillT$ does
  not vanish, and moreover, since the set of such points is dense in
  $S_0$, we conclude that $P(\mathcal{S})$ vanishes identically on
  $S_0$.

  \textbf{Step 3.} We now show that $P(\mathcal{S}), \underline{B}(\mathcal{S})$, and
  $\underline{A}(\mathcal{S})$ vanish on $\EventHorizonFuture$. To
  this end, we consider using the Bianchi equation in
  \Cref{eq:divergence-of-MST} for $\mathcal{S}$,
  \begin{equation}
    \CovariantDeriv^{\rho}\mathcal{S}_{\alpha\beta\mu\rho}
    ={} - 4\Lambda\frac{5Q\mathcal{F}^2+ 4\Lambda}{Q\mathcal{F}^2+8\Lambda}\mathcal{U}_{\alpha\beta\mu\rho}\mathcal{F}^{-4}\KillT^{\sigma}\mathcal{F}^{\gamma\delta}\tensor[]{\mathcal{S}}{_{\gamma\delta\sigma}^{\rho}}
    + Q\KillT^{\sigma}\left(\frac{2}{3}\mathcal{I}_{\alpha\beta\mu\rho}\mathcal{F}^{\gamma\delta}\tensor[]{\mathcal{S}}{_{\gamma\delta\sigma}^{\rho}}
      - \mathcal{F}_{\mu\rho}\tensor[]{\mathcal{S}}{_{\alpha\beta\sigma}^{\rho}}
    \right).
  \end{equation}
  Now assume without loss of generality that $e_4$ is geodesic along
  $\EventHorizonFuture$. Since we have already shown that
  $B(\mathcal{S}) = A(\mathcal{S}) = 0$ along $\EventHorizonFuture$,
  we have that
  \begin{equation*}
    \CovariantDeriv_4P(\mathcal{S}) = -2\mathcal{J}(\mathcal{S})_{434}.
  \end{equation*}
  Since $P(\mathcal{S})$ vanishes on $S_0$, to show that
  $P(\mathcal{S})$ vanishes on $\EventHorizonFuture$, it suffices to
  show that $\mathcal{J}(\mathcal{S})_{434} $ vanishes on
  $\EventHorizonFuture$. Then, it is clear that
  \begin{equation*}
    \mathcal{J}_{434}=
     4\Lambda\frac{5Q\mathcal{F}^2+ 4\Lambda}{Q\mathcal{F}^2+8\Lambda}\mathcal{U}_{344\rho}\mathcal{F}^{-4}\KillT^{\sigma}\mathcal{F}^{\gamma\delta}\tensor[]{\mathcal{S}}{_{\gamma\delta\sigma}^{\rho}}
    - Q\KillT^{\sigma}\left(\frac{2}{3}\mathcal{I}_{344\rho}\mathcal{F}^{\gamma\delta}\tensor[]{\mathcal{S}}{_{\gamma\delta\sigma}^{\rho}}
      - \mathcal{F}_{4\rho}\tensor[]{\mathcal{S}}{_{34\sigma}^{\rho}}
    \right).
  \end{equation*}
  Then we see that each term is only nonzero if the indices are
  chosen so that $\rho=3$. But observe that
  \begin{align*}
    \mathcal{F}^{\gamma\delta}\tensor[]{\mathcal{S}}{_{\gamma\delta\sigma}^3}
    ={}& 12\mathcal{F}_{43} \tensor[]{\mathcal{S}}{_{34\sigma}^3}.
  \end{align*}
  Thus, we see that $\mathcal{J}(\mathcal{S})_{434}$ is a linear
  combination of $A(\mathcal{S})$ and $B(\mathcal{S})$. Since we
  already know $A(\mathcal{S})$ and $B(\mathcal{S})$ both vanish along
  $\EventHorizonFuture$, it then follows that
  $\mathcal{J}(\mathcal{S})_{434}$ and thus also $P(\mathcal{S})$
  vanishes identically on $\EventHorizonFuture$.

  To show that $\underline{B}(\mathcal{S})$ vanishes on
  $\EventHorizonFuture$, we will use a similar strategy. Using the
  divergence property of the Mars-Simon tensor in
  \Cref{prop:divergence-of-MST}, we have that
  \begin{equation*}
    \nabla_4\underline{B}(\mathcal{S})_a
    = \mathcal{J}(\mathcal{S})_{4a3},
  \end{equation*}
  where we recall that $\nabla_4$ is the horizontal derivative in the
  $e_4$ direction.  We already have that $\underline{B}(\mathcal{S})$
  vanishes on $S_0$. Thus, to prove that $\underline{B}(\mathcal{S})$
  vanishes along $\EventHorizonFuture$, it suffices to show that
  $\mathcal{J}(\mathcal{S})_{4a3}$ vanishes identically on
  $\EventHorizonFuture$.

  To this end, we can compute that
  \begin{align*}
    \mathcal{J}(\mathcal{S})_{4a3}
    ={}&  4\Lambda\frac{5Q\mathcal{F}^2+ 4\Lambda}{Q\mathcal{F}^2+8\Lambda}\mathcal{U}_{a34\rho}\mathcal{F}^{-4}\KillT^{\sigma}\mathcal{F}^{\gamma\delta}\tensor[]{\mathcal{S}}{_{\gamma\delta\sigma}^{\rho}}
    - Q\KillT^{\sigma}\left(\frac{2}{3}\mathcal{I}_{a34\rho}\mathcal{F}^{\gamma\delta}\tensor[]{\mathcal{S}}{_{\gamma\delta\sigma}^{\rho}}
      - \mathcal{F}_{4\rho}\tensor[]{\mathcal{S}}{_{a3\sigma}^{\rho}}
    \right).
  \end{align*}
  Since we already have that $A(\mathcal{S}), B(\mathcal{S}),$ and
  $P(\mathcal{S})$ vanish on $\EventHorizonFuture$, we have that
  \begin{align*}
    \mathcal{S}_{34a4} = \mathcal{S}_{b4a3} = \mathcal{S}_{ab34} = \mathcal{S}_{abcd} = \mathcal{S}_{4bcd} = \mathcal{S}_{4b4c} =0
  \end{align*}
  on $\EventHorizonFuture$. But then it immediately follows that
  $\mathcal{J}(\mathcal{S})_{4a3}$ and thus also
  $\underline{B}(\mathcal{S})$ vanish identically on
  $\EventHorizonFuture$.

  Finally, we show that $\underline{A}(\mathcal{S})$ vanishes along
  $\EventHorizonFuture$. To this end, we observe again from
  \Cref{prop:divergence-of-MST} that $\underline{A}(\mathcal{S})$
  satisfies the transport equation
  \begin{equation*}
    \nabla_4\underline{A}_{ab}(\mathcal{S})
    = \mathcal{J}_{a3b}. 
  \end{equation*}
  Then, since $\underline{A}(\mathcal{S})$ vanishes on $S_0$ and all
  the other null components of $\mathcal{S}$ vanish, we have that
  $\underline{A}(\mathcal{S})$ vanishes along the entire event horizon.

  Similar arguments suffice to show that $B(\mathcal{S})$ and
  $A(\mathcal{S})$ vanish along $\EventHorizonPast$.  This concludes
  the proof of \Cref{prop:horizon:main}.
\end{proof}

\subsection{Vanishing of \texorpdfstring{$\mathcal{S}$}{S} in a neighborhood of the bifurcate sphere}
\label{sec:MST:nbhd-horizon}

In this section, we will show that $\mathcal{S}$ vanishes in a
neighborhood of the bifurcate sphere $S_0$. The contents of this
section are extremely similar to those of Section 6 in
\cite{ionescuUniquenessSmoothStationary2009}, but are included for the
sake of completeness.

\begin{prop}
  \label{prop:nbhd-horizon:main}
  There exists $r_1 = r_1(A_0)>0$ such that
  \begin{equation*}
    \mathcal{S} = 0,\qquad \text{in } \mathbf{O}_{r_1}\bigcap \mathbf{E}.
  \end{equation*}
\end{prop}
In Proposition \ref{prop:nbhd-horizon:main}, a neighborhood
$\mathbf{O}_{r_1}$ of $S_0$ where $\mathcal{S}$ vanishes is
constructed.  
The main property that is used is that the relevant
Killing horizon is pseudoconvex, and thus, that a unique
continuation argument based on Carleman estimates can be applied. 
The arguments are
also similar to those used to prove the rigidity of Kerr in
\cite[Section 6]{ionescuUniquenessSmoothStationary2009}.

To prove \Cref{prop:nbhd-horizon:main}, we use the following main
Carleman estimate.
\begin{lemma}
  \label{lemma:nbhd-horizon:Carleman}
  There is $\varepsilon\in (0, \varepsilon_2)$ sufficiently small and
  $C(\varepsilon)$ sufficiently large such that for any
  $x_0\in S_0$, any $\lambda\ge C(\varepsilon)$,
  and any $\phi\in C_0^{\infty}(B_{\varepsilon^{10}}(x_0))$
  \begin{equation}
    \label{eq:nbhd-horizon:Carleman}
    \lambda\norm*{e^{-\lambda f_{\varepsilon}}\phi}_{L^2}
    + \norm*{e^{-\lambda f_{\varepsilon}}\abs*{D^1\phi}}_{L^2}
    \le C(\varepsilon)\lambda^{-\frac{1}{2}}\norm*{e^{-\lambda f_{\varepsilon}}\Box_{\Metric}\phi}_{L^2},
  \end{equation}
  where $f_{\varepsilon}=\ln (h_{\varepsilon} + \varepsilon^{12}N^{x_0})$, where
  \begin{equation*}
    h_{\varepsilon} = \varepsilon^{-1}(u_{+}+\varepsilon)(u_-+\varepsilon),
  \end{equation*}
  and we recall the definition of $N^{x_0}$ from \Cref{eq:N-x0:def}.
\end{lemma}
\begin{proof}
  Recalling the definition of $\mathbf{O}_{\varepsilon^2}$ from
  \Cref{def:O-set}, we see that for $\varepsilon$ sufficiently small,
  $B_{\varepsilon^{10}}(x_0) \subset \mathbf{O}_{\varepsilon^2}$, and
  thus the weight $f_{\varepsilon}$ is well-defined in
  $B_{\varepsilon^{10}}(x_0)$. Then, using
  \Cref{lemma:nbhd-horizon:pseudo-convexity} and
  \Cref{prop:carleman-estimate} directly yields the result. 
\end{proof}

We can now prove \Cref{prop:nbhd-horizon:main}.
\begin{proof}[Proof of \Cref{prop:nbhd-horizon:main}]
  Using \Cref{lemma:nbhd-horizon:Carleman}, there exist constants
  $\varepsilon(A_0)\in (0, \varepsilon_0)$ and $C(\varepsilon)$ such
  that for any $x_0\in S_0$, $\lambda\ge C(\varepsilon)$, and any
  $\phi\in C_0^{\infty}(B_{\varepsilon^{10}}(x_0))$.
  \begin{equation*}
    \lambda\norm*{e^{-\lambda f_{\varepsilon}}\phi}_{L^2}
    + \norm*{e^{-\lambda f_{\varepsilon}}\abs*{D^1\phi}}_{L^2}
    \le C(\varepsilon)\lambda^{-\frac{1}{2}}\norm*{e^{-\lambda f_{\varepsilon}}\Box_{\Metric}\phi}_{L^2},
  \end{equation*}
  where
  \begin{equation}
    \label{eq:nbhd-horizon:f-epsilon-def}
    f_{\varepsilon} = \ln \left(\varepsilon^{-1}(u_++\varepsilon)(u_-+\varepsilon) + \varepsilon^{12}N^{x_0}\right).
  \end{equation}
  Moreover, using \Cref{coro:regularity:S0}, we see that for
  $\varepsilon$ sufficiently small, $\abs*{R-J\sigma_0}>0$ on
  $B_{\varepsilon^{10}}(x_0)$, so the coefficients of the wave equation of
  $\mathcal{S}$ in \Cref{theorem:MST-wave-equation} are in fact
  regular on $B_{\varepsilon^{10}}(x_0)$. We now fix $\varepsilon$ so that
  all subsequent implicit constants in this proof may depend on
  $\varepsilon$ and $A_0$. In what follows, we will show that
  $\mathcal{S}=0$ in $B_{\varepsilon^{40}}(x_0)\bigcap \mathbf{E}$ for
  any $x_0\in X_0$, which suffices to prove the proposition.

  We now fix $x_0\in S_0$ on the bifurcate sphere, and consider the
  family of smooth functions
  \begin{equation*}
    \begin{split}
      \phi_{(j_1,j_2,j_3,j_4)}:B_{\varepsilon^{10}}(x_0)\mapsto{}& \Complex,\\
      \phi_{(j_1,j_2,j_3,j_4)}(x_0) ={}& \mathcal{S}\left(\partial_{j_1}, \partial_{j_2}, \partial_{j_3}, \partial_{j_4}\right)(x_0),
    \end{split}    
  \end{equation*}
  where the vectorfields $\partial_{\alpha}$ are induced by the
  coordinate chart $\Phi^{x_0}$. Now let
  $\chi \in C_0^{\infty}(\Real)$ such that $\chi:\Real\mapsto [0,1]$,
  $\supp\chi \in \left[\frac{1}{2}, \infty\right)]$, and $\chi(x)=1$
  for $x\in \left[\frac{3}{4}, \infty\right)$. Then for
  $\mathbf{j}=(j_1,j_2,j_3,j_4)\in \curlyBrace*{1,2,3,4}^4$ define
  $\phi^{\delta,\varepsilon}_{\mathbf{j}}\in
  C_0^{\infty}\left(B_{\varepsilon^{10}}(x_0)\bigcap
    \mathbf{E}\right)$ for $\delta\in (0,1]$
  \begin{equation*}
    \begin{split}
      \phi^{\delta,\varepsilon}_{\mathbf{j}} \vcentcolon={}& \phi_{\mathbf{j}}\cdot\widetilde{\chi}_{\delta,\varepsilon},\\
      \widetilde{\chi}_{\delta,\varepsilon}\vcentcolon={}& \mathbf{1}_{\mathbf{E}}\chi\left(\frac{u_+u_-}{\delta}\right)\left(1-\chi\left(\frac{N^{x_0}}{\varepsilon^{20}}\right)\right),
    \end{split}
  \end{equation*}
  so that derivatives of $\widetilde{\chi}_{\delta,\varepsilon}$
  vanish outside the set
  $\mathbf{A}_{\delta}\bigcup \widetilde{\mathbf{B}}_{\varepsilon}$, where
  \begin{equation}
    \label{eq:nbhd-horizon:Abf-Bbf-def}
    \begin{split}
      \mathbf{A}_{\delta} ={}& \curlyBrace*{x\in B_{\varepsilon^{10}}(x_0)\bigcap \mathbf{E}: u_+(x)u_-(x) \in \left(\frac{\delta}{2}, \delta\right)};\\
      \mathbf{B}_{\varepsilon} ={}& \curlyBrace*{x\in B_{\varepsilon^{10}}(x_0)\bigcap \mathbf{E}: N^{x_0}(x)\in \left(\frac{\varepsilon^{20}}{2}, \varepsilon^{20}\right)}.
    \end{split}
  \end{equation}
   We can then calculate that
  \begin{equation*}
    \Box_{\Metric}\phi^{\delta,\varepsilon}_{\mathbf{j}}
    ={} \widetilde{\chi}_{\delta,\varepsilon}\Box_{\Metric}\phi_{\mathbf{j}}
    + 2\CovariantDeriv_{\alpha}\phi_{\mathbf{j}}\CovariantDeriv^{\alpha}\widetilde{\chi}_{\delta,\varepsilon}
    + \phi_{\mathbf{j}}\Box_{\Metric}\widetilde{\chi}_{\delta,\varepsilon}
    .
  \end{equation*}
  In what follows, we will first take $\delta\to 0$, then take
  $\lambda\to \infty$.
  
  Applying \Cref{lemma:nbhd-horizon:Carleman} for any
  $\mathbf{j}=(j_1,j_2,j_3,j_4)\in \curlyBrace*{0,1,2,3}^4$ we have
  \begin{equation}
    \label{eq:nbhd-horizon:carleman-application}
    \begin{split}
      \lambda\norm*{e^{-\lambda f_{\varepsilon}}\widetilde{\chi}_{\delta,\varepsilon}\phi_{\mathbf{j}}}_{L^2}
      + \norm*{e^{-\lambda f_{\varepsilon}}\widetilde{\chi}_{\delta,\varepsilon}\abs*{D^1\phi_{\mathbf{j}}}}_{L^2}
      \lesssim{}& \lambda^{-\frac{1}{2}}\norm*{e^{-\lambda f_{\varepsilon}}\widetilde{\chi}_{\delta,\varepsilon}\Box_{\Metric}\phi_{\mathbf{j}}}_{L^2}
      + \norm*{e^{-\lambda f_{\varepsilon}}\CovariantDeriv_{\alpha}\phi_{\mathbf{j}}\CovariantDeriv^{\alpha}\widetilde{\chi}_{\delta,\varepsilon}}_{L^2}\\
      & + \norm*{e^{-\lambda f_{\varepsilon}}\phi_{\mathbf{j}}\left(\abs*{\Box_{\Metric}\widetilde{\chi}_{\delta,\varepsilon}} + \abs*{D^1\widetilde{\chi}_{\delta,\varepsilon}}\right)}_{L^2},
    \end{split}
  \end{equation}
  for any $\lambda$ sufficiently large. We can now estimate
  $\abs*{\Box_{\Metric}\phi_{\mathbf{j}}}$ using
  \Cref{theorem:MST-wave-equation} and the fact that
  $R-J\sigma_0\neq 0$ in $B_{\varepsilon^{10}}(x_0)$ from
  \Cref{coro:regularity:S0} to write
  \begin{equation*}
    \abs*{\Box_{\Metric}\phi_{\mathbf{j}}} \lesssim \sum_{\mathbf{\ell}\in \{0,1,2,3\}^4}\left(
      \abs*{D^1\phi_{\mathbf{\ell}}}
      + \abs*{\phi_{\mathbf{\ell}}}
    \right)
  \end{equation*}
  for some large implicit constant. Then, summing
  \Cref{eq:nbhd-horizon:carleman-application} over
  $\mathbf{j}\in \{0,1,2,3\}^4$, we can write that
  \begin{equation*}
    \begin{split}
      &\sum_{\mathbf{j}\in \{0,1,2,3\}^4}\left( \lambda\norm*{e^{-\lambda f_{\varepsilon}}\widetilde{\chi}_{\delta,\varepsilon}\phi_{\mathbf{j}}}_{L^2}
        + \norm*{e^{-\lambda f_{\varepsilon}}\widetilde{\chi}_{\delta,\varepsilon}\abs*{D^1\phi_{\mathbf{j}}}}_{L^2} \right)\\
      \lesssim{}& \sum_{\mathbf{j}\in \{0,1,2,3\}^4}\left( \lambda^{-\frac{1}{2}}\norm*{e^{-\lambda f_{\varepsilon}}\widetilde{\chi}_{\delta,\varepsilon}\phi_{\mathbf{j}}}_{L^2}
                  + \lambda^{-\frac{1}{2}}\norm*{e^{-\lambda f_{\varepsilon}}\widetilde{\chi}_{\delta,\varepsilon}\abs*{D^1\phi_{\mathbf{j}}}}_{L^2} \right)\\
      & +\sum_{\mathbf{j}\in \{0,1,2,3\}^4} \left( \norm*{e^{-\lambda f_{\varepsilon}}\CovariantDeriv_{\alpha}\phi_{\mathbf{j}}\CovariantDeriv^{\alpha}\widetilde{\chi}_{\delta,\varepsilon}}_{L^2}
        + \norm*{e^{-\lambda f_{\varepsilon}}\phi_{\mathbf{j}}\left(\abs*{\Box_{\Metric}\widetilde{\chi}_{\delta,\varepsilon}} + \abs*{D^1\widetilde{\chi}_{\delta,\varepsilon}}\right)}_{L^2} \right).
    \end{split}
  \end{equation*}
  The key observation is that for $\lambda$ sufficiently large, the
  first two terms on the \RHS{} can be absorbed into the \LHS. More
  specifically, for $\lambda$ sufficiently large and $0<\delta\le 1$,
  \begin{equation}
    \label{eq:nbhd-horizon:main-eq-with-error}
    \begin{split}
      &\sum_{\mathbf{j}\in \{0,1,2,3\}^4}\left( \lambda\norm*{e^{-\lambda f_{\varepsilon}}\widetilde{\chi}_{\delta,\varepsilon}\phi_{\mathbf{j}}}_{L^2}
        + \norm*{e^{-\lambda f_{\varepsilon}}\widetilde{\chi}_{\delta,\varepsilon}\abs*{D^1\phi_{\mathbf{j}}}}_{L^2} \right)\\
      \lesssim{}& \sum_{\mathbf{j}\in \{0,1,2,3\}^4} \left( \norm*{e^{-\lambda f_{\varepsilon}}\CovariantDeriv_{\alpha}\phi_{\mathbf{j}}\CovariantDeriv^{\alpha}\widetilde{\chi}_{\delta,\varepsilon}}_{L^2}
        + \norm*{e^{-\lambda f_{\varepsilon}}\phi_{\mathbf{j}}\left(\abs*{\Box_{\Metric}\widetilde{\chi}_{\delta,\varepsilon}} + \abs*{D^1\widetilde{\chi}_{\delta,\varepsilon}}\right)}_{L^2} \right).
    \end{split}
  \end{equation}
  We now take $\delta\to 0$. To this end, we will first control each
  term on the right-hand side of
  \Cref{eq:nbhd-horizon:main-eq-with-error}.

  We first show that
  \begin{equation}
    \label{eq:nbhd-horizon:cutoff-error}
    \abs*{\Box_{\Metric}\widetilde{\chi}_{\delta,\varepsilon}}
    + \abs*{D^1\widetilde{\chi}_{\delta,\varepsilon}}
    \lesssim \mathbf{1}_{\mathbf{B}_{\varepsilon}} + \delta^{-1}\mathbf{1}_{\mathbf{A}_{\delta}}. 
  \end{equation}
  We can first directly compute from the definition of
  $\widetilde{\chi}_{\delta,\varepsilon}$ that 
  \begin{equation}
    \label{eq:nbhd-horizon:cutoff-error:1}
    \abs*{D^1\widetilde{\chi}_{\delta,\varepsilon}}
    \lesssim \mathbf{1}_{\mathbf{B}_{\varepsilon}}
    + \delta^{-1}\mathbf{1}_{\mathbf{A}_{\delta}}. 
  \end{equation}
  We can also compute that
  \begin{align}
      \abs*{\Box_{\Metric}\widetilde{\chi}_{\delta,\varepsilon}}
      \le{}& \abs*{\Box_{\Metric}\left(\mathbf{1}_{\mathbf{E}}\chi\left(\frac{u_+u_-}{\delta}\right)\right)}\left(1-\chi\left(\frac{N^{x_0}}{\varepsilon^{20}}\right)\right)
      + C\left(\mathbf{1}_{\mathbf{B}_{\varepsilon}}
             + \delta^{-1}\mathbf{1}_{\mathbf{A}_{\delta}}\right)\notag\\
    \lesssim{}& \left(\mathbf{1}_{\mathbf{B}_{\varepsilon}}
           + \delta^{-1}\mathbf{1}_{\mathbf{A}_{\delta}}\right)
                + \delta^{-2}\mathbf{1}_{\mathbf{E}\bigcap B_{\varepsilon^{10}(x_0)}}\abs*{\CovariantDeriv(u_+u_-)}_{\Metric}^2\notag\\
    \lesssim{}& \left(\mathbf{1}_{\mathbf{B}_{\varepsilon}}
           + \delta^{-1}\mathbf{1}_{\mathbf{A}_{\delta}}\right)
             .\label{eq:nbhd-horizon:cutoff-error:2}
  \end{align}
  Combining \Cref{eq:nbhd-horizon:cutoff-error:1} and
  \Cref{eq:nbhd-horizon:cutoff-error:2} directly yields
  \Cref{eq:nbhd-horizon:cutoff-error}.

  We now estimate the first term on the \RHS{} of
  \Cref{eq:nbhd-horizon:main-eq-with-error}.  Observe that since
  \begin{equation*}
    \phi_{\mathbf{j}}=0,\qquad
    \mathbf{O}_{\varepsilon_2}\bigcap \EventHorizonPast\bigcup \EventHorizonFuture,
  \end{equation*}
  we have from \Cref{eq:smooth-extension:event-future} and
  \Cref{eq:smooth-extension:event-past} that there are smooth
  functions
  $\phi'_{\mathbf{j}}: \mathbf{O}_{\varepsilon_2}\mapsto \Complex$
  such that
  \begin{equation}
    \label{eq:phi-horizon-extension}
    \phi_{\mathbf{j}} = u_+u_-\phi'_{\mathbf{j}}\qquad\text{in }\mathbf{O}_{\varepsilon_2}. 
  \end{equation}
  We can also easily estimate
  \begin{equation}
    \label{eq:nbhd-horizon:cutoff-error:extension-term}
    \abs*{\CovariantDeriv\phi_{\mathbf{j}} \CovariantDeriv \widetilde{\chi}_{\delta,\varepsilon}}
    \le C(\phi')\left(\mathbf{1}_{\mathbf{B}_{\varepsilon}} + \mathbf{1}_{\mathbf{A}_{\delta}}\right),
  \end{equation}
  from \Cref{eq:phi-horizon-extension}. As a result, combining
  \Cref{eq:nbhd-horizon:cutoff-error} and
  \Cref{eq:nbhd-horizon:cutoff-error:extension-term}, we have that
  \begin{equation*}
    \abs*{\CovariantDeriv\phi_{\mathbf{j}} \CovariantDeriv \widetilde{\chi}_{\delta,\varepsilon}}
    + \abs*{\phi_{\mathbf{j}}}\left(\abs*{\Box_{\Metric}\widetilde{\chi}_{\delta,\varepsilon}} + \abs*{D^1\widetilde{\chi}_{\delta,\varepsilon}}\right)
     \lesssim_{\phi'}\mathbf{1}_{\mathbf{B}_{\varepsilon}} + \mathbf{1}_{\mathbf{A}_{\delta}}. 
   \end{equation*}
   As a result, we can rewrite \Cref{eq:nbhd-horizon:main-eq-with-error} as
   \begin{equation}
    \label{eq:nbhd-horizon:main-eq-for-delta-to-0}
    \begin{split}
      \sum_{\mathbf{j}\in \{0,1,2,3\}^4}\left( \lambda\norm*{e^{-\lambda f_{\varepsilon}}\widetilde{\chi}_{\delta,\varepsilon}\phi_{\mathbf{j}}}_{L^2}
        + \norm*{e^{-\lambda f_{\varepsilon}}\widetilde{\chi}_{\delta,\varepsilon}\abs*{D^1\phi_{\mathbf{j}}}}_{L^2} \right)
      \lesssim_{\phi'}{}&   \norm*{e^{-\lambda f_{\varepsilon}}\left( \mathbf{1}_{\mathbf{B}_{\varepsilon}} + \mathbf{1}_{\mathbf{A}_{\delta}} \right)}_{L^2}.
    \end{split}
  \end{equation}
  Taking the limit as $\delta\to 0$ on both sides of
  \Cref{eq:nbhd-horizon:main-eq-for-delta-to-0} then yields that for
  $\lambda$ sufficiently large,
  \begin{equation}
    \label{eq:nbhd-horizon:main-eq-delta-is-0}
    \begin{split}
      \sum_{\mathbf{j}\in \{0,1,2,3\}^4} \lambda\norm*{e^{-\lambda f_{\varepsilon}}\mathbf{1}_{B_{\frac{\varepsilon^{10}}{2}}}\phi_{\mathbf{j}}}_{L^2}        
      \lesssim_{\phi'}{}&   \norm*{e^{-\lambda f_{\varepsilon}}\mathbf{1}_{\mathbf{B}_{\varepsilon}}}_{L^2}.
    \end{split}
  \end{equation}
  We can now use the definition of $f_{\varepsilon}$ in
  \Cref{eq:nbhd-horizon:f-epsilon-def} to see that
  \begin{equation*}
    \sup_{\mathbf{B}_{\varepsilon}}e^{-\lambda f_{\varepsilon}}
    \le e^{-\lambda \ln\frac{\varepsilon+\varepsilon^{32}}{2}}
    \le \inf_{B_{\varepsilon^{40}}(x_0)\bigcap \mathbf{E}}e^{-\lambda f_{\varepsilon}}.
  \end{equation*}
  Then it follows from \Cref{eq:nbhd-horizon:main-eq-delta-is-0} that
  \begin{equation*}
    \lambda\sum_{\mathbf{j}\in \{0,1,2,3\}^4} \norm*{\mathbf{1}_{B_{\varepsilon^{40}}\bigcap \mathbf{E}}\phi_{\mathbf{j}}}_{L^2}
    \lesssim_{\phi'}\norm*{\mathbf{1}_{\mathbf{B}_{\varepsilon}}}_{L^2}.
  \end{equation*}
  Now taking the limit as $\lambda\to \infty$ yields that 
  $\phi_{\mathbf{j}}=0$ in
  $B_{\varepsilon^{40}}(x_0)\bigcap \mathbf{E}$, which completes the
  proof of \Cref{prop:nbhd-horizon:main}. 
\end{proof}
As a consequence of \Cref{prop:nbhd-horizon:main}, we can show some
additional properties of $y$ along the event horizon
$\EventHorizon$. We first show the following transport equations for $y$.
\begin{lemma}
  \label{lemma:y-z-full-deriv}
  Let
  \begin{equation*}
    \mathbf{N} = \closure\mathbf{O}_{r_1}\cap \{0\le u_{+}, u_- < r_1\},
  \end{equation*}
  where $\mathbf{O}_{r_1}$ is the set in
  \Cref{prop:nbhd-horizon:main}. Then in $\mathbf{N}$,
  \begin{equation}
    \label{eq:y-z-full-deriv}
    \CovariantDeriv y
    = - \frac{1}{4}\trace X P e_3
    - \frac{1}{4}\overline{\trace \underline{X}} P e_4,\qquad
    \CovariantDeriv z
    = \frac{1}{2\ImagUnit}\left(P \overline{H}_{\beta} + P \underline{H}_{\beta}\right),
  \end{equation}
  where we recall the definitions of the complexified Ricci
  coefficients in \Cref{def:complex-notation}.  In particular, we thus
  have that
  \begin{equation*}
    \Metric(\CovariantDeriv z, \CovariantDeriv z) \ge 0.
  \end{equation*}
  We also have that
  \begin{equation}
    \label{eq:T-L-LBar-components}
    \KillT\cdot e_4
    = \frac{1}{2}\trace X P,\qquad
    \KillT\cdot e_3
    = - \frac{1}{2}\overline{\trace \underline{X}}P,
  \end{equation}
  and
  \begin{equation}
    \label{eq:trX-trXBar-PSquared-in-terms-of-y-z}
    -\frac{1}{4}\trace X \overline{\trace \underline{X}} P^2
    = \frac{1}{y^2+z^2}\left(k - by + cy^2 - \frac{\Lambda}{3}y^4\right)
    = (\KillT\cdot e_4)(\KillT\cdot e_3).
  \end{equation}
\end{lemma}
\begin{proof}
  From \Cref{prop:nbhd-horizon:main}, we have that $\mathcal{S}=0$ in
  $\mathbf{N}$, so then from \Cref{coro:FCal-null-components}, we have
  that
  \begin{gather*}
    \mathcal{W}_{4a4b}  =
    \mathcal{W}_{3a3b}  =
    \mathcal{W}_{a434}  =
    \mathcal{W}_{a343}  = 0,\\
    \frac{1}{2}\in^{ab}\mathcal{W}_{a3b4} = \frac{1}{3}QR^2, \qquad \delta^{ab}\mathcal{W}_{a3b4}=0.
  \end{gather*}
  This implies that
  \begin{equation*}
    A(\mathcal{W})
    = \overline{A}(\mathcal{W})
    = B(\mathcal{W})
    = \overline{B}(\mathcal{W}) =0,    
  \end{equation*}
  and
  \begin{align*}
    P(\mathcal{W}) ={}& \frac{1}{3}QR^2\\
    ={}& -\frac{1}{P}R - \frac{\Lambda}{3R^2}\\
    ={}& \frac{b}{2P^3}.
  \end{align*}

  Then, the first four Bianchi identities \Cref{eq:Bianchi:nabla3-A},
  \Cref{eq:Bianchi:nabla4-ABar}, \Cref{eq:Bianchi:nabla4-B}, and
  \Cref{eq:Bianchi:nabla3-BBar} imply that
  \begin{equation*}
    \Xi = \underline{\Xi} = \widehat{X} = \widehat{\underline{X}} = 0.
  \end{equation*}
  On the other hand, the next four Bianchi identities in
  \Cref{eq:Bianchi:nabla3-B}, \Cref{eq:Bianchi:nabla4-BBar},
  \Cref{eq:Bianchi:nabla4-P}, and \Cref{eq:Bianchi:nabla3-P} imply
  that
  \begin{align*}
    \ComplexDeriv \overline{P(\mathcal{W})}
    &= - 3\overline{P(\mathcal{W})}H    
      ,
    \\
    \ComplexDeriv P(\mathcal{W})
    &= -3P(\mathcal{W})\HBar   
      ,
    \\
    \CovariantDeriv_4 P(\mathcal{W})
    &= - \frac{3}{2}\Trace X P(\mathcal{W})
      ,
    \\
    \CovariantDeriv_3 P(\mathcal{W})
    &= - \frac{3}{2}\overline{\Trace\XBar}P(\mathcal{W})
      .
  \end{align*}
  These simplify to
  \begin{equation}
    \label{eq:P:null-derivatives}
    \begin{split}
      \overline{\ComplexDeriv} P ={}& P\overline{H}\\
      \ComplexDeriv P ={}& P \underline{H}\\
      \CovariantDeriv_4P ={}& \frac{1}{2}\trace X P \\
      \CovariantDeriv_3P ={}& \frac{1}{2}\overline{\trace \underline{X}} P.
    \end{split}    
  \end{equation}
  Using \Cref{lemma:y-z-derivatives:S-small} with $\varepsilon_{\mathcal{S}}=0$, 
  we see that
  \begin{equation*}
    \begin{split}
      \nabla z ={}& \frac{1}{2\ImagUnit}\left(P \underline{H} + P \overline{H}\right),\\
      \LeftDual{\nabla} z ={}& \frac{1}{2}\left(P \underline{H} - P \overline{H} \right),\\
      \CovariantDeriv_4y ={}& \frac{1}{2}\trace X P ,\\
      \CovariantDeriv_3y ={}& \frac{1}{2}\overline{\trace \underline{X}} P.
    \end{split}    
  \end{equation*}
  In particular we see that
  \begin{equation}    
    \label{eq:P-H-HBar-relations}
    \trace X P = \overline{\trace X P}, \qquad
    \overline{\trace \underline{X}} P = \trace \underline{X} \overline{P}.
  \end{equation}  
  Using \Cref{lemma:y-z-derivatives:S-small} with $\varepsilon_{\mathcal{S}}=0$, we can thus express
  \begin{equation*}
    \CovariantDeriv_{\beta}y
    = - \frac{1}{4}\trace X P (e_3)_{\beta}
    - \frac{1}{4}\overline{\trace \underline{X}} P (e_4)_{\beta},\qquad
    \CovariantDeriv_{\beta}z
    = \frac{1}{2\ImagUnit}\left(P \overline{H}_{\beta} + P \underline{H}_{\beta}\right),
  \end{equation*}
  as stated in \Cref{eq:y-z-full-deriv}.

  The identities in \Cref{eq:T-L-LBar-components} and
  \Cref{eq:trX-trXBar-PSquared-in-terms-of-y-z} follow immediately
  from \Cref{eq:y-z-full-deriv},
  \Cref{eq:y-z-derivatives:basic:S-small}, and
  \Cref{eq:nabla-y-contracted:S-small} with
  $\varepsilon_{\mathcal{S}}=0$.
\end{proof}

As a consequence of \Cref{prop:nbhd-horizon:main}, we show that $y$ is
constant on the event horizon $\EventHorizon$ and
that it increases in $\mathbf{E}$.
\begin{lemma}
  \label{lemma:y-constant-on-horizons}
  We have that
  \begin{equation}
    \label{eq:y-value-on-horizons}
    \begin{gathered}
      y = y_{S_0}\qquad \text{ on }\EventHorizon\bigcap \mathbf{O}_{r_1},
    \end{gathered}    
  \end{equation}
  Finally, for sufficiently small $\varepsilon = \varepsilon(A_0)>0$,
  \begin{equation}
    \label{eq:y-local-S0}
    y > y_{S_0} + C(A_0)^{-1}u_+u_- \qquad \text{ on } \mathbf{O}_{\varepsilon}\bigcap \mathbf{E}.
  \end{equation}
\end{lemma}

\begin{proof}
  Let $\mathbf{N} = \mathbf{O}_{r_1}$ denote the set constructed in
  \Cref{prop:nbhd-horizon:main}. By its construction in
  \Cref{prop:nbhd-horizon:main}, $\mathcal{S}=0$ in
  $\mathbf{N}$. We also recall from
  \Cref{lemma:horizon:Fcal-null-components} that $B(\mathcal{F})$
  vanishes along $\EventHorizon^+$, and $\underline{B}(\mathcal{F})$
  vanishes along both $\EventHorizon^-$. As a result, we have that for
  any vectorfield $\widetilde{L}$ tangent to the null generators of
  $\EventHorizon^+$, there exists a constant $C\in \Real$ such that
  \begin{equation*}
    \mathcal{F}\cdot\widetilde{L} = C \widetilde{L}.
  \end{equation*}
  thus, $\widetilde{L}$ is parallel to either $\bm{\ell}_-$ or $\bm{\ell}_+$ on
  $\EventHorizon^+$. Similarly, any vectorfield tangent to the null
  generators of $\EventHorizon^-$ must also be parallel to either
  $\bm{\ell}_-$ or $\bm{\ell}_+$. Thus, $(e_1,e_2)$ must be tangent to the bifurcate
  sphere $S_0$. Recall from \Cref{lemma:y-z-derivatives:S-small} that
  $\nabla y = 0$, since we have that $\mathcal{S}=0$ on
  $\mathbf{N}$. Thus we have that $y$ is in fact constant on $S_0$.

  To show that $y$ is in fact constant along the entire horizon
  $\EventHorizon$, we use the transport equations implied by Lemma
  \ref{lemma:y-z-full-deriv},
  \begin{equation*}
    \CovariantDeriv_{\bm{\ell}_+}y = \frac{1}{2}\Trace X P, \qquad
    \CovariantDeriv_{\bm{\ell}_-}y = \frac{1}{2}\overline{\Trace \underline{X}}P.
  \end{equation*}
  Since we already know that $L$ and $\underline{L}$ must be tangent
  to $\bm{\ell}_-$ or $\bm{\ell}_+$, then the fact that the horizons
  are nonexpanding allows us to transport $y$ from $S_0$ along
  $\EventHorizonFuture$ and $\EventHorizonPast$.  Thus $y$ is a
  constant on $\EventHorizon\bigcap \mathbf{O}_{r_1}$.

  Using \Cref{eq:trX-trXBar-PSquared-in-terms-of-y-z} on $S_0$ and
  the fact that $\KillT$ is tangent to $S_0$, it follows that for any $p\in S_0$
  \begin{equation*}
    \Delta(y(p)) = 0.
  \end{equation*}
  Using then \ref{ass:E1} and the subextremality assumption proves
  \Cref{eq:y-value-on-horizons}.

  To prove \Cref{eq:y-local-S0}, we begin by observing that combining
  \Cref{eq:y-value-on-horizons},
  \Cref{eq:smooth-extension:event-future} and
  \Cref{eq:smooth-extension:event-past}, we have that
  \begin{equation}
    \label{eq:y-local-around-S0-expression}
    y = y_{S_0} + u_+u_-\cdot y',
  \end{equation}
  for some smooth function $y': \mathbf{O}\to \Real$, with
  $\abs*{D^1y'}\le \widetilde{C}$. 

  Next we calculate
  $\CovariantDeriv^{\alpha}\CovariantDeriv_{\alpha}P$. We can compute
  that since $\mathcal{S}=0$ in $\mathbf{N}$,
  \begin{align*}
    \CovariantDeriv^{\alpha}\CovariantDeriv_{\alpha}P
    ={}& \CovariantDeriv^{\alpha}\mathbf{P}_{\alpha}\\
    ={}& \frac{1}{2}\CovariantDeriv^{\alpha}\frac{\bm{\sigma}_{\alpha}}{R}\\
    ={}& -\frac{1}{2R^2}\CovariantDeriv^{\alpha}R\bm{\sigma}_{\alpha}
         + \frac{1}{2R}\CovariantDeriv^{\alpha}\bm{\sigma}_{\alpha}\\
    ={}& -\frac{1}{2R^2}(2JR-\Lambda)\mathbf{P}^{\alpha}\bm{\sigma}_{\alpha}
         - \frac{1}{2R}\left(\mathcal{F}^2 + 2\Lambda N\right)\\
    = {}& - \frac{1}{4R^3}(2JR-\Lambda)\bm{\sigma}^{\alpha}\bm{\sigma}_{\alpha}
          - \frac{1}{2R}(2\Lambda N - 4R^2)\\
    ={}& \frac{1}{4R^2}(2JR-\Lambda)N\mathcal{F}^2
         - \frac{1}{2R}(2\Lambda N - 4R^2)\\
    ={}& 2(R-JN),
  \end{align*}
  where we used \Cref{eq:deriv-R}, \Cref{eq:divergence-sigma}, the
  definition of $R$ in \Cref{eq:R:def}, and
  \Cref{lemma:sigma-basic-props}.

  Since from Proposition \ref{prop:nbhd-horizon:main}, we have that
  $\mathcal{S}=0$ on $\mathbf{O}_{\varepsilon}$. As a result we apply
  the expressions for $J,R,N$ in terms of $P$ from
  \Cref{eq:R-sigma-in-terms-of-P:S-small} with
  $\varepsilon_{\mathcal{S}}=0$, recalling that $N = \Re\sigma$, to
  compute that
  \begin{equation*}
    2(R-JN)=\frac{-b + 2cy - \frac{4}{3}y^3\Lambda}{y^2+z^2}
    = \frac{\partial_y\Delta}{y^2+z^2}
    .
  \end{equation*}
  Thus, we have that
  \begin{equation}
    \label{eq:wave-op-y}
    \CovariantDeriv^{\alpha}\CovariantDeriv_{\alpha}y  = \frac{\partial_y\Delta}{y^2+z^2}.
  \end{equation}
  Then, substituting \Cref{eq:y-local-around-S0-expression} in to
  \Cref{eq:wave-op-y}, and evaluating at a point on $S_0$, we see that
  \begin{equation*}
    \frac{\partial_y\Delta(y_{S_0})}{y^2_0+z^2}
    = 2\CovariantDeriv^{\alpha}u_+\CovariantDeriv_{\alpha}u_-\cdot y'
    = 2 y'.
  \end{equation*}
  Since we assumed that $\partial_y\Delta(y_{S_0})>0$ in
  \Cref{eq:S0-y-Delta-assumptions}, we see that for
  $\varepsilon\in (0, r_1)$ sufficiently small,
  \begin{equation*}
    y > y_{S_0} + \widetilde{C}^{-1}u_+u_-,\qquad \text{ in }\mathbf{O}_{\varepsilon}\bigcap \mathbf{E},
  \end{equation*}
  as desired. 
\end{proof}

\subsection{Extension of \texorpdfstring{$\mathcal{S}$}{S} in \texorpdfstring{$\mathbf{E}$}{E}}
\label{sec:MST:main-extension}

The main goal of this section will be to show that there exists some $y_{\mathcal{S}} > y_{*}$ such that 
\begin{equation}
  \label{eq:bootstrap-argument:main-goals}
  \mathcal{S}=0 \qquad \text{on }\underline{\Sigma}_{y_{\mathcal{S}}}\bigcap \mathbf{E}.
\end{equation}
This will imply that $\mathcal{S}=0$ on $\mathbf{E}$. Then, we will
use the fact that we have calibrated our constants
$(b_{S_0},c_{S_0},k_{S_0})$ to match those of \KdS{} and the main result of
\cite{marsSpacetimeCharacterizationKerrNUTAde2015} to prove \Cref{thm:main:e1c1}.

We now make the following bootstrap assumption that we will
improve. Let
\begin{equation}
  \label{eq:S-extension:BSA}
  R_0\vcentcolon= \sup\{R\in \Real^{+}: \mathcal{S} = 0 \text{ in }\mathcal{U}_R\}.
\end{equation}
We will show that in fact $R_0 > y_{*}$, where $y_{*}$ is the unique $y$
maximizer of $\Delta$ in $\mathbf{E}$.

We first initialize our bootstrap argument.
\begin{prop}
  \label{prop:bootstrap-init}
  There exists some $R_1\in\Real^+$,
  $R_1\ge y_{S_0} + C^{-1}$ 
  for some $C = C(A_0) > 0 $, such that $\mathcal{S}=0$ in
  $\mathcal{U}_{R_1}$.
\end{prop}

\begin{proof}
  Let $\varepsilon$ be as chosen in
  \Cref{lemma:y-constant-on-horizons}. It follows from
  \Cref{prop:nbhd-horizon:main} that $\mathcal{S}=0$ in
  $\mathbf{O}_{\varepsilon}\bigcap \mathbf{E}$. Recall from
  \Cref{eq:A0-bounding-uPlus-uMinus} that we have that
  \begin{equation*}
    \frac{u_+}{u_-} + \frac{u_-}{u_+}\le A_0\qquad \text{in }\Sigma_0\bigcap \mathbf{E} \bigcap \mathbf{O}_{\varepsilon}. 
  \end{equation*}
  Then, using \Cref{eq:y-local-around-S0-expression}, we have that
  \begin{equation*}
    y-y_{S_0}\in\left[
      \widetilde{C}^{-1}\left(u_+^2 + u_-^2\right),
      \widetilde{C}\left(u_+^2 + u_-^2\right)
    \right]\qquad \text{in }\Sigma_0\bigcap \mathbf{E} \bigcap \mathbf{O}_{\varepsilon}. 
  \end{equation*}
  Thus, for $R_1$ sufficiently close to $y_{S_0}$, the set
  $\mathcal{U}_{R_1}$ is contained in $\mathbf{O}_{\varepsilon}$, and the
  proposition follows.
\end{proof}

We now define
\begin{equation}
  \label{eq:NR2:def}
  \mathbf{N}_{R_0} \vcentcolon= \text{ connected component of } \left[\left(\bigcup_{t\in\Real}\Phi_t^*\mathcal{U}_{R_0}\right)\bigcup \mathbf{O}_{r_1}\right]\bigcap \widetilde{\mathcal{M}} \text{ containing }\mathcal{U}_{R_0},
\end{equation}
where $r_1$ is as constructed in \Cref{prop:nbhd-horizon:main}.  Since
$\KillT$ is a Killing vectorfield,
$\LieDerivative_{\KillT}\mathcal{S}=0$ in $\widetilde{\mathcal{M}}$.
Then, we have from our induction hypothesis that $\mathcal{S}=0$ in
$\mathcal{U}_{R_0}$, and from our assumptions on stationarity that
$\KillT$ does not vanish in $\mathbf{E}$. As a result, it follows that
\begin{equation}
  \label{eq:S-vanishing-in-N-R2}
  \mathcal{S}=0\text{ in }\mathbf{N}_{R_0}\bigcap \mathbf{E}.
\end{equation}
Let $\delta_{R_0}$ be sufficiently small so that
\begin{equation*}
  y(p) \in \left( \frac{y_{S_0} + R_1}{2}, 2R_0 \right)\quad \text{for any } x\in B_{\delta_{R_0}}(x_0).
\end{equation*}

In view of \Cref{eq:A-tilde-def}, there exists some
$\delta_{R_0} > C(R_0)^{-1}$ sufficiently small such that the
set $(-\delta_{R_0},\delta_{R_0})\times (B_{\delta_{R_0}}(x_0)\bigcap \Sigma_0)$ is
diffeomorphic to the set
$\bigcap_{\abs*{t}<\delta_{R_0}}\Phi_t(B_{\delta_{R_0}}(x_0)\bigcap \Sigma_0)$. We
let
\begin{equation}
  \label{eq:Q-projection:def}
  \begin{split}
    \QProj:\bigcap_{\abs*{t}<\delta_{R_0}}\Phi_t\left(B_{\delta_{R_0}}(x_0)\bigcap \Sigma_0\right)&\to B_{\delta_{R_0}}(x_0)\bigcap \Sigma_0\\
    \Phi(x)&\mapsto x
  \end{split}
\end{equation}
denote the induced smooth projection which takes every point
$\Phi(x)\mapsto x$.

We first prove the following auxiliary lemma.
\begin{lemma}
  \label{lemma:main-prop:aux}
  Let $x_0\in \partial_{\Sigma_0\bigcap\mathbf{E}} \mathcal{U}_{R_0}
  $. Then there exists some $0< r_0$ such that
  \begin{equation*}
    \curlyBrace*{x\in B_{r_0}(x_0): y(x) < R_0}\subset \bigcup_{\abs*{t}<\delta_{R_0}}\Phi_t(\mathcal{U}_{R_0}).
  \end{equation*}
\end{lemma}
\begin{proof}
  Since \Cref{eq:S-vanishing-in-N-R2} holds, we can infer from 
  \Cref{lemma:nabla-y-contracted:S-small} that \Cref{eq:nabla-y-contracted:S-small} 
  holds with $\varepsilon_{\mathcal{S}}=0$ in $\mathbf{N}_{R_0}$. As a result, there exists some 
  $r' \le C(R_0)^{-1}$ such that in $B_{r'}(x_0)$
  \begin{equation*}
    \CovariantDeriv_{\alpha}y\CovariantDeriv^{\alpha}y \ge C(R_0)^{-1}.
  \end{equation*}
  As a result, there exists some
  $r_0 = r_0(A_0,\widetilde{A}_{\widetilde{C}^{-1}}, R_0)>0$ and an
  open set $B_{r_0}(x_0)\subset B'\subset B_{r'(x_0)}$ such that the set
  $\mathcal{B}\vcentcolon=\curlyBrace*{x\in B':y(x)<R_0}$ is connected.

  Then we have that the set
  $\QProj(\mathcal{B})\subset B_{r'(x_0)}\bigcap \Sigma_0$ is connected
  and contains the set $\mathcal{B}$. Since $y(\QProj(x)) = y(x)$, it
  follows from the definition of $\mathcal{U}_{R_0}$ that
  \begin{equation*}
    \QProj(\{x\in B':y(x) < R_0\}) \subset \mathcal{U}_{R_0}. 
  \end{equation*}
  The result then follows from the definition of $\QProj$ in \Cref{eq:Q-projection:def}.
\end{proof}

As previously mentioned, the main step in the proof of
\Cref{prop:main-prop} is to show that the induction hypothesis can in
fact be extended, which is encapsulated in the following proposition. 
\begin{prop}
  \label{prop:main-extension}
  Let
  $x_0\in \partial_{\Sigma_0\bigcap \mathbf{E}}(\mathcal{U}_{R_0})$,
  where we recall that we have assumed that $R_0\le y_*$.  Then, there
  exists some
  $r_3 = r_3(A_0, \widetilde{A}_{\widetilde{C}^{-1}}, R_0)\in (0,
  r_0)$ such that $\mathcal{S}=0$ in $B_{r_3}(x_0)$.
\end{prop}
The main ingredient needed to prove \Cref{prop:main-extension} is a
Carleman inequality.
\begin{lemma}
  \label{lemma:main-extension:Carleman}
  Let
  $x_0\in \partial_{\Sigma_0\bigcap \mathbf{E}}(\mathcal{U}_{R_0})$,
  where we recall that we assumed that $R_0<y_{\mathcal{S}}$.  There is some
  $0<\varepsilon<r_0$ sufficiently small and some $C(\varepsilon)$
  sufficiently large such that for any $\lambda\ge C(\varepsilon)$ and
  any $\phi\in C_0^{\infty}(B_{\varepsilon^{10}}(x_0))$,
  \begin{equation}
    \label{eq:main-extension:Carleman}
    \begin{split}
      \lambda \norm*{e^{-\lambda \widetilde{f_{\varepsilon}}}\phi}_{L^2}
      + \norm*{e^{-\lambda\widetilde{f}_{\varepsilon}}\abs*{D^1\phi}}_{L^2}
      \le{}& \widetilde{C}_{\varepsilon}\left( \lambda^{-\frac{1}{2}}\norm*{e^{- \lambda\widetilde{f}_{\varepsilon}}\Box_{\Metric}\phi}_{L^2}
             + \varepsilon^{-6}\norm*{e^{-\lambda \widetilde{f}_{\varepsilon}}\KillT(\phi)}_{L^2} \right),
    \end{split}
  \end{equation}
  where,
  \begin{equation}
    \label{eq:main-extension:f-tilde-epsilon-def}
    \widetilde{f}_{\varepsilon} = \ln\left(y - R_0 + \varepsilon + \varepsilon^{12}N^{x_0}\right),
  \end{equation}
  where we recall the definition of $N^{x_0}$ from \Cref{eq:N-x0:def}.
\end{lemma}
\begin{proof}
  We will apply \Cref{prop:carleman-estimate} with
  \begin{equation*}
    \{\mathbf{V}_i\}_{i=1}^1 = \{\KillT\},\qquad
    h_{\varepsilon }= y-R_0+\varepsilon, \qquad e_{\varepsilon} = \varepsilon^{12}N^{x_0}.
  \end{equation*}
  From the definition in \Cref{def:negligible perturbation}, it is
  clear that $e_{\varepsilon}$ is a negligible perturbation if
  $\varepsilon$ is sufficiently small.

  We now show that for some $\varepsilon_1$ sufficiently small,
  $\{h_{\varepsilon}\}_{\varepsilon\in (0,\varepsilon_1)}$ forms a
  $\KillT$-pseudoconvex family of weights. From the definition
  of $h_{\varepsilon}$, we have that
  \begin{equation*}
    h_{\varepsilon}(x_0) = \varepsilon,\qquad
    \KillT(h_{\varepsilon})(x_0)  = 0,\qquad
    \abs*{D^jh_{\varepsilon}}\lesssim 1,
  \end{equation*}
  so \Cref{eq:strict-T-Z-null-convexity:cond1} is satisfied for
  $\varepsilon_1$ sufficiently small. The conditions in
  \Cref{eq:strict-T-Z-null-convexity:cond2} and
  \Cref{eq:strict-T-Z-null-convexity:main-condition} are then
  satisfied from \Cref{prop:y-pseudoconvexity} and
  \Cref{lemma:nabla-y-contracted:S-small} since
  $y_0<y(x_0)<y_{\mathcal{S}}$. Thus, for some $\varepsilon_1$
  sufficiently small,
  $\{h_{\varepsilon}\}_{\varepsilon\in (0,\varepsilon_1)}$ forms a
  $\KillT$-pseudoconvex family of
  weights. \Cref{prop:carleman-estimate} then directly yields the
  result.
\end{proof}

We also prove the following lemma, which tells us that the
coefficients of the wave equation for $\mathcal{S}$ has regular
coefficients. Recall from \Cref{eq:MST-wave equation} that the only
possible irregularity in the coefficients of \Cref{eq:MST-wave
  equation} is the potential vanishing of $R-J\sigma_0$
\begin{lemma}
  \label{prop:boundedness-of-regularity-qty}
  There is some constant $C=C(A_0,R_0)$ that depends only on $A_0$ and $R_0$ such
  that for any $x_0\in \partial_{\Sigma_0\bigcap \mathbf{E}}(\mathcal{U}_{R_0})$,
  \begin{equation}
   \label{eq:regularity-qty-non-degeneracy}
    \abs*{\frac{1}{R-J\sigma_0}}\ge C^{-1}.
  \end{equation} 
\end{lemma}
\begin{proof}
  Assume for the sake of contradicton that there exists some point
  $x_0\in \partial_{\Sigma_0\bigcap \mathbf{E}}(\mathcal{U}_{R_0})$,
  such that $y(x_0)>y_{S_0}$ and $\frac{1}{R-J\sigma_0}$ is
  unbounded. This is equivalent to assuming that
  $(R-J\sigma_0)(x_0)=0$. Recall from the definition of $k$ in
  \Cref{eq:b-c-k:def} that
  \begin{align*}
    \CovariantDeriv_{\alpha}z\CovariantDeriv^{\alpha}z
    &= \frac{k - cz^2 - \frac{\Lambda}{3}z^4}{y^2+z^2}.
  \end{align*}
  Since $S_0\subset \mathbf{N}_{R_0}$, we know that in fact, 
  \begin{align*}
    \CovariantDeriv_{\alpha}z\CovariantDeriv^{\alpha}z
    &= \frac{k_{S_0} - c_{S_0}z^2 - \frac{\Lambda}{3}z^4}{y^2+z^2}    
      ,
  \end{align*}
  from the constancy of $c, k$ in
  \Cref{lemma:b-c-k-almost-constant:S-small} and
  \Cref{lemma:k-almost-constant} when $\varepsilon_{\mathcal{S}}=0$. 

  Now recalling the definition of $\Delta$ in \Cref{eq:Delta:def},
  we can write that
  \begin{equation*}
    c_{S_0}y^2 +\Delta =  k_{S_0} + y\partial_y\Delta + \Lambda y^4.
  \end{equation*}
  We also at the same time have that $y_{S_0}>0$,
  $\partial_y\Delta(y_{S_0}) > 0$, $\Delta(y_{S_0})=0$ from the
  subextremality assumption.

  Then, from \Cref{lemma:y-z-full-deriv}, we have that for
  $\varepsilon_{\mathcal{S}}$ sufficiently small, 
  \begin{equation}
    \label{eq:regularity-qty-non-degeneracy:aux S0 props}
    c_{S_0}z^2 \le k_{S_0},\qquad
    c_{S_0}y_{S_0}^2 > k_{S_0} + \Lambda y_{S_0}^4,\qquad
    \abs*{\frac{z}{y_{S_0}}}\le 1. 
  \end{equation}
  Since $y(x_0)>y_{S_0}$, we must have that $\abs*{\frac{z(x_0)}{y(x_0)}}\le 1$.
  In particular, this implies that $\abs*{\arg P} \le \frac{\pi}{4}$. As a
  result, we have that
  \begin{equation*}
    \abs*{\frac{1}{P^3} + \frac{4\Lambda}{3b_{S_0}}} \ge \frac{\Lambda}{6b_{S_0}}.
  \end{equation*}
  Moreover, recall from~\Cref{eq:R-sigma-in-terms-of-P:S-small} when
  $\varepsilon_{\mathcal{S}}=0$ that when $\mathcal{S}=0$ that
  $R = \frac{b_{S_0}}{2P^2} - \frac{\Lambda}{3}P$, and
  $\sigma_0 = -\frac{b_{S_0}}{P} - \frac{\Lambda}{3}P^2$, and
  moreover, that $b=b_{S_0}$ is constant from
  \Cref{lemma:b-c-k-almost-constant:S-small}.
  
  Then, since $\abs*{P}>y_{S_0}$ since $R_0>y_{S_0}$, we have that for
  $\varepsilon_{\mathcal{S}}$ sufficiently small,
  \begin{equation*}
    \abs*{R-J\sigma_0}
    =\abs*{\frac{b_{S_0}}{2}P}\abs*{\frac{1}{P^3} - \frac{4\Lambda}{3b_{S_0}}}
    \ge \frac{\Lambda}{12}y_{S_0}.
  \end{equation*}
  So in fact, there exists some $C$ sufficiently large so that
  \Cref{eq:regularity-qty-non-degeneracy} holds. 
\end{proof}
 
We now prove \Cref{prop:main-extension}.
\begin{proof}[Proof of \Cref{prop:main-extension}]
  We will fix $\varepsilon$ throughout the proof. We now show that
  $\mathcal{S}=0$ in the set
  $B_{\varepsilon^{100}} = B_{\varepsilon^{100}}(x_0)$.

  From the wave equation for $\mathcal{S}$ in
  \Cref{theorem:MST-wave-equation}, the fact that $\KillT$ is Killing,
  and \Cref{prop:boundedness-of-regularity-qty}, we have that there
  exist smooth tensor fields $\mathcal{A}$ and $\mathcal{B}$ such that
  \begin{equation}
    \label{eq:main-extension:MST-relations}
    \begin{split}
      \Box_{\Metric}\mathcal{S}_{\alpha_1\cdots \alpha_4}
      ={}& \mathcal{S}_{\beta_1\cdots\beta_4}\tensor[]{\mathcal{A}}{^{\beta_1\cdots\beta_4}_{\alpha_1\cdots\alpha_4}}
           + \CovariantDeriv_{\beta_5}\mathcal{S}_{\beta_1\cdots\beta_4}\tensor[]{\mathcal{B}}{^{\beta_1\cdots\beta_5}_{\alpha_1\cdots\alpha_4}},\\
      \LieDerivative_{\KillT}\mathcal{S}={}&0,
    \end{split}
  \end{equation}
  in $B_{\varepsilon^{10}}(x_0)$. Using \Cref{lemma:main-prop:aux} and
  the bootstrap assumption that $\mathcal{S}$ vanishes in
  $\mathcal{U}_{R_0}$, we have that
  \begin{equation}
    \label{eq:main-extension:bootstrap-assumption}
    \mathcal{S}=0\quad \text{ in }\curlyBrace*{x\in B_{\varepsilon^{10}}(x_0):y(x)<R_0}. 
  \end{equation}
  Now, for $\mathbf{j}=(j_1,j_2,j_3,j_4)\in \{0,1,2,3\}^4$, we
  consider the vectorfields $\partial_{\alpha}$ induced by the
  coordinate chart $\Phi^{x_0}$ (recall the definition in
  \Cref{eq:Phi-x0:def}), and define the smooth functions
  $\phi_{\mathbf{j}}$ by
  \begin{equation*}
    \begin{split}
      \phi_{\mathbf{j}}: B_{\varepsilon^{10}}(x_0)&\to \Complex\\
      x&\mapsto \mathcal{S}(\partial_{j_1},\partial_{j_2},\partial_{j_3},\partial_{j_4})(x).
    \end{split}    
  \end{equation*}
  Now, we let $\widetilde{\chi}: \Real\to [0,1]$ denote a smooth function
  supported in $\left[\frac{1}{2},\infty\right)$ and equal to $1$ in
  $\left[\frac{3}{4},\infty\right)$. We then define 
  $\phi^{\varepsilon}_{\mathbf{j}}\in
  C^{\infty}_0(B_{\varepsilon^{10}}(x_0))$ to be the localization of $\phi_{\mathbf{j}}$ by
  \begin{equation*}
    \phi^{\varepsilon}_{\mathbf{j}}\vcentcolon= \phi_{\mathbf{j}}\widetilde{\chi}_{\varepsilon},
  \end{equation*}
  where 
  \begin{equation}
     \label{eq:main-extension:cutoff-def}
    \widetilde{\chi}_{\varepsilon}= 1- \widetilde{\chi}(\varepsilon^{-40}N^{x_0}(x)).
  \end{equation}
  It is then a quick computation to see that
  \begin{equation*}
    \begin{split}
      \Box_{\Metric}\phi^{\varepsilon}_{\mathbf{j}}
      ={}& \widetilde{\chi}_{\varepsilon}\Box_{\Metric}\phi_{\mathbf{j}}
           + 2\CovariantDeriv_{\alpha}\phi_{\mathbf{j}} \CovariantDeriv^{\alpha}\widetilde{\chi}_{\varepsilon}
           +\phi_{\mathbf{j}}\Box_{\Metric}\widetilde{\chi}_{\varepsilon},\\
      \KillT\left(\phi^{\varepsilon}_{\mathbf{j}}\right)
      ={}& \widetilde{\chi}_{\varepsilon}\KillT(\phi_{\mathbf{j}})
           + \phi_{\mathbf{j}}\KillT(\widetilde{\chi}_{\varepsilon}).
    \end{split}
  \end{equation*}
  We then use the Carleman inequality in \Cref{eq:main-extension:Carleman}
  to see that for any
  $\{j_1,j_2,j_3,j_4\}\in \{0,1,2,3\}^4$,
  \begin{equation}
    \label{eq:main-extension:Carleman-with-cutoffs}
    \begin{split}
      &\lambda\norm*{e^{-\lambda\widetilde{f}_{\varepsilon}}\widetilde{\chi}_{\varepsilon}\phi_{\mathbf{j}}}_{L^2}
        + \norm*{e^{-\lambda\widetilde{f}_{\varepsilon}}\widetilde{\chi}_{\varepsilon}\abs*{D^1\phi_{\mathbf{j}}}}_{L^2}\\
      \lesssim{}& C(R_0)\lambda^{-\frac{1}{2}}\norm*{e^{-\lambda\widetilde{f}_{\varepsilon}}\widetilde{\chi}_{\varepsilon}\Box_{\Metric}\phi_{\mathbf{j}}}_{L^2}
                  + C(R_0)\norm*{e^{-\lambda\widetilde{f}_{\varepsilon}}\widetilde{\chi}_{\varepsilon}\KillT(\phi_{\mathbf{j}})}_{L^2}
                  + C(R_0)\norm*{e^{-\lambda\widetilde{f}_{\varepsilon}}\CovariantDeriv_{\alpha}\phi_{\mathbf{j}}\CovariantDeriv^{\alpha}\widetilde{\chi}_{\varepsilon}}_{L^2}\\
      & + C(R_0)\norm*{e^{-\lambda\widetilde{f}_{\varepsilon}}\phi_{\mathbf{j}}\left(\abs*{\Box_{\Metric}\widetilde{\chi}_{\varepsilon}}+ \abs*{D^1\widetilde{\chi}_{\varepsilon}}\right)}_{L^2},
    \end{split}
  \end{equation}
  for any $\lambda\ge C(R_0)$. From the equations the
  Mars-Simon tensor $\mathcal{S}$ satisfies in
  \Cref{eq:main-extension:MST-relations}, we have that
  \begin{equation}
    \label{eq:main-extenion:MST-inequalities}
    \begin{split}
      \abs*{\Box_{\Metric}\phi_{\mathbf{j}}}&\le C(R_0)\sum_{\ell\in \{0,1,2,3\}^4}\left(\abs*{D^1\phi_{\ell}} + \abs*{\phi_{\ell}}\right),\\
      \abs*{\KillT(\phi_{\mathbf{j}})}&\le C(R_0)\sum_{\ell\in \{0,1,2,3\}^4}\abs*{\phi_{\ell}}.
    \end{split}
  \end{equation}
  We sum the inequalities in
  \Cref{eq:main-extension:Carleman-with-cutoffs} over the indices
  $\mathbf{j}\in \{0,1,2,3\}^4$. The critical observation is
  that the first three terms in the right-hand side of
  \Cref{eq:main-extension:Carleman-with-cutoffs} can be absorbed in
  the left-hand side of
  \Cref{eq:main-extension:Carleman-with-cutoffs} using
  \Cref{eq:main-extenion:MST-inequalities} for $\lambda$ sufficiently
  large. Thus, for any $\lambda\ge C(R_0)$,
  \begin{equation}
    \label{eq:main-extension:aux0}
    \begin{split}
      &\lambda\sum_{\mathbf{j}\in \{0,1,2,3\}^4}\norm*{e^{-\lambda\widetilde{f}_{\varepsilon}}\widetilde{\chi}_{\varepsilon}\phi_{\mathbf{j}}}_{L^2}
        + \norm*{e^{-\lambda\widetilde{f}_{\varepsilon}}\widetilde{\chi}_{\varepsilon}\abs*{D^1\phi_{\mathbf{j}}}}_{L^2}\\
      \lesssim{}& 
       \sum_{\mathbf{j}\in \{0,1,2,3\}^4}\left( \norm*{e^{-\lambda\widetilde{f}_{\varepsilon}}\CovariantDeriv_{\alpha}\phi_{\mathbf{j}}\CovariantDeriv^{\alpha}\widetilde{\chi}_{\varepsilon}}_{L^2}
        + \norm*{e^{-\lambda\widetilde{f}_{\varepsilon}}\phi_{\mathbf{j}}\left(\abs*{\Box_{\Metric}\widetilde{\chi}_{\varepsilon}}+ \abs*{D^1\widetilde{\chi}_{\varepsilon}}\right)}_{L^2} \right),
    \end{split}
  \end{equation}
  Using \Cref{eq:main-extension:bootstrap-assumption}, and the
  definition of the cutoff $\widetilde{\chi}_{\varepsilon}$ in
  \Cref{eq:main-extension:cutoff-def}, we have
  \begin{equation}
    \label{eq:main-extension:aux1}
    \abs*{\CovariantDeriv_{\alpha}\phi_{\mathbf{j}}\CovariantDeriv^{\alpha}\widetilde{\chi}_{\varepsilon}}
    + \phi_{\mathbf{j}} \left(\abs*{\Box_{\Metric}\widetilde{\chi}_{\varepsilon}} + \abs*{D^1\widetilde{\chi}_{\varepsilon}}\right)
    \lesssim \mathbf{1}_{\curlyBrace*{x\in B_{\varepsilon^{10}}(x_0): y(x)\ge R_0, N^{x_0}(x)\ge \varepsilon^{50}}}.
  \end{equation}
  Using the definition of $\widetilde{f}_{\varepsilon}$ in
  \Cref{eq:main-extension:f-tilde-epsilon-def}, we observe that 
  \begin{equation}
    \label{eq:main-extension:aux2}
    \inf_{B_{\varepsilon^{100}}(x_0)}e^{-\lambda\widetilde{f}_{\varepsilon}}
    \ge e^{-\lambda\ln (\varepsilon + \varepsilon^{70})}
    \ge \sup_{\curlyBrace*{x\in B_{\varepsilon^{10}}(x_0): y(x)\ge R_0, N^{x_0}(x)\ge \varepsilon^{50}}}e^{-\lambda\widetilde{f}_{\varepsilon}}.
  \end{equation}
  Then it follows by combining \Cref{eq:main-extension:aux1} and
  \Cref{eq:main-extension:aux2} with \Cref{eq:main-extension:aux0}
  that
  \begin{equation*}
    \lambda\sum_{\mathbf{j}\in \{0,1,2,3\}^4}\norm*{\mathbf{1}_{B_{\varepsilon^{100}}(x_0)}\phi_{\mathbf{j}}}_{L^2}
    \le \sum_{\mathbf{j}\in \{0,1,2,3\}^4}\norm*{\mathbf{1}_{\curlyBrace*{x\in B_{\varepsilon^{10}}(x_0):y(x)\ge R_0,N^{x_0}(x)\ge \varepsilon^{50}}}}_{L^2}
  \end{equation*}
  for any $\lambda$ sufficiently large. \Cref{prop:main-extension}
  then follows by letting $\lambda\to\infty$.
\end{proof}
We can now close the proof of \Cref{prop:main-prop}.
\begin{proof}[Proof of \Cref{prop:main-prop}]
  To prove \Cref{prop:main-prop}, we assume for the sake of
  contradiction that $R_0 \le y_{*}$, where $R_0$ is as defined in
  \Cref{eq:S-extension:BSA}.
  Then, to prove \Cref{prop:main-prop}, it suffices to advance the
  bootstrap assumption for $R' = R_0 + r',$ where $r'>0$ can only
  depend on the constants $A_0$, $\widetilde{A}_{\varepsilon}$, and $R_0$. In
  particular, $r'$ cannot depend on the point
  $x_0\in \partial_{\Sigma_0\cap \mathbf{E}} \mathcal{U}_{R_0}$.

  But using the second inclusion in
  \Cref{eq:UBar-VBar-behavior:growth}, it then follows from
  \Cref{prop:main-extension} that $\mathcal{S}$ vanishes in a small
  neighborhood of $\mathcal{U}_{R_0+\delta_{R_0}^2}$, concluding the proof
  of \Cref{prop:main-prop}.
\end{proof}

We can now prove \Cref{thm:main:e1c1}.
\begin{proof}[Proof of \Cref{thm:main:e1c1}]
  From \Cref{prop:main-prop}, we have that for some $\delta>0$
  sufficiently small, $\mathcal{S}=0$ on
  $\mathcal{U}_{y_{*}+\delta}$. On the other hand, from
  \Cref{coro:main-prop:c1}, we also have that $\mathcal{S}=0$ on
  $\underline{\mathcal{U}}_{y_{*}-\delta}$. Since from
  \Cref{eq:UBar-VBar-behavior:U-V-equal}, we have that
  $\mathcal{U}_{y_{*}+\delta} = \mathcal{V}_{y_{*}+\delta}$,
  $\underline{\mathcal{U}}_{y_{*}-\delta} =
  \underline{\mathcal{V}}_{y_{*}-\delta}$, and it is clear that
  $\Sigma_0\bigcap\mathbf{E}\subset \mathcal{V}_{y_{*}+\delta}\bigcup
  \underline{\mathcal{V}}_{y_{*}-\delta}$, we have that in fact
  $\mathcal{S}=0$ on $\Sigma_0\bigcap \mathbf{E}$. Then, Theorem 1 of
  \cite{marsSpacetimeCharacterizationKerrNUTAde2015} concludes that
  $\mathcal{M}$ must be isometrically diffeomorphic to $\Metric_{M,a,\Lambda}$,
  with $(M,a)$ given by \Cref{eq:compatibility-conditions}. 
\end{proof}

\section{Extension of \texorpdfstring{$\underline{\mathbf{K}}$}{K} using \texorpdfstring{\ref{ass:C2}}{C2}}
\label{sec:HawkingVF}

In this section, we describe the extension of the existence of a
second Killing vectorfield $\underline{\mathbf{K}}$, which will
essentially play the role of the Hawking vectorfield associated to the
cosmological horizon, to $\mathbf{E}$ using \ref{ass:C2}. The main
proposition is as follows.
\begin{prop}
  \label{prop:Hawking-extension:main}
  Under the stationary, smooth bifurcate sphere, subextremality, and
  \ref{ass:C2} assumptions of \Cref{sec:assumptions}, there exists
  some $y_{\underline{\mathbf{K}}}< y_{*}$ such that there exists a
  second Killing vectorfield $\underline{\mathbf{K}}$ in the exterior
  region $\underline{\mathbf{E}}_{y_{\underline{\mathbf{K}}}}$ such
  that in $\underline{\mathbf{E}}_{y_{\underline{\mathbf{K}}}}$,
  \begin{equation}
    \label{eq:Hawking-extension:props}
    \LieDerivative_{\underline{\mathbf{K}}}\Metric=0,\qquad
    [\KillT, \underline{\mathbf{K}}]=0,\qquad
    \underline{\mathbf{K}}\cdot \bm{\sigma}=0.
  \end{equation}
  Moreover, there exists a rotational Killing vector
  $\underline{\KillPhi}$ in the span of $\KillT$ and
  $\underline{\mathbf{K}}$.
\end{prop}

As a corollary, we have the case where one takes \ref{ass:E2} in the place of \ref{ass:C2}. The proof is almost identical.
\begin{corollary}
  \label{coro:Hawking-extension:main:E2}
  Under the stationary, smooth bifurcate sphere, subextremality, and
  \ref{ass:E2} assumptions of \Cref{sec:assumptions}, there exists
  some $y_{\mathbf{K}}> y_{*}$ such that there
  exists a second Killing vectorfield $\mathbf{K}$ on
  $\mathbf{E}_{y_{\mathbf{K}}}$ such that in
  $\mathbf{E}_{y_{\mathbf{K}}}$,
  \begin{equation}
    \label{eq:Hawking-extension:props:e2}
    \LieDerivative_{\mathbf{K}}\Metric=0,\qquad
    [\KillT, \mathbf{K}]=0,\qquad
    \mathbf{K}\cdot \bm{\sigma}=0.    
  \end{equation}
  Moreover, there exists a rotational Killing vector $\KillPhi$ in the
  span of $\KillT$ and $\mathbf{K}$.
\end{corollary}

We list below the main key steps of the extension procedure. In view
of the smallness assumption on $\mathcal{S}$, the procedure here
closely follows that of \cite{alexakisUniquenessSmoothStationary2010}.

\begin{enumerate}
\item We first construct $\underline{\mathbf{K}}$ in a lightcone
  emanating from $\underline{S}_0$. This construction is standard and
  involves solving a characteristic initial value problem. The proof
  is largely similar to the proofs in
  \cite{friedrichRigidityTheoremSpacetimes1999,
    alexakisHawkingsLocalRigidity2010}. 
\item To extend $\underline{\mathbf{K}}$ into $\mathbf{E}$, we will
  use several unique continuation arguments. To facilitate this, we
  first derive a good wave-transport system involving a
  renormalization of
  $\DeformationTensor{0}{\underline{\mathbf{K}}}{}$. This is almost
  identical to the wave-transport system derived in
  \cite{alexakisUniquenessSmoothStationary2010}, and the differences
  arise from the fact that in the current setting we are interested in
  solutions to Einstein's equation with $\Lambda>0$, while the authors
  in \cite{alexakisUniquenessSmoothStationary2010} were interested in
  $\Lambda=0$ solutions.
\item As before in the case of the extension of $\mathcal{S}$, we
  start with a unique continuation argument at the horizon. In this
  case, using the fact that the cosmological horizon has a
  null-bifurcate geometry similar to that of the event horizon, we can
  use the pseudoconvexity of the cosmological horizon to extend
  $\underline{\mathbf{K}}$ as a Killing vector in a neighborhood of
  $\underline{S}_0$.
\item We then use another unique continuation argument to extend the
  vanishing of $\DeformationTensor{0}{\underline{\mathbf{K}}}{}$ into
  $\mathbf{E}$.  To this end, we show that, under the assumption that
  $\mathcal{S}$ is small (given by \ref{ass:C2}), we can construct a
  $\KillT$-pseudoconvex foliation using the scalar quantity
  $y$. Similar to the results in \Cref{sec:MST:main-extension}, the
  $\KillT$-pseudoconvexity does not extend into the entirety of
  $\mathbf{E}$ and breaks down at the ergoregion adjacent to the event
  horizon. However, the (critical) fact that we can construct a
  $\KillT$-pseudoconvex foliation up to some
  $y_{\underline{\mathbf{K}}}<y_{*}$ is a consequence of the assumed
  subextremality of the spacetime. 
\item Finally, we reconstruct the rotational vectorfield from some
  combination of $\KillT$ and $\underline{\mathbf{K}}$.
\end{enumerate}

\begin{definition}
  Define
  \begin{equation}
    \label{eq:construction-Killing-VF:causal-regions-def}
    \begin{gathered}
      I^{++}\vcentcolon= \curlyBrace*{p\in \underline{\mathbf{O}}: \underline{u}_+(p)\ge 0, \underline{u}_-(p)\ge 0},\qquad
      I^{--}\vcentcolon= \curlyBrace*{p\in \underline{\mathbf{O}}: \underline{u}_+(p)\le 0, \underline{u}_-(p)\le 0},\\
      I^{+-}\vcentcolon= \curlyBrace*{p\in \underline{\mathbf{O}}: \underline{u}_+(p)\ge 0, \underline{u}_-(p)\le 0},\qquad
      I^{-+}\vcentcolon= \curlyBrace*{p\in \underline{\mathbf{O}}: \underline{u}_+(p)\le 0, \underline{u}_-(p)\ge 0}.
    \end{gathered}    
  \end{equation}
  Clearly, $I^{+-}$ and $I^{-+}$ coincide with the causal future and
  causal past sets of $\underline{S}_{0}$ in $\mathbf{O}$
  respectively.
\end{definition}

\subsection{Construction of the Hawking vectorfield in the causal region}
\label{sec:HawkingVF:construction}

In this section, we construct the Hawking Killing vectorfield in the
causal region $I^{+-}\bigcup I^{-+}$. This involves solving a
characteristic initial data problem. This is largely similar to known
proofs in the $\Lambda=0$ setting, but we reprove the result here for
the sake of completeness, noting that the main source of change from
the $\Lambda=0$ setting comes from the fact that the null structure
equations for $\nabla_3\Trace\chi, \nabla_4\Trace\underline{\chi}$
change when $\Lambda\neq 0$.

\begin{prop}
  \label{prop:construction-Killing-VF:causal-region}
  Under the assumptions of \Cref{thm:main}, there exists a small neighborhood $\mathbf{O}$ of
  $\underline{S}$ such that there exists a smooth Killing vectorfield
  $\underline{\mathbf{K}}$ in $\mathbf{O}\bigcap I^{+-}\bigcap I^{-+}$
  such that
  \begin{equation*}
    \underline{\mathbf{K}} = \underline{u}_+\underline{L}_- - \underline{u}_-\underline{L}_+\qquad
    (\CosmologicalHorizonFuture\bigcup \CosmologicalHorizonPast)\bigcap \underline{\mathbf{O}}.
  \end{equation*}
\end{prop}
\begin{proof}
  We construct $\underline{\mathbf{K}}$ as the solution to the following characteristic initial-value problem.
  \begin{equation}
    \label{eq:construction-Killing-VF:causal-region:characteristic-IVP}    
    \left( \Box_{\Metric}+\Lambda \right)\underline{\mathbf{K}} = 0,\qquad
    \underline{\mathbf{K}} = \underline{u}_+\underline{L}_- - \underline{u}_-\underline{L}_+,\qquad
    (\CosmologicalHorizonFuture\bigcup \CosmologicalHorizonPast)\bigcap \underline{\mathbf{O}}.
  \end{equation}
  It is well-known that the characteristic initial value problem for
  wave equations of the type in
  \Cref{eq:construction-Killing-VF:causal-region:characteristic-IVP}
  is well-posed
  \cite{lukLocalExistenceCharacteristic2012,rendallReductionCharacteristicInitial1990}. Thus,
  we have that $\underline{\mathbf{K}}$ is smooth and well-defined in
  the domain of dependence of $\CosmologicalHorizon$
  in $\underline{\mathbf{O}}$.

  To show that $\underline{\mathbf{K}}$ is Killing, we observe that
  using the Einstein vacuum equations with
  cosmological constant $\Lambda$, we have that
  \begin{align*}
    \CovariantDeriv_{\alpha}\CovariantDeriv_{\mu}\CovariantDeriv^{\mu}\underline{\mathbf{K}}_{\beta}
    ={}& \CovariantDeriv_{\mu}\CovariantDeriv_{\alpha}\CovariantDeriv^{\mu}\underline{\mathbf{K}}_{\beta}
         + \tensor[]{\Riem}{_{\alpha\mu}^{\mu}_{\nu}}\CovariantDeriv^{\nu}\underline{\mathbf{K}}_{\beta}
         + \tensor[]{\Riem}{_{\beta\mu\alpha}^{\nu}}\CovariantDeriv^{\mu}\underline{\mathbf{K}}_{\nu}\\
    ={}& \CovariantDeriv^{\mu}\CovariantDeriv_{\mu}\CovariantDeriv_{\alpha}\underline{\mathbf{K}}_{\beta}
         + \CovariantDeriv^{\mu}\tensor[]{\Riem}{_{\alpha\mu\beta}^{\nu}}\underline{\mathbf{K}}_{\nu}
         - \Lambda \CovariantDeriv_{\beta}\underline{\mathbf{K}}_{\alpha}
         +  \tensor[]{\Riem}{_{\beta\mu\alpha}^{\nu}}\CovariantDeriv^{\mu}\underline{\mathbf{K}}_{\nu}
         + \tensor[]{\Riem}{_{\alpha\mu\beta}^{\nu}}\CovariantDeriv^{\mu}\underline{\mathbf{K}}_{\nu}
         .
  \end{align*}
  Using the Bianchi identity, the Einstein equations, and symmetrizing
  over $\alpha, \beta$, we have that
  \begin{equation*}
    \Box_{\Metric}\DeformationTensor{0}{\underline{\mathbf{K}}}{_{\alpha\beta}}
    - 2\Lambda \DeformationTensor{0}{\underline{\mathbf{K}}}{_{\alpha\beta}}
    ={} \CovariantDeriv_{(\alpha}\Box_{\Metric}\underline{\mathbf{K}}_{\beta)}    
    + 2\tensor[]{\Riem}{^{\mu}_{\alpha\beta}^{\nu}}\DeformationTensor{0}{\underline{\mathbf{K}}}{_{\mu\nu}}.
  \end{equation*}
  Then using the wave equation for $\underline{\mathbf{K}}$ in \Cref{eq:construction-Killing-VF:causal-region:characteristic-IVP}, we have that 
  \begin{equation}
    \label{eq:construction-Killing-VF:causal-region:wave-eq-for-deformation-tensor}
    \left( \Box_{\Metric}-\Lambda \right)\DeformationTensor{0}{\underline{\mathbf{K}}}{_{\alpha\beta}}
    ={} 2\tensor[]{\Riem}{^{\mu}_{\alpha\beta}^{\nu}}\DeformationTensor{0}{\underline{\mathbf{K}}}{_{\mu\nu}}.
  \end{equation}
  In view of the uniqueness property for characteristic initial value
  problems, we see that to show that
  $\DeformationTensor{0}{\underline{\mathbf{K}}}{} = 0$ in
  $\underline{\mathbf{O}}\bigcap \left( I^{+-}\bigcup I^{-+} \right)$, it suffices to show that
  $\DeformationTensor{0}{\underline{\mathbf{K}}}{}$ vanishes on
  $(\CosmologicalHorizonFuture\bigcup \CosmologicalHorizonPast)\bigcap
  \underline{\mathbf{O}}$. Without loss of generality, we will just
  show that $\DeformationTensor{0}{\underline{\mathbf{K}}}{} = 0$ on
  $\CosmologicalHorizonFuture$.

  We now fix a null frame
  $(e_3,e_4) = (\underline{L}_-, -\underline{L}_+)$. 
  Observe from
  \Cref{eq:construction-Killing-VF:causal-region:characteristic-IVP}
  that $\underline{\mathbf{K}} = \underline{u}_+\underline{L}_-$ on
  $\CosmologicalHorizonFuture\bigcap \underline{\mathbf{O}}$ and is thus tangent to
  the null generators of $\CosmologicalHorizonFuture$. It thus follows that
  \begin{equation}
    \label{eq:construction-Killing-VF:causal-region:D-K-aux:future-aux}
    \CovariantDeriv_3\underline{\mathbf{K}}_4 = -2, \qquad
    \CovariantDeriv_3\underline{\mathbf{K}}_3
    = \CovariantDeriv_a\underline{\mathbf{K}}_3
    = \CovariantDeriv_3\underline{\mathbf{K}}_a
    = \CovariantDeriv_a\underline{\mathbf{K}}_b
    =0. 
  \end{equation}
  This immediately implies that on
  $\CosmologicalHorizonFuture\bigcap \underline{\mathbf{O}}$, 
  \begin{equation*}
    \DeformationTensor{0}{\underline{\mathbf{K}}}{_{33}}
    = \DeformationTensor{0}{\underline{\mathbf{K}}}{_{3a}}
    = \DeformationTensor{0}{\underline{\mathbf{K}}}{_{ab}}
    =0. 
  \end{equation*}
  We can similarly observe using the fact that
  $\underline{\mathbf{K}}=-\underline{u}_-\underline{L}_+$ on $\CosmologicalHorizonPast$ that
  \begin{equation}
    \label{eq:construction-Killing-VF:causal-region:D-K-aux:past-aux}
    \begin{gathered}
      \CovariantDeriv_4\underline{\mathbf{K}}_3 = 2, \qquad
      \CovariantDeriv_4\underline{\mathbf{K}}_4
      = \CovariantDeriv_a\underline{\mathbf{K}}_4
      = \CovariantDeriv_4\underline{\mathbf{K}}_a
      = \CovariantDeriv_a\underline{\mathbf{K}}_b
      =0\\
      \DeformationTensor{0}{\underline{\mathbf{K}}}{_{44}}
      = \DeformationTensor{0}{\underline{\mathbf{K}}}{_{4a}}
      = \DeformationTensor{0}{\underline{\mathbf{K}}}{_{ab}}
      =0. 
    \end{gathered}
  \end{equation}
  It remains to recover the vanishing of the remaining components. To
  this end, we observe that
  \Cref{eq:construction-Killing-VF:causal-region:characteristic-IVP}
  implies that on $\CosmologicalHorizonFuture\bigcap \underline{\mathbf{O}}$,
  \begin{equation*}
    \frac{1}{2}\CovariantDeriv_3\CovariantDeriv_4\underline{\mathbf{K}}
    + \frac{1}{2}\CovariantDeriv_4\CovariantDeriv_3\underline{\mathbf{K}}
    = \CovariantDeriv^a\CovariantDeriv_a\underline{\mathbf{K}}
    + \Lambda\underline{\mathbf{K}}.
  \end{equation*}

  From the definition of the Riemann curvature tensor, we have that
  $[\CovariantDeriv_3,\CovariantDeriv_4]\underline{\mathbf{K}}_{\mu} =
  \Riem_{34\mu\nu}\underline{\mathbf{K}}^{\nu}$. As a result, on $\CosmologicalHorizonFuture\bigcap \underline{\mathbf{O}}$,
  \begin{equation}
    \label{eq:construction-Killing-VF:causal-region:DD-K:aux}
    \CovariantDeriv_3\CovariantDeriv_4\underline{\mathbf{K}}_{\mu}   
    = \CovariantDeriv^a\CovariantDeriv_a\underline{\mathbf{K}}_{\mu}
    -\frac{1}{2}\Riem_{43\mu\nu}\underline{\mathbf{K}}^{\nu}
    + \underline{\mathbf{K}}_{\mu}
    .
  \end{equation}
  We first recover $\DeformationTensor{0}{\underline{\mathbf{K}}}{_{34}}$.
  Contracting \Cref{eq:construction-Killing-VF:causal-region:DD-K:aux} with $e_3$ and using \Cref{eq:construction-Killing-VF:causal-region:D-K-aux:future-aux}, we see that on $\CosmologicalHorizonFuture\bigcap \underline{\mathbf{O}}$, $\CovariantDeriv_3\CovariantDeriv_4\underline{\mathbf{K}}_3=0$. In addition, from \Cref{eq:construction-Killing-VF:causal-region:D-K-aux:past-aux}, we have also that $\CovariantDeriv_4\underline{\mathbf{K}}_3=2$ on $\underline{S}_0$. Then the fact that we assumed that the cosmological horizon is non-expanding gives us that on $\CosmologicalHorizonFuture$,
  \begin{equation*}
    \CovariantDeriv_3\CovariantDeriv_4\underline{\mathbf{K}}_3= \underline{L}_-(\CovariantDeriv_4\underline{\mathbf{K}}_3).
  \end{equation*}
  This immediately yields that $\CovariantDeriv_4\underline{\mathbf{K}}_3=2$ on all of $\CosmologicalHorizonFuture$, which now implies that on $\CosmologicalHorizonFuture$,
  \begin{equation*}
    \DeformationTensor{0}{\underline{\mathbf{K}}}{_{34}}=0.    
  \end{equation*}
  Next, contracting \Cref{eq:construction-Killing-VF:causal-region:DD-K:aux} with $e_a$ to recover $\DeformationTensor{0}{\underline{\mathbf{K}}}{_{a4}}$, we see from \Cref{eq:construction-Killing-VF:causal-region:D-K-aux:future-aux} and \Cref{prop:non-expanding-null-hypersurface:basic-props} that on $\CosmologicalHorizonFuture$, 
  \begin{equation*}
    \CovariantDeriv_a\CovariantDeriv_b\underline{\mathbf{K}}_c=0,\qquad \Riem_{34a\nu}\underline{\mathbf{K}}^{\nu}=0.
  \end{equation*}
  Again using the fact that $\underline{\mathbf{K}} = \underline{u}_+\underline{L}_-$ on $\CosmologicalHorizonFuture\bigcap \underline{\mathbf{O}}$, we have that
  \begin{equation*}
    2\DeformationTensor{0}{\underline{\mathbf{K}}}{_{a4}}
    = \CovariantDeriv_4\underline{\mathbf{K}}_a - 2\underline{u}_+\zeta_a,
  \end{equation*}
  where we used the fact that it follows from the Ricci formulas in \Cref{eq:Ricci-formulas} that
  \begin{equation*}
    \CovariantDeriv_a\underline{L}_-= \zeta_ae_3.
  \end{equation*}
  Thus, we can derive that
  \begin{align*}
    0={}&\CovariantDeriv_3\CovariantDeriv_4\underline{\mathbf{K}}_b\\
    ={}& e_3(\CovariantDeriv_4\underline{\mathbf{K}}_b)
         - \CovariantDeriv_4\underline{\mathbf{K}}_{\CovariantDeriv_3e_b}
         - \CovariantDeriv_{\CovariantDeriv_3e_4}\underline{\mathbf{K}}_b\\
    ={}& 2e_3\left(
         \DeformationTensor{0}{\underline{\mathbf{K}}}{_{b4}}
         +\underline{u}_+\zeta_b
         \right)
         - \CovariantDeriv_4\underline{\mathbf{K}}_{\nabla_3e_b}
         - \zeta_b\CovariantDeriv_4\underline{\mathbf{K}}_3\\
    ={}& 2\nabla_3\left(\DeformationTensor{0}{\underline{\mathbf{K}}}{_{b4}}
         + \underline{u}_+\zeta_b\right)
         - 2\zeta_b,\\
    ={}& 2\nabla_3\DeformationTensor{0}{\underline{\mathbf{K}}}{_{b4}}
         +2\underline{u}_+\nabla_4\zeta_b.
  \end{align*}
  On the other hand, along $\CosmologicalHorizonFuture$, combining \Cref{prop:non-expanding-null-hypersurface:basic-props} and \Cref{eq:null-structure:nab4-zeta}, we see that $\zeta$ verifies the transport equation
  \begin{equation*}
    \nabla_4\zeta_a = -\beta(\Weyl)_a = 0.
  \end{equation*}
  As a result, we in fact have that along $\CosmologicalHorizonFuture$, 
  \begin{equation*}
    \nabla_3\DeformationTensor{0}{\underline{\mathbf{K}}}{_{b4}}=0.
  \end{equation*}
  Since $\DeformationTensor{0}{\underline{\mathbf{K}}}{_{b4}}=0$ on $\underline{S}_0$, it thus follows that actually,
  \begin{equation*}
    \DeformationTensor{0}{\underline{\mathbf{K}}}{_{b4}}=0
  \end{equation*}
  along $\CosmologicalHorizonFuture$.

  Finally, to show that $\DeformationTensor{0}{\underline{\mathbf{K}}}{_{44}}$ also vanishes along $\CosmologicalHorizonFuture$, we contract \Cref{eq:construction-Killing-VF:causal-region:DD-K:aux} with $e_4$. Then, we have that
  \begin{equation}
    \label{eq:construction-Killing-VF:causal-region:D4D3K4:aux-0}
    \CovariantDeriv_3\CovariantDeriv_4\underline{\mathbf{K}}_4
    = \CovariantDeriv^a\CovariantDeriv_a\underline{\mathbf{K}}_4
    - 2\left( \rho(\Weyl)-\frac{\Lambda}{3} \right)\underline{u}_+
    -2\Lambda \underline{u}_+
    ,
  \end{equation}
  where we used that
  \begin{equation*}
    \Riem_{3434} = 4\rho(\Weyl) - \frac{4\Lambda}{3}
  \end{equation*}
  Using that we already know that
  $\DeformationTensor{0}{\underline{\mathbf{K}}}{_{b4}}$ vanishes
  along $\CosmologicalHorizonFuture$, the Ricci formulas in \Cref{eq:Ricci-formulas}, the fact that $\zeta=-\eta$ on $\CosmologicalHorizonFuture$ from \Cref{coro:non-expanding-null-hypersurface:horizon-qtys}, we have that
  \begin{align*}
    \CovariantDeriv_3\CovariantDeriv_4\underline{\mathbf{K}}_4
    ={}& e_3\left( \CovariantDeriv_4\underline{\mathbf{K}}_4 \right)
         - \CovariantDeriv_4\underline{\mathbf{K}}_{\CovariantDeriv_3e_4}
         - \CovariantDeriv_{\CovariantDeriv_3e_4}\underline{\mathbf{K}}_4\\
    ={}& e_3\DeformationTensor{0}{\underline{\mathbf{K}}}{_{44}}
         - 4\zeta^b\DeformationTensor{0}{\underline{\mathbf{K}}}{_{4b}}\\
    ={}& e_3\DeformationTensor{0}{\underline{\mathbf{K}}}{_{44}},
  \end{align*}
  so that combining with \Cref{eq:construction-Killing-VF:causal-region:D4D3K4:aux-0}, we have that
  \begin{equation}
    \label{eq:construction-Killing-VF:causal-region:D4D3K4:aux-1}
    \underline{L}_-\left( \DeformationTensor{0}{\underline{\mathbf{K}}}{_{44}} \right)
    =
    \CovariantDeriv^a\CovariantDeriv_a\underline{\mathbf{K}}_4
    -\left(2\rho(\Weyl) + \frac{4\Lambda}{3}\right)\underline{u}_+.
  \end{equation}
  On the other hand, we can also calculate that
  \begin{align*}
    \CovariantDeriv_b\CovariantDeriv_a\underline{\mathbf{K}}_4
    ={}& -e_b\left( \CovariantDeriv_a\underline{\mathbf{K}}_4 \right)
         - \CovariantDeriv_{\CovariantDeriv_be_a}\underline{\mathbf{K}}_4
         -\CovariantDeriv_a\underline{\mathbf{K}}_{\CovariantDeriv_be_4}\\
    ={}& -2e_b(\underline{u}_+\zeta_a)
         - \frac{1}{2}\chi_{ba}\CovariantDeriv_3\underline{\mathbf{K}}_4
         + \zeta_b\CovariantDeriv_a\underline{\mathbf{K}}_4\\
    ={}& -2e_b(\underline{u}_+\zeta_a)
         + \chi_{ba}
         -2\underline{u}_+\zeta_a\zeta_b.
  \end{align*}
  We thus have that
  \begin{equation*}
    \CovariantDeriv^a\CovariantDeriv_a\underline{\mathbf{K}}_4
    ={} -2\underline{u}_+(\Divergence\zeta + \abs*{\zeta}^2)
    + \Trace \chi
    .
  \end{equation*}
  Plugging this into \Cref{eq:construction-Killing-VF:causal-region:D4D3K4:aux-1}, we have that 
  \begin{equation*}
    \underline{L}_-\left(\DeformationTensor{0}{\underline{\mathbf{K}}}{_{44}}\right)
    = -2\underline{u}_+(\Divergence\zeta + \abs*{\zeta}^2)
    + \Trace \chi
    -\left(2\rho(\Weyl) + \frac{4\Lambda}{3}\right)\underline{u}_+,
  \end{equation*}
  where we used from
  \Cref{coro:non-expanding-null-hypersurface:horizon-qtys} that
  
  On the other hand, we see that the null structure equation \Cref{eq:null-structure:nab3-tr-chi} on $\CosmologicalHorizonFuture$ reduces to
  \begin{equation*}
    \begin{split}
      \nabla_3\Trace \chi      
      ={}& 2\Divergence \zeta
           + 2  \abs*{\zeta}^2 
           + 2\rho(\Weyl)
           + \frac{4\Lambda}{3}.
    \end{split}
  \end{equation*}
  Then we have that
  \begin{align*}
    \underline{L}_-\underline{L}_-(
    \DeformationTensor{0}{\underline{\mathbf{K}}}{_{44}})
    ={}& -2\Divergence \zeta
         -2\abs*{\zeta}^2
         -\left(2\rho(\Weyl) - \frac{2\Lambda}{3}\right)
         + \underline{L}_-\Trace\chi\\
       & + \underline{u}_+\underline{L}_-\left(
         -2\Divergence \zeta
         -2\abs*{\zeta}^2
         -\left(2\rho(\Weyl) - \frac{2\Lambda}{3}\right)
         \right)
         \\
    ={}& -2\Divergence \zeta
         -2\abs*{\zeta}^2
         -\left(2\rho(\Weyl) + \frac{4\Lambda}{3}\right)
         + 2\underline{u}_+\underline{L}_-\left(
         \Divergence \zeta-\abs*{\zeta}^2
         -\rho(\Weyl) 
         \right)\\
       &  + 2\Divergence \zeta
         +2\abs*{\zeta}^2
         + 2\rho(\Weyl) + \frac{4\Lambda}{3}
    \\
    ={}& 2\underline{u}_+\underline{L}_-\left(
         \Divergence \zeta-\abs*{\zeta}^2
         -\rho(\Weyl) 
         \right)
         .
  \end{align*}
  Then, we see that since $e_3 = \underline{L}_-$ is geodesic on
  $\CosmologicalHorizonFuture$, and from
  \Cref{coro:non-expanding-null-hypersurface:horizon-qtys}, we see
  that on $\CosmologicalHorizonFuture$, the equations
  \Cref{eq:null-structure:nab3-zeta} and \Cref{eq:Bianchi:nabla3-P}
  reduce to
  \begin{equation*}
    \nabla_3\zeta = 0,\qquad \nabla_3\rho =0.
  \end{equation*}
  Using \Cref{lemma:horizon-generator-hor-deriv-commutator}, this also  gives us that on $\CosmologicalHorizonFuture$,
  \begin{equation*}
    \nabla_3\Divergence\zeta = 0.
  \end{equation*}
  We therefore have that on $\CosmologicalHorizonFuture$
  \begin{equation*}
    \underline{L}_-\underline{L}_-(
    \DeformationTensor{0}{\underline{\mathbf{K}}}{_{44}})=0.
  \end{equation*}
  Since we also know that
  \begin{equation*}
    \evalAt*{\underline{L}_-(
      \DeformationTensor{0}{\underline{\mathbf{K}}}{_{44}})}_{\underline{S}_0}
    = \evalAt*{\DeformationTensor{0}{\underline{\mathbf{K}}}{_{44}}}_{\underline{S}_0}
    =0,
  \end{equation*}
  we in fact have that
  \begin{equation*}
    \DeformationTensor{0}{\underline{\mathbf{K}}}{_{44}}=0
  \end{equation*}
  on the entirety of $\CosmologicalHorizonFuture$, as desired.
\end{proof}

\begin{corollary}
  \label{coro:K-commutation-horizon-generators}
  For $\underline{\mathbf{K}}$ as constructed in
  \Cref{prop:construction-Killing-VF:causal-region}, we have that in a
  neighborhood of $\underline{S}_0$ in $I^{+-}\bigcap I^{-+}$,
  \begin{equation}
    \label{eq:K-commutation-horizon-generators}
    [\underline{L}_-,\underline{\mathbf{K}}] = -\underline{L}_-,\qquad
    [\underline{L}_+,\underline{\mathbf{K}}] = -\underline{L}_+.
  \end{equation}
\end{corollary}
\begin{proof}
  We prove the first equation in
  \Cref{eq:K-commutation-horizon-generators}. The second follows
  similarly.  Let us define
  \begin{equation*}
    W\vcentcolon= [\underline{L}_-, \underline{\mathbf{K}}] + \underline{L}_-
    = -\LieDerivative_{\underline{\mathbf{K}}}\underline{L}_- + \underline{L}_-.
  \end{equation*}
  Then, we see that to show that $W=0$ in a
  neighborhood of $\underline{S}_0$ in $I^{+-}\bigcap I^{-+}$, it suffices to show that
  \begin{gather}
    \label{eq:K-commutation-horizon-generators:initial-data}
    W=0,\qquad \text{on }\CosmologicalHorizonPast\bigcap \underline{\mathbf{O}},\\
    \label{eq:K-commutation-horizon-generators:transport}
    \CovariantDeriv_{\underline{L}_-}W=-\CovariantDeriv_W\underline{L}_-.
  \end{gather}

  Using the fact that $\underline{\mathbf{K}}$ is Killing, so
  $\LieDerivative_{\underline{\mathbf{K}}}$ commutes with covariant
  derivatives, we have that
  \begin{align*}
    \CovariantDeriv_{\underline{L}_-}W
    ={}& \CovariantDeriv_{\underline{L}_-}\left(-\LieDerivative_{\underline{\mathbf{K}}}\underline{L}_- + \underline{L}_-\right)\\
    ={}& -\CovariantDeriv_{\underline{L}_-}\LieDerivative_{\underline{\mathbf{K}}}\underline{L}_- \\
    ={}& -\LieDerivative_{\underline{\mathbf{K}}}\CovariantDeriv_{\underline{L}_-}\underline{L}_-
         + \CovariantDeriv_{\LieDerivative_{\underline{\mathbf{K}}}}\underline{L}_-\\
    ={}& -\CovariantDeriv_W\underline{L}_-,
  \end{align*}
  where we used the fact that $\underline{L}_-$ is geodesic. This
  proves \Cref{eq:K-commutation-horizon-generators:transport}.

  To show that \Cref{eq:K-commutation-horizon-generators:initial-data}
  holds, we observe that
  $\underline{\mathbf{K}}=-\underline{u}_-\underline{L}_+$ on
  $\CosmologicalHorizonPast\bigcap \underline{\mathbf{O}}$, so $W=0$
  is equivalent to 
  \begin{equation}
    \label{eq:K-commutation-horizon-generators:initial-data:reduction}
    \CovariantDeriv_3\underline{\mathbf{K}}_{\mu}
    - \underline{u}_-\CovariantDeriv_4\left( \underline{L}_- \right)_{\mu}
    + \left( \underline{L}_- \right)_{\mu} = 0.
  \end{equation}
  We check
  \Cref{eq:K-commutation-horizon-generators:initial-data:reduction}.
  We first observe that from the definition of $\zeta$ in
  \Cref{eq:ricci-components:def},
  \begin{equation*}
    \CovariantDeriv_3\underline{\mathbf{K}}_a = \CovariantDeriv_a\underline{\mathbf{K}}_3 = \underline{u}_-\Metric(\CovariantDeriv_ae_4,e_3)
    = 2 \underline{u}_-\zeta_a.
  \end{equation*}
  But from the Ricci formula in \Cref{eq:Ricci-formulas}, we also have that
  \begin{equation*}
    \CovariantDeriv_4(\underline{L}_-)_a
    = \Metric(e_a,\CovariantDeriv_4e_3)
    = -2\underline{\eta}_a.
  \end{equation*}
  But then from
  \Cref{coro:non-expanding-null-hypersurface:horizon-qtys}, we have
  that
  \begin{equation*}
    \CovariantDeriv_3\underline{\mathbf{K}}_a
    - \underline{u}_-\CovariantDeriv_4\left( \underline{L}_- \right)_a
    + \left( \underline{L}_- \right)_a
    = 2 \underline{u}_-\zeta_a
    + 2\underline{u}_-\underline{\eta}_a
    = 0.
  \end{equation*}
  Next, we observe that
  $\CovariantDeriv_3\underline{\mathbf{K}}_3=\DeformationTensor{0}{\underline{\mathbf{K}}}{_{33}}
  = 0$ in view of \Cref{prop:construction-Killing-VF:causal-region}.
  We also have from the Ricci formula in \Cref{eq:Ricci-formulas} that
  \begin{equation*}
    D_4(\underline{L}_-)_3 = \Metric(e_3,\CovariantDeriv_4e_3)
    = 0.
  \end{equation*}
  Thus, we have that
  \begin{equation*}
    \CovariantDeriv_3\underline{\mathbf{K}}_3
    - \underline{u}_-\CovariantDeriv_4\left( \underline{L}_- \right)_3
    + \left( \underline{L}_- \right)_3
    = 0.
  \end{equation*}
  Finally, we see that
  $\CovariantDeriv_3\underline{\mathbf{K}}_4 =
  -\CovariantDeriv_4\underline{\mathbf{K}}_3 = -2$ on
  $\CosmologicalHorizonPast\bigcap \underline{\mathbf{O}}$ from
  \Cref{eq:construction-Killing-VF:causal-region:D-K-aux:past-aux}. On
  the other hand, from the Ricci formula in \Cref{eq:Ricci-formulas},
  \begin{equation*}
    \CovariantDeriv_4( \underline{L}_- )_4 = \Metric(e_4, \CovariantDeriv_4e_3)=0.
  \end{equation*}
  Finally, from the fact that
  $\Metric(e_4,\underline{L}_-) = \Metric(e_4,e_3) = -2$, we have that
  \begin{equation*}
    \CovariantDeriv_3\underline{\mathbf{K}}_4
    - \underline{u}_-\CovariantDeriv_4\left( \underline{L}_- \right)_4
    + \left( \underline{L}_- \right)_4
    = 2 - 2 = 0,
  \end{equation*}
  as desired.
\end{proof}
 
\subsection{Wave transport system}

In this section, we derive the main wave-transport system that will be
used to extend the Hawking vectorfield. 

We first define the zeroth-order and first-order deformation tensors.
\begin{definition}
  \label{def:DefTens}
  Let $X$ be a smooth vectorfield on $\mathcal{M}$. 
  We define the \emph{(zeroth-order) deformation tensor of $X$} by
  \begin{equation}
    \label{eq:DefTens-order-0}
    \DeformationTensor{0}{X}{_{\alpha\beta}}
    \vcentcolon= \frac{1}{2}\left(\CovariantDeriv_{\alpha}X_{\beta} + \CovariantDeriv_{\beta}X_{\alpha}\right)
    = \frac{1}{2}\LieDerivative_X\Metric_{\alpha\beta}.
  \end{equation}
  We also define the \emph{first-order deformation tensor of $X$} by 
  \begin{equation}
    \label{eq:DefTens-order-1}
    \DeformationTensor{1}{X}{_{\alpha\beta\gamma}}
    \vcentcolon= 
      \CovariantDeriv_{\beta}\DeformationTensor{0}{X}{_{\alpha\gamma}}
      + \CovariantDeriv_{\alpha}\DeformationTensor{0}{X}{_{\beta\gamma}}
      - \CovariantDeriv_{\gamma}\DeformationTensor{0}{X}{_{\alpha\beta}}
    .
  \end{equation}
\end{definition}

When extending the Hawking vectorfield, we will not be able to extend
just the vanishing of its deformation tensor. To find a suitable
wave-transport system of equations, we need several auxiliary
quantities that we now define.
\begin{definition}
  \label{def:VF-extension:tensor-qtys-def}
  Let $\omega$ be a specific two-form that will be determined later.
  Then define
  \begin{align}
    \label{eq:VF-extension:tensor-qtys-def:wave-qty}
    \LieZRiem &\vcentcolon= \LieDerivative_{\underline{\mathbf{K}}}\Riem - B\odot \Riem,\\
    \label{eq:VF-extension:tensor-qtys-def:B}
    B &\vcentcolon= \DeformationTensor{0}{\underline{\mathbf{K}}}{} + \omega,\\
    \label{eq:VF-extension:tensor-qtys-def:B-dot}
    \dot{B} &\vcentcolon= \CovariantDeriv_{\GeodesicVF}B,\\
    \label{eq:VF-extension:tensor-qtys-def:P}
    P_{\mu\nu\sigma} &\vcentcolon= \CovariantDeriv_\mu\DeformationTensor{0}{\underline{\mathbf{K}}}{_{\nu\sigma}}
              - \CovariantDeriv_\sigma\DeformationTensor{0}{\underline{\mathbf{K}}}{_{\mu\nu}}
              - \CovariantDeriv_\nu\DeformationTensor{0}{\underline{\mathbf{K}}}{_{\mu\sigma}},
  \end{align}
  where by $\odot$ we refer to the Nomizu product of a $(0,2)$-tensor
  and a $(0,4)$-tensor, so that
  \begin{equation*}
    (B\odot \Riem)_{\alpha\beta\mu\nu}
    ={} \tensor[]{B}{_{\alpha}^{\sigma}}\Riem_{\sigma\beta\mu\nu}
    + \tensor[]{B}{_{\beta}^{\sigma}}\Riem_{\alpha\sigma\mu\nu}
    + \tensor[]{B}{_{\mu}^{\sigma}}\Riem_{\alpha\beta\sigma\nu}
    + \tensor[]{B}{_{\nu}^{\sigma}}\Riem_{\alpha\beta\mu\sigma}.
  \end{equation*}
\end{definition}

We begin with some basic computations for a vacuum cosmological spacetime.
\begin{lemma}
  \label{lemma:Riem-Ric-basic-eqns}
  In a vacuum cosmological spacetime, the following relations hold
  true.
  \begin{align}
    \CovariantDeriv^\alpha\Riem_{\alpha\beta\mu\nu}
    ={}& 0, \label{eq:basic-eqns:div-Riem}\\
    \Box_g\Ric_{\mu\nu}={}& 0, \label{eq:basic-eqns:wave-Ric}\\
      \left(\Box_g - 2\Lambda\right)\Riem_{\alpha\beta\mu\nu}
    ={}& \tensor[]{\Riem}{^\gamma_{\alpha\beta}^\lambda}\Riem_{\gamma\lambda \mu\nu}
    - \tensor[]{\Riem}{^\gamma_{\beta\alpha}^\lambda}\Riem_{\gamma\lambda \mu\nu}
    + \tensor[]{\Riem}{^\gamma_{\alpha\mu}^\lambda}\Riem_{\beta\gamma \nu\lambda} \notag\\
    & - \tensor[]{\Riem}{^\gamma_{\alpha\nu}^\lambda}\Riem_{\beta\gamma \mu\lambda}
    + \tensor[]{\Riem}{^\gamma_{\beta\mu}^\lambda}\Riem_{\gamma \alpha\nu\lambda}
    - \tensor[]{\Riem}{^\gamma_{\beta\nu}^\lambda}\Riem_{\gamma \alpha\mu\lambda}.
                           \label{eq:basic-eqns:wave-Riem}
  \end{align}
\end{lemma}
\begin{proof}
  We first prove \cref{eq:basic-eqns:div-Riem}
  \begin{align*}
    \CovariantDeriv^\alpha\Riem_{\alpha\beta\mu\nu}
    ={}& \CovariantDeriv_\mu\Riem_{\beta\nu} - \CovariantDeriv_\nu\Riem_{\beta\mu}\\
    ={}& 0 ,
  \end{align*}
  where we used ($\Lambda$-EVE) in the second equality.

  \Cref{eq:basic-eqns:wave-Ric} follows directly from ($\Lambda$-EVE).

  To prove the wave equation for $\Riem$ in
  \Cref{eq:basic-eqns:wave-Riem}, we use the second Bianchi identity
  and the divergence property in
  \cref{eq:basic-eqns:div-Riem}. Differentiating and contracting the
  second Bianchi identity, we get
  \begin{equation*}    
    \Box_g\Riem_{\alpha\beta\mu\nu}
    + \CovariantDeriv^\gamma\CovariantDeriv_\alpha\Riem_{\beta \gamma\mu\nu}
    + \CovariantDeriv^\gamma\CovariantDeriv_\beta\Riem_{\gamma\alpha \mu\nu}
    = 0.
  \end{equation*}
  Commuting the derivatives in the second and third terms and using
  the divergence property in \cref{eq:basic-eqns:div-Riem}, we have that
  \begin{align*}
    \Box_g\Riem_{\alpha\beta\mu\nu}
    ={}& \tensor[]{\Ric}{_\alpha ^\gamma}\Riem_{\beta\gamma \nu \mu}
         - \tensor[]{\Ric}{_\beta^\gamma }\Riem_{\gamma \alpha \mu \nu}
         + \tensor[]{\Riem}{^\lambda_{\alpha \beta}^\gamma }\Riem_{\lambda\gamma \mu\nu}
         - \tensor[]{\Riem}{^\lambda_{\beta\alpha }^\gamma }\Riem_{\lambda\gamma \mu\nu}\\
     & + \tensor[]{\Riem}{^\lambda_{\alpha \mu}^\gamma }\Riem_{\beta \lambda\nu\gamma }
       - \tensor[]{\Riem}{^\lambda_{\alpha \nu}^\gamma }\Riem_{\beta \lambda\mu\gamma }
       + \tensor[]{\Riem}{^\lambda_{\beta \mu}^\gamma }\Riem_{i\alpha \nu\gamma }
       - \tensor[]{\Riem}{^\lambda_{\beta \nu}^\gamma }\Riem_{i\alpha \mu\gamma }.
  \end{align*}
  Then, using \EVE, we have that
  \begin{align*}
    \Box_g\Riem_{\alpha\beta\mu\nu}
    ={}& \Lambda \Riem_{\beta\alpha\nu\mu}
         - \Lambda \Riem_{\beta\alpha\mu\nu}
         + \tensor[]{\Riem}{^\lambda_{\alpha \mu }^\gamma }\Riem_{\lambda \gamma \mu\nu}
         - \tensor[]{\Riem}{^\lambda_{\mu \alpha }^\gamma }\Riem_{\lambda \gamma \mu\nu}\\
       & + \tensor[]{\Riem}{^\lambda_{\alpha\mu}^\gamma }\Riem_{\mu \lambda \nu\gamma}
         - \tensor[]{\Riem}{^\lambda_{\alpha\nu}^\gamma }\Riem_{\mu \lambda \mu\gamma}
         + \tensor[]{\Riem}{^\lambda_{\mu\mu}^\gamma }\Riem_{\lambda \alpha \nu\gamma}
         - \tensor[]{\Riem}{^\lambda_{\mu\nu}^\gamma }\Riem_{\lambda \alpha \mu\gamma}\\
    ={}& 2\Lambda \Riem_{\mu \alpha \nu\mu}
         + \tensor[]{\Riem}{^\lambda_{\alpha \mu }^\gamma }\Riem_{\lambda \gamma \mu\nu}
         - \tensor[]{\Riem}{^\lambda_{\mu \alpha }^\gamma }\Riem_{\lambda \gamma \mu\nu}\\
       & + \tensor[]{\Riem}{^\lambda_{\alpha\mu}^\gamma }\Riem_{\mu \lambda \nu\gamma}
         - \tensor[]{\Riem}{^\lambda_{\alpha\nu}^\gamma }\Riem_{\mu \lambda \mu\gamma}
         + \tensor[]{\Riem}{^\lambda_{\mu\mu}^\gamma }\Riem_{\lambda \alpha \nu\gamma}
         - \tensor[]{\Riem}{^\lambda_{\mu\nu}^\gamma }\Riem_{\lambda \alpha \mu\gamma}.
  \end{align*}
  Moving the first term on the \RHS{} to the \LHS, we then have that
  \begin{align*}
    \left(\Box_g - 2\Lambda\right)\Riem_{\alpha\beta\mu\nu}
    ={}& \tensor[]{\Riem}{^\lambda_{\alpha\beta}^\gamma}\Riem_{\lambda \gamma \mu\nu}
    - \tensor[]{\Riem}{^\lambda_{\beta \alpha }^\gamma}\Riem_{\lambda \gamma \mu\nu}
    + \tensor[]{\Riem}{^\lambda_{\alpha \mu}^\gamma}\Riem_{\beta \lambda \nu\gamma}\\
    & - \tensor[]{\Riem}{^\lambda_{\alpha\nu}^\gamma}\Riem_{\beta \lambda \mu\gamma}
    + \tensor[]{\Riem}{^\lambda_{\beta \mu}^\gamma}\Riem_{\lambda \alpha \nu\gamma}
    - \tensor[]{\Riem}{^\lambda_{\beta\nu}^\gamma}\Riem_{\lambda \alpha \mu\gamma},
  \end{align*}
  as desired.
\end{proof}

We now prove the main equations that we will use subsequently.
\begin{lemma}
  \label{lemma:VF-extension:wave}
  For $\LieZRiem$ as defined in
  \Cref{eq:VF-extension:tensor-qtys-def:wave-qty}, we have that
  \begin{equation*}
    \Box_{\Metric}\LieZRiem
    = \AdmissibleRHS(B, \dot{B}, P, \CovariantDeriv B,\CovariantDeriv \dot{B},\CovariantDeriv P, \LieZRiem).
  \end{equation*}
\end{lemma}

\begin{proof}
  Observe that
  \begin{align*}
    &\left(\Box_g -2\Lambda\right) \LieZRiem_{\alpha\beta\mu\nu }\\
    ={}& \left(\Box_g -2\Lambda\right)\left(
         \LieDerivative_{\underline{\mathbf{K}}}\Riem_{\alpha\beta\mu\nu}
         - B\odot \Riem_{\alpha\beta\mu\nu}\right)\\
    ={}& \CovariantDeriv^\rho\left(\LieDerivative_{\underline{\mathbf{K}}} \CovariantDeriv_\rho\Riem_{\alpha\beta\mu\nu}
         - \CovariantDeriv_\rho\left(B\odot \Riem\right)_{\alpha\beta\mu\nu}
         + \DeformationTensor{1}{\underline{\mathbf{K}}}{_{\alpha\rho\sigma}}\tensor[]{\Riem}{^\sigma_{\beta\mu\nu}}
         + \DeformationTensor{1}{\underline{\mathbf{K}}}{_{\beta\rho\sigma}}\tensor[]{\Riem}{_{\alpha}^\sigma_{\mu\nu}}\right)\\
       &  + \CovariantDeriv^\rho\left(\DeformationTensor{1}{\underline{\mathbf{K}}}{_{\mu\rho\sigma}}\tensor[]{\Riem}{_{\alpha\beta}^\sigma_{\nu}}
         + \DeformationTensor{1}{\underline{\mathbf{K}}}{_{\nu\rho\sigma}}\tensor[]{\Riem}{_{\alpha\beta \mu}^\sigma}
         \right)
         - 2\Lambda \LieDerivative_{\underline{\mathbf{K}}}\Riem_{\alpha\beta\mu\nu}
         + 2\Lambda B \odot \Riem_{\alpha\beta\mu\nu}\\
    ={}&\LieDerivative_{\underline{\mathbf{K}}}\left(\Box_g -2\Lambda\right) \Riem_{\alpha\beta\mu\nu}
         + \DeformationTensor{1}{\underline{\mathbf{K}}}{_\rho^\rho_\sigma}\CovariantDeriv^\sigma\tensor[]{\Riem}{_{\alpha\beta\mu\nu}}
         + \DeformationTensor{1}{\underline{\mathbf{K}}}{_\alpha^\rho_\sigma}\CovariantDeriv_\rho\tensor[]{\Riem}{^\sigma_{\beta\mu\nu}}
         + \DeformationTensor{1}{\underline{\mathbf{K}}}{_\nu^\rho_\sigma}\CovariantDeriv_\rho\tensor[]{\Riem}{_\alpha^\sigma_{\mu\nu}}\\
       &  + \DeformationTensor{1}{\underline{\mathbf{K}}}{_\mu^\rho_\sigma}\CovariantDeriv_\rho\tensor[]{\Riem}{_{\alpha\beta}^\sigma_\nu}
         + \DeformationTensor{1}{\underline{\mathbf{K}}}{_\nu^\rho_\sigma}\CovariantDeriv_\rho\tensor[]{\Riem}{_{\alpha\beta\mu}^\sigma}
        - \left(\Box_g -2\Lambda\right)\left(B\odot \Riem\right)_{\alpha\beta\mu\nu}
         + \CovariantDeriv^\rho\DeformationTensor{1}{\underline{\mathbf{K}}}{_{\alpha\rho\sigma}}\tensor[]{\Riem}{^\sigma_{\beta\mu\nu}}\\
       & + \CovariantDeriv^\rho\DeformationTensor{1}{\underline{\mathbf{K}}}{_{\beta\rho\sigma}}\tensor[]{\Riem}{_{\alpha}^\sigma_{\mu\nu}}
         + \CovariantDeriv^\rho\DeformationTensor{1}{\underline{\mathbf{K}}}{_{\mu\rho\sigma}}\tensor[]{\Riem}{_{\alpha\beta}^\sigma_{\nu}}
         + \CovariantDeriv^\rho\DeformationTensor{1}{\underline{\mathbf{K}}}{_{\nu\rho\sigma}}\tensor[]{\Riem}{_{\alpha\beta\mu}^\sigma}.
  \end{align*}
  Using the fact that $\Gamma = \AdmissibleRHS(B, \CovariantDeriv B)$, we can write
  \begin{equation*}
    \left(\Box_g - 2\Lambda\right) \LieZRiem_{ijkl}
    ={} \LieDerivative_{\underline{\mathbf{K}}}\left(\Box_g - 2\Lambda\right)\Riem_{ijkl}
    - \left(\Box_g - 2\Lambda\right)\left(B\odot \Riem\right)_{ijkl}
    + \AdmissibleRHS(B, \CovariantDeriv P, \LieZRiem).
  \end{equation*}
  Moreover, observe that
  \begin{equation*}
    \LieDerivative_{\underline{\mathbf{K}}}\left(\Box_g - 2\Lambda\right)\Riem
    = \LieDerivative_{\underline{\mathbf{K}}}\AdmissibleRHS(\Riem)
    = \AdmissibleRHS(\LieDerivative_{\underline{\mathbf{K}}}\Riem)
    = \AdmissibleRHS(B, \LieZRiem).
  \end{equation*}
  Then we observe that by using \Cref{eq:basic-eqns:wave-Riem}, we can
  write that
  \begin{align*}
    \Box_g\left(B\odot \Riem\right)_{\alpha\beta\mu\nu}
    ={}& \CovariantDeriv^\rho\CovariantDeriv_\rho\left(
         \tensor[]{B}{_\alpha^\sigma}\Riem_{\sigma\beta\mu\nu}
         + \tensor[]{B}{_\beta^\sigma}\Riem_{\alpha\sigma\mu\nu}
         + \tensor[]{B}{_\mu^\sigma}\Riem_{\alpha\beta\sigma\nu}
         + \tensor[]{B}{_\nu^\sigma}\Riem_{\alpha\beta\mu\sigma}
         \right) \\
    ={}& \CovariantDeriv^\rho\left(
         \CovariantDeriv_\rho\tensor[]{B}{_\alpha^\sigma}\Riem_{\sigma\beta\mu\nu}
         + \tensor[]{B}{_\alpha^\sigma}\CovariantDeriv_\rho\Riem_{\sigma\beta\mu\nu}
         + \CovariantDeriv_\rho\tensor[]{B}{_\beta^\sigma}\Riem_{\alpha\sigma\mu\nu}
         + \tensor[]{B}{_\beta^\sigma}\CovariantDeriv_\rho\Riem_{\alpha\sigma\mu\nu}\right)\\
       &+ \CovariantDeriv^\rho\left(\CovariantDeriv_\rho\tensor[]{B}{_\mu^\sigma}\Riem_{\alpha\mu\sigma\nu}
         + \tensor[]{B}{_\mu^\sigma}\CovariantDeriv_\rho\Riem_{\alpha\beta\sigma\nu}
         + \CovariantDeriv_\rho\tensor[]{B}{_\nu^\sigma}\Riem_{\alpha\beta\mu\sigma}
         + \tensor[]{B}{_\nu^\sigma}\CovariantDeriv_\rho\Riem_{\alpha\beta\mu\sigma}
         \right)\\
    ={}& \Box_g\tensor[]{B}{_\alpha^\sigma}\Riem_{\sigma\beta\mu\nu}
         + \Box_g\tensor[]{B}{_\beta^\sigma}\Riem_{\alpha\sigma\mu\nu}
         + \Box_g\tensor[]{B}{_\mu^\sigma}\Riem_{\alpha\beta\sigma\nu}
         + \Box_g\tensor[]{B}{_\nu^\sigma}\Riem_{\alpha\beta\mu\sigma}\\
       & + \tensor[]{B}{_\alpha^\sigma}\Box_g\Riem_{\sigma\beta\mu\nu}
         + \tensor[]{B}{_\beta^\sigma}\Box_g\Riem_{\alpha\sigma\mu\nu}
         + \tensor[]{B}{_\mu^\sigma}\Box_g\Riem_{\alpha\beta\sigma\nu}
         + \tensor[]{B}{_\nu^\sigma}\Box_g\Riem_{\alpha\beta\mu\sigma}\\
       & +  2\CovariantDeriv^\rho\tensor[]{B}{_\alpha^\sigma}\CovariantDeriv_\rho\Riem_{\sigma\beta\mu\nu}
         + 2\CovariantDeriv^\rho\tensor[]{B}{_\beta^\sigma}\CovariantDeriv_\rho\Riem_{\alpha\sigma\mu\nu}
         + 2\CovariantDeriv^\rho\tensor[]{B}{_\mu^\sigma}\CovariantDeriv_\rho\Riem_{\alpha\beta\sigma\nu}
         + 2\CovariantDeriv^\rho\tensor[]{B}{_\nu^\sigma}\CovariantDeriv_\rho\Riem_{\alpha\beta\mu\sigma}\\
    ={}& \Box_g\tensor[]{B}{_\alpha^\sigma}\Riem_{\sigma\beta\mu\nu}
         + \Box_g\tensor[]{B}{_\beta^\sigma}\Riem_{\alpha\sigma\mu\nu}
         + \Box_g\tensor[]{B}{_\mu^\sigma}\Riem_{\alpha\beta\sigma\nu}
         + \Box_g\tensor[]{B}{_\nu^\sigma}\Riem_{\alpha\beta\mu\sigma}
     + \AdmissibleRHS\left(B, \CovariantDeriv B\right).
  \end{align*}
  Thus, we have that
  \begin{align*}
    \left(\Box_g - 2\Lambda\right)\LieZRiem_{\alpha\beta\mu\nu }
    ={}&  \CovariantDeriv_\rho\left(
         \DeformationTensor{1}{\underline{\mathbf{K}}}{_\alpha^\rho_\sigma}
         - \CovariantDeriv^\rho B_{\alpha\sigma}
         \right)\tensor[]{\Riem}{^\sigma_{\beta\mu\nu}}
         + \CovariantDeriv_\rho\left(
         \DeformationTensor{1}{\underline{\mathbf{K}}}{_\nu^\rho_\sigma}
         - \CovariantDeriv^\rho B_{\beta\sigma}
         \right)\tensor[]{\Riem}{_\alpha^\sigma_{\mu\nu}}\\
        & + \CovariantDeriv_\rho\left( 
         \DeformationTensor{1}{\underline{\mathbf{K}}}{_\mu^\rho_\sigma}
         - \CovariantDeriv^\rho B_{\mu\sigma}
         \right)\tensor[]{\Riem}{_{\mu\nu}^\sigma_\nu}
        + \CovariantDeriv_\rho\left(
         \DeformationTensor{1}{\underline{\mathbf{K}}}{_\nu^\rho_\sigma}
         - \CovariantDeriv^\rho B_{\nu\sigma}
         \right)\tensor[]{\Riem}{_{\alpha\beta\mu}^\sigma}
         + \AdmissibleRHS\left(B, \CovariantDeriv B, \LieZRiem\right). 
  \end{align*}
  From the definition of $P$, we have that
  \begin{equation*}
    P_{\mu\nu\rho} = \DeformationTensor{1}{\underline{\mathbf{K}}}{_{\mu\nu\rho}} - \CovariantDeriv_\nu B_{\mu\rho}.
  \end{equation*}
  As a result, we have that
  \begin{align*}
    \left(\Box_g - 2\Lambda\right)\LieZRiem_{\alpha\beta\mu\nu}
    ={}& \CovariantDeriv^\rho P_{\alpha\rho\sigma}\tensor[]{\Riem}{^\sigma_{\beta\mu\nu}}
         + \CovariantDeriv^\rho P_{\beta\rho\sigma}\tensor[]{\Riem}{_\alpha^\sigma_{\mu\nu}}\\
       & + \CovariantDeriv^\rho P_{\mu\rho\sigma}\tensor[]{\Riem}{_{\alpha\beta}^\sigma_\nu}
         + \CovariantDeriv^\rho P_{\nu\rho\sigma}\tensor[]{\Riem}{_{\alpha\beta\mu}^\sigma}
         + \AdmissibleRHS\left(B, \CovariantDeriv B, \LieZRiem\right)\\
    ={}& \AdmissibleRHS\left(B, \CovariantDeriv B, \CovariantDeriv P, \LieZRiem\right),
  \end{align*}
  as desired. 
\end{proof}

\begin{lemma}
  \label{lemma:VF-extension:transport}
  Given the vectorfield $\underline{\mathbf{K}}$, extended to
  $\mathcal{M}$ by
  \begin{equation*}
    \CovariantDeriv_{\GeodesicVF}\CovariantDeriv_{\GeodesicVF}\underline{\mathbf{K}}
    = \Riem(\GeodesicVF, \underline{\mathbf{K}})\GeodesicVF.
  \end{equation*}
  Then we have that
  \begin{equation*}
    \GeodesicVF\cdot\DeformationTensor{0}{\underline{\mathbf{K}}}{} = 0\qquad \text{in } \Manifold. 
  \end{equation*}
  Moreover, defining $\omega$ as the solution of the transport equation
  \begin{equation}
    \label{eq:VF-extension:transport:omega-def}
    \CovariantDeriv_{\GeodesicVF}\omega_{\mu\nu} = \DeformationTensor{0}{\underline{\mathbf{K}}}{_{\mu\rho}}\CovariantDeriv_\nu\GeodesicVF^\rho - \DeformationTensor{0}{\underline{\mathbf{K}}}{_{\nu\rho}}\CovariantDeriv_\mu\GeodesicVF^\rho,
  \end{equation}
  with $\omega=0$ in $\underline{\mathbf{O}}$, then
  \begin{equation*}
    \GeodesicVF^\rho P_{\mu\nu\rho} = 0,\qquad
    \GeodesicVF^\mu\omega_{\mu\nu}=0,\qquad
    \text{in }\Manifold.
  \end{equation*}
  Moreover, in $\Manifold$ we have that
  we have that
  \begin{align*}
    \CovariantDeriv_{\GeodesicVF}B_{\mu\nu} ={}& \dot{B}_{\mu\nu},\\
    \CovariantDeriv_{\GeodesicVF}\dot{B}_{\mu\nu}={}& \GeodesicVF^\sigma\GeodesicVF^\rho\left(\LieDerivative_{\underline{\mathbf{K}}}\Riem\right)_{\sigma\mu\nu\rho}
                        - 2\dot{B}_{\rho\nu}\CovariantDeriv_\mu\GeodesicVF^\rho
                        - \DeformationTensor{0}{\underline{\mathbf{K}}}{_\nu^\rho}\GeodesicVF^\sigma\GeodesicVF^\alpha\Riem_{\sigma\mu\rho \alpha}\\
    \CovariantDeriv_{\GeodesicVF}P_{\mu\nu\sigma} ={}& 2\GeodesicVF^\rho \LieZRiem_{\mu\nu\sigma\rho}
                    + 2\GeodesicVF^\rho\tensor[]{B}{_\sigma^\gamma}\Riem_{\mu\nu\gamma\rho}
                    - \CovariantDeriv_\sigma\GeodesicVF^\rho P_{\mu\nu\rho}.
  \end{align*}
\end{lemma}
\begin{proof}
  See the proofs of Lemma 2.6 and Proposition 2.7 in
  \cite{ionescuLocalExtensionKilling2013}. The proofs in particular
  are independent of the value of $\Lambda$, so do not change in our
  current setting.
\end{proof}
\begin{corollary}
  \label{coro:VF-extension:wave-transport}
  Under the assumptions of \Cref{lemma:VF-extension:transport}, the
  quantities $\LieZRiem$, $B$, $\dot{B}$, and $P$ as defined in
  \Cref{def:VF-extension:tensor-qtys-def} satisfy the following system
  of equations
  \begin{equation}
    \label{eq:VF-extension:wave-transport}
    \begin{split}
      \Box_{\Metric}\LieZRiem
      ={}& \AdmissibleRHS(B, \dot{B}, P, \CovariantDeriv B,\CovariantDeriv \dot{B},\CovariantDeriv P, \LieZRiem),\\
      \CovariantDeriv_{\GeodesicVF}B ={}& \AdmissibleRHS(\dot{B}),\\
      \CovariantDeriv_{\GeodesicVF}\dot{B}={}& \AdmissibleRHS(B, \dot{B}, \LieZRiem),\\
      \CovariantDeriv_{\GeodesicVF}P ={}& \AdmissibleRHS(B, P, \LieZRiem).
    \end{split}
  \end{equation}
\end{corollary}
\begin{proof}
  Combining the results of \Cref{lemma:VF-extension:wave} and \Cref{lemma:VF-extension:transport} directly yield the result. 
\end{proof}

\subsection{Extension of the Hawking vectorfield in a neighborhood of
  the cosmological horizon}

In this section, we show that the Hawking vectorfield
$\underline{\mathbf{K}}$ extends to a neighborhood of the cosmological
horizon. This follows a similar approach to
\cite{alexakisHawkingsLocalRigidity2010,
  alexakisUniquenessSmoothStationary2010,
  ionescuLocalExtensionKilling2013}.

\begin{prop}
  \label{prop:Hawking-extension:nbhd-horizon}
  There exists a neighborhood $\underline{\mathbf{O}}'$ of the bifurcation sphere $\underline{S}_0$ and a constant $\underline{\varepsilon}>0$ such that $\underline{\mathbf{O}}\subset \underline{\mathbf{O}}'$ and a smooth vectorfield $\underline{\mathbf{K}}$ in $\underline{\mathbf{O}}'$ such that $\underline{\mathbf{K}} = \underline{u}_+\underline{L}_- - \underline{u}_-\underline{L}_+$ on $\CosmologicalHorizon\bigcap \underline{\mathbf{O}}'$, and
  \begin{equation}
    \label{eq:Hawking-extension:nbhd-horizon:props}
    \begin{split}
      \LieDerivative_{\underline{\mathbf{K}}}\Metric=0,\qquad
      [\KillT, \underline{\mathbf{K}}] = 0,\qquad
      \underline{\mathbf{K}}\cdot\bm{\sigma}=0,\qquad \text{in }\underline{\mathbf{O}}',
    \end{split}
  \end{equation}
  and
  \begin{equation}
    \label{eq:Hawking-extension:nbhd-horizon:timelike}
    \begin{split}
      \Metric(\underline{\mathbf{K}}, \underline{\mathbf{K}})\le -\underline{c}(2-r)^2\qquad \text{on }\underline{\Sigma}_2\bigcap \underline{\mathbf{O}}'.
    \end{split}
  \end{equation}
\end{prop}
\begin{remark}
  Observe that \Cref{prop:Hawking-extension:nbhd-horizon} implies that
  there exists some $\underline{y}_0<y_{\underline{S}_0}$ such that
  $\underline{\mathbf{K}}$ as described in
  \Cref{prop:Hawking-extension:nbhd-horizon} exists on
  $\underline{\mathcal{U}}_{\underline{y}_0}$.
\end{remark}
We will then extend $\underline{\mathbf{K}}$ as a solution to
\begin{equation} 
  \label{eq:HawkingVF-extension-eqn}
  \CovariantDeriv_{\underline{L}_-}\CovariantDeriv_{\underline{L}_-}\underline{\mathbf{K}} =
  \Riem(\underline{L}_-,\underline{\mathbf{K}})\underline{L}_-
\end{equation}

\begin{lemma}
  \label{lemma:VF-extension:nbhd-of-horizon}
  There exists some $\underline{r}_1(A)>0$ such that $\underline{\mathbf{K}}$ extends as a Killing vectorfield in $\underline{\mathbf{O}}_{\underline{r}_1}\cap \mathbf{E}$. 
\end{lemma}
\begin{proof}
  We introduce, $G_i = \LieZRiem$, $H_j = B, \CovariantDeriv B, \dot{B}, \CovariantDeriv \dot{B}, P, \CovariantDeriv P$, so that \Cref{eq:VF-extension:wave-transport} can be rewritten as
  \begin{equation}
    \label{eq:VF-extension:nbhd-of-horizon:wave-transport-system}
    \begin{split}
      \Box_{\Metric}G_i &= \AdmissibleRHS(G_I, H_J),\\
      \CovariantDeriv_{\underline{L}_-}H_j &= \AdmissibleRHS(G_I, H_J), 
    \end{split}
  \end{equation}
  where we set $\GeodesicVF = \underline{L}_-$
  .  Fix
  some $x_0\in \underline{S}_0$ and define
  \begin{equation*}
    h_{\varepsilon} = \varepsilon^{-1}(\underline{u}_++\varepsilon)(\underline{u}_-+\varepsilon), \qquad
    e_{\varepsilon} = \varepsilon^{10}N^{x_0},
  \end{equation*}
  where $(\underline{u}_+,\underline{u}_-)$ are the optical functions
  defined in \Cref{eq:geodesic-u-uBar-qtys}, and $N^{x_0}$ is as
  defined in \Cref{eq:N-x0:def}.

  It is clear that $e_{\varepsilon}$ is a negligible perturbation as
  defined in \Cref{def:negligible perturbation} for $\varepsilon$
  sufficiently small. Moreover, it is easy to verify that
  \Cref{eq:Carleman-estimate:transport:conditions} is verified for
  $\varepsilon$ sufficiently small and $\GeodesicVF = \underline{L}_-$.

  Then from the Carleman estimates in \Cref{prop:carleman-estimate}
  and \Cref{lemma:Carleman-estimate:transport}, since we have that the
  cosmological horizon is pseudoconvex from
  \Cref{lemma:nbhd-horizon:pseudo-convexity}, we have that there is
  some $\varepsilon=\varepsilon(A_0)$ such that
  \begin{equation}
    \label{eq:VF-extension:horizon:Carleman-estimates}
    \begin{split}
      \lambda\norm*{e^{-\lambda f_{\varepsilon}}\phi}_{L^2}
      + \norm*{e^{-\lambda f_{\varepsilon}}\abs*{D\phi}}_{L^2}
      &\lesssim \lambda^{-\frac{1}{2}}\norm*{e^{-\lambda f_{\varepsilon}}\Box_{\Metric}\phi}_{L^2},\\
      \norm*{e^{-\lambda f_{\varepsilon}}\phi}_{L^2}
      &\lesssim \lambda^{-1}\norm*{e^{-\lambda f_{\varepsilon}}\underline{L}_-(\phi)}_{L^2},
    \end{split}
  \end{equation}
  for any $\phi\in C^{\infty}_0(B_{\varepsilon^{10}}(x_0))$ and any
  $\lambda$ sufficiently large, and where
  $f_{\varepsilon} = \ln(h_{\varepsilon} + e_{\varepsilon})$. Let now
  $\chi: \Real\to [0,1]$ denote a smooth cutoff function supported in
  $\left[\frac{1}{2}, \infty\right)$ and equal to $1$ in
  $\left[\frac{3}{4},\infty\right)$.  For $\delta\in (0, 1]$,
  $i=1,\cdots,I, j=1,\cdots, J$, define
  \begin{equation}
    \label{eq:VF-extension:horizon:Gdelta-Hdelta-def}
    \begin{split}
      \widetilde{\chi}_{\delta,\varepsilon}
      &\vcentcolon=\bm{1}_{I^{+-}_{c}}\chi(\underline{u}_+\underline{u}_-\delta^{-1})\left(1-\chi\left(\frac{N^{x_0}}{\varepsilon^{20}}\right)\right),\\
      G^{\delta,\varepsilon}_i &\vcentcolon= G_i \widetilde{\chi}_{\delta,\varepsilon},\\
      H^{\delta,\varepsilon}_j &\vcentcolon= H_j \widetilde{\chi}_{\delta,\varepsilon}.
    \end{split}
  \end{equation}
  It is clear that
  $G^{\delta,\varepsilon}_i, H^{\delta,\varepsilon}_j\in
  C_0^{\infty}(B_{\varepsilon^{10}}\cap \mathbf{E})$. We will now
  apply the Carleman inequalities in
  \Cref{eq:VF-extension:horizon:Carleman-estimates}, and then take the
  limits $\delta\to 0$ and $\lambda\to 0$ in that order.

  We first rewrite the equations in \Cref{eq:VF-extension:nbhd-of-horizon:wave-transport-system} using the cutoff quantities introduced in \Cref{eq:VF-extension:horizon:Gdelta-Hdelta-def}.
  \begin{equation}
    \label{eq:VF-extension:nbhd-of-horizon:wave-transport-system:renormalized}
    \begin{split}
      \Box_{\Metric}G^{\delta, \varepsilon}_i
      ={}& \widetilde{\chi}_{\delta,\varepsilon}\Box_{\Metric}G_i
           + 2 \CovariantDeriv_{\alpha}G_i \CovariantDeriv^{\alpha}\widetilde{\chi}_{\delta,\varepsilon}
           + G_i\Box_{\Metric}\widetilde{\chi}_{\delta,\varepsilon},\\
      \underline{L}_-(H^{\delta,\varepsilon}_j)
      ={}& \widetilde{\chi}_{\delta,\varepsilon} \underline{L}_-(H_j) + H_j\underline{L}_-(\widetilde{\chi}_{\delta,\varepsilon}). 
    \end{split}
  \end{equation}
  We can then apply the Carleman inequalities in \Cref{eq:VF-extension:horizon:Carleman-estimates} to see that for $\lambda$ sufficiently large,
  \begin{equation}
    \label{eq:VF-extension:nbhd-of-horizon:consequence-of-Carleman:wave}
    \begin{split}
      &\lambda\norm*{e^{-\lambda f_{\varepsilon}}\widetilde{\chi}_{\delta,\varepsilon}G_i}_{L^2}
    + \norm*{e^{-\lambda f_{\varepsilon}}\widetilde{\chi}_{\delta,\varepsilon}\abs*{D^1G_i}}_{L^2}\\
    \lesssim{}& \lambda^{-\frac{1}{2}}\norm*{e^{-\lambda f_{\varepsilon}}\widetilde{\chi}_{\delta,\varepsilon}\Box_{\Metric}G_i}_{L^2}
    +\norm*{e^{-\lambda f_{\varepsilon}}\CovariantDeriv_{\alpha} G_i \CovariantDeriv^{\alpha}\widetilde{\chi}_{\delta,\varepsilon}}_{L^2}
    + \norm*{e^{-\lambda f_{\varepsilon}}G_i\left(\abs*{\Box_{\Metric}\widetilde{\chi}_{\delta,\varepsilon}} + \abs*{D^1\widetilde{\chi}_{\delta,\varepsilon}}\right)}_{L^2}.
    \end{split}    
  \end{equation}
  and
  \begin{equation}
    \label{eq:VF-extension:nbhd-of-horizon:consequence-of-Carleman:transport}
    \norm*{e^{-\lambda f_{\varepsilon}}\widetilde{\chi}_{\delta,\varepsilon}H_j}_{L^2}
    \lesssim \lambda^{-1}\norm*{e^{-\lambda f_{\varepsilon}}\widetilde{\chi}_{\delta,\varepsilon}\underline{L}_-(H_j)}_{L^2}
    + \lambda^{-1}\norm*{e^{-\lambda f_{\varepsilon}}H_j\underline{L}_-(\widetilde{\chi}_{\delta,\varepsilon})}_{L^2}.
  \end{equation}
  From the equations in
  \Cref{eq:VF-extension:nbhd-of-horizon:wave-transport-system}, we see
  that we must have that
  \begin{equation}
    \label{eq:VF-extension:nbhd-of-horizon:estimates-from-eqns}
    \begin{split}
      \abs*{\Box_{\Metric}G_i}
      &\lesssim \sum_{\ell=1}^{\abs*{I}}\left(\abs*{D^1G_{\ell}} + \abs*{G_{\ell}}\right)
        + \sum_{m=1}^{\abs*{J}}\abs*{H_m},\\
      \abs*{\underline{L}_-(H_j)}
      &\lesssim \sum_{\ell=1}^{\abs*{I}}\left(\abs*{D^1G_{\ell}} + \abs*{G_{\ell}}\right)
        + \sum_{m=1}^{\abs*{J}}\abs*{H_m}.
    \end{split}
  \end{equation}
  Then, combining \Cref{eq:VF-extension:nbhd-of-horizon:consequence-of-Carleman:wave} and \Cref{eq:VF-extension:nbhd-of-horizon:consequence-of-Carleman:transport}, using \Cref{eq:VF-extension:nbhd-of-horizon:estimates-from-eqns}, and choosing $\lambda$ sufficiently large, we have that 
  \begin{equation}
    \label{eq:VF-extension:nbhd-of-horizon:combined-estimate}
    \begin{split}
      &\lambda\sum_{i=1}^{\abs*{I}}\norm*{e^{-\lambda f_{\varepsilon}}\widetilde{\chi}_{\delta,\varepsilon}G_i}_{L^2}
        + \sum_{j=1}^J\norm*{e^{-\lambda f_{\varepsilon}}\widetilde{\chi}_{\delta,\varepsilon}H_j}_{L^2}\\
      \lesssim{}&\lambda^{-1}\sum_{j=1}^{\abs*{J}}\norm*{e^{-\lambda f_{\varepsilon}}H_j\abs*{D^1\widetilde{\chi}_{\delta,\varepsilon}}}_{L^2}
                  + \sum_{i=1}^{\abs*{I}}\norm*{e^{-\lambda f_{\varepsilon}}\CovariantDeriv_{\alpha} G_i\CovariantDeriv^{\alpha}\widetilde{\chi}_{\delta,\varepsilon}}_{L^2}\\
                &  + \sum_{i=1}^{\abs*{I}} \norm*{e^{-\lambda f_{\varepsilon}}G_i\left(\Box_{\Metric}\widetilde{\chi}_{\delta,\varepsilon} + \abs*{D^1\widetilde{\chi}_{\delta,\varepsilon}}\right)}_{L^2}
                  .
    \end{split}    
  \end{equation}
  Then letting $\delta\to 0$, and $\lambda\to \infty$, and making use of the fact that for $C$ sufficiently large, 
  \begin{equation*}
    e^{\frac{\lambda}{C}}\sup_{\curlyBrace*{x\in B_{\varepsilon^{10}(x_0)\cap I^{++}}:N^{x_0}\ge \frac{\varepsilon^{20}}{2}}}e^{-\lambda f_{\varepsilon}}
    \le \inf_{B_{\varepsilon^{40}(x_0)}\cap I^{++}}e^{-\lambda f_{\varepsilon}}, 
  \end{equation*}
  we then obtain that $\bm{1}_{B_{\varepsilon^{40}}\cap I^{++}}G_i=\bm{1}_{B_{\varepsilon^{40}}\cap I^{++}}H_j =0$, as desired. 
\end{proof}

In \Cref{prop:Hawking-extension:nbhd-horizon}, we not only want to
construct the vectorfield $\underline{\mathbf{K}}$ so that it is
Killing, but also such that it satisfies several nice properties we
will make use of in the extension procedure.
\begin{lemma}
  \label{lemma:S0:Hawking-commutes-with-T}
  There is an open set
  $\underline{S}_0\subset \underline{\mathbf{O}}'\subset \underline{\mathbf{O}}$ such that in $\underline{\mathbf{O}}'$, 
  \begin{equation*}
    [\KillT, \underline{\mathbf{K}}] = \underline{\mathbf{K}}\cdot \bm{\sigma} = 0.
  \end{equation*}
\end{lemma}
\begin{proof}
  See the proof of Proposition 5.1 in \cite{alexakisHawkingsLocalRigidity2010}. The proof relies neither on the way that $\underline{\mathbf{K}}$ was extended, nor on the cosmological constant, but only on the fact that $\KillT$ is both Killing and tangent to the horizons, which are non-expanding null hypersurfaces. 
\end{proof}

To conclude the proof of \Cref{prop:Hawking-extension:nbhd-horizon}, we also have to show that the vectorfield $\underline{\mathbf{K}}$ is timelike. 
\begin{lemma}
  \label{lemma:Hawking-extension:nbhd-horizon:timelike}
  Let $\underline{\mathbf{K}}$ be as constructed above in a
  neighborhood $\underline{\mathbf{O}}$ of $\underline{S}_0$. Then there exists a neighborhood $\underline{\mathbf{O}}'\subset\underline{\mathbf{O}}$ of $\underline{S}_0$ such that
  \begin{equation}
     \label{eq:Hawking-extension:nbhd-horizon:timelike-prop}
    \Metric(\underline{\mathbf{K}}, \underline{\mathbf{K}})\le \underline{u}_-\underline{u}_+\quad \text{in }(I^{++}\bigcup I^{--})\bigcup \underline{\mathbf{O}}'.
  \end{equation}
  In particular, $\underline{\mathbf{K}}$ is timelike in the set $\underline{\mathbf{O}}'\backslash(I^{+-}\bigcup I^{-+})$. 
\end{lemma}
\begin{proof}
  Since $\underline{\mathbf{K}}$ is Killing in
  $\underline{\mathbf{O}}$, we have that on $\underline{S}_0$,
  \begin{align*}
    \Box_{\Metric}(\underline{\mathbf{K}}\cdot\underline{\mathbf{K}})
    &= 2\CovariantDeriv^{\alpha}(\underline{\mathbf{K}}^{\beta}\CovariantDeriv_{\alpha}\underline{\mathbf{K}}_{\beta})
      + 2\underline{\mathbf{K}}\cdot\Box_{\Metric}\underline{\mathbf{K}}
    \\
    &= 2\CovariantDeriv^{\alpha}\underline{\mathbf{K}}^{\beta}\CovariantDeriv_{\alpha}\underline{\mathbf{K}}_{\beta}
      - 2\Lambda \underline{\mathbf{K}}\cdot \underline{\mathbf{K}}
    \\
     &= 2\CovariantDeriv^4\underline{\mathbf{K}}^3\CovariantDeriv_4\underline{\mathbf{K}}_3\\
    &=-4,
  \end{align*}
  where the third equality follows from the fact that
  \Cref{prop:construction-Killing-VF:causal-region} implies that
  $\underline{\mathbf{K}}\cdot\underline{\mathbf{K}}=0$ on
  $\CosmologicalHorizon\bigcap \underline{\mathbf{O}}$, and the last
  equality follows from
  \Cref{eq:construction-Killing-VF:causal-region:D-K-aux:future-aux} and
  \Cref{eq:construction-Killing-VF:causal-region:D-K-aux:past-aux}.
  As a result, we have that there exists some
  $f\in C^{\infty}(\underline{\mathbf{O}}, \Real)$ such that
  $\underline{\mathbf{K}}\cdot\underline{\mathbf{K}}=\underline{u}_+\underline{u}_-f$. Then,
  using the above equation on $\underline{S}_0$ and the fact that
  $\underline{u}_\pm=0$ on $\underline{S}_0$, we have that
  \begin{equation*}
    -4 = \CovariantDeriv^{\alpha}\CovariantDeriv_{\alpha}(\underline{u}_+\underline{u}_-f)
    = 2 f \CovariantDeriv^{\alpha}\underline{u}_+\CovariantDeriv_{\alpha}\underline{u}_-
    = 2f \underline{L}_-\underline{u}_-\underline{L}_+\underline{u}_+
    = -2f.
  \end{equation*}
  Thus, we see that $f=2$ on $\underline{S}_0$, and \Cref{lemma:Hawking-extension:nbhd-horizon:timelike} follows in a sufficiently small $\underline{\mathbf{O}}'$. 
\end{proof}

Putting together the results of in this section proves
\Cref{prop:Hawking-extension:nbhd-horizon}.
\begin{proof}[Proof of \Cref{prop:Hawking-extension:nbhd-horizon}]
  Combining the results of \Cref{lemma:VF-extension:nbhd-of-horizon},
  \Cref{lemma:S0:Hawking-commutes-with-T},
  \Cref{lemma:Hawking-extension:nbhd-horizon:timelike} proves
  \Cref{prop:Hawking-extension:nbhd-horizon}.
\end{proof}

\subsection{Extension of the Hawking vectorfield in \texorpdfstring{$\mathbf{E}$}{E}}

In this section, we will extend the second Killing vector
$\underline{\mathbf{K}}$ into the exterior region $\mathbf{E}$. As
discussed before, the extension cannot be done in the entirety of
$\mathbf{E}$. In particular, our extension procedure breaks down at
the ergoregion of the event horizon.

The main thrust of the proof of \Cref{prop:Hawking-extension:main} is
the following proposition.
\begin{prop}
  \label{prop:Hawking-extension}
  There exists some $y_{\underline{\mathbf{K}}}$ such that in
  $\underline{\mathcal{U}}_R$ for any
  $y_{\underline{\mathbf{K}}}\le R\le \underline{y}_0$, there is a
  smooth vectorfield $\underline{\mathbf{K}}$ that agrees with the
  vectorfield $\underline{\mathbf{K}}$ defined in
  \Cref{prop:Hawking-extension:nbhd-horizon} in
  $\underline{\mathcal{U}}_{\underline{y}_0}$ such
  that
  \begin{equation}
    \label{eq:Hawking-extension:main:props}
    \LieDerivative_{\underline{\mathbf{K}}}\Metric=0,\qquad
    [\KillT, \underline{\mathbf{K}}]=0,\qquad
    \underline{\mathbf{K}}\cdot \bm{\sigma}=0.
  \end{equation}
\end{prop}
To prove \Cref{prop:Hawking-extension} we make the following bootstrap
assumption. Let $\underline{R}_0$ be the infimum of all $y<\underline{y}_0$ such that
there is a smooth vectorfield $\underline{\mathbf{K}}$ defined in the connected open set $\underline{\mathbf{E}}_y$ which agrees with the vectorfield $\underline{\mathbf{K}}$ constructed in \Cref{lemma:VF-extension:nbhd-of-horizon} in $\underline{\mathcal{U}}_{\underline{y}_0}$ and such that \Cref{eq:Hawking-extension:props} are satisfied in $\underline{\mathbf{E}}_y$. 

The bootstrap assumption is already verified for $\underline{R}=\underline{y}_0$ from \Cref{lemma:VF-extension:nbhd-of-horizon} by simply extending $\underline{\mathbf{K}}$ from $\underline{\mathcal{U}}_{\underline{y}_0}$ to $\underline{\mathbf{E}}_{\underline{y}_0}$ by solving $[\KillT, \underline{\mathbf{K}}]=0$, where we recall that $\KillT$ does not vanish in $\mathbf{E}$ by assumption.

We now would like to show that in fact, the bootstrap assumption can be extended, i.e. assuming that \Cref{prop:Hawking-extension} holds for $\underline{R}=\underline{R}_0$ for some $y_{*}\le \underline{R}_0<\underline{y}_0$, we would like to show that it in fact holds for $\underline{R}=\underline{R}_0-\delta$ for some $\delta>0$ sufficiently small. To do this, we will extend $\underline{\mathbf{K}}$ as a solution of an ordinary differential equation. To this end, let us first define an auxiliary vectorfield we will use to extend $\underline{\mathbf{K}}$. 

We define the vectorfield $\mathbf{Y}' = \mathbf{Y}'(x)$ that we will use to extend $\underline{\mathbf{K}}$ as the vectorfield such that
\begin{equation}
  \label{eq:Hawking-extension:Y-geodesic-def}
  \begin{cases}
    \mathbf{Y}'(x) = \mathbf{Y}(x), &x\in \partial_{\Sigma_0\bigcap \mathbf{E}}(\underline{\mathcal{U}}_{\underline{R}}),\\
  \CovariantDeriv_{\mathbf{Y}'}\mathbf{Y}'=0, &\text{in }\widetilde{\mathbf{O}}_{\delta,\underline{R}},
  \end{cases}
\end{equation}
where
\begin{equation}
  \label{eq:Hawking-extension:Otilde-def}
  \widetilde{\mathbf{O}}_{\delta,\underline{R}} \vcentcolon= \bigcup_{x\in \partial_{\underline{\Sigma}_2}(\underline{\mathcal{U}}_{\underline{R}})}B_{\delta}(x).
\end{equation}
We observe that for $\delta$ sufficiently small, and $R\le \underline{y}_0$, it is clear that 
\begin{equation*}  
  \mathbf{E}_{\underline{R}}\bigcap \widetilde{\mathbf{O}}_{\delta,\underline{R}}
  = \curlyBrace*{p\in \widetilde{\mathbf{O}}_{\delta,\underline{R}}: y(x)>\underline{R}}.
\end{equation*}

We first state the main Carleman inequality.
\begin{lemma}
  \label{lemma:Hawking-extension:Carleman}
  Let
  $x_0\in \partial_{\Sigma_0\bigcap \mathbf{E}}
  (\underline{\mathcal{U}}_{\underline{R}})$, where $y_{\underline{\mathbf{K}}}<\underline{R}<\underline{y}_0$. Then there is
  some $0<\varepsilon$ sufficiently small and some
  $C(\varepsilon)$ sufficiently large such that for any
  $\lambda\ge C(\varepsilon)$ and any
  $\phi,\psi \in C^{\infty}_0(B_{\varepsilon^{10}}(x_0))$,
  \begin{equation}
    \label{eq:Hawking-extension:Carleman}
    \begin{split}
      \lambda\norm*{e^{-\lambda \widetilde{f}_{\varepsilon}}\phi}_{L^2}
      + \norm*{e^{-\lambda\widetilde{f}_{\varepsilon}}\abs*{D^1\phi}}_{L^2}
      \le{}& C(\varepsilon) \lambda^{-\frac{1}{2}}\norm*{e^{- \lambda\widetilde{f}_{\varepsilon}}\Box_{\Metric}\phi}_{L^2}
             + \varepsilon^{-6}\norm*{e^{-\lambda \widetilde{f}_{\varepsilon}}\KillT(\phi)}_{L^2} ,
      \\
      \norm*{e^{-\lambda \widetilde{f}_{\varepsilon}}\psi}_{L^2}
      \le{}&  C(\varepsilon)\lambda^{-1}\norm*{e^{-\lambda \widetilde{f}_{\varepsilon}}\mathbf{Y}'(\psi)}_{L^2},
    \end{split}
  \end{equation}
  where $\mathbf{Y}'$ is defined as in \Cref{eq:Hawking-extension:Y-geodesic-def}, and where
  \begin{equation}
    \label{eq:Hawking-extension:Carleman:f-def}
    \widetilde{f}_{\varepsilon} = \ln(-(y-\underline{R})+\varepsilon+ \varepsilon^{12}N^{x_0}),
  \end{equation}
  where we recall the definition of $N^{x_0}$ from \Cref{eq:N-x0:def}.
\end{lemma}
\begin{proof}
  To prove the first equation in \Cref{eq:Hawking-extension:Carleman}, we apply \Cref{prop:carleman-estimate} with
  \begin{equation*}
    \{\mathbf{V}_i\}_{i=1}^1 = \{\KillT\},\qquad
    h_{\varepsilon }= -(y-y_{\underline{\mathbf{K}}})+\varepsilon, \qquad e_{\varepsilon} = \varepsilon^{12}N^{x_0}.
  \end{equation*}
  From the definition in \Cref{def:negligible perturbation}, it is
  clear that $e_{\varepsilon}$ is a negligible perturbation if
  $\varepsilon$ is sufficiently small.

  We now show that for some $\varepsilon_1$ sufficiently small,
  $\{h_{\varepsilon}\}_{\varepsilon\in (0,\varepsilon_1)}$ forms a
  $\KillT$-pseudoconvex family of weights. From the definition of
  $h_{\varepsilon}$, we have that
  \begin{equation*}
    h_{\varepsilon}(x_0) = \varepsilon,\qquad
    \KillT(h_{\varepsilon})(x_0) = 0,\qquad
    \abs*{D^jh_{\varepsilon}}\lesssim 1, 
  \end{equation*}
  so so \Cref{eq:strict-T-Z-null-convexity:cond1} is satisfied for
  $\varepsilon_1$ sufficiently small. The conditions in
  \Cref{eq:strict-T-Z-null-convexity:cond2} and
  \Cref{eq:strict-T-Z-null-convexity:main-condition} are then
  satisfied from \Cref{prop:y-pseudoconvexity} since from our
  assumptions, we have that \Cref{eq:S-smallness-assumption} holds
  since $y(x_0)>y_{*}$. Thus, for some $\varepsilon_1$ sufficiently
  small, $\{h_{\varepsilon}\}_{\varepsilon\in (0,\varepsilon_1)}$
  forms a $\KillT$-pseudoconvex family of
  weights. \Cref{prop:carleman-estimate} then directly yields the
  first estimate in \Cref{eq:Hawking-extension:Carleman}.

  To prove the second estimate in
  \Cref{eq:Hawking-extension:Carleman}, we observe that for
  $\mathbf{Y}'$ as defined in
  \Cref{eq:Hawking-extension:Y-geodesic-def}, we have from
  \Cref{eq:nabla-y-contracted:S-small} and the fact that
  $y_{\underline{\mathbf{K}}}<y(x_0)<\underline{y}_0$ that for
  $\varepsilon_{\mathcal{S}}$ sufficiently small, there exists some
  $c>0$ such that
  \begin{equation*}
    \abs*{\mathbf{Y}'(h_{\varepsilon})(x_0)}\ge c.
  \end{equation*}
  This means that \Cref{eq:Carleman-estimate:transport:conditions} are satisfied, and thus the second estimate in \Cref{eq:Hawking-extension:Carleman} follows as a direct application of \Cref{eq:Carleman-estimate:transport}.
\end{proof}

We now introduce the main extension lemma. We will apply this multiple times throughout the proof of \Cref{prop:Hawking-extension}.
\begin{lemma}
  \label{lemma:Hawking-extension:unique-continuation}
  Let $\delta>0, x_0\in \partial_{\underline{\Sigma}_2}(\mathcal{U}_{\underline{R}_0})$ and $\phi_{\mathbf{j}}, \psi_{\mathbf{i}}\in C^{\infty}\left( B_{\delta}(x_0), \Real \right))$, where $\mathbf{i}\in \mathbf{I}, \mathbf{j}\in \mathbf{J}$ and $\mathbf{I}, \mathbf{J}$ are sets of indices, such that $\phi=\curlyBrace*{\phi_{\mathbf{j}}}_{\mathbf{j}\in \mathbf{J}}, \curlyBrace*{\psi_{\mathbf{i}}}_{\mathbf{i}\in \mathbf{I}}$ verify the schematic equations
  \begin{equation}
    \label{eq:Hawking-extension:unique-continuation:wave-transport}
    \begin{split}
      \Box_{\Metric}\phi ={}& \AdmissibleRHS(\phi, \psi, \partial \phi),\\
      \KillT(\phi)={}& \AdmissibleRHS(\phi, \psi),\\
      \mathbf{Y}'(\psi)={}& \AdmissibleRHS(\phi, \psi, \partial \phi).
    \end{split}    
  \end{equation}
  If $\phi, \psi=0$ in $B_{\delta}(x_0)\bigcap \mathbf{E}_{\underline{R}_0}=\curlyBrace*{x\in B_{\delta}(x_0): y(x) > \underline{R}_0}$, then $\phi=0$ and $\psi=0$ in $B_{\widetilde{\delta}}(x_0)$ for some $\widetilde{\delta}\in (0,\delta)$ sufficiently small.
\end{lemma}

\begin{proof}
  Let $\varepsilon^{10}=\delta$, and let
  $\widetilde{\chi}:\Real\to [0,1]$ denote a smooth cutoff function
  supported in $[\frac{1}{2},\infty)$ and equal to $1$ in
  $[\frac{3}{4},\infty)$. Then define
  \begin{equation*}
    \phi_{\mathbf{j}}^{\varepsilon}\vcentcolon= \phi_{\mathbf{j}}\widetilde{\chi}_{\varepsilon},\qquad
    \psi_{\mathbf{i}}^{\varepsilon}\vcentcolon= \psi_{\mathbf{i}}\widetilde{\chi}_{\varepsilon},
  \end{equation*}
  where
  \begin{equation*}
    \widetilde{\chi}_{\varepsilon} = 1-\widetilde{\chi}(\varepsilon^{-40}N^{x_0}(x)).
  \end{equation*}

  Then we see that
  \begin{equation*}
    \begin{split}
      \Box_{\Metric}\phi_{\mathbf{j}}^{\varepsilon}
      ={}& \widetilde{\chi}_{\varepsilon}\Box_{\Metric}\phi_{\mathbf{j}}
      + 2\CovariantDeriv_{\alpha}\phi_{\mathbf{j}}\CovariantDeriv^{\alpha}\widetilde{\chi}_{\varepsilon}
           + \phi_{\mathbf{j}}\Box_{\Metric}\widetilde{\chi}_{\varepsilon},\\
      \KillT(\phi_{\mathbf{j}}^{\varepsilon})
      ={}& \widetilde{\chi}_{\varepsilon}\KillT(\phi_{\mathbf{j}})
           + \KillT(\widetilde{\chi}_{\varepsilon})\phi_{\mathbf{j}},\\
      \mathbf{Y}'(\psi_{\mathbf{i}}^{\varepsilon})
      ={}& \widetilde{\chi}_{\varepsilon}\cdot \mathbf{Y}'(\psi_{\mathbf{i}})
           + \psi_{\mathbf{i}}\cdot \mathbf{Y}'(\widetilde{\chi}_{\varepsilon}).
    \end{split}
  \end{equation*}
  The Carleman estimates in \Cref{eq:Hawking-extension:Carleman} then give that
  \begin{equation}
    \label{eq:Hawking-extension:Carleman-with-cutoff}
    \begin{split}
      \lambda\norm*{e^{-\lambda\widetilde{f}_{\varepsilon}}\widetilde{\chi}_{\varepsilon}\phi_{\mathbf{j}}}_{L^2}
        + \norm*{e^{-\lambda\widetilde{f}_{\varepsilon}}\widetilde{\chi}_{\varepsilon}\abs*{D^1\phi_{\mathbf{j}}}}_{L^2}
      \lesssim{}& \lambda^{-\frac{1}{2}}\norm*{e^{-\lambda\widetilde{f}_{\varepsilon}}\widetilde{\chi}_{\varepsilon}\Box_{\Metric}\phi_{\mathbf{j}}}_{L^2}
                  + \norm*{e^{-\lambda\widetilde{f}_{\varepsilon}}\phi_{\mathbf{j}}\left(\abs*{\Box_{\Metric}\widetilde{\chi}_{\varepsilon}}+ \abs*{D^1\widetilde{\chi}_{\varepsilon}}\right)}_{L^2}
                  \\
                & + \norm*{e^{-\lambda\widetilde{f}_{\varepsilon}}\CovariantDeriv_{\alpha}\phi_{\mathbf{j}}\CovariantDeriv^{\alpha}\widetilde{\chi}_{\varepsilon}}_{L^2}
                  + \norm*{e^{-\lambda\widetilde{f}_{\varepsilon}}\widetilde{\chi}_{\varepsilon}\KillT(\phi_{\mathbf{j}})}_{L^2}
                  ,\\
            \lambda^{\frac{1}{2}}\norm*{e^{-\lambda \widetilde{f}_{\varepsilon}}\widetilde{\chi}_{\varepsilon}\psi_{\mathbf{i}}}_{L^2}
      \lesssim{}& \lambda^{-\frac{1}{2}}\norm*{e^{-\lambda\widetilde{f}_{\varepsilon}}\widetilde{\chi}_{\varepsilon}\mathbf{Y}'(\psi_{\mathbf{i}})}_{L^2}
                  + \lambda^{-\frac{1}{2}}\norm*{e^{-\lambda \widetilde{f}_{\varepsilon}}\psi_{\mathbf{i}}\abs*{D\widetilde{\chi}_{\varepsilon}}}_{L^2},
    \end{split}
  \end{equation}  
  for $\lambda$ sufficiently large, where we recall the definition of $\widetilde{f}_{\varepsilon}$ in \Cref{eq:Hawking-extension:Carleman:f-def}. From the equations in \Cref{eq:Hawking-extension:unique-continuation:wave-transport}, we have that
  \begin{align*}
    \abs*{\Box_{\Metric}\phi_{\mathbf{j}}}
    &\lesssim \sum_{\mathbf{j}}\left(\abs*{D \phi_{\mathbf{j}}} + \abs*{\phi_{\mathbf{j}}}\right)
      + \sum_{\mathbf{i}}\abs*{\psi_{\mathbf{i}}}\\
    \abs*{\KillT(\phi_{\mathbf{j}})}
    &\lesssim \sum_{\mathbf{j}}\abs*{\phi_{\mathbf{j}}}
      + \sum_{\mathbf{i}}\abs*{\psi_{\mathbf{i}}},\\
    \abs*{\mathbf{Y}'(\psi_{\mathbf{i}})}
    &\lesssim \sum_{\mathbf{j}}\left(\abs*{D\phi_{\mathbf{j}}} + \abs*{\phi_{\mathbf{j}}}\right)
      + \sum_{\mathbf{i}}\abs*{\psi_{\mathbf{i}}}.
  \end{align*}
    Then summing \Cref{eq:Hawking-extension:Carleman-with-cutoff} over all indices $\mathbf{j}, \mathbf{i}$, we observe that for $\lambda$ sufficiently large,
  \begin{equation}
    \label{eq:Hawking-extension:aux2}
    \begin{split}
      &\lambda \sum_{\mathbf{j}}\norm*{e^{-\lambda \widetilde{f}_{\varepsilon}}\widetilde{\chi}_{\varepsilon}\phi_{\mathbf{j}}}_{L^2}
      + \sum_{\mathbf{j}}\norm*{e^{-\lambda \widetilde{f}_{\varepsilon}}\widetilde{\chi}_{\varepsilon}\abs*{D\phi_{\mathbf{j}}}}_{L^2}
      + \lambda^{\frac{1}{2}}\sum_{\mathbf{i}}\norm*{e^{-\lambda \widetilde{f}_{\varepsilon}}\widetilde{\chi}_{\varepsilon}\psi_{\mathbf{i}}}_{L^2}\\
        \lesssim{}& \sum_{\mathbf{i}}\norm*{e^{-\lambda \widetilde{f}_{\varepsilon}}\psi_{\mathbf{i}}\abs*{D\widetilde{\chi}_{\varepsilon}}}_{L^2}
                    + \sum_{\mathbf{j}}\norm*{e^{-\lambda \widetilde{f}_{\varepsilon}}\phi_{\mathbf{j}}\abs*{D\widetilde{\chi}_{\varepsilon}}}_{L^2}
                    + \sum_{\mathbf{j}}\norm*{e^{-\lambda \widetilde{f}_{\varepsilon}}\phi_{\mathbf{j}}\left(\abs*{\Box_{\Metric}\widetilde{\chi}_{\varepsilon}}+\abs*{D\widetilde{\chi}_{\varepsilon}}\right)}_{L^2}.
    \end{split}
  \end{equation}
  Recall now that both $\Box_{\Metric}\widetilde{\chi}_{\varepsilon}$
  and $D\widetilde{\chi}_{\varepsilon}$ have compact support in
  $\curlyBrace*{x\in B_{\varepsilon^{10}}(x_0): N^{x_0}(x)\ge
    \varepsilon^{50}}$, and $\widetilde{\chi}_{\varepsilon}=1$ in
  $B_{\varepsilon^{100}}$. By assumption, $\phi_{\mathbf{j}}$ and
  $\psi_{\mathbf{i}}$ are supported in
  $\curlyBrace*{x\in B_{\delta}(x_0):y(x)\le \underline{R}_0}$. In
  addition, from the definition of $\widetilde{f}$ in
  \Cref{eq:Hawking-extension:Carleman:f-def},
  \begin{equation}
    \label{eq:Hawking-extension:aux3}
    \inf_{B_{\varepsilon^{100}}(x_0)}e^{-\lambda \widetilde{f}_{\varepsilon}}
    \ge e^{-\lambda \ln(\varepsilon+\varepsilon^{70})}
    \ge \sup_{\curlyBrace*{x\in B_{\varepsilon^{10}}(x_0):y(x)\le \underline{R}_0, N^{x_0}(x)\ge \varepsilon^{50}}}e^{-\lambda \widetilde{f}_{\varepsilon}}.
  \end{equation}
  It then follows using \Cref{eq:Hawking-extension:aux3} and \Cref{eq:Hawking-extension:aux2} that
  \begin{equation*}
    \lambda\left( \sum_{\mathbf{j}}\norm*{\mathbf{1}_{B_{\varepsilon^{100}}(x_0)}\phi_{\mathbf{j}}}_{L^2}
     + \sum_{\mathbf{i}}\norm*{\mathbf{1}_{B_{\varepsilon^{100}}(x_0)}\psi_{\mathbf{i}}}_{L^2} \right)
    \lesssim \norm*{\mathbf{1}_{\curlyBrace*{x\in B_{\varepsilon^{10}}(x_0):y(x)\le \underline{R}_0, N^{x_0}(x)\ge \varepsilon^{50}}}}_{L^2}.
  \end{equation*}
  Taking the limit $\lambda\to \infty$ then proves that
  $\phi_{\mathbf{j}}=\psi_{\mathbf{i}}=0$ on
  $B_{\varepsilon^{100}}(x_0)$.
\end{proof}

We will now prove \Cref{prop:Hawking-extension}.
\begin{proof}[Proof of \Cref{prop:Hawking-extension}]
  We will fix $\varepsilon$ throughout the proof. Let $x_0\in \partial_{\underline{\Sigma}_2}(\underline{\mathcal{U}}_{\underline{R}})$. From the wave-transport system for the variables $B, \dot{B}, P, \LieZRiem$ in \Cref{eq:VF-extension:wave-transport}, and the fact that $\KillT$ is Killing, we have that
  \begin{equation}
    \label{eq:Hawking-extension:main-system}
    \begin{split}
      \Box_{\Metric}\LieZRiem
      ={}& \AdmissibleRHS(B, \dot{B}, P, \CovariantDeriv B,\CovariantDeriv \dot{B},\CovariantDeriv P, \LieZRiem),\\
      \CovariantDeriv_{\GeodesicVF}B ={}& \AdmissibleRHS(\dot{B}),\\
      \CovariantDeriv_{\GeodesicVF}\dot{B}={}& \AdmissibleRHS(B, \dot{B}, \LieZRiem),\\
      \CovariantDeriv_{\GeodesicVF}P ={}& \AdmissibleRHS(B, P, \LieZRiem),\\
      \LieDerivative_{\KillT}\LieZRiem={}&0,
    \end{split}
  \end{equation}
  in $B_{\varepsilon^{10}}(x_0)$. We can then directly apply
  \Cref{lemma:Hawking-extension:unique-continuation} to see that
  $B, \dot{B}, P, \LieZRiem$ all vanish in
  $B_{\varepsilon^{100}}(x_0)$. Observe that from the definition of
  $\omega$ in \Cref{eq:VF-extension:transport:omega-def}, this in
  particular implies that
  $\LieDerivative_{\underline{\mathbf{K}}}\Metric=0$ and thus that
  $\underline{\mathbf{K}}$ extends as a Killing vectorfield in
  $B_{\varepsilon^{100}}(x_0)$. Since
  $x_0\in
  \partial_{\underline{\Sigma}_2}(\underline{\mathcal{U}}_{\underline{R}})$
  was arbitrary, we have that $\underline{\mathbf{K}}$ extends as a
  Killing vector to the open set
  $\widetilde{\mathbf{O}}_{\delta, \underline{R}_0}$ for some
  $\delta>0$. It remains to show that the additional properties in
  \Cref{eq:Hawking-extension:props} are also satisfied. To this end,
  we will perform another unique continuation argument.

  Since $\KillT, \underline{\mathbf{K}}$ are both Killing vectorfields
  on a solution to  Einstein's vacuum equations with
  cosmological constant $\Lambda$, we have that on $B_{\delta}(x_0)$,
  \begin{equation*}
    \Box_{\Metric}\KillT_{\mu} = -\tensor[]{\Ric}{^{\nu}_{\mu}}\KillT_{\nu}
    = -\Lambda \KillT_{\mu}, \qquad
    \Box_{\Metric}\underline{\mathbf{K}} = -\Lambda\underline{\mathbf{K}}. 
  \end{equation*}

  As a result, we can observe that
  \begin{align*}
    \Box_{\Metric}(\underline{\mathbf{K}}\cdot\bm{\sigma})
    &= -\Lambda \underline{\mathbf{K}}\cdot \bm{\sigma}
    + \underline{\mathbf{K}}\cdot \Box_{\Metric}\bm{\sigma}\\
    &=-\Lambda \underline{\mathbf{K}}\cdot \bm{\sigma}
      + \CovariantDeriv_{\underline{\mathbf{K}}} \Divergence \bm{\sigma}\\
    &= -\Lambda \underline{\mathbf{K}}\cdot \bm{\sigma}
      -\LieDerivative_{\underline{\mathbf{K}}}\mathcal{F}^2
      + 4\Lambda \Metric(\LieDerivative_{\underline{\mathbf{K}}}\KillT,\KillT)\\
    &= \AdmissibleRHS(\underline{\mathbf{K}}\cdot\bm{\sigma}, \LieDerivative_{\underline{\mathbf{K}}}\KillT)
    .
  \end{align*}
  Moreover, using the fact that $\LieDerivative_K$ commutes with covariant derivatives if $K$ is Killing, we have that 
  \begin{align*}
    \Box_{\Metric}(\LieDerivative_{\KillT}\underline{\mathbf{K}})
    =\LieDerivative_{\KillT}\Box_{\Metric}\underline{\mathbf{K}}
    = -\Lambda \LieDerivative_{\KillT}\underline{\mathbf{K}}.
  \end{align*}
  Similarly, defining
  $\widetilde{F}_{\mu\nu}\vcentcolon=
  \CovariantDeriv_{\mu}\LieDerivative_{\KillT}\underline{\mathbf{K}}_{\nu}$,
  we see that since the Killing vectors form a Lie algebra,
  \begin{align*}
    \CovariantDeriv_{\mathbf{Y}'}\LieDerivative_{\KillT}\underline{\mathbf{K}}
    = \mathbf{Y}'\cdot \widetilde{F}\qquad
    \CovariantDeriv_{\mathbf{Y}'}\widetilde{F}_{\mu\nu}
    = \Riem_{\mu\mathbf{Y}'\nu\LieDerivative_{\KillT}\underline{\mathbf{K}}},\qquad 
    \KillT(\underline{\mathbf{K}}\cdot\bm{\sigma})
    = \left( \LieDerivative_{\KillT}\underline{\mathbf{K}} \right)\cdot\bm{\sigma}
    + \underline{\mathbf{K}}\left(\KillT\cdot\bm{\sigma}\right)
    =\left( \LieDerivative_{\KillT}\underline{\mathbf{K}} \right)\cdot\bm{\sigma}.
  \end{align*}
  As a result, we have that $\underline{\mathbf{K}}\cdot\bm{\sigma}, \LieDerivative_{\KillT}\underline{\mathbf{K}}$ satisfy
  \begin{equation}
    \label{eq:Hawking-extension:wave transport system for aux properties}
    \begin{split}
      \Box_{\Metric}\left( \underline{\mathbf{K}}\cdot\bm{\sigma} \right)
      ={}& \AdmissibleRHS(\underline{\mathbf{K}}\cdot\bm{\sigma}, \LieDerivative_{\KillT}\underline{\mathbf{K}}, \widetilde{F}),\\
      \CovariantDeriv_{\mathbf{Y}'}\widetilde{F}={}& \AdmissibleRHS(\LieDerivative_{\KillT}\underline{\mathbf{K}}),\\
      \CovariantDeriv_{\mathbf{Y}'}\LieDerivative_{\KillT}\mathbf{K}={}& \AdmissibleRHS(\widetilde{F}),\\
      \KillT(\underline{\mathbf{K}}\cdot\bm{\sigma})={}& \AdmissibleRHS(\LieDerivative_{\KillT}\underline{\mathbf{K}}). 
    \end{split}
  \end{equation}
  Thus, applying \Cref{lemma:Hawking-extension:unique-continuation}
  with the system in~\Cref{eq:Hawking-extension:wave transport system
    for aux properties}, we have that
  $\LieDerivative_{\KillT}\underline{\mathbf{K}}=0,
  \underline{\mathbf{K}}\cdot\bm{\sigma}=0$ in
  $B_{\varepsilon^{100}}(x_0)$, as desired.
\end{proof}

We can now conclude the proof of \Cref{prop:Hawking-extension:main}.
\begin{proof}
  Assume for the sake of contradiction that $\underline{R}_0 >
  y_{*}$. But then using the first inclusion in
  \Cref{eq:UBar-VBar-behavior:growth}, it then follows from
  \Cref{prop:Hawking-extension} that $\underline{\mathbf{K}}$ is in
  fact a Killing vector in a small neighborhood of
  $\underline{\mathcal{U}}_{\underline{R}_0-\delta^2}$, and we can
  extend $\underline{\mathbf{K}}$ as a Killing vector to the region
  $\mathbf{E}_{\underline{R}_0-\delta^2}$ as a solution to
  $[\KillT, \underline{\mathbf{K}}]$, concluding the proof of
  \Cref{prop:Hawking-extension:main}.
\end{proof}

\subsection{\texorpdfstring{$\KillT$}{T} and \texorpdfstring{$\underline{\mathbf{K}}$}{K} span a timelike Killing vectorfield}

In order for us to use the additional Killing vector
$\underline{\mathbf{K}}$ constructed in \Cref{prop:Hawking-extension}
to help us prove \Cref{thm:main:e2c2} and \Cref{thm:main}, we will need
that $\KillT, \underline{\mathbf{K}}$ span a timelike Killing vectorfield.
\begin{prop}
  \label{prop:timelike-span-of-Killing-VFs}
  The vectorfield $\underline{\mathbf{K}}$ as constructed in
  \Cref{prop:Hawking-extension} does not vanish at any point in
  $\underline{\mathbf{E}}_{y_{\underline{\mathbf{K}}}}$ (recall \Cref{eq:E-bar-R:def}). Moreover, at any point
  $p\in \underline{\mathbf{E}}_{y_{\underline{\mathbf{K}}}}$, there is a timelike
  linear combination of the vectorfields $\KillT$ and
  $\underline{\mathbf{K}}$.
\end{prop}
We first prove a useful consequence of the Killing nature of
$\underline{\mathbf{K}}$.
\begin{lemma}
  \label{lemma:TKX-S-vanishes}
  In $\underline{\mathbf{E}}_{y_{\underline{\mathbf{K}}}}$ we have that
  \begin{equation}
    \label{eq:TKX-S-vanishes}
    \KillT^{\sigma}\underline{\mathbf{K}}^{\rho}\mathcal{X}^{\mu\nu}\mathcal{S}_{\mu\nu\sigma\rho}
    = 0.
  \end{equation}
\end{lemma}
\begin{proof}
  Recall from \Cref{eq:Hawking-extension:props} that
  $\LieDerivative_{\underline{\mathbf{K}}}\KillT=0$ in
  $\underline{\mathbf{E}}_{y_{\underline{\mathbf{K}}}}$. Then, using
  the fact that $\underline{\mathbf{K}}$ is Killing in
  $\underline{\mathbf{E}}_{y_{\underline{\mathbf{K}}}}$ and thus that
  $\LieDerivative_{\underline{\mathbf{K}}}$ commutes with covariant
  derivatives, we have that
  \begin{equation*}
    \LieDerivative_{\KillT}\mathcal{F} = 0.
  \end{equation*}
  This directly implies that
  \begin{equation}
    \label{eq:TKX-S-vanishes:aux1}
    \LieDerivative_{\KillT}R=0.
  \end{equation}
  But then, using \Cref{eq:deriv-R} and the fact that we know that
  $\underline{\mathbf{K}}\cdot \bm{\sigma} = 0$ in
  $\underline{\mathbf{E}}_{y_{\underline{\mathbf{K}}}}$ from
  \Cref{eq:Hawking-extension:props}, we have that
  \begin{equation}
    \label{eq:TKX-S-vanishes:aux2}
    \LieDerivative_{\KillT}R = - \frac{1}{4}\KillT^{\sigma}\underline{\mathbf{K}}^{\rho}\mathcal{X}^{\mu\nu}\mathcal{S}_{\mu\nu\sigma\rho}.
  \end{equation}
  Combining \Cref{eq:TKX-S-vanishes:aux1} and
  \Cref{eq:TKX-S-vanishes:aux2} directly yields
  \Cref{eq:TKX-S-vanishes}.
\end{proof} 
We are now ready to prove \Cref{prop:timelike-span-of-Killing-VFs}.
\begin{proof}[Proof of \Cref{prop:timelike-span-of-Killing-VFs}]
  We first show the first part of the proposition, that the Hawking
  vectorfield $\underline{\mathbf{K}}$ does not vanish at any point in
  $\underline{\mathbf{E}}_{y_{\underline{\mathbf{K}}}}$. It is clear
  that $\underline{\mathbf{K}}$ does not vanish at any point in
  $\underline{\mathbf{E}}_{\underline{y}_0}$ since
  $\underline{\mathbf{K}}$ is constructed as the solution of
  $[\KillT, \underline{\mathbf{K}}]=0$.

  To prove that $\underline{\mathbf{K}}$ does not vanish at any point
  in $\underline{\mathbf{E}}_{y_{\underline{\mathbf{K}}}}$, we define
  \begin{equation*}
    \underline{R}_1 \vcentcolon= \inf\curlyBrace*{R\in (y_{\underline{\mathbf{K}}}, \underline{y}_0): \underline{\mathbf{K}}\neq0 \quad\forall p\in \underline{\mathbf{E}}_R}.
  \end{equation*}
  Assume for the sake of contradiction that
  $\underline{R}_1> y_{\underline{\mathbf{K}}}$. Then there exists
  some point
  $x_0\in
  \partial_{\underline{\Sigma}_2}(\underline{\mathcal{U}}_{\underline{R}_1})$
  such that $\evalAt*{\underline{\mathbf{K}}}_{x_0}=0$. 
  Combining the fact that $\underline{\mathbf{K}}\cdot\bm{\sigma}$
  vanishes from \Cref{prop:Hawking-extension}, with
  \Cref{lemma:P-one-form-P-J-inverse-relation} and
  \Cref{lemma:TKX-S-vanishes}, we see that
  \begin{equation*}
    0 = \frac{1}{2R}\underline{\mathbf{K}}\cdot\bm{\sigma}
    = \underline{\mathbf{K}}\cdot \mathbf{P}
    = \CovariantDeriv_{\underline{\mathbf{K}}}P.
  \end{equation*}
  As a result, we have that
  \begin{equation}
    \label{eq:KBar-dot-Y-dot-Z-vanish}
    \underline{\mathbf{K}}\cdot \mathbf{Y}
    =  \underline{\mathbf{K}}\cdot \mathbf{Z}
    = 0,
  \end{equation}
  we have that $\underline{\mathbf{K}}\cdot\mathbf{Y}=0$, where
  $\mathbf{Y}^{\alpha} = \CovariantDeriv^{\alpha}y$ and
  $\mathbf{Z}^{\alpha} = \CovariantDeriv^{\alpha}z$. But then this
  implies that
  \begin{align*}
    [\underline{\mathbf{K}}, \mathbf{Y}]_{\beta}
    ={}& \underline{\mathbf{K}}^{\alpha}\CovariantDeriv_{\alpha}\mathbf{Y}_{\beta}
         - \mathbf{Y}^{\alpha}\CovariantDeriv_{\alpha}\underline{\mathbf{K}}_{\beta}\\
    ={}& \underline{\mathbf{K}}^{\alpha}\CovariantDeriv_{\beta}\CovariantDeriv_{\alpha}y
         - \CovariantDeriv^{\alpha}y\CovariantDeriv_{\beta}\underline{\mathbf{K}}_{\alpha}\\
    ={}& \CovariantDeriv_{\beta}\left(\underline{\mathbf{K}}\cdot \mathbf{Y}\right) =0.         
  \end{align*}
  Then, using \Cref{eq:nabla-y-contracted:S-small}, we have that
  $\mathbf{Y}$ is in fact spacelike. But then if
  $\underline{\mathbf{K}}$ vanishes at
  $p\in
  \partial_{\underline{\Sigma}_2}(\underline{\mathcal{U}}_{\underline{R}_1})$,
  then it must also vanish on the integral curve
  $\gamma_p(t), \abs*{t}\ll1$ of the vectorfield
  $\mathbf{Y}$. However, since $\mathbf{Y}$ is spacelike, this
  curve necessarily intersects the set
  $\underline{\mathbf{E}}_{\underline{R}'}$ for some
  $\underline{R}'>\underline{R}_1$, which contradicts our definition
  of $\underline{R}_1$, showing that $\underline{\mathbf{K}}$ does
  not vanish at any point in
  $\mathbf{E}_{y_{\underline{\mathbf{K}}}}$.

  We now show the second part of the proposition. Define
  \begin{equation*}
    \mathring{\mathcal{U}}\vcentcolon= \curlyBrace*{p\in \underline{\mathbf{E}}_{y_{\underline{\mathbf{K}}}}: \operatorname{Span}(\KillT, \underline{\mathbf{K}}) \text{ is not timelike} },
  \end{equation*}
  which is a closed set in
  $\underline{\mathbf{E}}_{y_{\underline{\mathbf{K}}}}$ consisting of
  orbits of the vectorfield $\KillT$. Moreover, we can see
  that
  $\mathring{\mathcal{U}}\subset
  \underline{\mathbf{E}}_{y_{\underline{\mathbf{K}}}} \backslash
  \underline{\mathbf{E}}_{\underline{y}_0}$, since we proved that
  $\underline{\mathbf{K}}$ is timelike in
  $\underline{\mathbf{E}}_{\underline{y}_0}$ (see
  \Cref{prop:Hawking-extension:nbhd-horizon}). Using
  \Cref{eq:basic-inclusions-UBar} and the fact that
  $\underline{y}_0<y_{\underline{S}_0}$, we have that in
  $\underline{\mathbf{E}}_{y_{\underline{\mathbf{K}}}}\backslash
  \underline{\mathbf{E}}_{\underline{y}_0}$, for some $C\gg 1$,
  \begin{equation*}
    y \le y_{\underline{S}_0} - C^{-1}.
  \end{equation*}
  Thus, in $\mathring{\mathcal{U}}$, we have that
  \begin{equation}
    \label{eq:timelike-span-of-Killing-VFs:Y-spacelike}
    \Metric(\mathbf{Y}, \mathbf{Y}) \gtrsim C^{-1}. 
  \end{equation}
  In addition, since in $\mathring{\mathcal{U}}$,
  $\Metric(\KillT, \KillT) > 0$, we must also have that 
  \begin{equation*}
    \Metric(\mathbf{Z}, \mathbf{Z})  >  \Metric(\KillT, \KillT) + \Metric(\mathbf{Y}, \mathbf{Y}).
  \end{equation*}
  Thus, we have that
  \begin{equation*}
    \Metric(\mathbf{Z}, \mathbf{Z})  \gtrsim C^{-1}.
  \end{equation*}
  On the other hand, we have from \Cref{lemma:y-z-derivatives:S-small}
  and \Cref{eq:KBar-dot-Y-dot-Z-vanish} that
  \begin{equation*}
    \Metric(\KillT, \mathbf{Y}) = \Metric(\KillT, \mathbf{Z}) = O(\varepsilon_{\mathcal{S}}), \qquad
    \Metric(\underline{\mathbf{K}}, \mathbf{Y})
    = \Metric(\underline{\mathbf{K}}, \mathbf{Z}) =  0. 
  \end{equation*}
  Then, since $\Metric$ is Lorentzian, it follows that for
  $\varepsilon_{\mathcal{S}}$ sufficiently small, $\KillT$,
  $\underline{\mathbf{K}}$, $\mathbf{Y}$, and $\mathbf{Z}$ cannot be
  linearly independent at any point $p\in \mathring{\mathcal{U}}$,
  since otherwise, the determinant of the matrix formed by the coefficients
  $\Metric(\KillT,\KillT), \Metric(\KillT, \underline{\mathbf{K}}),
  \Metric(\underline{\mathbf{K}},\underline{\mathbf{K}})$ would be
  negative, but this would imply that
  $\KillT + \underline{\mathbf{K}}$ is timelike at $p$, contradicting
  the fact that $p\in \mathring{\mathcal{U}}$. But then, recall from
  \Cref{lemma:y-z-derivatives:S-small} that
  $\Metric(\KillT, \mathbf{Y}) = O(\varepsilon_{\mathcal{S}})$, so
  that for $\varepsilon_0$ sufficiently small, the triplets
  $\curlyBrace*{\KillT, \mathbf{Y}, \mathbf{Z}}$ and
  $\curlyBrace*{\underline{\mathbf{K}}, \mathbf{Y}, \mathbf{Z}}$ are
  linearly independent. Thus, we must actually have that
  $\KillT, \underline{\mathbf{K}}$ are linearly dependent at all
  points in $\mathring{\mathcal{U}}$. Let us define the function
  $\mathfrak{k}:\mathring{\mathcal{U}}\to \Real$ such that
  \begin{equation*}
    \underline{\mathbf{K}} = \mathfrak{k}\KillT, \qquad \text{in }\mathring{\mathcal{U}}.
  \end{equation*}

  Now assume for the sake of contradiction that
  $\mathring{\mathcal{U}}\neq \emptyset$. Then there exists a point
  $p_0\in \mathring{\mathcal{U}}$ such that
  $y(p_0) = \sup_{p\in \mathring{\mathcal{U}}}y(p)$ since
  $\mathring{\mathcal{U}}\cap \underline{\Sigma}_2$ is compact.
  We may assume that $p_0\in \mathring{\mathcal{U}}\cap \underline{\Sigma}_2$.
  Then, from the definition of $\mathfrak{k}$ above, we have that
  there exists some $\mathfrak{k}_0\in \Real$ such that
  $\underline{\mathbf{K}}_{p_0} = \mathfrak{k}_0\KillT_{p_0}$.  Now
  let $\gamma_{p_0}(t):[-\delta, \delta]\to \mathbf{E}$ be a portion
  of the integral curve of $\mathbf{Y}$ that goes through $p_0$. Using
  \Cref{eq:timelike-span-of-Killing-VFs:Y-spacelike}, we then have
  that for $\delta$ sufficiently small, 
  \begin{equation*}
    t\in (0, \delta)\implies y(\gamma_{p_0}(t)) > y(p_0). 
  \end{equation*}
  Since we assumed that
  $y(p_0) = \sup_{p\in \mathring{\mathcal{U}}}y(p)$, we must therefore have that
  \begin{equation*}
    \mathring{\mathcal{U}}\bigcap \curlyBrace*{\gamma_{p_0}(t):t\in (0, \delta) } = \emptyset.
  \end{equation*}
  Therefore,
  \begin{equation*}
    \evalAt*{\Metric(\KillT, \KillT)}_p<0,\qquad p \in \curlyBrace*{\gamma_{p_0}(t):t\in (0, \delta) },
  \end{equation*}
  and moreover $\evalAt*{\Metric(\KillT,\KillT)}_{p_0}=0$. It then follows that
  \begin{equation*}
    \evalAt*{\mathbf{Y}(\Metric(\KillT, \KillT)(y^2+z^2))}_{p_0} \le 0.
  \end{equation*}

  Recalling that
  \begin{align*}
    -\Metric(\KillT, \KillT)
    = \Re\sigma = c_{S_0} + \Re\left(\frac{-b_{S_0}}{P}(1+\widetilde{b}) - \frac{\Lambda}{3}P^2\right) + O(\varepsilon_{\mathcal{S}}),
  \end{align*}
  from \Cref{lemma:R-sigma-in-terms-of-P:S-small}, and using \Cref{lemma:b-c-k-almost-constant:S-small}, we have that 
  \begin{equation*}
    -\Metric(\KillT, \KillT) = c_{S_0} - \frac{b_{S_0}y}{y^2+z^2} - \frac{\Lambda}{3}(y^2-z^2) + O(\varepsilon_{\mathcal{S}}). 
  \end{equation*}
  We can then evaluate using \Cref{lemma:y-z-derivatives:S-small} that
  \begin{align*}
    -\CovariantDeriv_{\mathbf{Y}}\left( \Metric(\KillT, \KillT)(y^2+z^2) \right)
    ={}& \left(2c_{S_0}y - b_{S_0} - \frac{4\Lambda}{3}y^3\right)\CovariantDeriv_{\mathbf{Y}}y
         + \left(2c_{S_0}z + \frac{4\Lambda}{3}z^3\right)\CovariantDeriv_{\mathbf{Y}}z\\
    ={}& \left(2c_{S_0}y - b_{S_0} - \frac{4\Lambda}{3}y^3\right)\CovariantDeriv_{\mathbf{Y}}y
         + O(\varepsilon_{\mathcal{S}}). 
  \end{align*}
  Thus, for $\varepsilon_{\mathcal{S}}$ sufficiently small, we in fact
  have that
  \begin{equation*}
    \CovariantDeriv_{\mathbf{Y}}\left( \Metric(\KillT, \KillT)(y^2+z^2) \right) > 0. 
  \end{equation*}
  But this is a contradiction, and thus, $\mathring{\mathcal{U}}=\emptyset$.
\end{proof}

\subsection{Uncovering the rotational vectorfield}

In this section, we show that the existence of the second Killing
vectorfield $\underline{\mathbf{K}}$ implies the existence of a
rotational vectorfield.
We first prove the existence of a rotational vectorfield in a
neighborhood of $\underline{S}_0$. 
\begin{prop}
  \label{prop:Z-existence}
  There is a constant $\underline{\lambda}_0\in \Real$ and an open neighborhood
  $\underline{\mathbf{O}}'\subset \underline{\mathbf{O}}$ of
  $\underline{S}_0$ such that the vectorfield
  \begin{equation*}
    \underline{\KillPhi} = \KillT + \underline{\lambda}_0\underline{\mathbf{K}}
  \end{equation*}
  has periodic orbits in $\underline{\mathbf{O}}'$. 
\end{prop}

To prove \Cref{prop:Z-existence}, we first prove two auxiliary lemmas.
\begin{lemma}
  \label{lemma:Z-existence:prelim}
  We have that
  \begin{equation}
    \label{eq:Z-existence:prelim}
    \begin{gathered}
      [\KillT, \underline{L}_{-}] = f_- \underline{L}_{-}\qquad \text{and}\qquad \underline{L}_-(f_-) = 0\qquad \text{on }\CosmologicalHorizonPast,\\
      [\KillT, \underline{L}_+] = f_+ \underline{L}_+\qquad \text{and}\qquad \underline{L}_+(f_+) = 0\qquad \text{on }\CosmologicalHorizonFuture.
    \end{gathered}    
  \end{equation}
\end{lemma}
\begin{proof}
  We prove the result on $\CosmologicalHorizonFuture$. The result on
  $\CosmologicalHorizonPast$ follows similarly.  Recall from
  \Cref{prop:non-expanding-null-hypersurface:basic-props} that on
  $\CosmologicalHorizonFuture$, $\underline{\chi}$ is symmetric and
  that $\KillT$ is both Killing and tangent to
  $\CosmologicalHorizonFuture$. As a result, we have that for every
  $X\in T\CosmologicalHorizonFuture$,
  \begin{align*}
    \Metric([\KillT, \underline{L}_-], X)
    ={}& \Metric(\CovariantDeriv_{\KillT}\underline{L}_-, X)
         - \Metric(\CovariantDeriv_{\underline{L}_-}\KillT, X)\\
    ={}& \Metric(\CovariantDeriv_{\KillT}\underline{L}_-, X)
         + \Metric(\CovariantDeriv_X\KillT, \underline{L}_-)\\
    ={}& \Metric(\CovariantDeriv_{\KillT}\underline{L}_-, X)
         - \Metric(\CovariantDeriv_X\KillT, \underline{L}_-)\\
    ={}& \underline{\chi}(\KillT, X)
         -\underline{\chi}(X, \KillT)=0.
  \end{align*}
  As a result, there exists some $f_-$ such that
  $[\KillT, \underline{L}_-] = f_-\underline{L}_-$. Since
  $\underline{L}_-$ is geodesic and $\KillT$ is Killing, we then have
  that
  \begin{align*}
    0 ={}& \LieDerivative_{\KillT}\left(\CovariantDeriv_{\underline{L}_-}\underline{L}_-\right)\\
    ={}& \CovariantDeriv_{\LieDerivative_{\KillT}\underline{L}_-}\underline{L}_-
         + \CovariantDeriv_{\underline{L}_-}\left( \LieDerivative_{\KillT}\underline{L}_- \right)\\
    ={}& \CovariantDeriv_{f_-\underline{L}_-}\underline{L}_-
         + \CovariantDeriv_{\underline{L}_-}(f_-\underline{L}_-)\\
    ={}& \underline{L}_-(f_-)\underline{L}_-.
  \end{align*}
  This prove the second equation in \Cref{eq:Z-existence:prelim}. 
\end{proof}

\begin{lemma}
  \label{lemma:Z-existence:main-lemma}
  There is a constant $\underline{t}_0>0$ such that $\Psi_{\underline{t}_0,\KillT} = \Id$ in
  $\underline{\mathbf{O}}'$. In addition, there is a constant
  $\lambda_0\in \Real$ and a choice of null frame
  $(\underline{L}_+, \underline{L}_-)$ along $\underline{S}_0$ satisfying
  \Cref{eq:L+L-:normalization-on-horizon} such that
  \begin{equation}
    \label{eq:Z-existence:main-lemma}
    [\KillT, \underline{L}_-] = \underline{\lambda}_0\underline{L}_-
    \qquad\text{ and }\qquad
    [\KillT, \underline{L}_+] = \underline{\lambda}_0\underline{L}_+
    \qquad\text{on }\underline{S}_0.
  \end{equation}
\end{lemma}
\begin{proof}
  The existence of the period $\underline{t}_0$ is a standard fact
  concerning Killing vectorfields on the sphere. We remark that in the
  special case that $\KillT=0$ on $\underline{S}_0$, any value of
  $\underline{t}_0>0$ is suitable. In this case, the conclusion of
  \Cref{prop:Z-existence} is that
  $\KillT + \underline{\lambda}_0\underline{\mathbf{K}}=0$ in
  $\underline{\mathbf{O}}'$ for some $\underline{\lambda}_0\in \Real$.

  Observe that given \Cref{eq:Z-existence:prelim}, it suffices to
  prove that there is some $\underline{\lambda}_0\in \Real$ and a choice of
  $\left( \underline{L}_+, \underline{L}_- \right)$ on $\underline{S}_0$ such that
  \begin{equation*}
    \Metric\left( [\KillT, \underline{L}_+], \underline{L}_- \right) = 2\underline{\lambda}_0,
    \qquad
    \Metric\left( [\KillT, \underline{L}_-], \underline{L}_+ \right) = 2\underline{\lambda}_0,
  \end{equation*}
  to prove \Cref{eq:Z-existence:main-lemma}.

  Observe that this is equivalent to showing that on $\underline{S}_0$,
  \begin{align*}
    \KillT^{\alpha}(\underline{L}_-)^{\beta}\CovariantDeriv_{\alpha}(\underline{L}_+)_{\beta}
    - (\underline{L}_+)^{\alpha}(\underline{L}_-)^{\beta}\CovariantDeriv_{\alpha}\KillT_{\beta}
    ={}& 2\underline{\lambda}_0,\\
    \KillT^{\alpha}(\underline{L}_+)^{\beta}\CovariantDeriv_{\alpha}(\underline{L}_-)_{\beta}
    - (\underline{L}_-)^{\alpha}(\underline{L}_+)^{\beta}\CovariantDeriv_{\alpha}\KillT_{\beta}
    ={}& 2\underline{\lambda}_0,
  \end{align*}
  which are equivalent identities since $\KillT$ is Killing, and are
  themselves equivalent to the identity
  \begin{equation}
    \label{eq:Z-existence:main-lemma:reduction}
    \underline{\lambda}_0 = F_{43} - 2\Metric(\zeta, \KillT)
  \end{equation}
  on $\underline{S}_0$, where we are denoting
  $(e_3,e_4) = (\underline{L}_-, -\underline{L}_+)$.
  Let us now define
  \begin{equation*}
    \widetilde{\lambda} = F_{43} - 2\Metric(\zeta, \KillT),
  \end{equation*}
  which we will show is constant on $\underline{S}_0$.

  Recall also from \Cref{eq:frame-transformation:zeta} that under a
  change of frame that only involves rescaling i.e. for
  $(e_4',e_3') = (\lambda e_4, \lambda^{-1}e_3)$, 
  \begin{equation*}
    \zeta' = \zeta - \nabla\log \lambda.
  \end{equation*}
  In the prime frame,
  \begin{align*}
    \widetilde{\lambda}' ={}& F_{4'3'} - 2\Metric(\zeta', \KillT)\\
    ={}& F_{43}-2\Metric(\zeta, \KillT)
         + \KillT(\log f)\\
    ={}& \widetilde{\lambda}+ \KillT(\log \lambda).
  \end{align*}
  We see then that defining $\widehat{\lambda}$ to be the average of
  $\widetilde{\lambda}$ along the integral curves of $\KillT$, we see
  that if we can find $\lambda$ such that
  \begin{equation*}
    \KillT(\log \lambda) = -\widetilde{\lambda} + \widehat{\lambda},
  \end{equation*}
  then $\widetilde{\lambda}'$ is constant. Thus, it only remains to
  prove that $\widehat{\lambda}$ is constant along $\underline{S}_0$.

  To this end, observe that since $\KillT$ is Killing, we have that
  \begin{equation*}
    \CovariantDeriv_{\alpha}\CovariantDeriv_{\beta}\KillT_{\gamma}
    = \KillT^{\lambda}\Riem_{\lambda\alpha\beta\gamma}.
  \end{equation*}
  Then, from the Ricci formula in \Cref{eq:Ricci-formulas} and
  \Cref{coro:non-expanding-null-hypersurface:horizon-qtys}, we see
  that on $\underline{S}_0$,
  \begin{equation*}
    \KillT^{\lambda}\Riem_{\lambda a 43}
    =\CovariantDeriv_{a}\CovariantDeriv_4\KillT_3
    = e_a\left( \CovariantDeriv_4\KillT_3 \right)
    = e_a\left( F_{43} \right). 
  \end{equation*}
  Next, observe that since $(\mathcal{M}, \Metric)$ is a solution to
  ($\Lambda$-EVE), we have that
  \begin{equation*}
    \Riem_{ab34} = \Weyl_{ab34} = 2\in_{ab}\LeftDual{\rho}(\Weyl). 
  \end{equation*}
  Thus, since $\KillT$ is tangent to $\underline{S}_0$,
  \begin{equation}
    \label{eq:Z-existence:main-lemma:ea-F43}
    e_a\left( F_{43} \right) = \KillT^{\lambda}\Riem_{\lambda a 43} = 2\KillT^b\in_{ab}\LeftDual{\rho}(\Weyl).  
  \end{equation}
  Since $\widetilde{\lambda}$ is clearly constant if $\KillT=0$
  identically on $\underline{S}_0$, we assume in the remainder of the
  proof that $\Theta$, subset of $\underline{S}_0$ on which $\KillT$
  vanishes, is finite. But then observe that
  \begin{align}
    e_a(\zeta, \KillT)
    ={}& \nabla_a\zeta_b\KillT^b + \zeta_b\nabla_a\KillT_b\notag\\
    ={}& \left( \nabla_a\zeta_b-\nabla_b\zeta_a \right)\KillT^b
         + \zeta^b\nabla_a\KillT_b
         + \nabla_{\KillT}\zeta_a\notag\\
    ={}& \in_{ab}\curl\zeta\KillT^b
         + \zeta^b\nabla_a\KillT_b
         + \nabla_{\KillT}\zeta_a.
         \label{eq:Z-existence:main-lemma:ea-zeta-T}
  \end{align}
  From \Cref{coro:non-expanding-null-hypersurface:horizon-qtys}, we
  have that on $\underline{S}_0$, $\chi = \underline{\chi} = 0$. Thus,
  we have from \Cref{eq:null-structure:curl-zeta} that on $\underline{S}_0$,
  \begin{equation}
    \label{eq:Z-existence:main-lemma:curl-zeta}
    \curl \zeta = \LeftDual{\rho}(\Weyl). 
  \end{equation}
  As a result, combining
  \Cref{eq:Z-existence:main-lemma:ea-F43,eq:Z-existence:main-lemma:ea-zeta-T,eq:Z-existence:main-lemma:curl-zeta},
  we have that
  \begin{align*}
    e_a\widetilde{\lambda}
    ={}& e_a(F_{43}) - 2e_a(\zeta\cdot\KillT)\\
    ={}& - 2\zeta^b\nabla_a\KillT_b
         - 2\nabla_{\KillT}\zeta_a.
  \end{align*}
  Next, consider the orthonormal frame $(e_1,e_2)$ on
  $\underline{S}_0\backslash \Theta$ given by fixing
  \begin{equation*}
    e_1 = -N^{-1}\KillT.
  \end{equation*}
  Since $e_1(N)=0$, we must have that
  \begin{equation*}
    \nabla_{\KillT}e_2 = -F_{12}e_1. 
  \end{equation*}
  But then
  \begin{align*}
    \frac{1}{2}\nabla_2\widetilde{\lambda}
    ={}& - \zeta^1\nabla_2\KillT_1
         - \zeta^2\nabla_2\KillT_2
         - \Metric(\nabla_{\KillT}\zeta, e_2)\\
    ={}& - \zeta^1F_{21}
         - \KillT \Metric(\zeta, e_2)
         + \Metric(\zeta, \nabla_{\KillT}e_2)\\
    ={}& - \KillT \Metric(\zeta, e_2) - \zeta^1F_{21} - \zeta^1F_{12}\\
    ={}& - \KillT(\zeta_2),
  \end{align*}
  so we have that
  \begin{equation}
    \label{eq:Z-existence:main-lemma:nabla-2-tilde-lambda}
    \nabla_2\widetilde{\lambda} = -2\KillT(\zeta_2).
  \end{equation}
  Let us now fix a non-trivial orbit $\gamma_0$ of $\KillT$ in
  $\underline{S}_0\backslash \Theta$. Then consider some vector
  $\mathbf{V}$ on $\gamma_0 $ such that
  $\Metric(\mathbf{V}, \mathbf{V})=1$, and extend it by parallel
  transport along the geodesics perpendicular to $\gamma_0$.  Now let
  $\varphi$ be such that $\mathbf{V}(\varphi)=1$, and $\varphi=0$ on
  $\gamma_0$. This defines a system of coordinates $(t, \varphi)$ in a
  neighborhood $U$ of $\gamma_0$ such that in $U$,
  \begin{equation*}
    \partial_t = \KillT,\qquad
    \nabla_{\partial_{\varphi}}\partial_{\varphi} =0,
  \end{equation*}
  and on $\gamma_0$,
  \begin{equation*}
    \Metric(\partial_t,\partial_{\varphi}) = 0, \qquad
    \Metric(\partial_{\varphi}, \partial_{\varphi}) = 1. 
  \end{equation*}
  We now recover the full form of the metric in $U$.  Since
  $\partial_t$ is Killing, we see that $N$ and
  $\Metric(\partial_{\varphi}, \partial_{\varphi})$ must be
  independent of $t$ in $U$. Moreover,
  \begin{align*}
    \partial_{\varphi}\Metric(\partial_t, \partial_{\varphi})
    ={}& \Metric(\nabla_{\partial_{\varphi}}\partial_t, \partial_{\varphi})
         + \Metric(\partial_t, \nabla_{\partial_{\varphi}}\partial_{\varphi})\\
    ={}& \Metric(\nabla_{\partial_{\varphi}}\partial_t, \partial_{\varphi})\\
    ={}&  \frac{1}{2}\partial_t\Metric(\partial_{\varphi}, \partial_{\varphi})\\
    ={}& 0.
  \end{align*}
  Therefore, since $\Metric(\partial_t,\partial_{\varphi})=0$ on
  $\gamma_0$, we must have that
  $\Metric(\partial_t,\partial_{\varphi})=0$ in $U$. Similarly,
  \begin{equation*}
    \partial_{\varphi}\Metric(\partial_{\varphi}, \partial_{\varphi})
    = 2\Metric(\nabla_{\partial_{\varphi}}\partial_{\varphi}, \partial_{\varphi})
    =0,
  \end{equation*}
  since $\partial_{\varphi}$ is geodesic, so $\Metric(\partial_{\varphi}, \partial_{\varphi})=1$ in all of $U$.
  We have thus proven that in $U$, the metric takes the form
  \begin{equation*}
    \Metric = d\varphi^2 - N(\varphi)dt^2.
  \end{equation*}
  Therefore, with $\KillT= \partial_t$, $e_2 = \partial_{\varphi}$, we
  see from \Cref{eq:Z-existence:main-lemma:nabla-2-tilde-lambda} that in $U$,
  \begin{equation*}
    \partial_{\varphi}\widetilde{\lambda} = - \partial_t\Metric(\zeta, \partial_{\varphi}).
  \end{equation*}
  Integrating in $t$ and using the fact that the orbits of $\KillT$
  are closed, we infer that $\widehat{\lambda}$ is constant along
  $\underline{S}_0$, as desired.
\end{proof}

We can now prove \Cref{prop:Z-existence}.
\begin{proof}[Proof of \Cref{prop:Z-existence}]
  From \Cref{lemma:Z-existence:prelim} and
  \Cref{lemma:Z-existence:main-lemma}, we have that
  \begin{equation*}
    \begin{gathered}
      [\KillT, \underline{L}_-] = \underline{\lambda}_0\underline{L}_-, \qquad \text{on } \CosmologicalHorizonPast\bigcap \underline{\mathbf{O}}',\\
      [\KillT, \underline{L}_+] = -\underline{\lambda}_0\underline{L}_+, \qquad \text{on } \CosmologicalHorizonFuture\bigcap \underline{\mathbf{O}}'.
    \end{gathered}    
  \end{equation*}
  Then, from \Cref{coro:K-commutation-horizon-generators}, we have
  that on $\CosmologicalHorizonPast\bigcap \underline{\mathbf{O}}'$, 
  \begin{equation*}
    [\underline{\KillPhi}, \underline{L}_-] = [\KillT + \underline{\lambda}_0\underline{\mathbf{K}}, \underline{L}_-] = \underline{\lambda}_0\underline{L}_- - \underline{\lambda}_0\underline{L}_-=0. 
  \end{equation*}
  Then, since $\underline{\KillPhi}$ is Killing and $\underline{L}_-$ is geodesic,
  we have that in $\underline{\mathbf{O}}'$,
  \begin{equation*}
    \CovariantDeriv_{\underline{L}_-}\LieDerivative_{\underline{\KillPhi}}\underline{L}_-
    = \LieDerivative_{\underline{\KillPhi}}(D_{\underline{L}_-}\underline{L}_-)
    - \CovariantDeriv_{\LieDerivative_{\underline{\KillPhi}}\underline{L}_-}\underline{L}_{-}
    = - \CovariantDeriv_{\LieDerivative_{\underline{\KillPhi}}\underline{L}_-}\underline{L}_{-}.
  \end{equation*}
  As a result, we have a transport equation for
  $[\underline{\KillPhi}, \underline{L}_-]$ that implies that
  $[\underline{\KillPhi}, \underline{L}_-]=0$ in $\underline{\mathbf{O}}'$.  A
  similar argument shows that $[\underline{\KillPhi}, \underline{L}_+]=0$ in
  $\underline{\mathbf{O}}'$. The conclusion then follows from the first
  claim in \Cref{lemma:Z-existence:main-lemma} and the fact that
  $[\underline{\KillPhi}, \underline{L}_-] = [\underline{\KillPhi}, \underline{L}_+]=0$ in
  $\underline{\mathbf{O}}'$.
\end{proof}

Next, we show that $\underline{\KillPhi}$ extends as a rotational vectorfield in
$\underline{\mathbf{E}}_{y_{\underline{\mathbf{K}}}}$, where we
already have extended $\underline{\mathbf{K}}$ as a Killing
vectorfield. 
\begin{prop}
  \label{prop:Z-extension}
  The rotational vectorfield $\underline{\KillPhi}$ extends to
  $\underline{\mathbf{E}}_{y_{\underline{\mathbf{K}}}}$ as a
  rotational vectorfield.
\end{prop}
\begin{proof}
  From \Cref{prop:Hawking-extension}, we see that in
  $\mathbf{E}\bigcup \underline{\mathbf{O}}'$,
  \begin{equation*}
    \LieDerivative_{\underline{\KillPhi}}\Metric=0, \qquad
    [\KillT, \underline{\KillPhi}]=[\underline{\mathbf{K}}, \underline{\KillPhi}]=[\mathbf{Y},\underline{\mathbf{K}}]=0, \qquad
    \underline{\KillPhi}\cdot \bm{\sigma} = \underline{\KillPhi}\cdot\mathbf{Y}=0,
  \end{equation*}
  Using \Cref{lemma:Z-existence:main-lemma}, we have that for any
  $p\in \underline{\mathbf{O}}'\bigcap \mathbf{E}$ and $s\in \Real$,
  \begin{equation}
    \label{eq:Z-extension:T-Lie-drag}
    \Psi_{s, \KillT}(p)
    = \Psi_{s,\KillT}\Psi_{\underline{t}_0,\underline{\KillPhi}}(p)
    = \Psi_{\underline{t}_0,\underline{\KillPhi}}\Psi_{s,\KillT}(p),
  \end{equation}
  where we used that $[\KillT, \underline{\KillPhi}]=0$. It follows then that
  $\Psi_{\underline{t}_0, \underline{\KillPhi}}(p)=p$ for any
  $p\in \underline{\mathbf{E}}_{\underline{y}_0}$.

  We now extend this to all of
  $\mathbf{E}_{y_{\underline{\mathbf{K}}}}$. Assume for the sake of
  contradiction that
  \begin{equation*}
    \underline{R}_0 \vcentcolon = \inf\curlyBrace*{y(p): \Psi_{\underline{t}_0,\underline{\KillPhi}}(p) = p} > y_{\underline{\mathbf{K}}}. 
  \end{equation*}
  But then, we have that $[\mathbf{Y}, \underline{\KillPhi}]=0$, so in fact
  \begin{equation*}
    \Psi_{s, \mathbf{Y}}(p)
    = \Psi_{s, \mathbf{Y}}\Psi_{\underline{t}_0, \underline{\KillPhi}}(p)
    = \Psi_{\underline{t}_0, \underline{\KillPhi}}\Psi_{s, \mathbf{Y}}(p)
  \end{equation*}
  for any $p\in $, and $s\in \Real$. 
  \begin{equation*}
    \Psi_{\underline{t}_0, \underline{\KillPhi}}(p) = p,\qquad \forall p\in \underline{\mathbf{E}}_{\underline{R}_0-\delta'},
  \end{equation*}
  for some $\delta'$ sufficiently small, which is a contradiction.
\end{proof}
\subsection{Proof of \texorpdfstring{\Cref{thm:main:e2c2}}{second main theorem}}

We are now ready to prove \Cref{thm:main:e2c2}. To this end, we will
show that the rotational vectorfields $\KillPhi$ and
$\underline{\KillPhi}$ generated by $\mathbf{K}$ and
$\underline{\mathbf{K}}$ respectively are in fact the same vectorfield
on the region where $\mathbf{K}$ and $\underline{\mathbf{K}}$ are both
defined. As a result, we will have a global vectorfield generating axisymmetry.

We now show that the existence of $\underline{\mathbf{K}}$ also
implies the existence of a vectorfield generating axisymmetry. Let
$\underline{\KillPhi}$ be the rotational vectorfield in
$\underline{\mathbf{O}}'$ constructed in \Cref{prop:Z-extension}, and
let $\KillPhi$ be the equivalent vectorfield constructed from
$\mathbf{K}$, so that for some $\delta$ sufficiently small,
\begin{align}
  \label{eq:KillPhi:identification:def}
  \KillPhi = \KillT + \lambda_0\mathbf{K},
  &\qquad \Psi_{t_0, \KillPhi}(p) = p\, \forall p\in \mathbf{E}_{y_{*}+\delta},\\
  \label{eq:KillPhiBar:identification:def}
  \underline{\KillPhi} = \KillT + \underline{\lambda}_0\underline{\mathbf{K}},
  &\qquad \Psi_{\underline{t}_0, \underline{\KillPhi}}(p) = p\, \forall p\in \underline{\mathbf{E}}_{y_{*}-\delta}.
\end{align}
We show now that in fact $\KillPhi$ and $\underline{\KillPhi}$ are the
same vectorfield. This would immediately imply that there exists a
global rotational vectorfield.

\begin{lemma}
  \label{lemma:identifying-rotational-VFs}
  Let $\KillPhi$ and $\underline{\KillPhi}$ be as defined in
  \Cref{eq:KillPhi:identification:def} and
  \Cref{eq:KillPhiBar:identification:def} respectively. Then
  $\KillPhi$ and $\underline{\KillPhi}$ are identical (up to a
  rescaling) in the region where they are both defined.
\end{lemma}
We remark that this lemma directly implies that
$\Span(\underline{\mathbf{K}}, \KillT)=\Span(\mathbf{K}, \KillT)$. 
\begin{proof}
  We first observe that for $\delta$ sufficiently small, both
  $\KillPhi$ and $\underline{\KillPhi}$ are defined in
  $\mathbf{E}_{y_{*}+\delta}\bigcap
  \underline{\mathbf{E}}_{y_{*}-\delta}$. Let us consider a domain
  $\mathbf{E}'\subset \mathbf{E}_{y_{*}+\delta}\bigcap
  \underline{\mathbf{E}}_{y_{*}-\delta}$ such that
  $\KillT, \mathbf{Y}, \mathbf{Z}, \KillPhi$ form a basis on
  $\mathbf{E}'$. 

  Combining
  \Cref{eq:KBar-dot-Y-dot-Z-vanish} (and the equivalent identity for
  $\mathbf{K}$) and \Cref{lemma:P-one-form-P-J-inverse-relation}, we
  see that
  \begin{equation*}
    \KillPhi\cdot\mathbf{Y}
    = \KillPhi\cdot\mathbf{Z}
    = \underline{\KillPhi}\cdot\mathbf{Y}
    = \underline{\KillPhi}\cdot\mathbf{Z}
    = 0.
  \end{equation*}
  Then, we see that on $\mathbf{E}'$,
  \begin{equation*}
    \underline{\KillPhi}
    = f\KillPhi + g\KillT,\qquad
    f, g\in C^{\infty}(\mathbf{E}';\Real),
  \end{equation*}
  where
  $\mathbf{Y}(f)=\mathbf{Y}(g)=\mathbf{Z}(f)=\mathbf{Z}(g)=0$. Since
  both $\KillPhi$ and $\underline{\KillPhi}$ commute with $\KillT$, we
  have that
  \begin{align*}
    0 ={}& \LieDerivative_{\KillT}\underline{\KillPhi}\\
    ={}& \KillT(f) \KillPhi
         + \KillT(g) \KillT,
  \end{align*}
  which implies that
  \begin{equation}
    \label{eq:identifying-rotational-VFs:DT-0}
    \KillT(f) = \KillT(g) = 0.
  \end{equation}  
  But since $\underline{\KillPhi}$ is itself Killing, we see that
  \begin{equation}
    \label{eq:identifying-rotational-VFs:aux}
    \KillPhi_{(\mu}\CovariantDeriv_{\nu)}f
    + \KillT_{(\mu}\CovariantDeriv_{\nu)}g
    = 0. 
  \end{equation}
  Taking the trace of \Cref{eq:identifying-rotational-VFs:aux}, we find that
  \begin{equation*}
    \KillPhi (f) + \KillT (g) = 0.
  \end{equation*}
  Along with \Cref{eq:identifying-rotational-VFs:DT-0}, this implies
  that $\KillPhi(f)=0$.  Finally, contracting
  \Cref{eq:identifying-rotational-VFs:aux} against $\KillT$ and
  $\KillPhi$, we find that
  \begin{equation*}
    \KillPhi (f)\left( \KillPhi\cdot\KillT  \right)
    + \KillT (f)\left( \KillPhi\cdot\KillPhi  \right)
    +  \KillPhi (g)\left( \KillT\cdot\KillT  \right)
    + \KillT (g)\left( \KillT\cdot\KillPhi  \right)
    = 0.
  \end{equation*}
  But since we know that $\KillT(f)=\KillT(g)=\KillPhi(f)=0$, and
  moreover, for $\delta$ sufficiently small, $\KillT\cdot\KillT<0$ in
  $\mathbf{E}'$ from the subextremality assumption, we have that
  \begin{equation*}
    \KillPhi (g)=0. 
  \end{equation*}
  This implies that $f$ and $g$ are both constants. But we see that
  $g$ must then vanish, since the orbits of $\underline{\KillPhi}$ are
  closed, but the orbits of $\KillT$ are not. 
\end{proof}

\section{Proof of \texorpdfstring{\Cref{thm:main}}{third main theorem}}
\label{sec:MST:global}

In this section, we show how to prove \Cref{thm:main}.  The main idea
for both of these proofs is to use the fact that from
\Cref{prop:Hawking-extension}, we have that there exists a second
Killing vectorfield $\underline{\mathbf{K}}$ on
$\underline{\mathbf{E}}_{y_{*}}$, and moreover, that
$\underline{\mathbf{K}}$ and $\KillT$ span a Killing vectorfield. As
can be seen from
\Cref{coro:pseudo-convexity:outside-ergoregion:multiple-vectors}, this
makes satisfying the requisite pseudoconvexity condition in
\Cref{eq:strict-T-Z-null-convexity:main-condition} much simpler.

We prove \Cref{thm:main} by proving the following
proposition. Afterwards, a direct application of Theorem 1 of
\cite{marsSpacetimeCharacterizationKerrNUTAde2015}  will suffice to
prove \Cref{thm:main}.
\begin{prop}
  \label{prop:main-prop:global}
  Under the stationary, smooth bifurcate sphere, subextremality, and
  the technical assumptions \ref{ass:E1} and \ref{ass:C2} of
  \Cref{sec:assumptions}, we have that $\mathcal{S} = 0$ in
  $\mathcal{U}_{y_{\underline{S}_0}}$.
\end{prop}

Recall that from \Cref{prop:main-prop}, we already know that
$\mathcal{S}=0$ in $\mathcal{M}_{\mathcal{S}}$. In order to extend the
vanishing of $\mathcal{S}$ to $\mathbf{E}$ and prove
\Cref{prop:main-prop:global}, we will rerun the bootstrap argument
used to prove \Cref{prop:main-prop} in
$\mathcal{M}_{\underline{\mathbf{K}}}$. Using the fact that we have
the existence of $\underline{\mathbf{K}}$ in
$\mathcal{M}_{\underline{\mathbf{K}}}$, we will be able to use
$\curlyBrace*{\KillT, \underline{\mathbf{K}}}$-pseudoconvexity to
improve the bootstrap assumption. 
To this end, recall the definition of $R_0$
in the bootstrap assumptions made in \Cref{eq:S-extension:BSA}. From
\Cref{prop:main-extension}, we have already shown that $R_0>y_{*}$. We
now show that in fact, $R_0 = y_{\underline{S}_0}$.

\begin{prop}
  \label{prop:main-extension:global}
  Let
  $x_0\in \partial_{\Sigma_0\bigcap \mathbf{E}}(\mathcal{U}_{R_0})$,
  where $R_0$ satisfies
  $y_{\underline{\mathbf{K}}} < y_{*}<R_0 < y_{\underline{S}_0}$.
  Then there exists some
  $r_3 = r_3(A_0, \widetilde{A}_{\widetilde{C}^{-1}}, R_0)\in (0,
  r_0)$ (potentially smaller than the $r_3$ found in
  \Cref{prop:main-extension}) such that $\mathcal{S}=0$ in
  $B_{r_3}(x_0)$.
\end{prop}
The proof of \Cref{prop:main-extension:global} is almost identical to
\Cref{prop:main-prop}. The only difference is that the
$\KillT$-pseudoconvexity of $y$ for $y< y_{\mathcal{S}}$ is replaced
by the $\{\KillT, \underline{\mathbf{K}}\}$-pseudoconvexity of $y$ for
$y> y_{\underline{\mathbf{K}}}$.  As before, the main ingredient
needed to prove \Cref{prop:main-extension:global} is a Carleman
inequality.
\begin{lemma}
  \label{lemma:main-extension:Carleman:global}
  Let
  $x_0\in \partial_{\Sigma_0\bigcap \mathbf{E}}(\mathcal{U}_{R_0})$,
  where $R_0$ satisfies
  $y_{\underline{\mathbf{K}}} < y_{*}<R_0 < y_{\underline{S}_0}$.
  There is some $0<\varepsilon<r_0$ sufficiently small and some
  $C(\varepsilon)$ sufficiently large such that for any
  $\lambda\ge C(\varepsilon)$ and any
  $\phi\in C_0^{\infty}(B_{\varepsilon^{10}}(x_0))$,
  \begin{equation}
    \label{eq:main-extension:Carleman:global}
    \begin{split}
      \lambda \norm*{e^{-\lambda \widetilde{f_{\varepsilon}}}\phi}_{L^2}
      + \norm*{e^{-\lambda\widetilde{f}_{\varepsilon}}\abs*{D^1\phi}}_{L^2}
      \le{}& C(\varepsilon)\lambda^{-\frac{1}{2}}\norm*{e^{- \lambda\widetilde{f}_{\varepsilon}}\Box_{\Metric}\phi}_{L^2}
             + \varepsilon^{-6}\norm*{e^{-\lambda \widetilde{f}_{\varepsilon}}\KillT(\phi)}_{L^2}\\
           & + \varepsilon^{-6}\norm*{e^{-\lambda \widetilde{f}_{\varepsilon}}\underline{\mathbf{K}}(\phi)}_{L^2}
             ,             
    \end{split}
  \end{equation}
  where, 
  \begin{equation}
    \label{eq:main-extension:f-tilde-epsilon-def:global}
    \widetilde{f}_{\varepsilon} = \ln\left(
      y_{\varepsilon}
      - y_0(x_0)
      + \varepsilon^{12}N^{x_0}\right).
  \end{equation}
\end{lemma}
\begin{proof}
  We will apply \Cref{prop:carleman-estimate} with
  \begin{equation*}
    \{\mathbf{V}_i\}_{i=1}^2 = \{\KillT, \underline{\mathbf{K}}\},
    \qquad
    h_{\varepsilon }=
    \widetilde{y}_{\varepsilon} - \widetilde{y}_0(x_0),
    \qquad
    e_{\varepsilon} = \varepsilon^{12}N^{x_0}.
  \end{equation*}
  From the definition in \Cref{def:negligible perturbation}, it is
  clear that $e_{\varepsilon}$ is a negligible perturbation if
  $\varepsilon$ is sufficiently small.

  We now show that for some $\varepsilon_1$ sufficiently small,
  $\{h_{\varepsilon}\}_{\varepsilon\in (0,\varepsilon_1)}$ forms a
  $\curlyBrace*{\KillT, \underline{\mathbf{K}}}$-pseudoconvex family
  of weights. From the definition of $h_{\varepsilon}$, we have that
  \begin{equation*}
    h_{\varepsilon}(x_0) = \varepsilon,\qquad
    \KillT(h_{\varepsilon})(x_0)  = \underline{\mathbf{K}}(h_{\varepsilon})(x_0)  =0,\qquad
    \abs*{D^jh_{\varepsilon}}\lesssim 1,
  \end{equation*}
  so \Cref{eq:strict-T-Z-null-convexity:cond1} is satisfied for
  $\varepsilon_1$ sufficiently small. Recall then from
  \Cref{prop:timelike-span-of-Killing-VFs} that in
  $\underline{\mathbf{E}}_{y_{\underline{\mathbf{K}}}}$,
  $\underline{\mathbf{K}}$ and $\KillT$ span a timelike Killing
  vectorfield. As a result, we have from
  \Cref{coro:pseudo-convexity:outside-ergoregion:multiple-vectors},
  that \Cref{eq:strict-T-Z-null-convexity:main-condition} is satisfied, with
  $\{\mathbf{V}\}_{i=1}^k = \{\KillT, \underline{\mathbf{K}}\}$.

  Observe that $\mathcal{S}(x_0)=0$, and thus
  \Cref{eq:nabla-y-contracted:S-small} directly implies the
  non-degeneracy condition in
  \Cref{eq:strict-T-Z-null-convexity:cond2} is satisfied for some
  $\varepsilon_1=\varepsilon_1(R_0)$ sufficiently small.  Thus, for
  some $\varepsilon_1$ sufficiently small,
  $\{h_{\varepsilon}\}_{\varepsilon\in (0,\varepsilon_1)}$ forms a
  $\{\KillT, \underline{\mathbf{K}}\}$-pseudoconvex family of
  weights. \Cref{prop:carleman-estimate} then directly yields the
  result.
\end{proof}

We now prove \Cref{prop:main-extension:global}.
\begin{proof}[Proof of \Cref{prop:main-extension:global}]
  Let us fix $\varepsilon$ throughout the proof. We will show that
  $\mathcal{S}=0$ in the set
  $B_{\varepsilon^{100}} = B_{\varepsilon^{100}}(x_0)$.  From the wave
  equation for $\mathcal{S}$ in \Cref{theorem:MST-wave-equation},
  the fact that $\KillT$ and $\underline{\mathbf{K}}$ are Killing
  vectorfields that moreover commute with each other (recall
  \Cref{prop:Hawking-extension}), and
  \Cref{prop:boundedness-of-regularity-qty}, we have that there exist
  smooth tensor fields $\mathcal{A}$ and $\mathcal{B}$ such that
  \begin{equation}
    \label{eq:main-extension:MST-relations:global}
    \begin{split}
      \Box_{\Metric}\mathcal{S}_{\alpha_1\cdots \alpha_4}
      ={}& \mathcal{S}_{\beta_1\cdots\beta_4}\tensor[]{\mathcal{A}}{^{\beta_1\cdots\beta_4}_{\alpha_1\cdots\alpha_4}}
           + \CovariantDeriv_{\beta_5}\mathcal{S}_{\beta_1\cdots\beta_4}\tensor[]{\mathcal{B}}{^{\beta_1\cdots\beta_5}_{\alpha_1\cdots\alpha_4}},\\
      \LieDerivative_{\KillT}\mathcal{S}={}&0,\\
      \LieDerivative_{\underline{\mathbf{K}}}\mathcal{S}={}&0
    \end{split}
  \end{equation}
  in $B_{\varepsilon^{10}}(x_0)$. Using \Cref{lemma:main-prop:aux} and
  the bootstrap assumption that $\mathcal{S}$ vanishes in
  $\mathcal{U}_{R_0}$, we have that
  \begin{equation}
    \label{eq:main-extension:bootstrap-assumption:global}
    \mathcal{S}=0\quad \text{ in }\curlyBrace*{x\in B_{\varepsilon^{10}}(x_0):y(x)<R_0}. 
  \end{equation}
  Now, for $\mathbf{j}=(j_1,j_2,j_3,j_4)\in \{0,1,2,3\}^4$, we consider the
  vectorfields $\partial_{\alpha}$ induced by the coordinate chart $\Phi^{x_0}$, and define
  the smooth functions $\phi_{\mathbf{j}}$ by
  \begin{equation*}
    \begin{split}
      \phi_{\mathbf{j}}: B_{\varepsilon^{10}}(x_0)&\to \Complex\\
      x&\mapsto \mathcal{S}(\partial_{j_1},\partial_{j_2},\partial_{j_3},\partial_{j_4})(x).
    \end{split}    
  \end{equation*}
  Now, we let $\chi: \Real\to [0,1]$ denote a smooth function
  supported in $\left[\frac{1}{2},\infty\right)$ and equal to $1$ in
  $\left[\frac{3}{4},\infty\right)$. We then define 
  $\phi^{\varepsilon}_{\mathbf{j}}\in
  C^{\infty}_0(B_{\varepsilon^{10}}(x_0))$ to be the localization of $\phi_{\mathbf{j}}$ by
  \begin{equation*}
    \phi^{\varepsilon}_{\mathbf{j}}\vcentcolon= \phi_{\mathbf{j}}\widetilde{\chi}_{\varepsilon},
  \end{equation*}
  where we recall from \Cref{eq:main-extension:cutoff-def} that
  \begin{equation*}
    \widetilde{\chi}_{\varepsilon}\vcentcolon= 1- \widetilde{\chi}(\varepsilon^{-40}N^{x_0}(x)).
  \end{equation*}
  It is then a quick computation to see that
  \begin{equation*}
    \begin{split}
      \Box_{\Metric}\phi^{\varepsilon}_{\mathbf{j}}
      ={}& \widetilde{\chi}_{\varepsilon}\Box_{\Metric}\phi_{\mathbf{j}}
           + 2\CovariantDeriv_{\alpha}\phi_{\mathbf{j}} \CovariantDeriv^{\alpha}\widetilde{\chi}_{\varepsilon}
           +\phi_{\mathbf{j}}\Box_{\Metric}\widetilde{\chi}_{\varepsilon},\\
      \KillT\left(\phi^{\varepsilon}_{\mathbf{j}}\right)
      ={}& \widetilde{\chi}_{\varepsilon}\KillT(\phi_{\mathbf{j}})
           + \phi_{\mathbf{j}}\KillT(\widetilde{\chi}_{\varepsilon}),\\
      \underline{\mathbf{K}}\left(\phi^{\varepsilon}_{\mathbf{j}}\right)
      ={}& \widetilde{\chi}_{\varepsilon}\underline{\mathbf{K}}(\phi_{\mathbf{j}})
           + \phi_{\mathbf{j}}\underline{\mathbf{K}}(\widetilde{\chi}_{\varepsilon})
           .
    \end{split}
  \end{equation*}
  We then use the Carleman inequality in
  \Cref{eq:main-extension:Carleman:global}
  to see that for any
  $\{j_1,j_2,j_3,j_4\}\in \{0,1,2,3\}^4$,
  \begin{equation}
    \label{eq:main-extension:Carleman-with-cutoffs:global}
    \begin{split}
      &\lambda\norm*{e^{-\lambda\widetilde{f}_{\varepsilon}}\widetilde{\chi}_{\varepsilon}\phi_{\mathbf{j}}}_{L^2}
        + \norm*{e^{-\lambda\widetilde{f}_{\varepsilon}}\widetilde{\chi}_{\varepsilon}\abs*{D^1\phi_{\mathbf{j}}}}_{L^2}\\
      \lesssim{}& C(R_0)\lambda^{-\frac{1}{2}}\norm*{e^{-\lambda\widetilde{f}_{\varepsilon}}\widetilde{\chi}_{\varepsilon}\Box_{\Metric}\phi_{\mathbf{j}}}_{L^2}
                  + C(R_0)\norm*{e^{-\lambda\widetilde{f}_{\varepsilon}}\widetilde{\chi}_{\varepsilon}\KillT(\phi_{\mathbf{j}})}_{L^2}
                  + C(R_0)\norm*{e^{-\lambda\widetilde{f}_{\varepsilon}}\widetilde{\chi}_{\varepsilon}\underline{\mathbf{K}}(\phi_{\mathbf{j}})}_{L^2}
      \\
                 & + C(R_0)\norm*{e^{-\lambda\widetilde{f}_{\varepsilon}}\CovariantDeriv_{\alpha}\phi_{\mathbf{j}}\CovariantDeriv^{\alpha}\widetilde{\chi}_{\varepsilon}}_{L^2}
       + C(R_0)\norm*{e^{-\lambda\widetilde{f}_{\varepsilon}}\phi_{\mathbf{j}}\left(\abs*{\Box_{\Metric}\widetilde{\chi}_{\varepsilon}}+ \abs*{D^1\widetilde{\chi}_{\varepsilon}}\right)}_{L^2},
    \end{split}
  \end{equation}
  for any $\lambda\ge C(R_0)$. From the equations the
  Mars-Simon tensor $\mathcal{S}$ satisfies in
  \Cref{eq:main-extension:MST-relations:global}, we have that
  \begin{equation}
    \label{eq:main-extenion:MST-inequalities:global}
    \begin{split}
      \abs*{\Box_{\Metric}\phi_{\mathbf{j}}}&\le C(R_0)\sum_{\ell\in \{0,1,2,3\}^4}\left(\abs*{D^1\phi_{\ell}} + \abs*{\phi_{\ell}}\right),\\
      \abs*{\KillT(\phi_{\mathbf{j}})}&\le C(R_0)\sum_{\ell\in \{0,1,2,3\}^4}\abs*{\phi_{\ell}},\\
       \abs*{\underline{\mathbf{K}}(\phi_{\mathbf{j}})}&\le C(R_0)\sum_{\ell\in \{0,1,2,3\}^4}\abs*{\phi_{\ell}}.
    \end{split}
  \end{equation}
  We sum the inequalities in
  \Cref{eq:main-extension:Carleman-with-cutoffs:global} over the indices
  $\mathbf{j}\in \{0,1,2,3\}^4$. The critical observation is
  that the first three terms in the right-hand side of
  \Cref{eq:main-extension:Carleman-with-cutoffs:global} can be absorbed in
  the left-hand side of
  \Cref{eq:main-extension:Carleman-with-cutoffs:global} using
  \Cref{eq:main-extenion:MST-inequalities:global} for $\lambda$ sufficiently
  large. Thus, for any $\lambda\ge C(R_0)$,
  \begin{equation}
    \label{eq:main-extension:aux0:global}
    \begin{split}
      &\lambda\sum_{\mathbf{j}\in \{0,1,2,3\}^4}\norm*{e^{-\lambda\widetilde{f}_{\varepsilon}}\widetilde{\chi}_{\varepsilon}\phi_{\mathbf{j}}}_{L^2}
        + \norm*{e^{-\lambda\widetilde{f}_{\varepsilon}}\widetilde{\chi}_{\varepsilon}\abs*{D^1\phi_{\mathbf{j}}}}_{L^2}\\
      \lesssim{}& 
       C(R_0)\sum_{\mathbf{j}\in \{0,1,2,3\}^4}\left( \norm*{e^{-\lambda\widetilde{f}_{\varepsilon}}\CovariantDeriv_{\alpha}\phi_{\mathbf{j}}\CovariantDeriv^{\alpha}\widetilde{\chi}_{\varepsilon}}_{L^2}
        + C(R_0)\norm*{e^{-\lambda\widetilde{f}_{\varepsilon}}\phi_{\mathbf{j}}\left(\abs*{\Box_{\Metric}\widetilde{\chi}_{\varepsilon}}+ \abs*{D^1\widetilde{\chi}_{\varepsilon}}\right)}_{L^2} \right).
    \end{split}
  \end{equation}
  Using \Cref{eq:main-extension:bootstrap-assumption:global}, and the
  definition of the cutoff $\widetilde{\chi}_{\varepsilon}$ in
  \Cref{eq:main-extension:cutoff-def}, we have
  \begin{equation}
    \label{eq:main-extension:aux1:global}
    \abs*{\CovariantDeriv_{\alpha}\phi_{\mathbf{j}}\CovariantDeriv^{\alpha}\widetilde{\chi}_{\varepsilon}}
    + \phi_{\mathbf{j}} \left(\abs*{\Box_{\Metric}\widetilde{\chi}_{\varepsilon}} + \abs*{D^1\widetilde{\chi}_{\varepsilon}}\right)
    \le C(R_0)\mathbf{1}_{\curlyBrace*{x\in B_{\varepsilon^{10}}(x_0): y(x)\ge R_0, N^{x_0}(x)\ge \varepsilon^{50}}}.
  \end{equation}
  Using the definition of $\widetilde{f}_{\varepsilon}$ in
  \Cref{eq:main-extension:f-tilde-epsilon-def:global}, we observe that
  \begin{equation}
    \label{eq:main-extension:aux2:global}
    \inf_{B_{\varepsilon^{100}}(x_0)}e^{-\lambda\widetilde{f}_{\varepsilon}}
    \ge e^{-\lambda\ln (\varepsilon + \varepsilon^{70})}
    \ge \sup_{\curlyBrace*{x\in B_{\varepsilon^{10}}(x_0): y(x)\ge R_0, N^{x_0}(x)\ge \varepsilon^{50}}}e^{-\lambda\widetilde{f}_{\varepsilon}}.
  \end{equation}
  Then it follows by combining \Cref{eq:main-extension:aux1:global} and
  \Cref{eq:main-extension:aux2:global} with \Cref{eq:main-extension:aux0:global}
  that
  \begin{equation*}
    \lambda\sum_{\mathbf{j}\in \{0,1,2,3\}^4}\norm*{\mathbf{1}_{B_{\varepsilon^{100}}(x_0)}\phi_{\mathbf{j}}}_{L^2}
    \le C(R_0)\sum_{\mathbf{j}\in \{0,1,2,3\}^4}\norm*{\mathbf{1}_{\curlyBrace*{x\in B_{\varepsilon^{10}}(x_0):y(x)\ge R_0,N^{x_0}(x)\ge \varepsilon^{50}}}}_{L^2}
  \end{equation*}
  for any $\lambda\ge C(R_0)$. \Cref{prop:main-extension:global}
  then follows by letting $\lambda\to\infty$.
\end{proof}

We can now complete the proof of the main theorem.

\begin{proof}[Proof of \Cref{prop:main-prop:global}]
  Assume for the sake of contradiction that $R_0$ as defined in
  \Cref{eq:S-extension:BSA} satisfies $R_0 < y_{\underline{S}_0}$. But
  then using the second inclusion in
  \Cref{eq:UBar-VBar-behavior:growth}, it follows from
  \Cref{prop:main-extension:global} that $\mathcal{S}$ vanishes in a
  small neighborhood of $\mathcal{U}_{R_0+\delta^2}$, concluding the
  proof of \Cref{prop:main-prop:global}.
\end{proof}

\appendix

\section{Proofs of \texorpdfstring{\Cref{sec:Lambda-stationary-spacetimes}}{}}
The proofs in this section are all contained in
\cite{marsSpacetimeCharacterizationKerrNUTAde2015}, but we have
included them here for the sake of completeness.

\subsection{Proof of \texorpdfstring{\Cref{lemma:basic-computations:eta-T-sigma-null-decomp}}{}}
\label{appendix:lemma:basic-computations:eta-T-sigma-null-decomp}

We first list some basic properties that are easy to verify.
\begin{gather*}
  \tensor[]{\mathcal{F}}{_{\mu}^{\rho}}\mathcal{T}_{\nu\rho} = \frac{1}{8}\mathcal{F}^2\overline{\mathcal{F}}_{\mu\nu}, \\
  \KillT^{\mu}\mathcal{T}_{\mu\nu}=-\frac{1}{4}\bm{\eta}_{\nu}, \qquad
  \bm{\sigma}^{\mu}\mathcal{T}_{\mu\nu} = \frac{1}{8}\mathcal{F}^2\overline{\bm{\sigma}}_{\nu},\qquad
  \bm{\eta}^{\mu}\mathcal{T}_{\mu\nu} = -\frac{1}{16}\mathcal{F}^2\overline{\mathcal{F}}^2\KillT_{\nu},\\
  \KillT^{\mu}\KillT^{\nu}\mathcal{T}_{\mu\nu} = \frac{1}{8}\bm{\sigma}\cdot\overline{\bm{\sigma}},\qquad
  \KillT^{\mu}\bm{\sigma}^{\nu}\mathcal{T}_{\mu\nu}=0,\qquad
  \KillT^{\mu}\bm{\eta}^{\nu}\mathcal{T}_{\mu\nu} = -\frac{\Metric(\KillT,\KillT)}{16}\mathcal{F}^2\overline{\mathcal{F}}^2,\qquad
  \bm{\sigma}^{\mu}\bm{\eta}^{\nu}\mathcal{T}_{\mu\nu}=0,\\
  \bm{\sigma}^{\mu}\bm{\sigma}^{\nu}\mathcal{T}_{\mu\nu}=\frac{1}{8}\mathcal{F}^2\bm{\sigma}\cdot\bm{\sigma},\qquad
  \bm{\sigma}^{\mu}\overline{\bm{\sigma}}^{\nu}\mathcal{T}_{\mu\nu} = \frac{\Metric(\KillT,\KillT)}{8}\mathcal{F}^2\overline{\mathcal{F}}^2,\qquad
  \bm{\eta}^{\mu}\bm{\eta}^{\nu}\mathcal{T}_{\mu\nu}
  = \frac{1}{32}\mathcal{F}^2\overline{\mathcal{F}}^2\bm{\sigma}\cdot\bm{\sigma}.
\end{gather*}

Applying the Lanczos identity for double 2-forms to
$A_{\alpha\beta}B_{\mu\nu}$ where  $A_{\mu\nu}$ and
$B_{\mu\nu}$ are arbitrary 2-forms, we have that
\begin{equation*}
  A_{\alpha\beta}B_{\mu\nu}
  + \LeftDual{B}_{\alpha\beta}\LeftDual{A}_{\mu\nu}
  = L_{\mu\alpha}\Metric_{\beta\nu}
  - L_{\nu\alpha}\Metric_{\beta\mu}
  - L_{\mu\beta}\Metric_{\alpha\nu}
  + L_{\nu\beta}\Metric_{\alpha\mu},
\end{equation*}
where
\begin{equation*}
  L_{\mu\nu}\vcentcolon= \tensor[]{B}{_{\mu}^{\rho}}A_{\nu\rho} - \frac{1}{4}\Metric_{\mu\nu}B^{\rho\sigma}A_{\rho\sigma}. 
\end{equation*}
In particular, applying this with
$A=\mathcal{F}, B = \overline{\mathcal{F}}$, we have that
\begin{equation}
  \label{eq:basic-computations:main-Lanczos}
  \mathcal{F}_{\alpha\beta}\overline{\mathcal{F}}_{\mu\nu}
  + \overline{\mathcal{F}}_{\alpha\beta}\mathcal{F}_{\mu\nu}
  =2\left(
    \mathcal{T}_{\alpha\mu}\Metric_{\beta\nu}
    - \mathcal{T}_{\alpha\nu}\Metric_{\beta\mu}
    - \mathcal{T}_{\beta\mu}\Metric_{\alpha\nu}
  \right).
\end{equation}
Contracting \Cref{eq:basic-computations:main-Lanczos} with $\KillT^{\alpha}$, we have that
\begin{equation}
  \label{eq:basic-computations:main-Lanczos:T-contraction}
  8\mathcal{T}_{\beta[\mu}\KillT_{\nu]}
  =2\Metric_{\beta[\mu}\bm{\eta}_{\nu]}
  - \bm{\sigma}_{\beta}\overline{\mathcal{F}}_{\mu\nu}
  - \overline{\bm{\sigma}}_{\beta}\mathcal{F}_{\mu\nu}.
\end{equation}
Contracting \Cref{eq:basic-computations:main-Lanczos} with
$\bm{\sigma}^{\alpha}$, we have that
\begin{equation}
  \label{eq:basic-computations:main-Lanczos:sigma-contraction}
  8\mathcal{T}_{\beta[\mu}\bm{\sigma}_{\nu]}
  = \mathcal{F}^2
  \left( \KillT_{\beta}\overline{\mathcal{F}}_{\mu\nu}-\Metric_{\beta[\mu}\overline{\bm{\sigma}}_{\nu]} \right)
  - 2\bm{\eta}_{\beta}\mathcal{F}_{\mu\nu}. 
\end{equation}
Now fix the outgoing-ingoing null frame
$(e_4,e_3) = (\bm{\ell}_+,\bm{\ell}_-)$.  Contracting
\labelcref{eq:basic-computations:main-Lanczos:T-contraction,eq:basic-computations:main-Lanczos:sigma-contraction}
with $e_4$, we have that
\begin{equation}
  \label{eq:basic-computations:main-contractions:e4}
  \begin{split}
    -2R \overline{R} e_4\wedge \KillT
    ={}& e_4\wedge \bm{\eta}
         - 2\Metric(\KillT, e_4) \left( R\overline{\mathcal{F}} + \overline{R}\mathcal{F} \right),\\
    -R \overline{R} e_4\wedge \bm{\sigma}
    ={}& R^2e_4\wedge\overline{\bm{\sigma}}
         - 2R\Metric(\KillT, e_4) \left( R \overline{\mathcal{F}} + \overline{R} \mathcal{F} \right),
  \end{split}
\end{equation}
where we note that the third equation in
\Cref{eq:basic-computations:main-contractions:e4} is the first
equation in \cref{eq:basic-computations:main-contractions:e4}
multiplied by $R \overline{R}$.  We also have the following equations
which arise from contracting
\labelcref{eq:basic-computations:main-Lanczos:T-contraction,eq:basic-computations:main-Lanczos:sigma-contraction}
with $e_3$,
\begin{equation}
  \label{eq:basic-computations:main-contractions:e3}
  \begin{split}
    -2R \overline{R} e_3\wedge \KillT
    ={}& e_3\wedge \bm{\eta}
         + 2\Metric(\KillT, e_3) \left( R\overline{\mathcal{F}} + \overline{R}\mathcal{F} \right),\\
    -R \overline{R} e_3\wedge \bm{\sigma}
    ={}& R^2e_3\wedge\overline{\bm{\sigma}}
         - 2R\Metric(\KillT, e_3) \left( R \overline{\mathcal{F}} + \overline{R} \mathcal{F} \right).
  \end{split}
\end{equation}
Then,
\labelcref{eq:basic-computations:main-contractions:e4,eq:basic-computations:main-contractions:e3}
imply that
\begin{equation}
  \label{eq:basic-computations:main-wedge-eq}
  e_4\wedge\left(  \bm{\eta} + 2R \overline{R}\KillT - R \overline{\bm{\sigma}} - \overline{R} \bm{\sigma} \right)
  = e_3\wedge\left(  \bm{\eta} + 2R \overline{R}\KillT + R \overline{\bm{\sigma}} + \overline{R} \bm{\sigma} \right)=0,
\end{equation}
which implies that there exists some $A_+, A_-\in \Real$ such that
\begin{equation}
  \label{eq:basic-computations:main-eq}
  \begin{split}
    A_+e_4 &=   \bm{\eta} + 2R \overline{R}\KillT - R \overline{\bm{\sigma}} - \overline{R} \bm{\sigma}\\
    A_-e_3 &=  \bm{\eta} + 2R \overline{R}\KillT + R \overline{\bm{\sigma}} + \overline{R} \bm{\sigma},
  \end{split}
\end{equation}
where by contracting \Cref{eq:basic-computations:main-eq} with $e_3$,
$e_4$, and $\KillT$, we find that
\begin{equation}
  \label{eq:basic-computations:Apm-def}
  \begin{gathered}
    A_+ = -4 R \overline{R}\Metric(e_3,\KillT), \qquad
    A_- = -4 R \overline{R}\Metric(e_4,\KillT),\\
    A_+\Metric(e_4,\KillT) = A_-\Metric(e_3,\KillT)
    = 2\Metric(\KillT,\KillT) R \overline{R} - \frac{1}{2}\bm{\sigma}\cdot \overline{\bm{\sigma}} .
  \end{gathered}  
\end{equation}
Using
\labelcref{eq:basic-computations:main-eq,eq:basic-computations:Apm-def}
then yields \Cref{eq:basic-computations:eta-T-sigma-null-decomp}

\subsection{Proof of \texorpdfstring{\Cref{prop:divergence-of-MST}}{}}
\label{appendix:prop:divergence-of-MST}

Recall that in our sign convention regarding \Cref{eq:R:def}, we have that
\begin{equation*}
  R = -\frac{\ImagUnit}{2}\sqrt{\mathcal{F}^2}.
\end{equation*}
We first observe that
\begin{align*}
  Q \mathcal{U}_{\alpha\beta\mu\nu}
  ={}& \left(\frac{3J}{R} - \frac{\Lambda}{R^2}\right)\left(\mathcal{F}_{\alpha\beta}\mathcal{F}_{\mu\nu} - \frac{1}{3}\mathcal{F}^2\mathcal{I}_{\alpha\beta\mu\nu}\right)\\
  ={}& -(\Lambda-3JR)\left(\mathcal{X}_{\alpha\beta}\mathcal{X}_{\mu\nu} + \frac{4}{3}\mathcal{I}_{\alpha\beta\mu\nu}\right).
\end{align*}
Differentiating then yields that
\begin{equation}
  \label{eq:deriv-U}
  \CovariantDeriv_{\rho}\mathcal{U}_{\alpha\beta\mu\nu}
  ={} 3\CovariantDeriv_{\rho}(JR)\left(\mathcal{X}_{\alpha\beta}\mathcal{X}_{\mu\nu} + \frac{4}{3}\mathcal{I}_{\alpha\beta\mu\nu}\right)
  - (\Lambda-3JR)\left(\mathcal{X}_{\mu\nu}\CovariantDeriv_{\rho} \mathcal{X}_{\alpha\beta}
    + \mathcal{X}_{\alpha\beta}\CovariantDeriv_{\rho} \mathcal{X}_{\mu\nu}
  \right). 
\end{equation}
We first compute the covariant derivative of
$\mathcal{X}_{\mu\nu}$. To this end, observe that
\begin{align*}
  \CovariantDeriv_{\rho}\mathcal{F}_{\mu\nu}
  ={}&\KillT^{\sigma}\left(\mathcal{W}_{\mu\nu\sigma\rho}
       + \frac{4}{3}\Lambda\mathcal{I}_{\mu\nu\sigma\rho}
       \right)\\
  ={}& \KillT^{\sigma}\left(
       \mathcal{S}_{\mu\nu\sigma\rho}
       + Q \mathcal{U}_{\mu\nu\sigma\rho}
       + \frac{4}{3}\Lambda \mathcal{I}_{\mu\nu\sigma\rho}
       \right)\\
  ={}& \KillT^{\sigma}\mathcal{S}_{\mu\nu\sigma\rho}
       + 4JR\KillT^{\sigma}\mathcal{I}_{\mu\nu\sigma\rho}
       + \frac{3JR-\Lambda}{2R^2}\bm{\sigma}_{\rho}\mathcal{F}_{\mu\nu},\\
  \CovariantDeriv_{\rho}\mathcal{F}^2
  ={}& 2\mathcal{F}^{\mu\nu}\CovariantDeriv_{\rho}\mathcal{F}_{\mu\nu},\\
  ={}& -4(2JR-\Lambda)\bm{\sigma}_{\rho}
       + 2\KillT^{\sigma}\mathcal{F}^{\mu\nu}\mathcal{S}_{\mu\nu\sigma\rho}\\
  ={}& \frac{1}{3}\left(
       2Q\mathcal{F}^2 +4\Lambda
       \right)\bm{\sigma}_{\rho}
       + 2\KillT^{\sigma}\mathcal{F}^{\mu\nu}\mathcal{S}_{\mu\nu\sigma\rho}.
\end{align*}
Then, using \Cref{eq:deriv-R}, we have that 
\begin{align*}
  \CovariantDeriv_{\rho}\mathcal{X}_{\mu\nu}
  ={}& \frac{1}{R}\CovariantDeriv_{\rho}\mathcal{F}_{\mu\nu}
       - \frac{1}{R}\mathcal{X}_{\mu\nu}\CovariantDeriv_{\rho}R\\
  ={}& \frac{J}{2R}\bm{\sigma}_{\rho}\mathcal{X}_{\mu\nu}
       + 4J\KillT^{\sigma}\mathcal{I}_{\mu\nu\sigma\rho}
       + \frac{1}{R}\KillT^{\sigma}\mathcal{S}_{\mu\nu\sigma\rho}
       + \frac{1}{4R}\KillT^{\sigma}\mathcal{X}_{\mu\nu}\mathcal{X}^{\alpha\beta}\mathcal{S}_{\alpha\beta\sigma\rho},\\
  \CovariantDeriv^{\nu}\mathcal{X}_{\mu\nu}
  ={}& \frac{J}{2R}\bm{\sigma}^{\nu}\mathcal{X}_{\mu\nu}
       + 3J\KillT_{\mu}
       + \frac{1}{4R}\KillT^{\sigma}\tensor[]{\mathcal{X}}{_{\mu}^{\rho}}\mathcal{X}^{\alpha\beta}\mathcal{S}_{\alpha\beta\sigma\rho}\\
  ={}& 2J\KillT_{\mu}
       + \frac{1}{4R}\KillT^{\sigma}\tensor[]{\mathcal{X}}{_{\mu}^{\rho}}\mathcal{X}^{\alpha\beta}\mathcal{S}_{\alpha\beta\sigma\rho}. 
\end{align*}
Using that $J$ solves \Cref{eq:J:quadratic-eqn}, we can thus compute that
\begin{align*}
  \CovariantDeriv_{\rho}J
  ={} - \frac{J}{R- \sigma_0 J}\CovariantDeriv_{\rho}R + \frac{J^2}{2(R - \sigma_0 J)}\bm{\sigma}_{\rho}.
\end{align*}
As a result, we can compute that 
\begin{align}
  \CovariantDeriv_{\rho}\left(JR\right)
  ={}& R\CovariantDeriv_{\rho}J
       + J\CovariantDeriv_{\rho}R\notag\\
  ={}& \frac{J^2\sigma_0}{R-J\sigma_0}\left(\frac{R}{2\sigma_0}\bm{\sigma}_{\rho} - \CovariantDeriv_{\rho}R\right)\notag\\
  ={}& \frac{J}{R-J\sigma_0}\left(
       \frac{-3JR^2 + 2R\Lambda + \Lambda\sigma_0 J}{2R}\bm{\sigma}_{\rho}
       + \frac{1}{4}J\sigma_0 \KillT^{\gamma}\mathcal{X}^{\mu\nu} \mathcal{S}_{\mu\nu\gamma\rho}
       \right)\notag\\
  ={}& \frac{J}{2R}(3JR-\Lambda)\bm{\sigma}_{\rho}
       + \frac{J^2\sigma_0}{4(R-J\sigma_0)}\KillT^{\gamma}\mathcal{X}^{\mu\nu}\mathcal{S}_{\mu\nu\gamma\rho}
       . \label{eq:deriv-JR}
\end{align}
Combining \Cref{eq:deriv-U} and \Cref{eq:deriv-JR} yields the
result.

\section{The main formalism}
\label{sec:formalism}

The notation and formalism used is similar to that used in
\cite{ionescuUniquenessSmoothStationary2009}, but more closely follows
the formalism in \cite{giorgiGeneralFormalismStability2020}, in
particular in its use of the horizontal tensor formalism as opposed to
the Newman-Penrose formalism. We refer the reader to
\cite{giorgiGeneralFormalismStability2020} for a more thorough
introduction and discussion of the formalism.

\subsection{Horizontal structures}

Let $(\mathcal{M}, \Metric)$ be a $3+1$ dimensional smooth Lorentzian
manifold solving the vacuum Einstein equations with $\Lambda>0$. Then
consider some arbitrary null pair $e_3 = \underline{L}$, $e_4 = L$
normalized so that $\Metric(e_3,e_4)=-2$.

\begin{definition}
  \label{def:horizontal-vectorfield}
  We say that a vectorfield $X$ lies in the space of \emph{horizontal
  vectorfields} $\mathbf{O}(\mathcal{M})\subset T^{*}\mathcal{M}$ if
  \begin{equation*}
    \Metric(e_4, X) = \Metric(e_3,X) = 0.
  \end{equation*}
\end{definition}

\begin{definition}
  \label{def:hor-metric-volform}
  Let $X, Y\in \mathbf{O}(\mathcal{M})$, then we define the induced
  metric and induced volume form by
  \begin{equation*}
    \gamma(X,Y) = \Metric(X,Y),\qquad \in(X,Y) = \in(X,Y,e_3,e_4),
  \end{equation*}
  where $\in$ is the standard volume form on $\mathcal{M}$. If
  $\curlyBrace*{e_a}_{a=1}^2$ is an orthonormal basis of horizontal
  vectorfields, then $\gamma_{ab} = \delta_{ab}$, and we write
  $\in_{ab} = \in(e_a, e_b)$.
\end{definition}

$\mathbf{O}(\mathcal{M})$ is clearly a vector space, but it is not
necessarily closed under commutation. If it is closed under
commutation, i.e. if $X, Y\in \mathbf{O}(\mathcal{M})$ implies that
the commutator $[X,Y]\in \mathbf{O}(\mathcal{M})$, then we say that
$(e_3,e_4)$ are integrable.

\begin{definition}
  \label{def:horizontal-projection}
  Given $X\in T\mathcal{M}$, we denote by $\horProj{X}$ its
  \emph{horizontal projection},
  \begin{equation}
    \label{eq:horizontal-projection:def}
    \horProj{X} = X + \frac{1}{2}\Metric(X, e_3)e_4 + \frac{1}{2}\Metric(X,e_4)e_3.
  \end{equation}
  We also define the \emph{projection operator}
  \begin{equation*}
    \Pi^{\mu\nu} \vcentcolon= \Metric^{\mu\nu} + \frac{1}{2}\left(e_4^{\nu}e_3^{\mu} + e_4^{\mu}e_3^{\nu}\right).
  \end{equation*}
\end{definition}

\begin{definition}
  We say that a $k$-covariant tensor-field $U$ lies in the space of
  \emph{horizontal $k$-covariant tensors} $\mathbf{O}_k(\mathcal{M})$ if for
  any $\curlyBrace*{X_i}_{i=1}^k$,
  \begin{equation*}
    U(X_1,\cdots,X_k) = U(\horProj{X}_1,\cdots,\horProj{X}_k),
  \end{equation*}
  or in other words, if
  \begin{equation*}
    \tensor[]{\Pi}{_{\alpha_1}^{\beta_1}}\cdots\tensor[]{\Pi}{_{\alpha_k}^{\beta_k}}U_{\beta_1\cdots\beta_k}
    = U_{\alpha_1\cdots \alpha_k}. 
  \end{equation*}
\end{definition}

\begin{definition}
  \label{def:hor-trace-atrace}
  Given a horizontal 2-covariant tensor $U$ and an arbitrary
  orthonormal horizontal frame $\{e_a\}_{a=1,2}$, the \emph{trace of a
  horizontal 2-tensor $U$} is defined by
  \begin{equation}
    \label{eq:hor-trace:def}
    \Trace U = \delta^{ab}U_{ab}.
  \end{equation}
  Similarly, we define the \emph{anti-trace of $U$} by
  \begin{equation}
    \label{eq:hor-atrace:def}
    \aTrace{U} = \in^{ab}U_{ab}.
  \end{equation}
  A general horizontal 2-tensor can be decomposed as
  \begin{equation}
    \label{eq:two-tensor-horizontal-decomp}
    U_{ab} = \widehat{U}_{ab} + \frac{1}{2}\delta_{ab}\Trace U + \frac{1}{2}\in_{ab}\aTrace{U},
  \end{equation}
  where $\widehat{U}$ is symmetric and trace-free.
\end{definition}

We can also define the horizontal covariant derivative. 
\begin{definition}
  \label{def:horizontal-covariant-derivative}
  Given $X, Y\in \mathbf{O}(\mathcal{M})$, we define the \emph{horizontal
    covariant derivative} $\nabla_XY$ as
  \begin{equation*}
    \nabla_XY \vcentcolon= \horProj{\CovariantDeriv_XY}. 
  \end{equation*}
  We can extend this definition to the case where $X=e_3,e_4$ in the
  natural way.
  \begin{equation*}
    \nabla_3X \vcentcolon= \horProj{\CovariantDeriv_3X}, \qquad
    \nabla_4X \vcentcolon= \horProj{\CovariantDeriv_4X}. 
  \end{equation*}
  Given a horizontal $k$-covariant tensor-field $U$, we define its
  horizontal covariant derivative by the formula
  \begin{equation*}
    \nabla_ZU(X_1,\cdots, X_k)
    = Z(U(X_1,\cdots, X_k))
    - U(\nabla_ZX_1,\cdots, X_k)
    -\cdots - U(X_1,\cdots, \nabla_ZX_k).
  \end{equation*}
\end{definition}

\subsection{Ricci coefficients and null components of the Weyl tensor}

We define the following null Ricci coefficients.
\begin{equation}
  \label{eq:ricci-components:def}
  \begin{gathered}
    \underline{\chi}_{ab}=\Metric(\CovariantDeriv_ae_3, e_b),\qquad \chi_{ab}=\Metric(\CovariantDeriv_ae_4, e_b),\\
    \underline{\xi}_a=\frac{1}{2} \Metric(\CovariantDeriv_3 e_3 , e_a),\qquad \xi_a=\frac{1}{2} \Metric(\CovariantDeriv_4 e_4, e_a),\\
    \underline{\omega}=\frac{1}{4} \Metric(\CovariantDeriv_3e_3 , e_4),\qquad \omega=\frac{1}{4} \Metric(\CovariantDeriv_4 e_4, e_3),\qquad \\
    \underline{\eta}_a=\frac{1}{2}\Metric(\CovariantDeriv_4 e_3, e_a),\qquad \quad \eta_a=\frac{1}{2} \Metric(\CovariantDeriv_3 e_4, e_a),\qquad\\
    \zeta_a=\frac{1}{2} \Metric(\CovariantDeriv_{e_a}e_4,  e_3).
  \end{gathered}
\end{equation}

We also have the following Ricci formulas. 
\begin{equation}
  \label{eq:Ricci-formulas}
  \begin{split}
    \CovariantDeriv_a e_b &= \nabla_ae_b + \frac{1}{2}\chi_{ab}e_3 + \frac{1}{2}\chiBar_{ab}e_4,\\
    \CovariantDeriv_a e_4 &= \chi_{ab}e_b -\zeta e_4,\\
    \CovariantDeriv_a e_3 &= \chiBar_{ab}e_b +\zeta e_3,\\
    \CovariantDeriv_3 e_a &= \nabla_3 e_a + \eta_{a}e_3 + \xiBar_a e_4,\\
    \CovariantDeriv_3 e_3 &= -2\omegaBar e_3 + 2\xiBar_b e_b,\\
    \CovariantDeriv_3 e_4 &= 2\omegaBar e_4 + 2\eta_b e_b,\\
    \CovariantDeriv_4 e_a &= \nabla_4 e_a + \etaBar_a e_4 + \xi_a e_3,\\
    \CovariantDeriv_4 e_4 &= -2\omega e_4 + 2\xi_b e_b,\\
    \CovariantDeriv_4 e_3 &= 2\omega e_3 + 2\etaBar_b e_b.
  \end{split}  
\end{equation}

We also define the following null Weyl tensor components for any given
Weyl field $\mathbf{W}$.
\begin{equation}
  \label{eq:weyl-components:def}
  \begin{split}
    \alpha(\mathbf{W})(X,Y) &= \mathbf{W}(e_4,X,e_4,Y),\\
    \underline{\alpha}(\mathbf{W})(X,Y) &= \mathbf{W}(e_3,X,e_3,Y),\\
    \beta(\mathbf{W})(X) &= \frac{1}{2}\mathbf{W}(X,e_4,e_3,e_4),\\
    \underline{\beta}(\mathbf{W})(X) &=  \frac{1}{2}\mathbf{W}(X,e_3,e_3,e_4),\\
    \varrho(\mathbf{W})(X,Y) &= \mathbf{W}(X,e_3,Y,e_4),\\
    \rho(\mathbf{W}) &= \frac{1}{4}\mathbf{W}(e_4,e_3,e_4,e_3),\\
    \LeftDual{\rho}(\mathbf{W}) &= \frac{1}{4}\mathbf{W}(e_3,e_4,e_3,e_4),
  \end{split}  
\end{equation}

\subsection{Non-complexified equations}

For $\alpha, \beta, \rho, \underline{\beta}, \underline{\alpha}$ the
null components of the Weyl tensor $\mathbf{W}$, we have the following
null structure equations. 
\begin{align}
  \label{eq:null-structure:nab4-tr-chi}
    \nabla_4\Trace\chi
  + \frac{1}{2}\left(\Trace\chi^2 - \aTrace{\chi}^2\right)
  + 2 \omega\Trace\chi
  ={}& -\widehat{\chi}\cdot \widehat{\chi}
    + 2\Divergence \xi
  + 2\xi\cdot (\etaBar + \eta + 2\zeta),\\
  \label{eq:null-structure:nab4-atr-chi}
  \nabla_4\aTrace{\chi}
  + \Trace \chi\aTrace{\chi}
  + 2\omega\aTrace{\chi},
  ={}& 2\curl \xi + 2\xi\wedge\left(-\underline{\eta} + \eta - 2\zeta\right)\\
  \label{eq:null-structure:nab4-TF-chi}
  \nabla_4\widehat{\chi}
  +\Trace \chi \widehat{\chi}
  + 2\omega\widehat{\chi}
  ={}& \nabla\otimes \underline{\xi} + \underline{\xi}\otimes \left(\eta + \underline{\eta} + 2\zeta\right)
       - \underline{\alpha};
\end{align}

\begin{align}
  \label{eq:null-structure:nab3-tr-chib}
    \nabla_3\Trace\underline{\chi}
  + \frac{1}{2}\left(\Trace\underline{\chi}^2 - \aTrace{\underline{\chi}}^2\right)
  + 2 \underline{\omega}\Trace\underline{\chi}
  ={}& -\widehat{\underline{\chi}}\cdot \widehat{\underline{\chi}}
    + 2\Divergence \underline{\xi}
  + 2\xi\cdot (\etaBar + \eta - 2\zeta),\\
  \label{eq:null-structure:nab3-atr-chib}
  \nabla_3\aTrace{\underline{\chi}}
  + \Trace \chi\aTrace{\underline{\chi}}
  + 2\underline{\omega}\aTrace{\underline{\chi}},
  ={}& 2\curl \underline{\xi} + 2\underline{\xi}\wedge\left(-\underline{\eta} + \eta + 2\zeta\right)\\
  \label{eq:null-structure:nab3-TF-chib}
  \nabla_3\widehat{\underline{\chi}}
  +\Trace \chi \widehat{\underline{\chi}}
  + 2\underline{\omega}\widehat{\underline{\chi}}
  ={}& \nabla\otimes \underline{\xi} + \underline{\xi}\otimes \left(\eta + \underline{\eta} - 2\zeta\right)
       - \alpha;
\end{align}

\begin{align}  
  \nabla_3\Trace \chi
  + \widehat{\underline{\chi}}\cdot\widehat{\chi}
  + \frac{1}{2}\left(\Trace \underline{\chi}\Trace\chi - \aTrace{\underline{\chi}}\aTrace{\chi}\right)
  -2\underline{\omega}\Trace\chi
  ={}& 2\Divergence \eta
       + 2 \left( \xi\cdot\underline{\xi} + \abs*{\eta}^2 \right)\notag\\
      &+ 2\rho
       + \frac{4\Lambda}{3},\label{eq:null-structure:nab3-tr-chi}\\
  \label{eq:null-structure:nab3-atr-chi}
  \nabla_3\aTrace{\chi}
  + \widehat{\underline{\chi}}\wedge\widehat{\chi}
  + \frac{1}{2}\left(  \aTrace{\underline{\chi}}\Trace\chi + \Trace\underline{\chi}\aTrace{\chi} \right)
  - 2\underline{\omega}\aTrace{\chi}
  ={}& 2\curl\eta + 2\underline{\xi}\wedge\xi
       - 2\LeftDual{\rho},\\
  \label{eq:null-structure:nab3-TF-chi}
  \nabla_3\widehat{\chi}
  + \frac{1}{2}\left(  \Trace \underline{\chi}\widehat{\chi} + \Trace\chi\widehat{\underline{\chi}} \right)
  + \frac{1}{2}\left(  -\LeftDual{\widehat{\underline{\chi}}} \aTrace{\chi} + \LeftDual{\widehat{\chi}}\aTrace{\underline{\chi}} \right)
  - 2\underline{\omega}\widehat{\chi}
  ={}& \nabla\otimes \eta
       + \eta\otimes \eta
       + \underline{\xi}\otimes\xi,
\end{align}

\begin{align}
  \nabla_4\Trace \underline{\chi}
  + \widehat{\chi}\cdot\widehat{\underline{\chi}}
  + \frac{1}{2}\left(\Trace\chi\Trace \underline{\chi} - \aTrace{\chi}\aTrace{\underline{\chi}}\right)
  -2\omega\Trace\underline{\chi}
  ={}& 2\Divergence \underline{\eta}
       + 2 \left( \xi\cdot\underline{\xi} + \abs*{\eta}^2 \right)\notag\\
     & + 2\rho
       + \frac{4\Lambda}{3},\label{eq:null-structure:nab4-tr-chib}\\
  \label{eq:null-structure:nab4-atr-chib}
  \nabla_4\aTrace{\underline{\chi}}
  + \widehat{\chi}\wedge\widehat{\underline{\chi}}
  + \frac{1}{2}\left(  \Trace\chi\aTrace{\underline{\chi}} + \aTrace{\chi}\Trace\underline{\chi} \right)
  - 2\omega\aTrace{\underline{\chi}}
  ={}& 2\curl\eta + 2\underline{\xi}\wedge\xi
       - 2\LeftDual{\rho},\\
  \label{eq:null-structure:nab4-TF-chib}
  \nabla_4\widehat{\underline{\chi}}
  + \frac{1}{2}\left(  \Trace \chi\widehat{\underline{\chi}} + \Trace\underline{\chi}\widehat{\chi} \right)
  + \frac{1}{2}\left(  -\LeftDual{\widehat{\chi}} \aTrace{\underline{\chi}} + \LeftDual{\widehat{\underline{\chi}}}\aTrace{\chi} \right)
  - 2\omega\widehat{\underline{\chi}}
  ={}& \nabla\otimes \underline{\eta}
       + \underline{\eta}\otimes \underline{\eta}
       + \xi\otimes\underline{\xi},
\end{align}

\begin{align}
  \nabla_3\zeta + 2\nabla\underline{\omega}
  ={}& -\widehat{\underline{\chi}}\cdot(\zeta+\eta)
       -\frac{1}{2}\Trace\underline{\chi}\left(\zeta+\eta\right)
       -\frac{1}{2}\aTrace{\underline{\chi}}\left(\LeftDual{\zeta} + \LeftDual{\eta}\right)
       + 2\underline{\omega}(\zeta-\eta)\notag\\
     & + \widehat{\chi}\cdot\underline{\xi}
       + \frac{1}{2} \Trace\chi\underline{\xi}
       + \frac{1}{2} \aTrace{\chi}\LeftDual{\underline{\xi}}
       + 2\omega\underline{\xi}
       - \underline{\beta},  \label{eq:null-structure:nab3-zeta}\\
  \nabla_4\zeta
  - 2\nabla \omega
  ={}& \chiTF\cdot (-\zeta +\underline{\eta})
       + \frac{1}{2}\Trace\chi(-\zeta+\underline{\eta})
       + \frac{1}{2}\aTrace{\underline{\chi}}\left(-\LeftDual{\zeta} + \LeftDual{\underline{\eta}}\right)
       + 2\omega\left(\zeta + \underline{\eta}\right)\notag\\
     & - \widehat{\underline{\chi}}\cdot \xi
       - \frac{1}{2}\Trace\underline{\chi}\xi
       - \frac{1}{2}\aTrace{\underline{\chi}}\LeftDual{\xi}
       - 2\underline{\omega}\xi
       -\beta. \label{eq:null-structure:nab4-zeta}
\end{align}

\begin{align}
  \label{eq:null-structure:nab4-xib}
  \nabla_4\underline{\xi}
  -\nabla_3\underline{\eta}
  &= - \chiBar\cdot(\eta-\underline{\eta})
    + 4\omega\underline{\xi}
    -\underline{\beta},\\
  \label{eq:null-structure:nab3-xi}
  \nabla_3\xi
  - \nabla_4\eta
  &= \chi\cdot(\eta-\underline{\eta})
    + 4\underline{\omega}\xi
    +\beta. 
\end{align}

\begin{equation}
  \label{eq:null-structure:nab4-nab3-omega}
  \nabla_4\underline{\omega}
  + \nabla_3\omega
  = \rho
  - \frac{\Lambda}{3}
  + 4\omega\underline{\omega}
  + \xi\cdot\underline{\xi}
  + \zeta\cdot (\eta-\underline{\eta})
  -\eta\cdot\underline{\eta}.
\end{equation}

\begin{align}
  \label{eq:null-structure:div-TF-chi}
  \Divergence \widehat{\chi}
  + \zeta \cdot \widehat{\chi}
  &= \frac{1}{2}\nabla \Trace\chi
    + \frac{1}{2}\Trace\chi \zeta
    - \frac{1}{2}\LeftDual{\nabla}\aTrace{\chi}
    - \frac{1}{2}\aTrace{\chi}\LeftDual{\zeta}
    - \aTrace{\chi}\LeftDual{\eta}
    - \aTrace{\chi}\LeftDual{\xi}
    - \beta,\\
  \label{eq:null-structure:div-TF-chib}
  \Divergence\widehat{\underline{\chi}}
  - \zeta\cdot\widehat{\chi}
  &= \frac{1}{2}\nabla\Trace\underline{\chi}
    -\frac{1}{2}\Trace\underline{\chi}\zeta
    -\frac{1}{2}\LeftDual{\nabla}\aTrace{\underline{\chi}}
    +\frac{1}{2}\aTrace{\underline{\chi}}\LeftDual{\zeta}
    -\aTrace{\underline{\chi}}\LeftDual{\underline{\eta}}
    -\aTrace{\underline{\chi}}\LeftDual{\underline{\xi}}
    +\underline{\beta}. 
\end{align}

\begin{equation}
  \label{eq:null-structure:curl-zeta}
  \curl\zeta = - \frac{1}{2}\widehat{\chi}\wedge\widehat{\underline{\chi}}
  + \frac{1}{4}\left(
    \Trace\chi \aTrace{\underline{\chi}}
    - \Trace\underline{\chi} \aTrace{\chi}
  \right)
  + \omega \aTrace{\underline{\chi}}
  - \underline{\omega} \aTrace{\chi}
  + \LeftDual{\rho}. 
\end{equation}

\subsection{Complex notation}
\label{sec:complex-notation}


\begin{definition}
  \label{def:complex-notation}
  Let $\mathbf{W}$ be a Weyl field, and
  $\mathcal{W} = \mathbf{W} + \ImagUnit\LeftDual{\mathbf{W}}$
  be its complexification.  We then recall the definition of the
  following complex components of $\mathcal{W}$.
  \begin{equation*}
    \begin{gathered}
      A(\mathcal{W})\vcentcolon=  \alpha(\mathbf{W}) + \ImagUnit \LeftDual{\alpha}(\mathbf{W}), \qquad
      B(\mathcal{W})\vcentcolon= \beta(\mathbf{W}) + \ImagUnit \LeftDual{\beta}(\mathbf{W}),\\
      \underline{A}(\mathcal{W})\vcentcolon= \underline{\alpha}(\mathbf{W}) + \ImagUnit\LeftDual{\underline{\alpha}}(\mathbf{W}),\qquad
      \underline{B}(\mathcal{W})\vcentcolon= \underline{\beta}(\mathbf{W}) + \ImagUnit \LeftDual{\underline{\beta}}(\mathbf{W}),\\
      P(\mathcal{W})\vcentcolon= \rho(\mathbf{W}) + \ImagUnit\LeftDual{\rho}(\mathbf{W}).
    \end{gathered}
  \end{equation*}
  We also have the following complexified Ricci coefficients (recall
  the definition of the Ricci coefficients from
  \Cref{eq:ricci-components:def}).
  \begin{equation*}
    \begin{gathered}
      X= \chi + \ImagUnit \LeftDual{\chi},\qquad
      \underline{X}= \underline{\chi} + \ImagUnit \LeftDual{\underline{\chi}},\\
      H = \eta + \ImagUnit \LeftDual{\eta},\qquad
      \underline{H} = \underline{\eta} + \ImagUnit \LeftDual{\underline{\eta}},\\
      \Xi = \xi + \ImagUnit \LeftDual{\xi},\qquad
      \underline{\Xi} = \underline{\xi} + \ImagUnit \LeftDual{\underline{\xi}},\\
      Z = \zeta + \ImagUnit \LeftDual{\zeta}.
    \end{gathered}
  \end{equation*}
\end{definition}
In particular, note that
\begin{equation*}
  \begin{gathered}
    \Trace X = \Trace \chi - \ImagUnit \aTrace{\chi},\qquad
    \Trace \underline{X} = \Trace \underline{\chi} - \ImagUnit \aTrace{\underline{\chi}},\qquad
    \widehat{X} = \widehat{\chi} +\ImagUnit\LeftDual{\widehat{\chi}},\qquad
    \widehat{\underline{X}} = \widehat{\underline{\chi}} +\ImagUnit\LeftDual{\widehat{\underline{\chi}}}.
  \end{gathered}  
\end{equation*}

\begin{definition}
  \label{def:complex-derivatives}
  We define derivatives of complex quantities as follows.
  \begin{enumerate}
  \item For two scalar real-valued functions $a$ and $b$, we define
    \begin{equation*}
      \ComplexDeriv(a+\ImagUnit b)
      \vcentcolon= \left(\nabla + \ImagUnit\LeftDual{\nabla}\right)(a+\ImagUnit b).
    \end{equation*}
  \item For a real-valued $1$-form $f$, we define
    \begin{equation*}
      \ComplexDeriv\cdot (f+\ImagUnit\LeftDual{f})
      \vcentcolon= \left(\nabla + \ImagUnit\LeftDual{\nabla}\right)\cdot (f + \ImagUnit\LeftDual{f}),
    \end{equation*}
    and
    \begin{equation*}
      \ComplexDeriv\SymTracelessTensorProd(f+\ImagUnit\LeftDual{f})
      \vcentcolon= \left(\nabla + \ImagUnit\LeftDual{\nabla}\right)\SymTracelessTensorProd(f+ \ImagUnit\LeftDual{f}).
    \end{equation*}
  \item For a symmetric traceless 2-tensor $u$, we define
    \begin{equation*}
      \ComplexDeriv\cdot (u+\ImagUnit\LeftDual{u})
      \vcentcolon= (\nabla + \ImagUnit\LeftDual{\nabla})\cdot (u+\ImagUnit \LeftDual{u}). 
    \end{equation*}
  \end{enumerate}
\end{definition}

\subsection{Null structure and null Bianchi equations}

For any general null frame, we have the following null Bianchi equations.
\begin{prop}
  Let $W$ be a Weyl tensor such that
  \begin{equation*}
    \CovariantDeriv^{\alpha}\mathcal{W}_{\alpha\beta\mu\nu} = \mathcal{J}_{\beta\mu\nu}.
  \end{equation*}
  Then, denoting
  \begin{equation*}
    A = A(\mathcal{W}), \qquad \underline{A}=\underline{A}(\mathcal{W}),\qquad
    B = B(\mathcal{W}), \qquad\underline{B}=\underline{B}(\mathcal{W}),\qquad
    P = P(\mathcal{W}), 
  \end{equation*}
  to condense the notation slightly, we have the following null
  Bianchi equations
  \begin{align}
    \nabla_3A
    - \frac{1}{2}\ComplexDeriv\SymTracelessTensorProd B
    &= -\frac{1}{2}\Trace\XBar A
      + 4\omegaBar A
      + \frac{1}{2}\left(Z + 4H\right)\SymTracelessTensorProd B
      - 3\overline{P}\XHat
      + \mathcal{J}_{ba4}
      + \mathcal{J}_{ab4}
      - \frac{1}{2}\delta_{ab}\mathcal{J}_{434}
      , \label{eq:Bianchi:nabla3-A}    
    \\
    \nabla_4 \ABar
    + \frac{1}{2}\ComplexDeriv\SymTracelessTensorProd \BBar
    &= - \frac{1}{2}\Trace X \ABar
      + 4\omega \ABar
      + \frac{1}{2}\left(Z - 4\HBar\right)\SymTracelessTensorProd \BBar
      - 3P\XHatBar
      + \mathcal{J}_{ba3}
      + \mathcal{J}_{ab3}
      - \frac{1}{2}\delta_{ab}\mathcal{J}_{343}
      , \label{eq:Bianchi:nabla4-ABar}
    \\
    \nabla_4 B
    - \frac{1}{2}\overline{\ComplexDeriv}\cdot A
    &= -2 \overline{\Trace X}B
      - 2\omega B
      + \frac{1}{2} A \cdot \left(\overline{2Z + \HBar}\right)
      + 3 \overline{P}\Xi
      - \mathcal{J}_{4a4}
      , \label{eq:Bianchi:nabla4-B}
    \\
    \nabla_3 \BBar
    + \frac{1}{2}\overline{\ComplexDeriv}\cdot \ABar
    &= -2 \overline{\Trace \XBar}\BBar
      - 2 \omegaBar \BBar
      - \frac{1}{2}\ABar\cdot \left(-2Z+H\right)
      - 3P \XiBar
      + \mathcal{J}_{3a3}
      , \label{eq:Bianchi:nabla3-BBar}
    \\
    \nabla_3 B
    - \ComplexDeriv \overline{P}
    &= - \Trace \XBar B
      + 2\omegaBar B
      +\overline{\BBar}\cdot \XHat
      + 3\overline{P}H
      + \frac{1}{2}A\cdot \overline{\XiBar}
      + \mathcal{J}_{3a4}
      , \label{eq:Bianchi:nabla3-B}
    \\
    \nabla_4\BBar
    + \ComplexDeriv P
    &= - \Trace X\BBar
      +2\omega\BBar
      +\overline{B}\cdot \XHatBar
      -3P\HBar
      -\frac{1}{2}\ABar\cdot \overline{\Xi}
      -\mathcal{J}_{4a3}
      , \label{eq:Bianchi:nabla4-BBar}
    \\
    \nabla_4 P
    -\frac{1}{2}\ComplexDeriv\cdot \overline{B}
    &= - \frac{3}{2}\Trace X P
      + \frac{1}{2}\left(2\HBar + Z\right)\cdot \overline{B}
      - \overline{\Xi}\cdot\BBar
      - \frac{1}{4}\XHatBar \cdot \overline{A}
      - \frac{1}{2}\mathcal{J}_{434} 
      , \label{eq:Bianchi:nabla4-P}
    \\
    \nabla_3 P
    + \frac{1}{2} \overline{\ComplexDeriv}\cdot \BBar
    &= - \frac{3}{2}\overline{\Trace\XBar}P
      -\frac{1}{2}\left(\overline{2H - Z}\right)\cdot \BBar
      + \XiBar\cdot \overline{B}
      -\frac{1}{4}\overline{\XHat}\cdot \ABar
      - \frac{1}{2}\mathcal{J}_{343}
      . \label{eq:Bianchi:nabla3-P}
  \end{align}  
\end{prop}

We also have the following null structure equations.
\begin{prop}
  \label{prop:null-structure-eqns:complex}
  Denoting by $\mathcal{W}$ the complexified Weyl tensor, and 
  \begin{equation*}
    A = A(\mathcal{W}), \qquad \underline{A}=\underline{A}(\mathcal{W}),\qquad
    B = B(\mathcal{W}), \qquad\underline{B}=\underline{B}(\mathcal{W}),\qquad
    P = P(\mathcal{W}), 
  \end{equation*}
  we have that
  \begin{align}
    \nabla_{3}\Trace \XBar
    + \frac{1}{2}(\Trace\XBar)^{2}
    + 2\omegaBar\Trace \XBar
    &= \ComplexDeriv\cdot \overline{\XiBar}
      + \XiBar\cdot \overline{\HBar}
      + \overline{\XiBar}\cdot (H-2Z)
      -\frac{1}{2}\XHatBar\cdot \overline{\XHatBar},
      \label{eq:null-structure:D3-tr-XBar}
    \\
    \nabla_3 \XHatBar
    + \Re(\Trace \XBar) \XHatBar
    + 2\omegaBar\XHatBar
    &= \frac{1}{2}\ComplexDeriv\SymTracelessTensorProd\XiBar
      + \frac{1}{2}\XiBar\SymTracelessTensorProd(H + \HBar -2Z)
      -\ABar,
      \label{eq:null-structure:D3-XHatBar}
    \\
    \nabla_4\Trace X
    + \frac{1}{2}(\Trace X)^2
    + 2\omega \Trace X
    &= \ComplexDeriv\cdot \XiBar
      + \Xi\cdot \HBar
      + \XiBar \cdot (H+2Z)
      - \frac{1}{2}\XHat \cdot \overline{\XHat},
      \label{eq:null-structure:D4-tr-X}
    \\
    \nabla_4\XHat + \Re(\Trace X)\XHat
    + 2\omega \XHat
    &= \frac{1}{2}\ComplexDeriv\SymTracelessTensorProd\Xi
      + \frac{1}{2}\Xi \SymTracelessTensorProd\left(\HBar+H+2Z\right)
      -A,
      \label{eq:null-structure:D4-XHat}
    \\
    \nabla_{3}\Trace X
    + \frac{1}{2}(\Trace\XBar)\Trace X
    - 2\omegaBar\Trace X
    &= \ComplexDeriv\cdot \overline{H}
      + H \cdot \overline{H}
      + 2 P
      + \frac{4\Lambda}{3}
      + \XiBar \cdot \overline{\Xi}
      -\frac{1}{2}\XHatBar\cdot \overline{\XHat},
      \label{eq:null-struture:D3-tr-X}
    \\
    \nabla_3 \XHat
    + \frac{1}{2}\Trace X \XHatBar
    - 2\omegaBar \XHat  
    &= \frac{1}{2}\ComplexDeriv\SymTracelessTensorProd H
      + \frac{1}{2}H\SymTracelessTensorProd H
      - \frac{1}{2}\overline{\Trace X}\XHatBar
      + \frac{1}{2}\XiBar\SymTracelessTensorProd \Xi,
  \end{align}

  \begin{align*}
    \nabla_4\Trace\XBar
    +\frac{1}{2}\Trace X\Trace\XBar
    -2\omega\Trace\XBar
    &= \ComplexDeriv\cdot \overline{\HBar}
      +\HBar\cdot\overline{\HBar}
      +2\overline{P}
      +\frac{4\Lambda}{3}
      + \Xi \cdot \overline{\XiBar}
      -\frac{1}{2}\XHat\cdot \overline{\XHatBar},\\
    \nabla_4\XHatBar
    +\frac{1}{2}\Trace X\XHatBar
    - 2\omega\XHatBar
    &= \frac{1}{2}\ComplexDeriv\SymTracelessTensorProd \HBar
      + \frac{1}{2}\HBar\SymTracelessTensorProd\HBar
      -\frac{1}{2}\overline{\Trace\XBar}\XHat
      + \frac{1}{2}\Xi\SymTracelessTensorProd\XiBar.
  \end{align*}
  \begin{align}
    \nabla_3 Z
    + \frac{1}{2}\Trace\XBar \left(Z+H\right)
    -2\omegaBar\left(Z-H\right)
    ={}& -2\ComplexDeriv\omegaBar
         - \frac{1}{2}\XBar\cdot \left(\ZBar+\HBar\right)
         + \frac{1}{2}\Trace X\XiBar
         + 2\omega\XiBar\\
       &- \BBar
         +\frac{1}{2}\overline{\XiBar}\cdot \XHat,
         \label{eq:null-structure:nabla-3-Z}
    \\
    \nabla_4 Z 
    + \frac{1}{2}\Trace X \left(Z-\HBar\right)
    - 2\omega\left(Z+\HBar\right)
    ={}&2\ComplexDeriv\omega
         + \frac{1}{2}\XHat\cdot\left(-\overline{Z} + \overline{\HBar}\right)
         - \frac{1}{2}\Trace\XBar \Xi
         - 2\omegaBar\Xi\\
       & - B
         - \frac{1}{2}\overline{\Xi}\cdot \XHatBar.
         \label{eq:null-structure:nabla-4-Z}
  \end{align}
  \begin{align*}
    \nabla_3\HBar
    -\nabla_4 \XiBar
    &= - \frac{1}{2}\overline{\Trace\XBar}\left(\HBar-H\right)
      -\frac{1}{2}\XHatBar\cdot \left(\overline{\HBar}- \overline{H}\right)
      -4\omega\XiBar
      +\BBar,\\
    \nabla_4 H
    -\nabla_3\Xi
    &= -\frac{1}{2}\overline{\Trace X}\left(H-\HBar\right)
      -\frac{1}{2}\XHat\cdot\left(\overline{H}-\overline{\HBar}\right)
      -4\omegaBar\Xi
      -B. 
  \end{align*}
  \begin{equation*}
    \nabla_4\omegaBar
    + \nabla_3\omega
    = \rho
    - \frac{\Lambda}{3}
    + 4\omega\omegaBar
    + \xi\cdot\xiBar
    + \zeta\cdot (\eta-\etaBar)
    -\eta\cdot\etaBar.
  \end{equation*}
  \begin{align}
    \frac{1}{2}\overline{\ComplexDeriv}\cdot\XHat
    + \frac{1}{2}\XHat\cdot\overline{Z}
    &= \frac{1}{2}\ComplexDeriv\overline{\Trace X}
      + \frac{1}{2}\overline{\Trace X}Z
      -\ImagUnit\Im\left(\Trace X\right)\left(H+\Xi\right)
      -B, \label{eq:null-structure:D-tr-X}
    \\
    \frac{1}{2}\overline{\ComplexDeriv}\cdot \XHatBar
    -\frac{1}{2}\XHatBar\cdot \ZBar
    &= \frac{1}{2}\ComplexDeriv \overline{\Trace\XBar}
      - \frac{1}{2}\overline{\Trace\XBar}\ZBar
      - \ImagUnit \Im\left(\Trace\XBar\right)\left(\HBar+\XiBar\right)
      +\BBar    
      . \label{eq:null-structure:D-tr-XBar}
  \end{align}
\end{prop}

\subsection{Frame transformation}
\label{sec:frame-transformation}

We recall the general null frame transformations for the horizontal
formalism. We refer the reader to
\cite{klainermanConstructionGCMSpheres2022} for a more in-depth
discussion and proofs of the results here.

\begin{lemma}[Lemma 3.1 of \cite{klainermanConstructionGCMSpheres2022}]
  \label{lemma:frame-transformation:general-form}
  Given two null frames $(e_1,e_2,e_3,e_4)$ and
  $(e_1',e_2',e_3',e_4')$, there exists a triplet of transition coefficients
  $(\lambda, f, \underline{f})\in \Real\times
  \mathbf{O}_1(\mathcal{M})\times \mathbf{O}_1(\mathcal{M})$ such that 
  \begin{equation*}
    \begin{split}
      e_4' ={}& \lambda \left( e_4 + f^be_b + \frac{1}{4}\abs*{f}^2e_3 \right),\\
      e_a' ={}& \left(\tensor[]{\delta}{_a^b} + \frac{1}{2}\underline{f}_af^b\right)e_b
                + \frac{1}{2}\underline{f}_ae_4
                + \left(\frac{1}{2}f_a + \frac{1}{8}\abs*{f}^2\underline{f}_a\right)e_3,\\
      e_3' ={}& \frac{1}{\lambda}\left(
                \left( 1 + \frac{1}{2}f\cdot \underline{f} + \frac{1}{16}\abs*{f}^2\abs*{\underline{f}}^2 \right)e_3
                + \left( \underline{f}^b + \frac{1}{4}\abs*{\underline{f}}^2f^b  \right)e_b
                + \frac{1}{4}\abs*{\underline{f}}^2e_4
                \right).
    \end{split}
  \end{equation*}
\end{lemma}
\begin{remark}
  Observe that the following identities hold.
  \begin{equation*}
    \begin{split}
      e_a' ={}& e_a + \frac{1}{2\lambda}\underline{f}_ae_4' + \frac{1}{2}f_ae_3,\\
      e_3' ={}& \frac{1}{\lambda}\left(e_3 + \underline{f}^ae_a' - \frac{1}{4\lambda}\abs*{\underline{f}}^2e_4' \right)
    \end{split}
  \end{equation*}
\end{remark}

We then have the following transformation formulas for the null Ricci
and curvature coefficients.
\begin{prop}[Proposition 3.3 of \cite{klainermanConstructionGCMSpheres2022}]
  \label{prop:frame-transformation}
  Let $(e_1,e_2,e_3,e_4)$ and $(e_1',e_2',e_3',e_4')$ be two null
  frames, and $(\lambda, f, \underline{f})$ the transition from
  $(e_1,e_2,e_3,e_4)$ to $(e_1',e_2',e_3',e_4')$ coefficients.
  \begin{itemize}
  \item The transformation formula for $\xi$ is given by:
    \begin{equation}
      \label{eq:frame-transformation:xi}
      \begin{split}
        \lambda^{-2}\xi'
        ={}& \xi
             + \frac{1\lambda}\nabla_4'f
             + \frac{1}{4}\left(\Trace \chi f  - \aTrace{\chi}\LeftDual{f}\right)
             + \omega f
             + \Err(\xi, \xi'),\\
        \Err(\xi, \xi') ={}& \frac{1}{2}f\cdot\widehat{\chi}
                             + \frac{1}{4}\abs*{f}^2\eta
                             + \frac{1}{2}\left(f\cdot\zeta\right)f
                             - \frac{1}{4}\abs*{f}^2\underline{\eta}\\
           & + \frac{1}{2\lambda^2}\left((f\cdot\xi')\underline{f} + \frac{1}{2}(f\cdot \underline{f})\xi'\right)
             + \LOT.
      \end{split}      
    \end{equation}
  \item The transformation  formula for $\underline{\xi}$ is given by
    \begin{equation}
      \label{eq:frame-transformation:xib}
      \begin{split}
        \lambda^2\underline{\xi}'
      ={}& \underline{\xi}
      + \frac{1}{2}\lambda\nabla_3'\underline{f}
      + \underline{\omega}\underline{f}
      + \frac{1}{4}\Trace \underline{\chi}\underline{f}
      - \frac{1}{4}\aTrace{\underline{\chi}}\LeftDual{\underline{f}}
      + \Err(\underline{\xi}, \underline{\xi}'),\\
      \Err(\underline{\xi}, \underline{\xi}')
      ={}& \frac{1}{2}\underline{f}\cdot\widehat{\underline{\chi}}
      - \frac{1}{2}\left(\underline{f}\cdot \zeta\right)\underline{f}
      + \frac{1}{4}\abs*{\underline{f}}^2\underline{\eta}
      - \frac{1}{4}\abs*{\underline{f}}^2\eta'
      + \LOT. 
      \end{split}      
    \end{equation}
  \item The transformation formulas for $\chi$ are given by
    \begin{equation}
      \label{eq:frame-transformation:tr-chi}
      \begin{split}
        \lambda^{-1}\Trace \chi'
        ={}& \Trace \chi
             + \nabla'\cdot f
             + f\cdot \eta
             + f\cdot \zeta
             + \Err(\Trace \chi, \Trace \chi'),\\
        \Err(\Trace \chi, \Trace \chi')
        ={}& \underline{f}\cdot\xi
             + \frac{1}{4}\underline{f}\cdot\left(f\Trace \chi - \LeftDual{f}\aTrace{\chi}\right)
             + \omega(f\cdot \underline{f})
             - \underline{\omega}\abs*{f}^2\\
           & - \frac{1}{4}\abs*{f}^2\Trace \underline{\chi}
             - \frac{1}{4\lambda}\left(f\cdot \underline{f}\right)\Trace \chi'
             + \frac{1}{4\lambda}\left(\underline{f}\wedge f\right)\aTrace{\chi'}
             + \LOT. 
      \end{split}      
    \end{equation}
    \begin{equation}
      \label{eq:frame-transformation:atr-chi}
      \begin{split}
        \lambda^{-1}\aTrace{ \chi'}
        ={}& \aTrace{ \chi}
             + \curl' f
             + f\wedge \eta
             + f\wedge \zeta
             + \Err(\aTrace{\chi}, \aTrace{ \chi'}),\\
        \Err(\aTrace{\chi}, \aTrace{ \chi'})
        ={}& \underline{f}\wedge\xi
             + \frac{1}{4}\left(\underline{f}\wedge f\Trace \chi - f\cdot\LeftDual{f}\aTrace{\chi}\right)
             + \omega(f\wedge \underline{f})\\
           & - \frac{1}{4}\abs*{f}^2\aTrace{\underline{\chi}}
             - \frac{1}{4\lambda}\left(f\cdot \underline{f}\right)\aTrace{\chi'}
             + \frac{1}{4\lambda}\left(\underline{f}\wedge f\right)\Trace{\chi'}
             + \LOT. 
      \end{split}      
    \end{equation}
    \begin{equation}
      \label{eq:frame-transformation:TF-chi}
      \begin{split}
        \lambda^{-1}\widehat{ \chi'}
        ={}& \widehat{ \chi}
             + \nabla'\otimes f
             + f\otimes \eta
             + f\otimes \zeta
             + \Err(\widehat{\chi}, \widehat{ \chi'}),\\
        \Err(\widehat{\chi}, \widehat{ \chi'})
        ={}& \underline{f}\otimes\xi
             + \frac{1}{4}\underline{f}\otimes\left( f\Trace \chi - \LeftDual{f}\aTrace{\chi}\right)
             + \omega(f\otimes \underline{f})
             - \underline{\omega}f\otimes f
        \\
           & - \frac{1}{4}\abs*{f}^2\aTrace{\underline{\chi}}
             + \frac{1}{4\lambda}f\otimes \underline{f}\Trace \chi'
             + \frac{1}{4\lambda}\left(\LeftDual{f}\otimes \underline{f}\right)\aTrace{\chi'}
        \\
           & + \frac{1}{2\lambda}\underline{f}\otimes\left(f\cdot \widehat{\chi}'\right)
             + \LOT. 
      \end{split}      
    \end{equation}
  \item The transformation formulas for $\underline{\chi}$ are given by
    \begin{equation}
      \label{eq:frame-transformation:tr-chib}
      \begin{split}
        \lambda\Trace \underline{\chi}'
        ={}& \Trace \underline{\chi}
             + \nabla'\cdot \underline{f}
             + \underline{f}\cdot\underline{\eta}
             - \underline{f}\cdot \zeta
             + \Err(\Trace \underline{\chi}, \Trace \underline{\chi}'),\\
        \Err(\Trace \chi, \Trace \chi')
        ={}& \frac{1}{2}\left(f\cdot \underline{f}\right)\Trace \underline{\chi}
             + f\cdot\underline{\xi}
             - \abs*{\underline{f}}^2\omega
             + (f\cdot \underline{f})\underline{\omega}
             - \frac{1}{4\lambda}\abs*{\underline{f}}^2\Trace\chi'
             + \LOT. 
      \end{split}      
    \end{equation}
    \begin{equation}
      \label{eq:frame-transformation:atr-chib}
      \begin{split}
        \lambda\aTrace{ \underline{\chi}'}
        ={}& \aTrace{ \underline{\chi}}
             + \curl' \underline{f}
             + \underline{f}\wedge \underline{\eta}
             - \zeta\wedge \underline{f}
             + \Err(\aTrace{\underline{\chi}}, \aTrace{ \underline{\chi}'}),\\
        \Err(\aTrace{\chi}, \aTrace{ \chi'})
        ={}& \underline{f}\wedge\xi
             + \frac{1}{4}\left(\underline{f}\wedge f\Trace \chi - f\cdot\LeftDual{f}\aTrace{\chi}\right)
             + \omega(f\wedge \underline{f})\\
           & - \frac{1}{4}\abs*{f}^2\aTrace{\underline{\chi}}
             - \frac{1}{4\lambda}\left(f\cdot \underline{f}\right)\aTrace{\chi'}
             + \frac{1}{4\lambda}\left(\underline{f}\wedge f\right)\Trace{\chi'}
             + \LOT. 
      \end{split}      
    \end{equation}
    \begin{equation}
      \label{eq:frame-transformation:TF-chib}
      \begin{split}
        \lambda\widehat{ \underline{\chi}'}
        ={}& \widehat{ \underline{\chi}}
             + \nabla'\otimes \underline{f}
             + \underline{f}\otimes \underline{\eta}
             - \underline{f}\otimes \underline{\zeta}
             + \Err(\widehat{\underline{\chi}}, \widehat{ \underline{\chi}'}),\\
        \Err(\widehat{\underline{\chi}}, \widehat{ \underline{\chi}'})
        ={}& \frac{1}{2}(f\otimes \underline{f})\Trace \underline{\chi}
             + f\otimes \underline{\xi}
             - (\underline{f}\otimes \underline{f})\omega
             + (f\otimes \underline{f})\underline{\omega}\\
             -& - \frac{1}{4\lambda}\abs*{\underline{f}}^2\widehat{\chi}'
             + \LOT. 
      \end{split}      
    \end{equation}
  \item The transformation formula for $\zeta$ is given by
    \begin{equation}
      \label{eq:frame-transformation:zeta}
      \begin{split}
        \zeta' ={}& \zeta - \nabla'\left(\log \lambda\right)
                    - \frac{1}{4}\Trace \underline{\chi} f
                    + \frac{1}{4}\aTrace{\underline{\chi}}\LeftDual{f}
                    + \omega \underline{f}
                    - \underline{\omega} f\\
                  &  + \frac{1}{4}\underline{f}\Trace \chi
                    + \frac{1}{4}\LeftDual{\underline{f}}\aTrace{\chi}
                    + \Err(\zeta , \zeta'),\\
        \Err(\zeta, \zeta')
        ={}& - \frac{1}{2}\widehat{\underline{\chi}}\cdot f
             + \frac{1}{2}(f\cdot\zeta)\underline{f}
             - \frac{1}{2}(f\cdot\underline{\eta})\underline{f}
             + \frac{1}{4}\underline{f}(f\cdot\eta)
             + \frac{1}{4}\underline{f}(f\cdot\zeta)\\
                  & + \frac{1}{4}\LeftDual{f}(f\wedge \eta)
                    + \frac{1}{4}\LeftDual{f}(f\wedge \zeta)
                    + \frac{1}{4}\underline{f}\nabla'\cdot f\\
                  & + \frac{1}{4}\LeftDual{\underline{f}}\curl'f
                    + \frac{1}{2\lambda} \underline{f}\cdot\widehat{\chi}'\\
                  & - \frac{1}{16\lambda}(f\cdot \underline{f})\underline{f}\Trace\chi'
                    + \frac{1}{16\lambda}(\underline{f}\wedge f)\underline{f}\aTrace{\chi'}\\
                  & - \frac{1}{16\lambda}\LeftDual{\underline{f}}(f\cdot \underline{f})\aTrace{\chi'}
                    + \frac{1}{16\lambda}\LeftDual{\underline{f}}(f\wedge \underline{f})\Trace{\chi'}
                    + \LOT.
      \end{split}
    \end{equation}
  \item The frame transformation for $\eta$ is given by
    \begin{equation}
      \label{eq:frame-transformation:eta}
      \begin{split}
        \eta'
        ={}& \eta
        + \frac{\lambda}{2}\nabla_3'f
        + \frac{1}{4}\underline{f}\Trace \chi
        - \frac{1}{4}\LeftDual{f}\aTrace{\chi}
        - \underline{\omega} f
             + \Err(\eta,\eta'),\\
        \Err(\eta, \eta')
        ={}& \frac{1}{2}(f\cdot\underline{f})\eta
             + \frac{1}{2}\underline{f}\cdot\widehat{\chi}
             + \frac{1}{2}f(\underline{f}\cdot\zeta)
             - (\underline{f}\cdot f)\eta'
             + \frac{1}{2}\underline{f}(f\cdot\eta')
             + \LOT. 
      \end{split}
    \end{equation}
  \item The transformation formula for $\underline{\eta}$ is given by
    \begin{equation}
      \label{eq:frame-transformation:etab}
      \begin{split}
        \underline{\eta}'
        ={}& \underline{\eta}
             + \frac{1}{2\lambda}\nabla_4'\underline{f}
             + \frac{1}{4}\Trace \underline{\chi}f
             - \frac{1}{4}\aTrace{\underline{\chi}}\LeftDual{f}
             - \omega \underline{f}
             + \Err(\underline{\eta}, \underline{\eta}'),\\
        \Err(\underline{\eta}, \underline{\eta}')
        ={}& \frac{1}{2}f\cdot\widehat{\underline{\chi}}
             + \frac{1}{2}(f\cdot\underline{\eta})\underline{f}
             - \frac{1}{4}(f\cdot\zeta)\underline{f}
             - \frac{1}{4\lambda^2}\abs*{\underline{f}}^2\xi' + \LOT.
      \end{split}      
    \end{equation}
  \item The transformation formula for $\omega$ is given by
    \begin{equation}
      \label{eq:frame-transformation:omega}
      \begin{split}
        \frac{1}{\lambda}\omega'
        ={}& \omega
             - \frac{1}{2\lambda}e_4'(\log\lambda)
             + \frac{1}{2}f\cdot(\zeta-\underline{\eta})
             + \Err(\omega, \omega')\\
        \Err(\omega, \omega')
        ={}& - \frac{1}{4}\abs*{f}^2\underline{\omega}
             - \frac{1}{8}\Trace \underline{\chi}\abs*{f}^2
             + \frac{1}{2\lambda^2}\underline{f}\cdot\xi'
             + \LOT.  
      \end{split}
    \end{equation}
  \item The transformation formula for $\underline{\omega}$ is given by
    \begin{equation}
      \label{eq:frame-transformation:omegab}
      \begin{split}
        \lambda\underline{\omega}'
        ={}& \underline{\omega}
             + \frac{\lambda}{2}e_3'(\log \lambda)
             - \frac{1}{2}\underline{f}\cdot(\zeta + \eta)
             + \Err(\underline{\omega}, \underline{\omega}'),\\
        \Err(\underline{\omega}, \underline{\omega}')
        ={}& f\cdot\underline{f}\underline{\omega}
             - \frac{1}{4}\abs*{\underline{f}}^2\omega
             + \frac{1}{2}f\cdot \underline{\xi}
             + \frac{1}{8}(f\cdot \underline{f})\Trace \underline{\chi}
             + \frac{1}{8}(\underline{f}\wedge f) \aTrace{\underline{\chi}}\\
           & - \frac{1}{8}\abs*{\underline{f}}^2\Trace\chi
             - \frac{1}{4}\lambda \underline{f}\cdot\nabla_3'f
             + \frac{1}{2}(\underline{f}\cdot f)(\underline{f}\cdot\eta')
             - \frac{1}{4}\abs*{\underline{f}}^2(f\cdot\eta')
             + \LOT,
      \end{split}
    \end{equation}
    where in the formulas above, $\LOT$ denotes expressions of the type
    \begin{equation*}
      \LOT = O\left( (f, \underline{f})^3 \right)\Gamma,
    \end{equation*}
    where $\Gamma$ is the set of Ricci coefficients, and where we
    emphasize that $\LOT$ contains no derivatives of $f, \underline{f}, \Gamma$.
  \end{itemize}
  We have the following transformation formulas for the null
  components of the Weyl tensor $\mathbf{W}$
  \begin{equation*}
    \alpha = \alpha(\mathbf{W}), \qquad
    \underline{\alpha}=\underline{\alpha}(\mathbf{W}),\qquad
    \beta = \beta(\mathbf{W}), \qquad
    \underline{\beta}=\underline{\beta}(\mathbf{W}),\qquad
    \rho = \rho(\mathbf{W}), 
  \end{equation*}
  \begin{itemize}
  \item The transformation formula for $\alpha, \underline{\alpha}$ are as follows
    \begin{equation}
      \label{eq:frame-transformation:alpha}
      \begin{split}
        \frac{1}{\lambda^2}\alpha'
        ={}& \alpha + \Err(\alpha, \alpha'),\\
        \Err(\alpha, \alpha')
        ={}& \left(f\otimes \beta - \LeftDual{f}\otimes \LeftDual{\beta}\right)
             + \left(f\otimes f
             - \frac{1}{2}\LeftDual{f}\otimes \LeftDual{f}\right)\rho\\
           & + \frac{3}{2}\left(f\otimes \LeftDual{f}\right)\LeftDual{\rho}
             + \LOT,
      \end{split}
    \end{equation}
    \begin{equation}
      \label{eq:frame-transformation:alphab}
      \begin{split}
        \lambda^2\underline{\alpha}'
        ={}& \underline{\alpha} + \Err(\underline{\alpha},\underline{\alpha}')\\
        \Err(\underline{\alpha},\underline{\alpha}')
        ={}& - \left(\underline{f}\otimes\underline{\beta} - \LeftDual{\underline{f}}\otimes \LeftDual{\underline{\beta}}\right)
             + \left(\underline{f}\otimes \underline{f} - \frac{1}{2}\LeftDual{\underline{f}}\otimes \LeftDual{\underline{f}}\right)\rho\\
           & + \frac{3}{2}\left(\underline{f}\otimes\LeftDual{\underline{f}}\right)\LeftDual{\rho}
             + \LOT.
      \end{split}
    \end{equation}
  \item The transformation formula for $\beta, \underline{\beta}$ are given by
    \begin{equation}
      \label{eq:frame-transformation:beta}
      \begin{split}
        \frac{1}{\lambda}\beta'
      ={}& \beta + \frac{3}{2}\left(f\rho + \LeftDual{f}\LeftDual{\rho}\right)
      + \Err(\beta, \beta')\\
      \Err(\beta, \beta')
      ={}& \frac{1}{2}\alpha\cdot \underline{f} + \LOT,
      \end{split}      
    \end{equation}
    \begin{equation}
      \label{eq:frame-transformation:betab}
      \begin{split}
        \lambda\underline{\beta}'
      ={}& \underline{\beta} - \frac{3}{2}\left(\underline{f}\rho + \LeftDual{\underline{f}}\LeftDual{\rho}\right)
      + \Err(\beta, \beta')\\
      \Err(\beta, \beta')
      ={}& -\frac{1}{2}\underline{\alpha}\cdot f + \LOT.
      \end{split}      
     \end{equation}
   \item The transformation formula for $\rho$ and $\LeftDual{\rho}$ are given by
     \begin{equation}
       \label{eq:frame-transformation:rho}
       \begin{split}
         \rho'
         ={}& \rho + \Err(\rho, \rho')\\
         \Err(\rho, \rho')
         ={}& \underline{f}\cdot\beta
              - f\cdot\underline{\beta}
              + \frac{3}{2}\rho(f\cdot\underline{f})
              - \frac{3}{2}\LeftDual{\rho}(f\wedge \underline{f})
              + \LOT
              ,
       \end{split}
     \end{equation}
     \begin{equation}
       \label{eq:frame-transformation:rho-dual}
       \begin{split}
         \LeftDual{\rho}'
         ={}& \LeftDual{\rho} + \Err(\LeftDual{\rho}, \LeftDual{\rho}')\\
         \Err(\LeftDual{\rho}, \LeftDual{\rho}')
         ={}& -\underline{f}\cdot\LeftDual{\beta}
              - f\cdot\LeftDual{\underline{\beta}}
              + \frac{3}{2}\LeftDual{\rho}(f\cdot\underline{f})
              + \frac{3}{2}\rho(f\wedge \underline{f})
              + \LOT,
       \end{split}
     \end{equation}
   \end{itemize}
   where for the transformation formulas of the curvature components
   $\alpha, \underline{\alpha}, \beta, \underline{\beta}$ above,
   $\LOT$ is used to denote expressions of the type
   \begin{equation*}
     \LOT = O\left((f, \underline{f})^3\right)(\rho, \LeftDual{\rho})
     + O\left((f, \underline{f})^2\right)(\alpha, \beta, \underline{\alpha}, \underline{\beta})
   \end{equation*}
   containing no derivatives of
   $f, \underline{f}, \alpha, \beta, \rho, \LeftDual{\rho},
   \underline{\beta}$, and $\underline{\alpha}$, and for the
   transformation formulas of the curvature components
   $\rho, \LeftDual{\rho}$, $\LOT$ is used to denote expressions of the type
   \begin{equation*}
     \LOT = O(f(f, \underline{f})^2)(\rho, \LeftDual{\rho})
     + O(f(f, \underline{f}))(\alpha, \beta, \underline{\alpha}, \underline{\beta}).
   \end{equation*}
 \end{prop}

 \begin{proof}
   We only provide a proof for the transformation formula for
   $\rho(\mathbf{W})$ in \Cref{eq:frame-transformation:rho}. The
   transformation formula for $\LeftDual{\rho}$ in
   \Cref{eq:frame-transformation:rho-dual} follows from that for
   $\rho$ in \Cref{eq:frame-transformation:rho} by symmetry
   considerations. For proofs of the remaining transformation
   formulas, we refer the readers to Appendix A of
   \cite{klainermanConstructionGCMSpheres2022}, where the
   transformation formulas for the Ricci coefficients can also be
   found.

   Observe that
   \begin{align*}
     \rho'(\mathbf{W})
     ={}& \mathbf{W}(e_4',e_3', e_4', e_3')\\
     ={}& \mathbf{W}\left(
          e_4', \frac{1}{\lambda}\left(e_3 + \underline{f}^ae_a' - \frac{1}{4\lambda}\abs*{\underline{f}}^2e_4'\right),
          e_4', e_3'
          \right)\\
     ={}& \frac{1}{\lambda}\mathbf{W}(e_4', e_3, e_4', e_3')
          + \frac{1}{\lambda}\underline{f}^a \mathbf{W}(e_4', e_a', e_4', e_3')\\
     ={}& \frac{1}{\lambda}\mathbf{W}\left(
     e_4', e_3, e_4', \frac{1}{\lambda}\left(e_3 + \underline{f}^ae_a' - \frac{1}{4\lambda}\abs*{\underline{f}}^2e_4'\right)
     \right)
     + \underline{f}^a\mathbf{W}\left(\frac{1}{\lambda}e_4', e_a', \frac{1}{\lambda}e_4', \lambda e_3'\right).
   \end{align*}
   As a result, we have that
   \begin{align*}
     4\rho'
     ={}& \mathbf{W}\left(\frac{1}{\lambda}e_4', e_3, \frac{1}{\lambda}e_4', e_3\right)
          + \underline{f}^a \mathbf{W}\left(
          \frac{1}{\lambda}e_4', e_3, \frac{1}{\lambda}e_4', e_a + \frac{1}{2\lambda}\underline{f}_ae_4' + \frac{1}{2}f_ae_3
          \right)\\
        & + \underline{f}^a\mathbf{W}\left(
          \frac{1}{\lambda}e_4', e_a + \frac{1}{2\lambda}\underline{f}_ae_4' + \frac{1}{2}f_ae_3,
          \frac{1}{\lambda}e_4',
          \lambda e_3'
          \right)\\
     ={}& \mathbf{W}\left(\frac{1}{\lambda}e_4', e_3, \frac{1}{\lambda}e_4', e_3\right)
          + \underline{f}^a \mathbf{W}\left(
          \frac{1}{\lambda}e_4',
          e_3,
          \frac{1}{\lambda}e_4',
          e_a  + \frac{1}{2}f_ae_3
          \right)\\
        & + \underline{f}^a\mathbf{W}\left(
          \frac{1}{\lambda}e_4',
          e_a  ,
          \frac{1}{\lambda}e_4',
          \lambda e_3'
          \right)
          + \frac{1}{2}\underline{f}\cdot f\mathbf{W}\left(
          \frac{1}{\lambda}e_4',
          e_3,
          \frac{1}{\lambda}e_4',
          \lambda e_3'
          \right)
     \\
     ={}& \mathbf{W}\left(\frac{1}{\lambda}e_4', e_3, \frac{1}{\lambda}e_4', e_3\right)
          + \underline{f}^a \mathbf{W}\left(
          \frac{1}{\lambda}e_4',
          e_3,
          \frac{1}{\lambda}e_4',
          e_a  + \frac{1}{2}f_ae_3
          \right)\\
        & + \underline{f}^a\mathbf{W}\left(
          \frac{1}{\lambda}e_4',
          e_a  ,
          \frac{1}{\lambda}e_4',
          \lambda e_3'
          \right)
          + \frac{1}{2}\underline{f}\cdot f\mathbf{W}\left(
          \frac{1}{\lambda}e_4',
          e_3,
          \frac{1}{\lambda}e_4',
          \lambda e_3'
          \right). 
   \end{align*}
   We can further compute that
   \begin{align*}
     \mathbf{W}\left(\frac{1}{\lambda}e_4', e_3, \frac{1}{\lambda}e_4', e_3\right)
     ={}&\mathbf{W}\left(
          e_4 + f^ae_a,
          e_3,
          e_4 + f^be_b,
          e_3
          \right)\\
     ={}& 4\rho
          + f^a\mathbf{W}(e_a, e_3, e_4, e_3)
          + f^b \mathbf{W}(e_4, e_3, e_b,e_3)
          + O(f^2)(\rho, \LeftDual{\rho})\\
     ={}& 4\rho - 4f\cdot\underline{\beta} + \LOT.
   \end{align*}
   We also have that
   \begin{align*}
     \underline{f}^a \mathbf{W}\left(\frac{1}{\lambda}e_4', e_3, \frac{1}{\lambda}e_4', e_a\right)
     ={}& \underline{f}^a\mathbf{W}\left(
          e_4 + f^ae_a,
          e_3,
          e_4 + f^be_b + \frac{1}{4}\abs*{f}^2e_3,
          e_a
          \right)\\
     ={}&\underline{f}^a\mathbf{W}\left(
          e_4 + f^ae_a,
          e_3,
          e_4 + f^be_b,
          e_a
          \right)
          + O(f^2\underline{f})\underline{\alpha}\\
     ={}& \underline{f}^a\mathbf{W}\left(e_4,e_3, e_4, e_a\right)
          + \underline{f}^af^b \mathbf{W}\left(
          e_b, e_3, e_4, e_a
          \right)
          + \underline{f}^af^c\mathbf{W}(e_4,e_3,e_c,e_a)
          + \LOT\\
     ={}& 2\underline{f}\cdot\beta
          - \underline{f}^af^b\left(-\rho\delta_{ba} + \LeftDual{\rho}\in_{ba}\right)
          + 2\underline{f}^af^c\in_{ac}\LeftDual{\rho}
          + \LOT\\
     ={}& 2\underline{f}\cdot\beta
          + \rho(f\cdot \underline{f})
          - 3\LeftDual{\rho}(f\wedge \underline{f})
          + \LOT.
   \end{align*}
   We also observe that
   \begin{align*}
     \underline{f}^a \mathbf{W}\left(\frac{1}{\lambda}e_4', e_3, \frac{1}{\lambda}e_4', \frac{1}{2}f_ae_3\right)
     ={}&  \frac{1}{2}f\cdot \underline{f} \mathbf{W}\left(
          e_4 + f^ae_a,
          e_3,
          e_4 + f^be_b,
          e_3
          \right)\\
     ={}& 2 f\cdot \underline{f}\rho
          + O(f(f, \underline{f})^2)(\alpha, \underline{\alpha}, \beta, \underline{\beta},\rho, \LeftDual{\rho}).
   \end{align*}
   Similarly, we can use the frame transformation formulas in
   \Cref{lemma:frame-transformation:general-form} to see that
   \begin{align*}
     &\frac{1}{2}\underline{f}\cdot f\mathbf{W}\left(
     \frac{1}{\lambda}e_4',
     e_3,
     \frac{1}{\lambda}e_4',
     \lambda e_3'
     \right)\\
     ={}& \frac{1}{2}\underline{f}\cdot f\mathbf{W}\left(
          \frac{1}{\lambda}e_4',
          e_3,
          \frac{1}{\lambda}e_4',
          e_3 + \underline{f}^ae_a' - \frac{1}{4\lambda}\abs{\underline{f}}^2e_4'
          \right)\\
     ={}& \frac{1}{2}\underline{f}\cdot f\mathbf{W}\left(
          \frac{1}{\lambda}e_4',
          e_3,
          \frac{1}{\lambda}e_4',
          e_3 + \underline{f}^ae_a'
          \right)\\
     ={}& \frac{1}{2}\underline{f}\cdot f\mathbf{W}\left(
          e_4 + f^be_b,
          e_3,
          e_4 + f^ce_c + \frac{1}{4}\abs*{f}^2e_3,
          e_3 + \underline{f}^ae_a'
          \right)\\
     ={}& \frac{1}{2}\underline{f}\cdot f\mathbf{W}\left(
          e_4 + f^be_b,
          e_3,
          e_4 + f^ce_c,
          e_3 + \underline{f}^ae_a'
          \right)
          + \frac{1}{2}\underline{f}\cdot f\mathbf{W}\left(
          e_4 + f^be_b,
          e_3,
          \frac{1}{4}\abs*{f}^2e_3,
          \underline{f}^ae_a'
          \right)\\
     ={}& 2\underline{f}\cdot f\rho
          + O(f(f, \underline{f})^2)(\alpha, \underline{\alpha},\beta,\underline{\beta},\rho, \LeftDual{\rho}).
   \end{align*}
      
   As a result, we have
   that
   \begin{align*}
     4\rho'
     ={}& 4\rho
          + 4\underline{f}\cdot\beta
          - 4f\cdot \underline{\beta}
          + 6\rho(f\cdot \underline{f})
          - 6\LeftDual{\rho}(f\wedge \underline{f})
          + \LOT,
   \end{align*}
   which yields the formula in \Cref{eq:frame-transformation:rho}.
 \end{proof}

 \begin{corollary}
   \label{coro:eqns-invariant-rescaling-transformation}
   Given two null frames $(e_1,e_2,e_3,e_4)$ and
   $(e_1',e_2',e_3',e_4')$, with transition coefficients
   $(\lambda, 0, 0)$ from $(e_1,e_2,e_3,e_4)$ to
   $(e_1',e_2',e_3',e_4')$, the null Bianchi and null structure
   equations are invariant under the frame transformation with
   transition coefficients $(\lambda, 0, 0)$. 
 \end{corollary}
 \begin{proof}
   Simple computation using the transformations in \Cref{prop:frame-transformation}.
 \end{proof}

\printbibliography

\end{document}

